\documentclass[aps,prl,notitlepage,twocolumn,superscriptaddress,10pt]{revtex4-2}

\makeatletter
\def\@hangfrom@section#1#2#3{\normalsize\@hangfrom{#1#2}#3}
\def\@hangfroms@section#1#2{\normalsize#1#2}
\makeatother

\usepackage{amsmath,amsfonts,amssymb,bm}
\usepackage{graphicx}
\usepackage{enumerate}

\usepackage{theorem}

\newtheorem{proposition}{Proposition}
\newtheorem{lemma}[proposition]{Lemma}
\newtheorem{theorem}[proposition]{Theorem}

\newenvironment{proof}{\noindent \textbf{{Proof~}}}{\hfill $\blacksquare$}
\newenvironment{sketchproof}{\textit{{Proof.---}}}{\hfill $\blacksquare$}

\usepackage{tikz}
\newenvironment{mytikz}{\begin{tikzpicture}[x=0.3pt,y=0.3pt,yscale=-1,xscale=1,baseline={([yshift=+0ex]current bounding box.center)}]}{\end{tikzpicture}}
\newenvironment{mytikz2}{\begin{tikzpicture}[x=0.4pt,y=0.4pt,yscale=-1,xscale=1,baseline={([yshift=+0ex]current bounding box.center)}]}{\end{tikzpicture}}
\newenvironment{mytikz3}{\begin{tikzpicture}[x=0.5pt,y=0.5pt,yscale=-1,xscale=1,baseline={([yshift=+0ex]current bounding box.center)}]}{\end{tikzpicture}}
\newenvironment{mytikz4}{\begin{tikzpicture}[x=0.45pt,y=0.45pt,yscale=-1,xscale=1,baseline={([yshift=+0ex]current bounding box.center)}]}{\end{tikzpicture}}

\definecolor{darkblue1}{rgb}{0.18,0.19,0.57}
\usepackage[colorlinks=true, allcolors=darkblue1]{hyperref}

\newcommand{\nc}{\newcommand}
\nc{\ket}[1]{|#1\rangle}
\nc{\bra}[1]{\langle#1|}
\nc{\ketbra}[2]{|#1\rangle\!\langle#2|}
\nc{\braket}[2]{\langle#1|#2\rangle}
\nc{\braoprket}[3]{\langle#1|#2|#3\rangle}
\nc{\opr}[1]{\operatorname{#1}}
\nc{\avg}[1]{\langle#1\rangle}
\nc{\ketbrasame}[1]{|#1\rangle\!\langle#1|}
\nc{\E}{\mathbb{E}}
\nc{\var}{\operatorname{Var}}

\nc{\hk}[1]{\textcolor{violet}{\textbf{[hk: #1]}}}
\nc{\ch}[1]{\textcolor{purple}{\textbf{[ch: #1]}}}
\nc{\ls}[1]{\textcolor{blue}{\textbf{[ls: #1]}}}
\usepackage[normalem]{ulem}
\nc{\hknew}[1]{\textcolor{violet}{#1}}

\begin{document}
\title{Absence of poor local minima in matrix product states}

\author{Hao-Kai Zhang}
\affiliation{Institute of Physics, Chinese Academy of Sciences, Beijing 100190, China}

\author{Chenghong Zhu}
\affiliation{The Hong Kong University of Science and Technology (Guangzhou), China}

\author{Shuo Liu}
\affiliation{Department of Physics, Princeton University, Princeton, New Jersey 08544, USA}

\author{Shi-Xin Zhang}
\affiliation{Institute of Physics, Chinese Academy of Sciences, Beijing 100190, China}

\author{Tao Xiang}
\email{txiang@iphy.ac.cn}
\affiliation{Institute of Physics, Chinese Academy of Sciences, Beijing 100190, China}

\date{\today}

\begin{abstract}
Quantum circuits suffer from severe trainability issues: even shallow circuits are swamped with poor local minima. Yet matrix product states (MPS), which can be prepared by sequential circuits, are remarkably trainable in practice---as demonstrated by decades of successful density matrix renormalization group calculations. In this work, we resolve this apparent paradox by proving that the energy landscapes of MPS are free from poor local minima, under the same setting where brickwork circuits are not. The key insight is that the gauge freedom of MPS creates an effective local overparametrization that causes local minima to concentrate near the global minimum, analogous to overparametrized classical neural networks. We rigorously prove that the local minimum distribution is invariant under moves of the orthogonality center of MPS representations. Numerical experiments further confirm that the optimization of sequential circuits converges to near-optimal solutions even for random Hamiltonians, in stark contrast to brickwork circuits. Our findings establish a theoretical understanding of the trainability of MPS, providing a valuable guide for designing variational quantum circuits and algorithms with better trainability in the future.
\end{abstract}


\maketitle

\textit{Introduction.---}The success of classical neural networks on a wide variety of tasks is largely attributed to their trainability. Despite the nonconvex nature of their loss landscapes, gradient descent often converges to a near-optimal solution. This can be described by the absence of poor local minima---the loss values of most local minima concentrate near that of the global minimum, which has been proved in certain overparametrized networks~\cite{Choromanska2015, Kawaguchi2016, Jacot2020}.

As the quantum analog of neural networks, variational quantum algorithms (VQAs)~\cite{Bharti2022, Cerezo2021a}, which aim to harness near-term quantum devices by optimizing parametrized quantum circuits for tasks ranging from ground state preparation to machine learning, are hindered by severe trainability issues. Deep circuits, despite their powerful expressibility, suffer from the barren plateau phenomenon~\cite{McClean2018, Cerezo2021, Zhang2024, Larocca2025, Bittel2021, You2021, Anschuetz2022, Anschuetz2023, Anschuetz2025, Arrasmith2020, Arrasmith2021, Liu2021a, Holmes2021, Zhang2023, Barthel2023a, Miao2023, Liu2023, Liu2024, Zhang2024a, Zhang2024b, Cerezo2025}, where circuit gradients vanish exponentially with system size. For shallow circuits that are free from barren plateaus~\cite{Cerezo2021, Zhang2024}, a more subtle issue remains---poor local minima~\cite{Bittel2021, You2021, Anschuetz2022, Anschuetz2023, Anschuetz2025}, which refers to the local minima with high loss values. It is proved that extensive poor local minima exist in the loss landscapes of certain shallow circuits, such as quantum convolutional neural networks and brickwork circuits of moderate depth~\cite{Anschuetz2022}.

A sharply distinct empirical fact is the successful optimization of matrix product states (MPS)~\cite{White1992, Schollwock2011, Orus2014, Orus2019, Cirac2021, Xiang2023, Aharonov2010, Landau2015, Arad2017}. As the core numerical method of solving one-dimensional quantum many-body physics in the past three decades, the algorithms that use MPS as variational ansatzes, such as the density matrix renormalization group (DMRG)~\cite{White1992} and the time-dependent variational principle (TDVP) in imaginary time~\cite{Haegeman2011, Haegeman2014, Haegeman2016, Hauru2021}, are routinely observed to converge to ground states even from random initialization. This empirical fact is intriguing because MPS can be exactly regarded as the output of sequential circuits~\cite{Schon2005, Schon2007, Banuls2008, Wei2022, Chen2023a}. The expressibility of MPS has been explained by the entanglement area law~\cite{Hastings2007, Eisert2010} in the ground states of gapped local Hamiltonians, whereas understanding the trainability of MPS remains an open question, namely, why the optimization of MPS so reliably finds the ground states. In particular, what property of MPS and sequential circuits is responsible for their benign loss landscapes, in contrast to those of other circuits?

In this work, we identify this property as the effective local overparametrization arising from the gauge structure of MPS. Since the physical state is invariant under inserting an invertible matrix and its inverse between adjacent tensors of an MPS representation, the same MPS can be expressed in mixed canonical forms with different orthogonality centers. This gauge degree of freedom does not change the state but substantially modifies the causal structure of the corresponding sequential circuit. Thus, for an arbitrary local observable, one can always move the orthogonality center nearby, so that the backward causal cone is overparametrized effectively by the surrounding local tensors. We formalize this mechanism rigorously through two theorems. Firstly, we prove that random canonical MPS ensembles with different orthogonality centers are statistically identical, implying that the induced distribution over physical states is independent of the orthogonality center. Secondly, after properly defining a probability measure over local minima, we prove that, for any given Hamiltonian, the energy landscapes of MPS representations with different orthogonality centers have identical distributions of local minimum values. We further conduct numerical experiments using random Hamiltonians and confirm that the optimization of sequential circuits indeed converges reliably to near-optimal solutions, while that in brickwork-type circuits is trapped by poor local minima.



\begin{figure}
    \centering
    \includegraphics[width=0.99\linewidth]{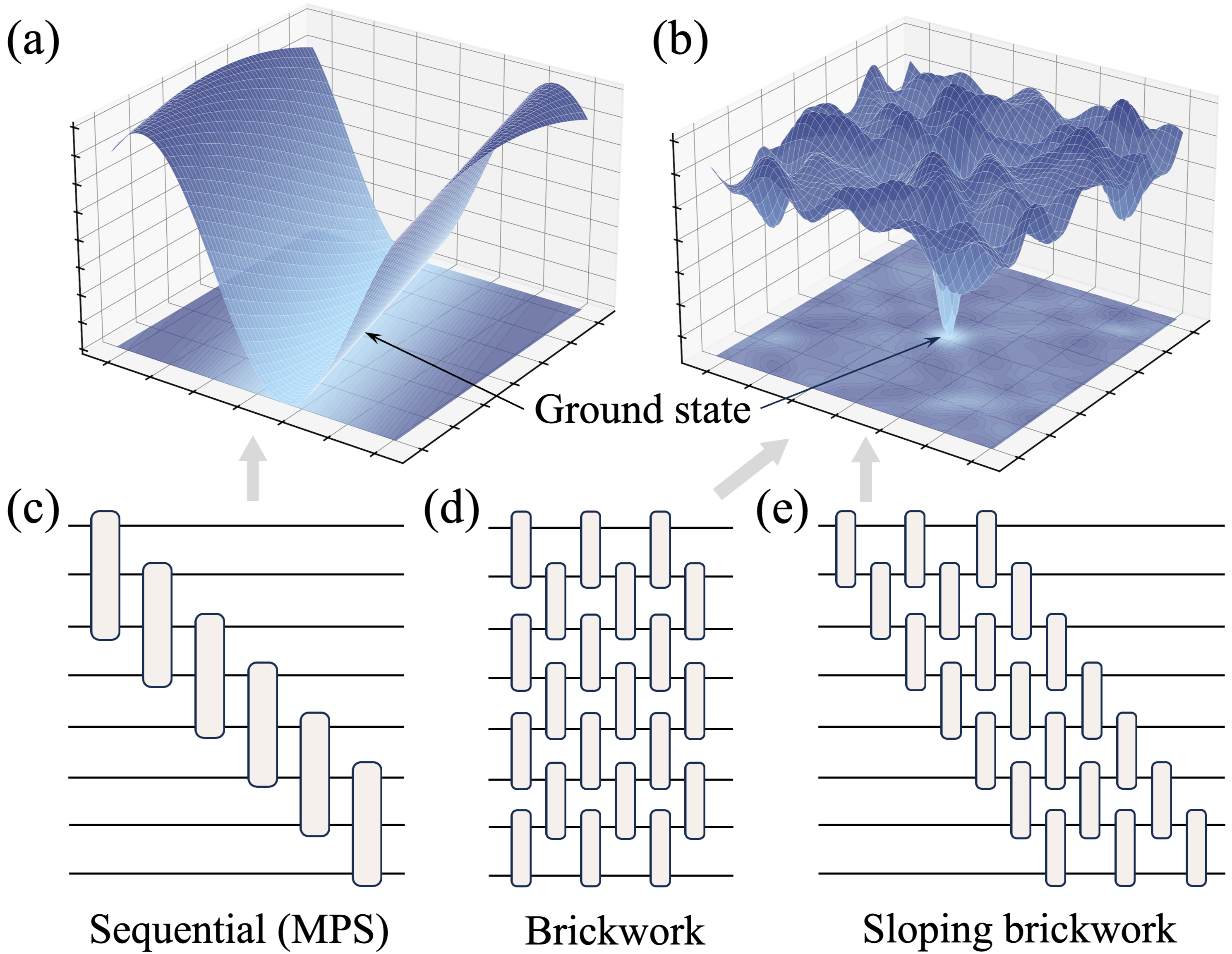}
    \caption{Energy landscape schematics in the absence (a) or presence (b) of poor local minima. (c) Sequential circuit, serving as the preparation circuit for MPS. (d) Brickwork circuit. (e) Sloping brickwork circuit, which reduces to sequential circuit when restricted to a single layer. Each block in the circuits represents a unitary that is universal on its support. The grey arrows indicate that the typical energy landscapes of sequential circuits are free from poor local minima, whereas those of brickwork-type circuits are not.}
    \label{fig:landscape_circuit}
\end{figure}


\textit{Basic setup.---}To find the ground state of a given quantum many-body system described by the Hamiltonian $H$, the variational method chooses a parametrized quantum state $\ket{\Psi(\bm{\theta})}$ as a variational ansatz and minimizes the energy function $\mathcal{E}(\bm{\theta}) = \braoprket{\Psi(\bm{\theta})}{H}{\Psi(\bm{\theta})}$. This method is expected to succeed if the ansatz is expressive enough to contain the ground state and the optimization process is smooth. In variational quantum algorithms, the ansatz is chosen as $\ket{\Psi(\bm{\theta})}=\mathbf{U}(\bm{\theta})\ket{0}^{\otimes N}$, where $\mathbf{U}(\bm{\theta})$ is a parametrized quantum circuit composed of unitary gates that can be run on quantum devices.

Matrix product states are a special class of states that often serve as ansatzes for solving one-dimensional quantum systems, whose wave function under the computational basis is given by a product of tensors $\ket{\Psi_\mathrm{MPS}(A)} = \sum_{n_1 \ldots n_N} A_{1,n_1} A_{2,n_2} \ldots A_{N,n_N} \ket{n_1 \ldots n_N}$, where $A_{i}$ is a $3$-order tensor of shape $D_{i-1}\times d \times D_i$. $d$ is the local Hilbert space dimension of the physical index $n_i$ and $D_i$ is the bond dimension of the virtual index at bond $i$. We focus on MPS with a uniform bond dimension $D$ for simplicity, which is frequently used in practice.


Different MPS representations may represent the same physical state due to the \textit{gauge} freedom of MPS, because inserting an arbitrary invertible matrix and its inverse between adjacent tensors does not change the wave function. Using this redundancy, the local tensors can be restricted to isometries without loss of generality. Specifically, a local tensor is called left (right) isometric if combining the physical index with the left (right) virtual index leads to an isometric matrix whose row or column vectors are orthonormal. An MPS representation $A$ is in its canonical form $A^{[i]}$ with orthogonality center at site $i$ if $\{A_1,\ldots,A_{i-1}\}$ are left isometric and $\{A_{i+1},\ldots,A_N\}$ are right isometric, which can be represented by the following tensor network diagram
\begin{equation}
\begin{mytikz4}
\draw [line width=0.5]    (235,100) -- (210,100) ;
\draw [line width=0.5]    (235,60) -- (235,100) ;
\draw [line width=0.5]    (260,100) -- (235,100) ;
\draw  [fill={rgb, 255:red, 245; green, 240; blue, 235 }  ,fill opacity=1 ][line width=0.5]  (217.32,103.54) -- (238.54,82.32) -- (252.68,96.46) -- (231.46,117.68) -- cycle ;
\draw [line width=0.5]    (285,100) -- (260,100) ;
\draw [line width=0.5]    (285,60) -- (285,100) ;
\draw [line width=0.5]    (310,100) -- (285,100) ;
\draw  [fill={rgb, 255:red, 245; green, 240; blue, 235 }  ,fill opacity=1 ][line width=0.5]  (270,85) -- (300,85) -- (299.31,115) -- (269.31,115) -- cycle ;
\draw [line width=0.5]    (335,100) -- (310,100) ;
\draw [line width=0.5]    (335,60) -- (335,100) ;
\draw [line width=0.5]    (360,100) -- (335,100) ;
\draw  [fill={rgb, 255:red, 245; green, 240; blue, 235 }  ,fill opacity=1 ][line width=0.5]  (331.46,82.32) -- (352.68,103.54) -- (338.54,117.68) -- (317.32,96.46) -- cycle ;
\draw [line width=0.5]    (425,100) -- (400,100) ;
\draw [line width=0.5]    (425,60) -- (425,100) ;
\draw [line width=0.5]    (450,100) -- (425,100) ;
\draw  [fill={rgb, 255:red, 245; green, 240; blue, 235 }  ,fill opacity=1 ][line width=0.5]  (421.46,82.32) -- (442.68,103.54) -- (428.54,117.68) -- (407.32,96.46) -- cycle ;
\draw [line width=0.5]    (475,100) -- (450,100) ;
\draw [line width=0.5]    (475,60) -- (475,100) ;
\draw [line width=0.5]    (145,100) -- (120,100) ;
\draw [line width=0.5]    (145,60) -- (145,100) ;
\draw [line width=0.5]    (170,100) -- (145,100) ;
\draw  [fill={rgb, 255:red, 245; green, 240; blue, 235 }  ,fill opacity=1 ][line width=0.5]  (127.32,103.54) -- (148.54,82.32) -- (162.68,96.46) -- (141.46,117.68) -- cycle ;
\draw [line width=0.5]    (120,100) -- (95,100) ;
\draw [line width=0.5]    (95,60) -- (95,100) ;

\draw (190,99) node   [align=left] {$ ...$};
\draw (380,99) node   [align=left] {$ ...$};
\end{mytikz4}\quad,
\end{equation}
where the tilted rectangles indicate the left and right isometric conditions. By embedding each isometric local tensor to a unitary block of support size $\beta=\log_d D + 1$, an MPS can be seen as the output of a sequential quantum circuit, where the blocks are arranged in a staircase pattern:
\begin{equation}\label{eq:sqc}
\begin{mytikz2}
\draw [line width=0.5]    (350,199.99) -- (350,130) ;
\draw [line width=0.5]    (350,110) -- (350,69.97) ;
\draw [line width=0.5]    (310,199.98) -- (310,130) ;
\draw [line width=0.5]    (310,110) -- (310,70.01) ;
\draw [line width=0.5]    (150,199.99) -- (150,130) ;
\draw [line width=0.5]    (230,200) -- (230,70) ;
\draw [line width=0.5]    (190,200) -- (190,70.02) ;
\draw [line width=0.5]    (150,110) -- (150,70.02) ;
\draw [line width=0.5]    (390,200) -- (390,70.02) ;
\draw [line width=0.5]    (270,200) -- (270,70) ;
\draw  [fill={rgb, 255:red, 245; green, 240; blue, 235 }  ,fill opacity=1 ][line width=0.5]  (130,140) -- (210,140) -- (210,160) -- (130,160) -- cycle ;
\draw  [fill={rgb, 255:red, 245; green, 240; blue, 235 }  ,fill opacity=1 ][line width=0.5]  (170.3,170) -- (290.3,170) -- (290.3,190) -- (170.3,190) -- cycle ;
\draw  [fill={rgb, 255:red, 245; green, 240; blue, 235 }  ,fill opacity=1 ][line width=0.5]  (250,140) -- (330,140) -- (330,160) -- (250,160) -- cycle ;
\draw  [fill={rgb, 255:red, 245; green, 240; blue, 235 }  ,fill opacity=1 ][line width=0.5]  (330,80) -- (410,80) -- (410,100) -- (330,100) -- cycle ;
\draw [line width=0.5]    (70,200.02) -- (70,70.02) ;
\draw [line width=0.5]    (110,199.98) -- (110,130) ;
\draw [line width=0.5]    (110,110) -- (110,70.01) ;
\draw  [fill={rgb, 255:red, 245; green, 240; blue, 235 }  ,fill opacity=1 ][line width=0.5]  (50,80.02) -- (130,80.02) -- (130,100.02) -- (50,100.02) -- cycle ;

\draw (130,120) node   [align=left] {$\displaystyle ...$};
\draw (330,120) node   [align=left] {$\displaystyle ...$};
\end{mytikz2}\quad.
\end{equation}
This diagram shows that moving the orthogonality center changes the causal structure of the circuit, although the physical state is unchanged, which is a central fact to our results below. Please refer to the Supplemental Material (SM) for detailed information on the basic setup and preliminaries~\cite{SM}.


\textit{Ensemble equivalence.---}To prove the invariance of local minimum distribution under moves of the orthogonality center, we first prove that the state distributions induced by different canonical MPS representations are identical. We use $\mathbb{V}_\mathrm{MPS}^{[i]}$ to denote the MPS ensemble generated by the random sequential circuit with orthogonality center at site $i$ as in Eq.\,\eqref{eq:sqc}, where each unitary is drawn independently from the Haar measure~\cite{Collins2006, Collins2016, Nahum2017, Zhou2019, Fisher2022, Lami2025a, Lami2025, Sauliere2026, Dowling2026}. The MPS ensemble is subject to the following theorem.

\begin{theorem}\label{thm:1}
The random MPS ensembles with different orthogonality centers are identical, i.e., $\mathbb{V}_\mathrm{MPS}^{[i]} = \mathbb{V}_\mathrm{MPS}^{[j]}$ for any two sites $i$ and $j$.
\end{theorem}

\begin{sketchproof}We give a sketch proof here with the detailed version left to SM~\cite{SM}. As different MPS ensembles share the same sample space and are induced from distributions over compact groups, they are identical as long as they have equal moments for arbitrary orders. According to the Weingarten formula, the $t$-th moment of a random unitary gate $U$ of shape $Dd\times Dd$ is given by
\begin{equation}
    \int \mathrm{d}\mu(U)~ U^{\otimes t} \otimes U^{* \otimes t} = \sum_{\sigma,\tau\in \mathcal{S}_{t}} W^{(t)}_{\sigma,\tau} (Dd) \ketbra{\sigma}{\tau},
\end{equation}
where $W^{(t)}(Dd)$ is the $t$-degree Weingarten matrix and $\ket{\sigma}$ denotes the permutation operator on the $t$ replicas from the permutation group $\mathcal{S}_t$. Thus, the $t$-th moment of the random MPS from $\mathbb{V}_\mathrm{MPS}^{[i]}$ can be diagrammatically represented as
\begin{equation}\label{eq:rmps}
\begin{mytikz4}

\draw [line width=0.5]    (430,140) -- (470,140) ;
\draw [line width=0.5]    (390,140) -- (430,140) ;
\draw [line width=0.5]    (470,140) -- (520,140) ;
\draw [line width=0.5]    (560,140) -- (520,140) ;
\draw [line width=0.5]    (560,140) -- (600,140) ;
\draw [line width=0.5]    (560,140) -- (560,80) ;
\draw  [fill={rgb, 255:red, 142; green, 142; blue, 142 }  ,fill opacity=1 ][line width=0.5]  (556.5,140) .. controls (556.5,138.07) and (558.07,136.5) .. (560,136.5) .. controls (561.93,136.5) and (563.5,138.07) .. (563.5,140) .. controls (563.5,141.93) and (561.93,143.5) .. (560,143.5) .. controls (558.07,143.5) and (556.5,141.93) .. (556.5,140) -- cycle ;
\draw [line width=0.5]    (340,140) -- (390,140) ;
\draw [line width=0.5]    (340,140) -- (300,140) ;
\draw [line width=0.5]    (260,140) -- (300,140) ;
\draw [line width=0.5]    (300,140) -- (300,80) ;
\draw  [fill={rgb, 255:red, 142; green, 142; blue, 142 }  ,fill opacity=1 ][line width=0.5]  (296.51,140.02) .. controls (296.51,138.08) and (298.07,136.52) .. (300.01,136.52) .. controls (301.94,136.52) and (303.51,138.08) .. (303.51,140.02) .. controls (303.51,141.95) and (301.94,143.52) .. (300.01,143.52) .. controls (298.07,143.52) and (296.51,141.95) .. (296.51,140.02) -- cycle ;
\draw [line width=0.5]    (430,139.79) -- (430,80) ;
\draw  [fill={rgb, 255:red, 142; green, 142; blue, 142 }  ,fill opacity=1 ][line width=0.5]  (426.5,140) .. controls (426.5,138.07) and (428.07,136.5) .. (430,136.5) .. controls (431.93,136.5) and (433.5,138.07) .. (433.5,140) .. controls (433.5,141.93) and (431.93,143.5) .. (430,143.5) .. controls (428.07,143.5) and (426.5,141.93) .. (426.5,140) -- cycle ;
\draw  [fill={rgb, 255:red, 245; green, 240; blue, 235 }  ,fill opacity=1 ][line width=0.5]  (260,155.02) -- (260,125.02) -- (275,140.02) -- cycle ;
\draw  [fill={rgb, 255:red, 245; green, 240; blue, 235 }  ,fill opacity=1 ][line width=0.5]  (390,155) -- (390,125) -- (405,140) -- cycle ;
\draw  [fill={rgb, 255:red, 245; green, 240; blue, 235 }  ,fill opacity=1 ][line width=0.5]  (390,125) -- (390,155) -- (375,140) -- cycle ;
\draw  [fill={rgb, 255:red, 245; green, 240; blue, 235 }  ,fill opacity=1 ][line width=0.5]  (470,155) -- (470,125) -- (485,140) -- cycle ;
\draw  [fill={rgb, 255:red, 245; green, 240; blue, 235 }  ,fill opacity=1 ][line width=0.5]  (470,125) -- (470,155) -- (455,140) -- cycle ;
\draw  [fill={rgb, 255:red, 245; green, 240; blue, 235 }  ,fill opacity=1 ][line width=0.5]  (600,125) -- (600,155) -- (585,140) -- cycle ;
\draw  [fill={rgb, 255:red, 245; green, 240; blue, 235 }  ,fill opacity=1 ][line width=0.5]  (505,125) -- (535,125) -- (535,155) -- (505,155) -- cycle ;
\draw  [fill={rgb, 255:red, 245; green, 240; blue, 235 }  ,fill opacity=1 ][line width=0.5]  (325,125) -- (355,125) -- (355,155) -- (325,155) -- cycle ;
\draw [line width=0.5]    (600,140) -- (620,140) ;
\draw [line width=0.5]    (240,140) -- (260,140) ;
\draw  [fill={rgb, 255:red, 245; green, 240; blue, 235 }  ,fill opacity=1 ][line width=0.5]  (545,100) -- (574.98,100) -- (559.99,114.99) -- cycle ;
\draw  [fill={rgb, 255:red, 245; green, 240; blue, 235 }  ,fill opacity=1 ][line width=0.5]  (415,100) -- (444.98,100) -- (429.99,114.99) -- cycle ;
\draw  [fill={rgb, 255:red, 245; green, 240; blue, 235 }  ,fill opacity=1 ][line width=0.5]  (285.01,100) -- (314.99,100) -- (300,114.99) -- cycle ;

\draw (635,140) node  [font=\normalsize] [align=left] {$\displaystyle ...$};
\draw (225,140) node  [font=\normalsize] [align=left] {$\displaystyle ...$};
\draw (430,165) node  [font=\footnotesize] [align=left] {$\displaystyle i$};
\draw (560,165) node  [font=\footnotesize] [align=left] {$\displaystyle i+1$};
\draw (300,165) node  [font=\footnotesize] [align=left] {$\displaystyle i-1$};

\end{mytikz4}\quad,
\end{equation}
times a factor that only depends on $t$, $d$ and $D$. The triangles represent the permutation tensor $S$ with $S_\sigma=\ket{\sigma}$. The upright squares represent $W^{(t)}(Dd)$. The squares with diagonal lines represent the $t$-degree Gram matrices $G^{(t)}_{\sigma,\tau}(D)=\braket{\sigma}{\tau}=D^{\# (\sigma^{-1}\tau)}$, where $\# (\sigma^{-1}\tau)$ counts the number of cyclic permutations in $\sigma^{-1}\tau$. The grey dots represent the copy tensors. Since the cyclic number is a class function in the group algebra $\mathbb{C}[\mathcal{S}_t]$, $W^{(t)}(Dd)$ commutes with $G^{(t)}(d)$. By exchanging them on the bond between $i$ and $i+1$, we obtain exactly the $t$-th moment of $\mathbb{V}_\mathrm{MPS}^{[i+1]}$. By induction, the proof is completed.
\end{sketchproof}

The above theorem implies that randomly initialized sequential circuits with different orthogonality centers not only generate the same set of MPS but also generate the same initialized distribution over it. To connect this statistical fact to optimization outcomes, we first give a practical definition of local minimum distribution.

\textit{Invariance of local minimum distribution.---}For an energy function $\mathcal{E}$, a parameter point $\bm{\theta}$ is a local minimum of $\mathcal{E}$ if there is an open neighborhood $\mathcal{N}(\bm{\theta})$ such that any point $\bm{\theta}'$ in the neighborhood satisfies $\mathcal{E}(\bm{\theta}) \leq \mathcal{E}(\bm{\theta}')$. For smooth functions, this is equivalent to zero gradient and positive semi-definite Hessian matrix. The strict positive definiteness is not required in order to include the local minima that are not isolated points but form continuous lines or areas.

Local minimum distribution refers to the distribution of the energy function values of local minima. Thus, it is necessary to define a measure over the local minimum set $\mathcal{LM}_\mathcal{E}$. To reflect the probability of converging to each local minimum, we use a practical measure by incorporating the optimization dynamics. For MPS, there is a standard choice of optimization algorithm---the imaginary-time TDVP algorithm~\cite{Haegeman2011}, which projects the exact imaginary-time evolution onto the tangent space of MPS. The one-site update algorithm of MPS, as a specific form of the DMRG algorithm, can be regarded as a particular limit of the TDVP algorithm by taking the infinite time step size after the Trotter decomposition of the tangent space projector~\cite{Haegeman2016}. In a more general context, the TDVP algorithm is equivalent to the quantum natural gradient descent~\cite{Stokes2020}, Riemannian gradient descent, and stochastic reconfiguration algorithm. Specifically, the evolution flow of parameters is given by
\begin{equation}
    \frac{\mathrm{d}}{\mathrm{d}t}\bm{\theta}(t) = - \mathrm{grad}\, \mathcal{E} (\bm{\theta}(t)),
\end{equation}
where $\mathrm{grad}\, \mathcal{E} $ is the quantum natural gradient, related to the standard gradient by $\mathrm{grad}\, \mathcal{E}(\bm{\theta})=\mathcal{I}^{+}(\bm{\theta}) \nabla \mathcal{E}(\bm{\theta})$. $\mathcal{I}^{+}(\bm{\theta})$ is the pseudo-inverse of the quantum Fisher information matrix~\cite{SM}. The long-time limit of this negative natural gradient flow defines an optimization map
\begin{equation}
    \mathcal{O}_\mathcal{E}(\bm{\theta}_0)=\lim_{t\rightarrow\infty}\bm{\theta}(t),
\end{equation}
where $\bm{\theta}_0$ represents the initial point. Given a random initialization ensemble $\Theta$ with measure $\mu_\Theta$, the push-forward measure under $\mathcal{O}_\mathcal{E}$ reflects the probability of converging to each local minimum. Namely, we assign each attractor a measure by the volume of its basin of attraction. The local minimum distribution is then given by the following cumulative distribution function
\begin{equation}\label{eq:PE}
    \mathcal{P}_\mathcal{E}(E;\Theta) = \mu_{\Theta}\left[ \mathcal{O}_\mathcal{E}^{-1}\left(\{ \bm{\theta}\in \mathcal{LM}_\mathcal{E} \mid \mathcal{E}(\bm{\theta})\leq E \}\right)\right],
\end{equation}
which directly describes the distribution of energy $E$ in the final converged results of optimization starting from $\Theta$. We use $\mathcal{A}_\mathrm{MPS}^{[i]}$ and $\mathbb{A}_\mathrm{MPS}^{[i]}$ to denote the MPS representation space and ensemble with orthogonality center at site $i$, respectively. The corresponding local minimum distribution obeys the following theorem.


\begin{figure}
    \centering
    \includegraphics[width=0.99\linewidth]{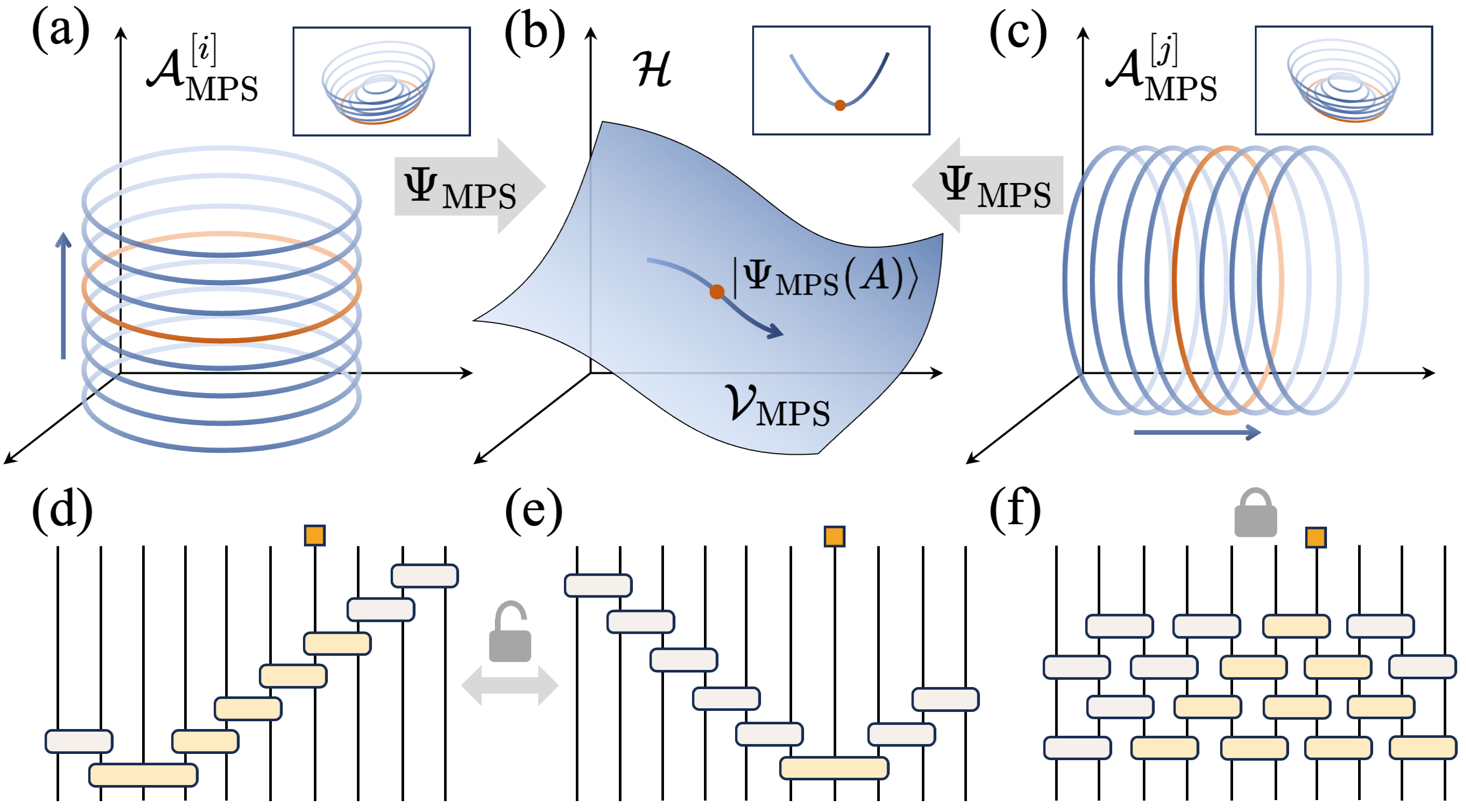}
    \caption{Local minimum correspondence. (a) and (c) depict the MPS representation spaces $\mathcal{A}_\mathrm{MPS}^{[i]}$ and $\mathcal{A}_\mathrm{MPS}^{[j]}$ with orthogonality centers at site $i$ and $j$, respectively. The rings represent the gauge orbits. (b) depicts the MPS variety $\mathcal{V}_\mathrm{MPS}$ as a subset in the Hilbert space $\mathcal{H}$. A gauge orbit is mapped to a single physical state by the contraction map $\Psi_\mathrm{MPS}$. The insets depict the respective energy landscapes. Each physical local minimum is mapped to a continuum of local minima in the representation space with the same energy. (d) and (e) depict the flexible causal structure in MPS and sequential circuits, which is absent in (f) brickwork circuits. The orange square represents a local observable and the colored blocks represent the backward causal cone of the observable.}
    \label{fig:mps_variety}
\end{figure}

\begin{theorem}\label{thm:2}
The local minimum distributions of the energy landscapes with respect to the MPS representation ensembles with different orthogonality centers are identical, i.e., $\mathcal{P}_\mathcal{E}(E; \mathbb{A}_\mathrm{MPS}^{[i]}) = \mathcal{P}_\mathcal{E}(E; \mathbb{A}_\mathrm{MPS}^{[j]})$ for any two sites $i,j$ and any Hamiltonian.
\end{theorem}

\begin{sketchproof}We sketch the main idea here and leave the details to SM~\cite{SM}. The set of MPS with fixed bond dimension form an algebraic variety $\mathcal{V}_\mathrm{MPS}$, with the full-rank stratum being a complex manifold $\mathcal{V}_\mathrm{MPS}^\mathrm{f}$ that is biholomorphic to the gauge orbit space, i.e., the quotient space of the full-rank MPS representation space with respect to the gauge group. With careful treatment of the rank-deficient singularities, an isoenergy local minimum correspondence can be established between the energy functions on $\mathcal{A}_\mathrm{MPS}^{[i]}$ and $\mathcal{A}_\mathrm{MPS}^{[j]}$, as illustrated in Fig.\,\ref{fig:mps_variety}. Since the TDVP evolution flow is independent of the MPS representation, the basin of attraction of the local minimum subset of $\mathcal{A}_\mathrm{MPS}^{[i]}$ at energy level $E$ has the same image in $\mathcal{V}_\mathrm{MPS}$ as that of $\mathcal{A}_\mathrm{MPS}^{[j]}$. According to Theorem~\ref{thm:1}, the image has the same measure under the two distributions induced by $\mathbb{A}_\mathrm{MPS}^{[i]}$ and $\mathbb{A}_\mathrm{MPS}^{[j]}$. Thus, the two local minimum subsets have the same basin-based measure, which completes the proof.
\end{sketchproof}

We remark that Theorem~\ref{thm:2} only implies that certain statistical properties of local minima are invariant under moves of the orthogonality center, while the exact energy landscape may deform drastically. For sequential circuits that are further parametrized by Pauli rotation angles, Theorem~\ref{thm:2} also holds true because the parametrization is a submersion up to coordinate singularities~\cite{SM}. 




\textit{Good local minima in MPS.---}Theorem~\ref{thm:2} allows us to analyze the local minimum distribution by choosing the most convenient gauge. We can thus derive some useful results on the trainability of MPS. We denote the Hamiltonian as $H = \sum_j H_j$, where $H_j$ is a Hermitian basis operator such as a Pauli string times a real coefficient.

In the simplest case, suppose that $H$ only contains one local operator $H_1$ supported on $s$ contiguous sites. Theorem~\ref{thm:2} guarantees that the local minimum distribution is the same as that of the sequential circuit whose orthogonality center is within the support of $H_1$, where all other unitaries outside the backward causal cone of $H_1$ are canceled out with their conjugates, as illustrated in Figs.\,\ref{fig:mps_variety}(d) and (e). Thus, when the number of independent parameters in the backward causal cone $2s(d-1)D^2 - 1$ is larger than the Hilbert space dimension of the support $(2d^s - 1)$, i.e., the bond dimension is larger than $D_c = \sqrt{d^s/[s(d-1)]}$, the reduced state is locally overparametrized, and hence the problem is equivalent to the optimization of a linear function over a Bloch ball, which is essentially a convex optimization problem where all the local minima are global minima. We call this effective local overparametrization because, in the original representation, the explicit orthogonality center can be far away from the support of $H_1$, so that the backward causal cone is too large to be overparametrized. From the perspective of circuit optimization, this simplest case is still non-trivial, because, as a counterexample, brickwork circuits of moderate depth are expected to face challenges in efficiently finding this simple ground state due to proliferation of poor local minima caused by underparametrization~\cite{Anschuetz2022}.

\begin{figure}
    \centering
    \includegraphics[width=0.99\linewidth]{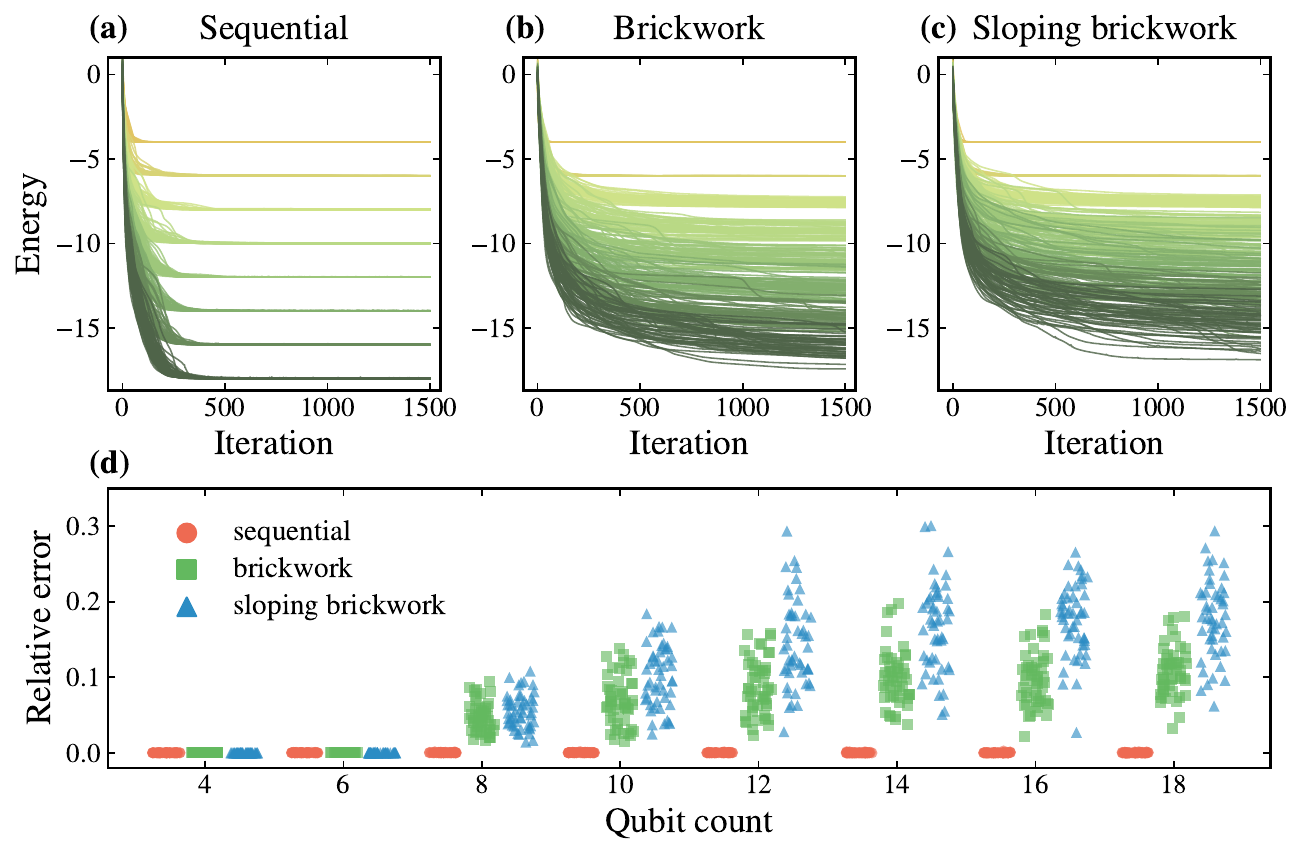}
    \caption{The numerical training curves of (a) sequential circuits of $\beta=3,L=1$, (b) brickwork circuits of $\beta=2,L=4$, and (c) sloping brickwork circuits of $\beta=2,L=4$. The block size $\beta$ and the number of layers $L$ are selected so that the numbers of parameters in different circuits are comparable. The loss function is the energy expectation of random backward-evolved Hamiltonians. The darkness of the color marks the system size from $4$ to $18$ qubits. The number of samples is $50$ for each system size. (d) The corresponding local minimum distributions in the form of scatter plots, with the data extracted from the final values in (a-c).}
    \label{fig:numerical_training}
\end{figure}

In the generic case where the Hamiltonian consists of multiple local subterms, the total energy function is a sum of elementary energy functions: $\mathcal{E}(\bm{\theta}) = \sum_j\mathcal{E}_j(\bm{\theta}) = \sum_j \braoprket{\Psi(\bm{\theta})}{H_j}{\Psi(\bm{\theta})}$. To preserve the benign landscape properties of the individual terms, a natural sufficient condition is the common local minimum condition: every local minimum of $\mathcal{E}$ is also a local minimum of each $\mathcal{E}_j$. Thus, if each $\mathcal{E}_j$ has local minima within a band of width $\epsilon_j$ of its global minimum, the local minima of $\mathcal{E}$ concentrate near its global minimum up to $\sum_j\epsilon_j$. This condition imposes constraints on the Hamiltonian and the ansatz. For example, if $\epsilon_j=0$, the Hamiltonian must be frustration-free and the ansatz must be expressive enough to represent the ground states. Specifically, a sufficient condition that ensures the total gradient vanishes only if each subterm gradient vanishes is the compatible gradient condition: there exists a positive constant $c$ such that $\|\mathrm{grad}\,\mathcal{E}(\bm{\theta})\|^2_\mathcal{I} \geq c\sum_j\|\mathrm{grad}\,\mathcal{E}_j(\bm{\theta})\|^2_\mathcal{I}$. The norm is taken with respect to the quantum Fisher information metric. Physically, this implies relatively weak frustration among the imaginary-time evolution directions governed by different subterms within the tangent space. We have numerically verified this condition across different Hamiltonians~\cite{SM}. In summary, provided sufficient expressibility, the typical energy landscape of MPS for weakly frustrated 1D local Hamiltonians is expected to be free from poor local minima, in stark contrast to brickwork circuits. We note that it is natural to expect that MPS are not free from poor local minima in the worst case, since the local Hamiltonian problem is $\mathsf{QMA}$-complete even in one dimension~\cite{Aharonov2009}.

\textit{Numerical experiments.---}In order to estimate the typical trainability of different circuits, we design a family of random Hamiltonians by $H_{\text{rand}} = \mathbf{V}^\dagger H_{Z} \mathbf{V}$, termed backward-evolved Hamiltonians, where $\mathbf{V}$ is a random circuit in the same shape as the reversed ansatz circuit and $H_{Z}=-\sum_j Z_j$. This design ensures that the ground state can always be captured by the ansatz, so that the expressibility is always sufficient, and hence the optimization performance is solely determined by the trainability. We mainly compare three circuit architectures, i.e., sequential circuits, brickwork circuits, and sloping brickwork circuits, as shown in Fig.\,\ref{fig:landscape_circuit}. The optimizer is chosen as the widely used Adam optimizer. The resulting local minimum distribution is expected to qualitatively align with the standard manifold-based definition in Eq.\,\eqref{eq:PE}, as different optimizers mainly influence the convergence speed and individual trajectories compared to the overall statistics. As shown in Fig.\,\ref{fig:numerical_training}, the optimization of sequential circuits always converges to the ground state energy $E_0=-N$, while those of brickwork circuits and sloping brickwork circuits get stuck in poor local minima with significant relative errors, especially when the system size increases. The technical details and additional numerical results can be found in SM~\cite{SM}.

\textit{Discussion.---}In this work, we attribute the success of MPS-based optimization algorithms to effective local overparametrization of MPS, which originates from the flexible causal structure enabled by the gauge freedom of MPS. By contrast, the brickwork-type circuits entangle different parts of the system rapidly in a global manner, which prevents them from possessing the flexibility as in MPS and sequential circuits. Amidst previous negative results on the local minimum problem in training quantum circuits~\cite{Anschuetz2022}, our findings establish a rare positive example with rigorous theoretical proofs beyond the global overparametrization~\cite{Larocca2023}. This suggests that, from the perspective of trainability, a preferable strategy to enhance circuit expressibility is to increase the size of each universal unitary block rather than simply stacking layers repeatedly. Furthermore, although our results mainly focus on MPS, the conclusions and techniques can be generalized to broader tensor network states which possess a well-defined orthogonality center that is movable without altering the ansatz space, such as tree tensor network states. Extending this framework to investigate the energy landscapes of higher-dimensional tensor network states beyond classically simulable regimes is a promising direction for future research.



\textit{Acknowledgments.---}We acknowledge the stimulating discussions with Yantao Wu, Jialin Chen and Zhen Fan. H.K.Z. was supported by the Postdoctoral Fellowship Program and China Postdoctoral Science Foundation (No. BX20250169) and Beijing Natural Science Foundation (No. 1264076). S.L. was supported by the Gordon and Betty Moore Foundation through Grant No. GBMF8685 towards the Princeton theory program, the Gordon and Betty Moore Foundation’s EPiQS Initiative (Grant No. GBMF11070), the Global Collaborative Network Grant at Princeton University, the Simons Investigator Grant No. 404513, the Princeton Global Network, the NSF-MERSEC (Grant No. MERSEC DMR 2011750), the Simons Collaboration on New Frontiers in Superconductivity (Grant No. SFI-MPS-NFS-00006741-01 and No. SFI-MPS-NFS-00006741-06), the Princeton Catalysis Initiative, the Schmidt Foundation at the Princeton University. S.X.Z. was supported by Quantum Science and Technology-National Science and Technology Major Project (No. 2024ZD0301700), the National Natural Science Foundation of China (No. 12574546) and the Chinese Academy of Sciences (No. YSBR-150).

\bibliographystyle{apsreve}
\bibliography{autoref}

\begin{thebibliography}{64}%
\makeatletter
\providecommand \@ifxundefined [1]{%
 \@ifx{#1\undefined}
}%
\providecommand \@ifnum [1]{%
 \ifnum #1\expandafter \@firstoftwo
 \else \expandafter \@secondoftwo
 \fi
}%
\providecommand \@ifx [1]{%
 \ifx #1\expandafter \@firstoftwo
 \else \expandafter \@secondoftwo
 \fi
}%
\providecommand \natexlab [1]{#1}%
\providecommand \bibnamefont  [1]{#1}%
\providecommand \bibfnamefont [1]{#1}%
\providecommand \citenamefont [1]{#1}%
\providecommand \href@noop [0]{\@secondoftwo}%
\providecommand \href [0]{\begingroup \@sanitize@url \@href}%
\providecommand \@href[1]{\@@startlink{#1}\@@href}%
\providecommand \@@href[1]{\endgroup#1\@@endlink}%
\providecommand \@sanitize@url [0]{\catcode `\\12\catcode `\$12\catcode `\&12\catcode `\#12\catcode `\^12\catcode `\_12\catcode `\%12\relax}%
\providecommand \@@startlink[1]{}%
\providecommand \@@endlink[0]{}%
\providecommand \url  [0]{\begingroup\@sanitize@url \@url }%
\providecommand \@url [1]{\endgroup\@href {#1}{\urlprefix }}%
\providecommand \urlprefix  [0]{URL }%
\providecommand \Eprint [0]{\href }%
\providecommand \doibase [0]{http://dx.doi.org/}%
\providecommand \selectlanguage [0]{\@gobble}%
\providecommand \bibinfo  [0]{\@secondoftwo}%
\providecommand \bibfield  [0]{\@secondoftwo}%
\providecommand \translation [1]{[#1]}%
\providecommand \BibitemOpen [0]{}%
\providecommand \bibitemStop [0]{}%
\providecommand \bibitemNoStop [0]{.\EOS\space}%
\providecommand \EOS [0]{\spacefactor3000\relax}%
\providecommand \BibitemShut  [1]{\csname bibitem#1\endcsname}%
\let\auto@bib@innerbib\@empty
\bibitem [{\citenamefont {Choromanska}\ \emph {et~al.}(2015)\citenamefont {Choromanska}, \citenamefont {Henaff}, \citenamefont {Mathieu}, \citenamefont {Arous},\ and\ \citenamefont {LeCun}}]{Choromanska2015}%
  \BibitemOpen
  \bibfield  {author} {\bibinfo {author} {\bibfnamefont {A.}~\bibnamefont {Choromanska}}, \bibinfo {author} {\bibfnamefont {M.}~\bibnamefont {Henaff}}, \bibinfo {author} {\bibfnamefont {M.}~\bibnamefont {Mathieu}}, \bibinfo {author} {\bibfnamefont {G.~B.}\ \bibnamefont {Arous}}, \ and\ \bibinfo {author} {\bibfnamefont {Y.}~\bibnamefont {LeCun}},\ }\bibfield  {title} {{The Loss Surfaces of Multilayer Networks},\ }\href {https://proceedings.mlr.press/v38/choromanska15} {\bibfield  {journal} {\bibinfo  {journal} {Journal of Machine Learning Research}\ }\textbf {\bibinfo {volume} {38}},\ \bibinfo {pages} {192} (\bibinfo {year} {2015})}\BibitemShut {NoStop}%
\bibitem [{\citenamefont {Kawaguchi}(2016)}]{Kawaguchi2016}%
  \BibitemOpen
  \bibfield  {author} {\bibinfo {author} {\bibfnamefont {K.}~\bibnamefont {Kawaguchi}},\ }\bibfield  {title} {{Deep Learning without Poor Local Minima},\ }\href {http://arxiv.org/abs/1605.07110} {\bibfield  {journal} {\bibinfo  {journal} {Advances in Neural Information Processing Systems}\ ,\ \bibinfo {pages} {586}} (\bibinfo {year} {2016})}\BibitemShut {NoStop}%
\bibitem [{\citenamefont {Jacot}\ \emph {et~al.}(2020)\citenamefont {Jacot}, \citenamefont {Gabriel},\ and\ \citenamefont {Hongler}}]{Jacot2020}%
  \BibitemOpen
  \bibfield  {author} {\bibinfo {author} {\bibfnamefont {A.}~\bibnamefont {Jacot}}, \bibinfo {author} {\bibfnamefont {F.}~\bibnamefont {Gabriel}}, \ and\ \bibinfo {author} {\bibfnamefont {C.}~\bibnamefont {Hongler}},\ }{Neural Tangent Kernel: Convergence and Generalization in Neural Networks},\ in\ \href {https://proceedings.neurips.cc/paper_files/paper/2018/file/5a4be1fa34e62bb8a6ec6b91d2462f5a-Paper.pdf} {\bibinfo {booktitle} {Advances in Neural Information Processing Systems}},\ \bibinfo {series and number} {\bibinfo {number} {4}}\ (\bibinfo {year} {2020})\BibitemShut {NoStop}%
\bibitem [{\citenamefont {Bharti}\ \emph {et~al.}(2022)\citenamefont {Bharti}, \citenamefont {Cervera-Lierta}, \citenamefont {Kyaw}, \citenamefont {Haug}, \citenamefont {Alperin-Lea}, \citenamefont {Anand}, \citenamefont {Degroote}, \citenamefont {Heimonen}, \citenamefont {Kottmann}, \citenamefont {Menke}, \citenamefont {Mok}, \citenamefont {Sim}, \citenamefont {Kwek},\ and\ \citenamefont {Aspuru-Guzik}}]{Bharti2022}%
  \BibitemOpen
  \bibfield  {author} {\bibinfo {author} {\bibfnamefont {K.}~\bibnamefont {Bharti}}, \bibinfo {author} {\bibfnamefont {A.}~\bibnamefont {Cervera-Lierta}}, \bibinfo {author} {\bibfnamefont {T.~H.}\ \bibnamefont {Kyaw}}, \bibinfo {author} {\bibfnamefont {T.}~\bibnamefont {Haug}}, \bibinfo {author} {\bibfnamefont {S.}~\bibnamefont {Alperin-Lea}}, \bibinfo {author} {\bibfnamefont {A.}~\bibnamefont {Anand}}, \bibinfo {author} {\bibfnamefont {M.}~\bibnamefont {Degroote}}, \bibinfo {author} {\bibfnamefont {H.}~\bibnamefont {Heimonen}}, \bibinfo {author} {\bibfnamefont {J.~S.}\ \bibnamefont {Kottmann}}, \bibinfo {author} {\bibfnamefont {T.}~\bibnamefont {Menke}}, \bibinfo {author} {\bibfnamefont {W.-K.}\ \bibnamefont {Mok}}, \bibinfo {author} {\bibfnamefont {S.}~\bibnamefont {Sim}}, \bibinfo {author} {\bibfnamefont {L.-C.}\ \bibnamefont {Kwek}}, \ and\ \bibinfo {author} {\bibfnamefont {A.}~\bibnamefont {Aspuru-Guzik}},\ }\bibfield  {title} {{Noisy intermediate-scale quantum algorithms},\ }\href {\doibase 10.1103/RevModPhys.94.015004} {\bibfield  {journal} {\bibinfo  {journal} {Reviews of Modern Physics}\ }\textbf {\bibinfo {volume} {94}},\ \bibinfo {pages} {015004} (\bibinfo {year} {2022})}\BibitemShut {NoStop}%
\bibitem [{\citenamefont {Cerezo}\ \emph {et~al.}(2021{\natexlab{a}})\citenamefont {Cerezo}, \citenamefont {Arrasmith}, \citenamefont {Babbush}, \citenamefont {Benjamin}, \citenamefont {Endo}, \citenamefont {Fujii}, \citenamefont {McClean}, \citenamefont {Mitarai}, \citenamefont {Yuan}, \citenamefont {Cincio},\ and\ \citenamefont {Coles}}]{Cerezo2021a}%
  \BibitemOpen
  \bibfield  {author} {\bibinfo {author} {\bibfnamefont {M.}~\bibnamefont {Cerezo}}, \bibinfo {author} {\bibfnamefont {A.}~\bibnamefont {Arrasmith}}, \bibinfo {author} {\bibfnamefont {R.}~\bibnamefont {Babbush}}, \bibinfo {author} {\bibfnamefont {S.~C.}\ \bibnamefont {Benjamin}}, \bibinfo {author} {\bibfnamefont {S.}~\bibnamefont {Endo}}, \bibinfo {author} {\bibfnamefont {K.}~\bibnamefont {Fujii}}, \bibinfo {author} {\bibfnamefont {J.~R.}\ \bibnamefont {McClean}}, \bibinfo {author} {\bibfnamefont {K.}~\bibnamefont {Mitarai}}, \bibinfo {author} {\bibfnamefont {X.}~\bibnamefont {Yuan}}, \bibinfo {author} {\bibfnamefont {L.}~\bibnamefont {Cincio}}, \ and\ \bibinfo {author} {\bibfnamefont {P.~J.}\ \bibnamefont {Coles}},\ }\bibfield  {title} {{Variational quantum algorithms},\ }\href {\doibase 10.1038/s42254-021-00348-9} {\bibfield  {journal} {\bibinfo  {journal} {Nature Reviews Physics}\ }\textbf {\bibinfo {volume} {3}},\ \bibinfo {pages} {625} (\bibinfo {year} {2021}{\natexlab{a}})}\BibitemShut {NoStop}%
\bibitem [{\citenamefont {McClean}\ \emph {et~al.}(2018)\citenamefont {McClean}, \citenamefont {Boixo}, \citenamefont {Smelyanskiy}, \citenamefont {Babbush},\ and\ \citenamefont {Neven}}]{McClean2018}%
  \BibitemOpen
  \bibfield  {author} {\bibinfo {author} {\bibfnamefont {J.~R.}\ \bibnamefont {McClean}}, \bibinfo {author} {\bibfnamefont {S.}~\bibnamefont {Boixo}}, \bibinfo {author} {\bibfnamefont {V.~N.}\ \bibnamefont {Smelyanskiy}}, \bibinfo {author} {\bibfnamefont {R.}~\bibnamefont {Babbush}}, \ and\ \bibinfo {author} {\bibfnamefont {H.}~\bibnamefont {Neven}},\ }\bibfield  {title} {{Barren plateaus in quantum neural network training landscapes},\ }\href {\doibase 10.1038/s41467-018-07090-4} {\bibfield  {journal} {\bibinfo  {journal} {Nature Communications}\ }\textbf {\bibinfo {volume} {9}},\ \bibinfo {pages} {1} (\bibinfo {year} {2018})}\BibitemShut {NoStop}%
\bibitem [{\citenamefont {Cerezo}\ \emph {et~al.}(2021{\natexlab{b}})\citenamefont {Cerezo}, \citenamefont {Sone}, \citenamefont {Volkoff}, \citenamefont {Cincio},\ and\ \citenamefont {Coles}}]{Cerezo2021}%
  \BibitemOpen
  \bibfield  {author} {\bibinfo {author} {\bibfnamefont {M.}~\bibnamefont {Cerezo}}, \bibinfo {author} {\bibfnamefont {A.}~\bibnamefont {Sone}}, \bibinfo {author} {\bibfnamefont {T.}~\bibnamefont {Volkoff}}, \bibinfo {author} {\bibfnamefont {L.}~\bibnamefont {Cincio}}, \ and\ \bibinfo {author} {\bibfnamefont {P.~J.}\ \bibnamefont {Coles}},\ }\bibfield  {title} {{Cost function dependent barren plateaus in shallow parametrized quantum circuits},\ }\href {\doibase 10.1038/s41467-021-21728-w} {\bibfield  {journal} {\bibinfo  {journal} {Nature Communications}\ }\textbf {\bibinfo {volume} {12}},\ \bibinfo {pages} {1791} (\bibinfo {year} {2021}{\natexlab{b}})}\BibitemShut {NoStop}%
\bibitem [{\citenamefont {Zhang}\ \emph {et~al.}(2024{\natexlab{a}})\citenamefont {Zhang}, \citenamefont {Liu},\ and\ \citenamefont {Zhang}}]{Zhang2024}%
  \BibitemOpen
  \bibfield  {author} {\bibinfo {author} {\bibfnamefont {H.-K.}\ \bibnamefont {Zhang}}, \bibinfo {author} {\bibfnamefont {S.}~\bibnamefont {Liu}}, \ and\ \bibinfo {author} {\bibfnamefont {S.-X.}\ \bibnamefont {Zhang}},\ }\bibfield  {title} {{Absence of Barren Plateaus in Finite Local-Depth Circuits with Long-Range Entanglement},\ }\href {\doibase 10.1103/PhysRevLett.132.150603} {\bibfield  {journal} {\bibinfo  {journal} {Physical Review Letters}\ }\textbf {\bibinfo {volume} {132}},\ \bibinfo {pages} {150603} (\bibinfo {year} {2024}{\natexlab{a}})}\BibitemShut {NoStop}%
\bibitem [{\citenamefont {Larocca}\ \emph {et~al.}(2025)\citenamefont {Larocca}, \citenamefont {Thanasilp}, \citenamefont {Wang}, \citenamefont {Sharma}, \citenamefont {Biamonte}, \citenamefont {Coles}, \citenamefont {Cincio}, \citenamefont {McClean}, \citenamefont {Holmes},\ and\ \citenamefont {Cerezo}}]{Larocca2025}%
  \BibitemOpen
  \bibfield  {author} {\bibinfo {author} {\bibfnamefont {M.}~\bibnamefont {Larocca}}, \bibinfo {author} {\bibfnamefont {S.}~\bibnamefont {Thanasilp}}, \bibinfo {author} {\bibfnamefont {S.}~\bibnamefont {Wang}}, \bibinfo {author} {\bibfnamefont {K.}~\bibnamefont {Sharma}}, \bibinfo {author} {\bibfnamefont {J.}~\bibnamefont {Biamonte}}, \bibinfo {author} {\bibfnamefont {P.~J.}\ \bibnamefont {Coles}}, \bibinfo {author} {\bibfnamefont {L.}~\bibnamefont {Cincio}}, \bibinfo {author} {\bibfnamefont {J.~R.}\ \bibnamefont {McClean}}, \bibinfo {author} {\bibfnamefont {Z.}~\bibnamefont {Holmes}}, \ and\ \bibinfo {author} {\bibfnamefont {M.}~\bibnamefont {Cerezo}},\ }\bibfield  {title} {{Barren plateaus in variational quantum computing},\ }\href {\doibase 10.1038/s42254-025-00813-9} {\bibfield  {journal} {\bibinfo  {journal} {Nature Reviews Physics}\ }\textbf {\bibinfo {volume} {7}},\ \bibinfo {pages} {174} (\bibinfo {year} {2025})}\BibitemShut {NoStop}%
\bibitem [{\citenamefont {Bittel}\ and\ \citenamefont {Kliesch}(2021)}]{Bittel2021}%
  \BibitemOpen
  \bibfield  {author} {\bibinfo {author} {\bibfnamefont {L.}~\bibnamefont {Bittel}}\ and\ \bibinfo {author} {\bibfnamefont {M.}~\bibnamefont {Kliesch}},\ }\bibfield  {title} {{Training Variational Quantum Algorithms Is NP-Hard},\ }\href {\doibase 10.1103/PhysRevLett.127.120502} {\bibfield  {journal} {\bibinfo  {journal} {Physical Review Letters}\ }\textbf {\bibinfo {volume} {127}},\ \bibinfo {pages} {120502} (\bibinfo {year} {2021})}\BibitemShut {NoStop}%
\bibitem [{\citenamefont {You}\ and\ \citenamefont {Wu}(2021)}]{You2021}%
  \BibitemOpen
  \bibfield  {author} {\bibinfo {author} {\bibfnamefont {X.}~\bibnamefont {You}}\ and\ \bibinfo {author} {\bibfnamefont {X.}~\bibnamefont {Wu}},\ }{Exponentially Many Local Minima in Quantum Neural Networks},\ in\ \href {http://arxiv.org/abs/2110.02479} {\bibinfo {booktitle} {Proceedings of Machine Learning Research}},\ Vol.\ \bibinfo {volume} {139}\ (\bibinfo {year} {2021})\ pp.\ \bibinfo {pages} {12144--12155}\BibitemShut {NoStop}%
\bibitem [{\citenamefont {Anschuetz}\ and\ \citenamefont {Kiani}(2022)}]{Anschuetz2022}%
  \BibitemOpen
  \bibfield  {author} {\bibinfo {author} {\bibfnamefont {E.~R.}\ \bibnamefont {Anschuetz}}\ and\ \bibinfo {author} {\bibfnamefont {B.~T.}\ \bibnamefont {Kiani}},\ }\bibfield  {title} {{Quantum variational algorithms are swamped with traps},\ }\href {https://www.nature.com/articles/s41467-022-35364-5} {\bibfield  {journal} {\bibinfo  {journal} {Nature Communications}\ }\textbf {\bibinfo {volume} {13}} (\bibinfo {year} {2022})}\BibitemShut {NoStop}%
\bibitem [{\citenamefont {Anschuetz}(2023)}]{Anschuetz2023}%
  \BibitemOpen
  \bibfield  {author} {\bibinfo {author} {\bibfnamefont {E.~R.}\ \bibnamefont {Anschuetz}},\ }{Critical Points in Quantum Generative Models},\ in\ \href {https://openreview.net/forum?id=2f1z55GVQN} {\bibinfo {booktitle} {International Conference on Learning Representations}}\ (\bibinfo {year} {2023})\ pp.\ \bibinfo {pages} {1--24}\BibitemShut {NoStop}%
\bibitem [{\citenamefont {Anschuetz}(2025)}]{Anschuetz2025}%
  \BibitemOpen
  \bibfield  {author} {\bibinfo {author} {\bibfnamefont {E.~R.}\ \bibnamefont {Anschuetz}},\ }{A Unified Theory of Quantum Neural Network Loss Landscapes},\ in\ \href {http://arxiv.org/abs/2408.11901} {\bibinfo {booktitle} {International Conference on Learning Representations}},\ \bibinfo {series and number} {\bibinfo {number} {1}}\ (\bibinfo {year} {2025})\ pp.\ \bibinfo {pages} {1--60}\BibitemShut {NoStop}%
\bibitem [{\citenamefont {Arrasmith}\ \emph {et~al.}(2021)\citenamefont {Arrasmith}, \citenamefont {Cerezo}, \citenamefont {Czarnik}, \citenamefont {Cincio},\ and\ \citenamefont {Coles}}]{Arrasmith2020}%
  \BibitemOpen
  \bibfield  {author} {\bibinfo {author} {\bibfnamefont {A.}~\bibnamefont {Arrasmith}}, \bibinfo {author} {\bibfnamefont {M.}~\bibnamefont {Cerezo}}, \bibinfo {author} {\bibfnamefont {P.}~\bibnamefont {Czarnik}}, \bibinfo {author} {\bibfnamefont {L.}~\bibnamefont {Cincio}}, \ and\ \bibinfo {author} {\bibfnamefont {P.~J.}\ \bibnamefont {Coles}},\ }\bibfield  {title} {{Effect of barren plateaus on gradient-free optimization},\ }\href {\doibase 10.22331/q-2021-10-05-558} {\bibfield  {journal} {\bibinfo  {journal} {Quantum}\ }\textbf {\bibinfo {volume} {5}},\ \bibinfo {pages} {558} (\bibinfo {year} {2021})}\BibitemShut {NoStop}%
\bibitem [{\citenamefont {Arrasmith}\ \emph {et~al.}(2022)\citenamefont {Arrasmith}, \citenamefont {Holmes}, \citenamefont {Cerezo},\ and\ \citenamefont {Coles}}]{Arrasmith2021}%
  \BibitemOpen
  \bibfield  {author} {\bibinfo {author} {\bibfnamefont {A.}~\bibnamefont {Arrasmith}}, \bibinfo {author} {\bibfnamefont {Z.}~\bibnamefont {Holmes}}, \bibinfo {author} {\bibfnamefont {M.}~\bibnamefont {Cerezo}}, \ and\ \bibinfo {author} {\bibfnamefont {P.~J.}\ \bibnamefont {Coles}},\ }\bibfield  {title} {{Equivalence of quantum barren plateaus to cost concentration and narrow gorges},\ }\href {\doibase 10.1088/2058-9565/ac7d06} {\bibfield  {journal} {\bibinfo  {journal} {Quantum Science and Technology}\ }\textbf {\bibinfo {volume} {7}},\ \bibinfo {pages} {045015} (\bibinfo {year} {2022})}\BibitemShut {NoStop}%
\bibitem [{\citenamefont {Liu}\ \emph {et~al.}(2022)\citenamefont {Liu}, \citenamefont {Yu}, \citenamefont {Duan},\ and\ \citenamefont {Deng}}]{Liu2021a}%
  \BibitemOpen
  \bibfield  {author} {\bibinfo {author} {\bibfnamefont {Z.}~\bibnamefont {Liu}}, \bibinfo {author} {\bibfnamefont {L.-W.}\ \bibnamefont {Yu}}, \bibinfo {author} {\bibfnamefont {L.-M.}\ \bibnamefont {Duan}}, \ and\ \bibinfo {author} {\bibfnamefont {D.-L.}\ \bibnamefont {Deng}},\ }\bibfield  {title} {{Presence and Absence of Barren Plateaus in Tensor-Network Based Machine Learning},\ }\href {\doibase 10.1103/PhysRevLett.129.270501} {\bibfield  {journal} {\bibinfo  {journal} {Physical Review Letters}\ }\textbf {\bibinfo {volume} {129}},\ \bibinfo {pages} {270501} (\bibinfo {year} {2022})}\BibitemShut {NoStop}%
\bibitem [{\citenamefont {Holmes}\ \emph {et~al.}(2022)\citenamefont {Holmes}, \citenamefont {Sharma}, \citenamefont {Cerezo},\ and\ \citenamefont {Coles}}]{Holmes2021}%
  \BibitemOpen
  \bibfield  {author} {\bibinfo {author} {\bibfnamefont {Z.}~\bibnamefont {Holmes}}, \bibinfo {author} {\bibfnamefont {K.}~\bibnamefont {Sharma}}, \bibinfo {author} {\bibfnamefont {M.}~\bibnamefont {Cerezo}}, \ and\ \bibinfo {author} {\bibfnamefont {P.~J.}\ \bibnamefont {Coles}},\ }\bibfield  {title} {{Connecting Ansatz Expressibility to Gradient Magnitudes and Barren Plateaus},\ }\href {\doibase 10.1103/PRXQuantum.3.010313} {\bibfield  {journal} {\bibinfo  {journal} {PRX Quantum}\ }\textbf {\bibinfo {volume} {3}},\ \bibinfo {pages} {010313} (\bibinfo {year} {2022})}\BibitemShut {NoStop}%
\bibitem [{\citenamefont {Zhang}\ \emph {et~al.}(2023{\natexlab{a}})\citenamefont {Zhang}, \citenamefont {Zhu}, \citenamefont {Jing},\ and\ \citenamefont {Wang}}]{Zhang2023}%
  \BibitemOpen
  \bibfield  {author} {\bibinfo {author} {\bibfnamefont {H.-K.}\ \bibnamefont {Zhang}}, \bibinfo {author} {\bibfnamefont {C.}~\bibnamefont {Zhu}}, \bibinfo {author} {\bibfnamefont {M.}~\bibnamefont {Jing}}, \ and\ \bibinfo {author} {\bibfnamefont {X.}~\bibnamefont {Wang}},\ }\bibfield  {title} {{Statistical Analysis of Quantum State Learning Process in Quantum Neural Networks},\ }\href {http://arxiv.org/abs/2309.14980} {\bibfield  {journal} {\bibinfo  {journal} {Advances in Neural Information Processing Systems}\ } (\bibinfo {year} {2023}{\natexlab{a}})}\BibitemShut {NoStop}%
\bibitem [{\citenamefont {Barthel}\ and\ \citenamefont {Miao}(2025)}]{Barthel2023a}%
  \BibitemOpen
  \bibfield  {author} {\bibinfo {author} {\bibfnamefont {T.}~\bibnamefont {Barthel}}\ and\ \bibinfo {author} {\bibfnamefont {Q.}~\bibnamefont {Miao}},\ }\bibfield  {title} {{Absence of Barren Plateaus and Scaling of Gradients in the Energy Optimization of Isometric Tensor Network States},\ }\href {\doibase 10.1007/s00220-024-05217-x} {\bibfield  {journal} {\bibinfo  {journal} {Communications in Mathematical Physics}\ }\textbf {\bibinfo {volume} {406}},\ \bibinfo {pages} {86} (\bibinfo {year} {2025})}\BibitemShut {NoStop}%
\bibitem [{\citenamefont {Miao}\ and\ \citenamefont {Barthel}(2024)}]{Miao2023}%
  \BibitemOpen
  \bibfield  {author} {\bibinfo {author} {\bibfnamefont {Q.}~\bibnamefont {Miao}}\ and\ \bibinfo {author} {\bibfnamefont {T.}~\bibnamefont {Barthel}},\ }\bibfield  {title} {{Isometric tensor network optimization for extensive Hamiltonians is free of barren plateaus},\ }\href {\doibase 10.1103/PhysRevA.109.L050402} {\bibfield  {journal} {\bibinfo  {journal} {Physical Review A}\ }\textbf {\bibinfo {volume} {109}},\ \bibinfo {pages} {L050402} (\bibinfo {year} {2024})}\BibitemShut {NoStop}%
\bibitem [{\citenamefont {Liu}\ \emph {et~al.}(2023)\citenamefont {Liu}, \citenamefont {Zhang}, \citenamefont {Jian},\ and\ \citenamefont {Yao}}]{Liu2023}%
  \BibitemOpen
  \bibfield  {author} {\bibinfo {author} {\bibfnamefont {S.}~\bibnamefont {Liu}}, \bibinfo {author} {\bibfnamefont {S.-X.}\ \bibnamefont {Zhang}}, \bibinfo {author} {\bibfnamefont {S.-K.}\ \bibnamefont {Jian}}, \ and\ \bibinfo {author} {\bibfnamefont {H.}~\bibnamefont {Yao}},\ }\bibfield  {title} {{Training variational quantum algorithms with random gate activation},\ }\href {\doibase 10.1103/PhysRevResearch.5.L032040} {\bibfield  {journal} {\bibinfo  {journal} {Physical Review Research}\ }\textbf {\bibinfo {volume} {5}},\ \bibinfo {pages} {L032040} (\bibinfo {year} {2023})}\BibitemShut {NoStop}%
\bibitem [{\citenamefont {Liu}\ \emph {et~al.}(2024)\citenamefont {Liu}, \citenamefont {Liu}, \citenamefont {Zhang}, \citenamefont {Huang},\ and\ \citenamefont {Wang}}]{Liu2024}%
  \BibitemOpen
  \bibfield  {author} {\bibinfo {author} {\bibfnamefont {X.}~\bibnamefont {Liu}}, \bibinfo {author} {\bibfnamefont {G.}~\bibnamefont {Liu}}, \bibinfo {author} {\bibfnamefont {H.-K.}\ \bibnamefont {Zhang}}, \bibinfo {author} {\bibfnamefont {J.}~\bibnamefont {Huang}}, \ and\ \bibinfo {author} {\bibfnamefont {X.}~\bibnamefont {Wang}},\ }\bibfield  {title} {{Mitigating Barren Plateaus of Variational Quantum Eigensolvers},\ }\href {\doibase 10.1109/TQE.2024.3383050} {\bibfield  {journal} {\bibinfo  {journal} {IEEE Transactions on Quantum Engineering}\ }\textbf {\bibinfo {volume} {5}},\ \bibinfo {pages} {1} (\bibinfo {year} {2024})}\BibitemShut {NoStop}%
\bibitem [{\citenamefont {Zhang}\ \emph {et~al.}(2024{\natexlab{b}})\citenamefont {Zhang}, \citenamefont {Zhu}, \citenamefont {Liu},\ and\ \citenamefont {Wang}}]{Zhang2024a}%
  \BibitemOpen
  \bibfield  {author} {\bibinfo {author} {\bibfnamefont {H.-K.}\ \bibnamefont {Zhang}}, \bibinfo {author} {\bibfnamefont {C.}~\bibnamefont {Zhu}}, \bibinfo {author} {\bibfnamefont {G.}~\bibnamefont {Liu}}, \ and\ \bibinfo {author} {\bibfnamefont {X.}~\bibnamefont {Wang}},\ }\bibfield  {title} {{Exponential Hardness of Optimization from the Locality in Quantum Neural Networks},\ }\href {\doibase 10.1609/aaai.v38i15.29614} {\bibfield  {journal} {\bibinfo  {journal} {Proceedings of the AAAI Conference on Artificial Intelligence}\ }\textbf {\bibinfo {volume} {38}},\ \bibinfo {pages} {16741} (\bibinfo {year} {2024}{\natexlab{b}})}\BibitemShut {NoStop}%
\bibitem [{\citenamefont {Zhang}\ \emph {et~al.}(2024{\natexlab{c}})\citenamefont {Zhang}, \citenamefont {Zhu},\ and\ \citenamefont {Wang}}]{Zhang2024b}%
  \BibitemOpen
  \bibfield  {author} {\bibinfo {author} {\bibfnamefont {H.-K.}\ \bibnamefont {Zhang}}, \bibinfo {author} {\bibfnamefont {C.}~\bibnamefont {Zhu}}, \ and\ \bibinfo {author} {\bibfnamefont {X.}~\bibnamefont {Wang}},\ }\bibfield  {title} {{Predicting quantum learnability from landscape fluctuation},\ }\href {http://arxiv.org/abs/2406.11805} {\bibfield  {journal} {\bibinfo  {journal} {arXiv:2406.11805}\ } (\bibinfo {year} {2024}{\natexlab{c}})}\BibitemShut {NoStop}%
\bibitem [{\citenamefont {Cerezo}\ \emph {et~al.}(2025)\citenamefont {Cerezo}, \citenamefont {Larocca}, \citenamefont {Garc{\'{i}}a-Mart{\'{i}}n}, \citenamefont {Diaz}, \citenamefont {Braccia}, \citenamefont {Fontana}, \citenamefont {Rudolph}, \citenamefont {Bermejo}, \citenamefont {Ijaz}, \citenamefont {Thanasilp}, \citenamefont {Anschuetz},\ and\ \citenamefont {Holmes}}]{Cerezo2025}%
  \BibitemOpen
  \bibfield  {author} {\bibinfo {author} {\bibfnamefont {M.}~\bibnamefont {Cerezo}}, \bibinfo {author} {\bibfnamefont {M.}~\bibnamefont {Larocca}}, \bibinfo {author} {\bibfnamefont {D.}~\bibnamefont {Garc{\'{i}}a-Mart{\'{i}}n}}, \bibinfo {author} {\bibfnamefont {N.~L.}\ \bibnamefont {Diaz}}, \bibinfo {author} {\bibfnamefont {P.}~\bibnamefont {Braccia}}, \bibinfo {author} {\bibfnamefont {E.}~\bibnamefont {Fontana}}, \bibinfo {author} {\bibfnamefont {M.~S.}\ \bibnamefont {Rudolph}}, \bibinfo {author} {\bibfnamefont {P.}~\bibnamefont {Bermejo}}, \bibinfo {author} {\bibfnamefont {A.}~\bibnamefont {Ijaz}}, \bibinfo {author} {\bibfnamefont {S.}~\bibnamefont {Thanasilp}}, \bibinfo {author} {\bibfnamefont {E.~R.}\ \bibnamefont {Anschuetz}}, \ and\ \bibinfo {author} {\bibfnamefont {Z.}~\bibnamefont {Holmes}},\ }\bibfield  {title} {{Does provable absence of barren plateaus imply classical simulability?},\ }\href {\doibase 10.1038/s41467-025-63099-6} {\bibfield  {journal} {\bibinfo  {journal} {Nature Communications}\ }\textbf {\bibinfo {volume} {16}},\ \bibinfo {pages} {7907} (\bibinfo {year} {2025})}\BibitemShut {NoStop}%
\bibitem [{\citenamefont {White}(1992)}]{White1992}%
  \BibitemOpen
  \bibfield  {author} {\bibinfo {author} {\bibfnamefont {S.~R.}\ \bibnamefont {White}},\ }\bibfield  {title} {{Density matrix formulation for quantum renormalization groups},\ }\href {\doibase 10.1103/PhysRevLett.69.2863} {\bibfield  {journal} {\bibinfo  {journal} {Physical Review Letters}\ }\textbf {\bibinfo {volume} {69}},\ \bibinfo {pages} {2863} (\bibinfo {year} {1992})}\BibitemShut {NoStop}%
\bibitem [{\citenamefont {Schollw{\"{o}}ck}(2011)}]{Schollwock2011}%
  \BibitemOpen
  \bibfield  {author} {\bibinfo {author} {\bibfnamefont {U.}~\bibnamefont {Schollw{\"{o}}ck}},\ }\bibfield  {title} {{The density-matrix renormalization group in the age of matrix product states},\ }\href {\doibase 10.1016/j.aop.2010.09.012} {\bibfield  {journal} {\bibinfo  {journal} {Annals of Physics}\ }\textbf {\bibinfo {volume} {326}},\ \bibinfo {pages} {96} (\bibinfo {year} {2011})}\BibitemShut {NoStop}%
\bibitem [{\citenamefont {Or{\'{u}}s}(2014)}]{Orus2014}%
  \BibitemOpen
  \bibfield  {author} {\bibinfo {author} {\bibfnamefont {R.}~\bibnamefont {Or{\'{u}}s}},\ }\bibfield  {title} {{A practical introduction to tensor networks: Matrix product states and projected entangled pair states},\ }\href {\doibase 10.1016/j.aop.2014.06.013} {\bibfield  {journal} {\bibinfo  {journal} {Annals of Physics}\ }\textbf {\bibinfo {volume} {349}},\ \bibinfo {pages} {117} (\bibinfo {year} {2014})}\BibitemShut {NoStop}%
\bibitem [{\citenamefont {Or{\'{u}}s}(2019)}]{Orus2019}%
  \BibitemOpen
  \bibfield  {author} {\bibinfo {author} {\bibfnamefont {R.}~\bibnamefont {Or{\'{u}}s}},\ }\bibfield  {title} {{Tensor networks for complex quantum systems},\ }\href {\doibase 10.1038/s42254-019-0086-7} {\bibfield  {journal} {\bibinfo  {journal} {Nature Reviews Physics}\ }\textbf {\bibinfo {volume} {1}},\ \bibinfo {pages} {538} (\bibinfo {year} {2019})}\BibitemShut {NoStop}%
\bibitem [{\citenamefont {Cirac}\ \emph {et~al.}(2021)\citenamefont {Cirac}, \citenamefont {P{\'{e}}rez-Garc{\'{i}}a}, \citenamefont {Schuch},\ and\ \citenamefont {Verstraete}}]{Cirac2021}%
  \BibitemOpen
  \bibfield  {author} {\bibinfo {author} {\bibfnamefont {J.~I.}\ \bibnamefont {Cirac}}, \bibinfo {author} {\bibfnamefont {D.}~\bibnamefont {P{\'{e}}rez-Garc{\'{i}}a}}, \bibinfo {author} {\bibfnamefont {N.}~\bibnamefont {Schuch}}, \ and\ \bibinfo {author} {\bibfnamefont {F.}~\bibnamefont {Verstraete}},\ }\bibfield  {title} {{Matrix product states and projected entangled pair states: Concepts, symmetries, theorems},\ }\href {\doibase 10.1103/RevModPhys.93.045003} {\bibfield  {journal} {\bibinfo  {journal} {Reviews of Modern Physics}\ }\textbf {\bibinfo {volume} {93}},\ \bibinfo {pages} {045003} (\bibinfo {year} {2021})}\BibitemShut {NoStop}%
\bibitem [{\citenamefont {Xiang}(2023)}]{Xiang2023}%
  \BibitemOpen
  \bibfield  {author} {\bibinfo {author} {\bibfnamefont {T.}~\bibnamefont {Xiang}},\ }{Density Matrix and Tensor Network Renormalization}\ (\bibinfo  {publisher} {Cambridge University Press},\ \bibinfo {year} {2023})\BibitemShut {NoStop}%
\bibitem [{\citenamefont {Aharonov}\ \emph {et~al.}(2010)\citenamefont {Aharonov}, \citenamefont {Arad},\ and\ \citenamefont {Irani}}]{Aharonov2010}%
  \BibitemOpen
  \bibfield  {author} {\bibinfo {author} {\bibfnamefont {D.}~\bibnamefont {Aharonov}}, \bibinfo {author} {\bibfnamefont {I.}~\bibnamefont {Arad}}, \ and\ \bibinfo {author} {\bibfnamefont {S.}~\bibnamefont {Irani}},\ }\bibfield  {title} {{Efficient algorithm for approximating one-dimensional ground states},\ }\href {\doibase 10.1103/PhysRevA.82.012315} {\bibfield  {journal} {\bibinfo  {journal} {Physical Review A}\ }\textbf {\bibinfo {volume} {82}},\ \bibinfo {pages} {012315} (\bibinfo {year} {2010})}\BibitemShut {NoStop}%
\bibitem [{\citenamefont {Landau}\ \emph {et~al.}(2015)\citenamefont {Landau}, \citenamefont {Vazirani},\ and\ \citenamefont {Vidick}}]{Landau2015}%
  \BibitemOpen
  \bibfield  {author} {\bibinfo {author} {\bibfnamefont {Z.}~\bibnamefont {Landau}}, \bibinfo {author} {\bibfnamefont {U.}~\bibnamefont {Vazirani}}, \ and\ \bibinfo {author} {\bibfnamefont {T.}~\bibnamefont {Vidick}},\ }\bibfield  {title} {{A polynomial time algorithm for the ground state of one-dimensional gapped local Hamiltonians},\ }\href {\doibase 10.1038/nphys3345} {\bibfield  {journal} {\bibinfo  {journal} {Nature Physics}\ }\textbf {\bibinfo {volume} {11}},\ \bibinfo {pages} {566} (\bibinfo {year} {2015})}\BibitemShut {NoStop}%
\bibitem [{\citenamefont {Arad}\ \emph {et~al.}(2017)\citenamefont {Arad}, \citenamefont {Landau}, \citenamefont {Vazirani},\ and\ \citenamefont {Vidick}}]{Arad2017}%
  \BibitemOpen
  \bibfield  {author} {\bibinfo {author} {\bibfnamefont {I.}~\bibnamefont {Arad}}, \bibinfo {author} {\bibfnamefont {Z.}~\bibnamefont {Landau}}, \bibinfo {author} {\bibfnamefont {U.}~\bibnamefont {Vazirani}}, \ and\ \bibinfo {author} {\bibfnamefont {T.}~\bibnamefont {Vidick}},\ }\bibfield  {title} {{Rigorous RG Algorithms and Area Laws for Low Energy Eigenstates in 1D},\ }\href {\doibase 10.1007/s00220-017-2973-z} {\bibfield  {journal} {\bibinfo  {journal} {Communications in Mathematical Physics}\ }\textbf {\bibinfo {volume} {356}},\ \bibinfo {pages} {65} (\bibinfo {year} {2017})}\BibitemShut {NoStop}%
\bibitem [{\citenamefont {Haegeman}\ \emph {et~al.}(2011)\citenamefont {Haegeman}, \citenamefont {Cirac}, \citenamefont {Osborne}, \citenamefont {Pi{\v{z}}orn}, \citenamefont {Verschelde},\ and\ \citenamefont {Verstraete}}]{Haegeman2011}%
  \BibitemOpen
  \bibfield  {author} {\bibinfo {author} {\bibfnamefont {J.}~\bibnamefont {Haegeman}}, \bibinfo {author} {\bibfnamefont {J.~I.}\ \bibnamefont {Cirac}}, \bibinfo {author} {\bibfnamefont {T.~J.}\ \bibnamefont {Osborne}}, \bibinfo {author} {\bibfnamefont {I.}~\bibnamefont {Pi{\v{z}}orn}}, \bibinfo {author} {\bibfnamefont {H.}~\bibnamefont {Verschelde}}, \ and\ \bibinfo {author} {\bibfnamefont {F.}~\bibnamefont {Verstraete}},\ }\bibfield  {title} {{Time-Dependent Variational Principle for Quantum Lattices},\ }\href {\doibase 10.1103/PhysRevLett.107.070601} {\bibfield  {journal} {\bibinfo  {journal} {Physical Review Letters}\ }\textbf {\bibinfo {volume} {107}},\ \bibinfo {pages} {070601} (\bibinfo {year} {2011})}\BibitemShut {NoStop}%
\bibitem [{\citenamefont {Haegeman}\ \emph {et~al.}(2014)\citenamefont {Haegeman}, \citenamefont {Marien}, \citenamefont {Osborne},\ and\ \citenamefont {Verstraete}}]{Haegeman2014}%
  \BibitemOpen
  \bibfield  {author} {\bibinfo {author} {\bibfnamefont {J.}~\bibnamefont {Haegeman}}, \bibinfo {author} {\bibfnamefont {M.}~\bibnamefont {Marien}}, \bibinfo {author} {\bibfnamefont {T.~J.}\ \bibnamefont {Osborne}}, \ and\ \bibinfo {author} {\bibfnamefont {F.}~\bibnamefont {Verstraete}},\ }\bibfield  {title} {{Geometry of matrix product states: Metric, parallel transport, and curvature},\ }\href {https://pubs.aip.org/jmp/article/55/2/021902/232472/Geometry-of-matrix-product-states-Metric-parallel} {\bibfield  {journal} {\bibinfo  {journal} {Journal of Mathematical Physics}\ }\textbf {\bibinfo {volume} {55}} (\bibinfo {year} {2014})}\BibitemShut {NoStop}%
\bibitem [{\citenamefont {Haegeman}\ \emph {et~al.}(2016)\citenamefont {Haegeman}, \citenamefont {Lubich}, \citenamefont {Oseledets}, \citenamefont {Vandereycken},\ and\ \citenamefont {Verstraete}}]{Haegeman2016}%
  \BibitemOpen
  \bibfield  {author} {\bibinfo {author} {\bibfnamefont {J.}~\bibnamefont {Haegeman}}, \bibinfo {author} {\bibfnamefont {C.}~\bibnamefont {Lubich}}, \bibinfo {author} {\bibfnamefont {I.}~\bibnamefont {Oseledets}}, \bibinfo {author} {\bibfnamefont {B.}~\bibnamefont {Vandereycken}}, \ and\ \bibinfo {author} {\bibfnamefont {F.}~\bibnamefont {Verstraete}},\ }\bibfield  {title} {{Unifying time evolution and optimization with matrix product states},\ }\href {\doibase 10.1103/PhysRevB.94.165116} {\bibfield  {journal} {\bibinfo  {journal} {Physical Review B}\ }\textbf {\bibinfo {volume} {94}},\ \bibinfo {pages} {165116} (\bibinfo {year} {2016})}\BibitemShut {NoStop}%
\bibitem [{\citenamefont {Hauru}\ \emph {et~al.}(2021)\citenamefont {Hauru}, \citenamefont {{Van Damme}},\ and\ \citenamefont {Haegeman}}]{Hauru2021}%
  \BibitemOpen
  \bibfield  {author} {\bibinfo {author} {\bibfnamefont {M.}~\bibnamefont {Hauru}}, \bibinfo {author} {\bibfnamefont {M.}~\bibnamefont {{Van Damme}}}, \ and\ \bibinfo {author} {\bibfnamefont {J.}~\bibnamefont {Haegeman}},\ }\bibfield  {title} {{Riemannian optimization of isometric tensor networks},\ }\href {\doibase 10.21468/SciPostPhys.10.2.040} {\bibfield  {journal} {\bibinfo  {journal} {SciPost Physics}\ }\textbf {\bibinfo {volume} {10}},\ \bibinfo {pages} {040} (\bibinfo {year} {2021})}\BibitemShut {NoStop}%
\bibitem [{\citenamefont {Sch{\"{o}}n}\ \emph {et~al.}(2005)\citenamefont {Sch{\"{o}}n}, \citenamefont {Solano}, \citenamefont {Verstraete}, \citenamefont {Cirac},\ and\ \citenamefont {Wolf}}]{Schon2005}%
  \BibitemOpen
  \bibfield  {author} {\bibinfo {author} {\bibfnamefont {C.}~\bibnamefont {Sch{\"{o}}n}}, \bibinfo {author} {\bibfnamefont {E.}~\bibnamefont {Solano}}, \bibinfo {author} {\bibfnamefont {F.}~\bibnamefont {Verstraete}}, \bibinfo {author} {\bibfnamefont {J.~I.}\ \bibnamefont {Cirac}}, \ and\ \bibinfo {author} {\bibfnamefont {M.~M.}\ \bibnamefont {Wolf}},\ }\bibfield  {title} {{Sequential Generation of Entangled Multiqubit States},\ }\href {\doibase 10.1103/PhysRevLett.95.110503} {\bibfield  {journal} {\bibinfo  {journal} {Physical Review Letters}\ }\textbf {\bibinfo {volume} {95}},\ \bibinfo {pages} {110503} (\bibinfo {year} {2005})}\BibitemShut {NoStop}%
\bibitem [{\citenamefont {Sch{\"{o}}n}\ \emph {et~al.}(2007)\citenamefont {Sch{\"{o}}n}, \citenamefont {Hammerer}, \citenamefont {Wolf}, \citenamefont {Cirac},\ and\ \citenamefont {Solano}}]{Schon2007}%
  \BibitemOpen
  \bibfield  {author} {\bibinfo {author} {\bibfnamefont {C.}~\bibnamefont {Sch{\"{o}}n}}, \bibinfo {author} {\bibfnamefont {K.}~\bibnamefont {Hammerer}}, \bibinfo {author} {\bibfnamefont {M.~M.}\ \bibnamefont {Wolf}}, \bibinfo {author} {\bibfnamefont {J.~I.}\ \bibnamefont {Cirac}}, \ and\ \bibinfo {author} {\bibfnamefont {E.}~\bibnamefont {Solano}},\ }\bibfield  {title} {{Sequential generation of matrix-product states in cavity QED},\ }\href {\doibase 10.1103/PhysRevA.75.032311} {\bibfield  {journal} {\bibinfo  {journal} {Physical Review A}\ }\textbf {\bibinfo {volume} {75}},\ \bibinfo {pages} {032311} (\bibinfo {year} {2007})}\BibitemShut {NoStop}%
\bibitem [{\citenamefont {Ba{\~{n}}uls}\ \emph {et~al.}(2008)\citenamefont {Ba{\~{n}}uls}, \citenamefont {P{\'{e}}rez-Garc{\'{i}}a}, \citenamefont {Wolf}, \citenamefont {Verstraete},\ and\ \citenamefont {Cirac}}]{Banuls2008}%
  \BibitemOpen
  \bibfield  {author} {\bibinfo {author} {\bibfnamefont {M.~C.}\ \bibnamefont {Ba{\~{n}}uls}}, \bibinfo {author} {\bibfnamefont {D.}~\bibnamefont {P{\'{e}}rez-Garc{\'{i}}a}}, \bibinfo {author} {\bibfnamefont {M.~M.}\ \bibnamefont {Wolf}}, \bibinfo {author} {\bibfnamefont {F.}~\bibnamefont {Verstraete}}, \ and\ \bibinfo {author} {\bibfnamefont {J.~I.}\ \bibnamefont {Cirac}},\ }\bibfield  {title} {{Sequentially generated states for the study of two-dimensional systems},\ }\href {\doibase 10.1103/PhysRevA.77.052306} {\bibfield  {journal} {\bibinfo  {journal} {Physical Review A}\ }\textbf {\bibinfo {volume} {77}},\ \bibinfo {pages} {052306} (\bibinfo {year} {2008})}\BibitemShut {NoStop}%
\bibitem [{\citenamefont {Wei}\ \emph {et~al.}(2022)\citenamefont {Wei}, \citenamefont {Malz},\ and\ \citenamefont {Cirac}}]{Wei2022}%
  \BibitemOpen
  \bibfield  {author} {\bibinfo {author} {\bibfnamefont {Z.-Y.}\ \bibnamefont {Wei}}, \bibinfo {author} {\bibfnamefont {D.}~\bibnamefont {Malz}}, \ and\ \bibinfo {author} {\bibfnamefont {J.~I.}\ \bibnamefont {Cirac}},\ }\bibfield  {title} {{Sequential Generation of Projected Entangled-Pair States},\ }\href {\doibase 10.1103/PhysRevLett.128.010607} {\bibfield  {journal} {\bibinfo  {journal} {Physical Review Letters}\ }\textbf {\bibinfo {volume} {128}},\ \bibinfo {pages} {010607} (\bibinfo {year} {2022})}\BibitemShut {NoStop}%
\bibitem [{\citenamefont {Chen}\ \emph {et~al.}(2023)\citenamefont {Chen}, \citenamefont {Dua}, \citenamefont {Hermele}, \citenamefont {Stephen}, \citenamefont {Tantivasadakarn}, \citenamefont {Vanhove},\ and\ \citenamefont {Zhao}}]{Chen2023a}%
  \BibitemOpen
  \bibfield  {author} {\bibinfo {author} {\bibfnamefont {X.}~\bibnamefont {Chen}}, \bibinfo {author} {\bibfnamefont {A.}~\bibnamefont {Dua}}, \bibinfo {author} {\bibfnamefont {M.}~\bibnamefont {Hermele}}, \bibinfo {author} {\bibfnamefont {D.~T.}\ \bibnamefont {Stephen}}, \bibinfo {author} {\bibfnamefont {N.}~\bibnamefont {Tantivasadakarn}}, \bibinfo {author} {\bibfnamefont {R.}~\bibnamefont {Vanhove}}, \ and\ \bibinfo {author} {\bibfnamefont {J.-Y.}\ \bibnamefont {Zhao}},\ }\bibfield  {title} {{Sequential Quantum Circuits as Maps between Gapped Phases},\ }\href {http://arxiv.org/abs/2307.01267} {\bibfield  {journal} {\bibinfo  {journal} {arXiv:2307.01267}\ } (\bibinfo {year} {2023})}\BibitemShut {NoStop}%
\bibitem [{\citenamefont {Hastings}(2007)}]{Hastings2007}%
  \BibitemOpen
  \bibfield  {author} {\bibinfo {author} {\bibfnamefont {M.~B.}\ \bibnamefont {Hastings}},\ }\bibfield  {title} {{An area law for one-dimensional quantum systems},\ }\href {\doibase 10.1088/1742-5468/2007/08/P08024} {\bibfield  {journal} {\bibinfo  {journal} {Journal of Statistical Mechanics: Theory and Experiment}\ }\textbf {\bibinfo {volume} {2007}},\ \bibinfo {pages} {P08024} (\bibinfo {year} {2007})}\BibitemShut {NoStop}%
\bibitem [{\citenamefont {Eisert}\ \emph {et~al.}(2010)\citenamefont {Eisert}, \citenamefont {Cramer},\ and\ \citenamefont {Plenio}}]{Eisert2010}%
  \BibitemOpen
  \bibfield  {author} {\bibinfo {author} {\bibfnamefont {J.}~\bibnamefont {Eisert}}, \bibinfo {author} {\bibfnamefont {M.}~\bibnamefont {Cramer}}, \ and\ \bibinfo {author} {\bibfnamefont {M.~B.}\ \bibnamefont {Plenio}},\ }\bibfield  {title} {{Colloquium : Area laws for the entanglement entropy},\ }\href {\doibase 10.1103/RevModPhys.82.277} {\bibfield  {journal} {\bibinfo  {journal} {Reviews of Modern Physics}\ }\textbf {\bibinfo {volume} {82}},\ \bibinfo {pages} {277} (\bibinfo {year} {2010})}\BibitemShut {NoStop}%
\bibitem [{SM()}]{SM}%
  \BibitemOpen
  \href@noop {} {{See the Supplemental Material for preliminaries on matrix product states and the Weingarten calculus, rigorous theorem proofs, and additional numerical results, which includes Ref.\,\cite{Drury2008, Kutschan2018, Zaletel2020, Zhang2022_z, Zhang2026}.}}\BibitemShut {Stop}%
\bibitem [{\citenamefont {Collins}\ and\ \citenamefont {{\'{S}}niady}(2006)}]{Collins2006}%
  \BibitemOpen
  \bibfield  {author} {\bibinfo {author} {\bibfnamefont {B.}~\bibnamefont {Collins}}\ and\ \bibinfo {author} {\bibfnamefont {P.}~\bibnamefont {{\'{S}}niady}},\ }\bibfield  {title} {{Integration with Respect to the Haar Measure on Unitary, Orthogonal and Symplectic Group},\ }\href {\doibase 10.1007/s00220-006-1554-3} {\bibfield  {journal} {\bibinfo  {journal} {Communications in Mathematical Physics}\ }\textbf {\bibinfo {volume} {264}},\ \bibinfo {pages} {773} (\bibinfo {year} {2006})}\BibitemShut {NoStop}%
\bibitem [{\citenamefont {Collins}\ and\ \citenamefont {Nechita}(2016)}]{Collins2016}%
  \BibitemOpen
  \bibfield  {author} {\bibinfo {author} {\bibfnamefont {B.}~\bibnamefont {Collins}}\ and\ \bibinfo {author} {\bibfnamefont {I.}~\bibnamefont {Nechita}},\ }\bibfield  {title} {{Random matrix techniques in quantum information theory},\ }\href {\doibase 10.1063/1.4936880} {\bibfield  {journal} {\bibinfo  {journal} {Journal of Mathematical Physics}\ }\textbf {\bibinfo {volume} {57}} (\bibinfo {year} {2016}),\ 10.1063/1.4936880}\BibitemShut {NoStop}%
\bibitem [{\citenamefont {Nahum}\ \emph {et~al.}(2017)\citenamefont {Nahum}, \citenamefont {Ruhman}, \citenamefont {Vijay},\ and\ \citenamefont {Haah}}]{Nahum2017}%
  \BibitemOpen
  \bibfield  {author} {\bibinfo {author} {\bibfnamefont {A.}~\bibnamefont {Nahum}}, \bibinfo {author} {\bibfnamefont {J.}~\bibnamefont {Ruhman}}, \bibinfo {author} {\bibfnamefont {S.}~\bibnamefont {Vijay}}, \ and\ \bibinfo {author} {\bibfnamefont {J.}~\bibnamefont {Haah}},\ }\bibfield  {title} {{Quantum entanglement growth under random unitary dynamics},\ }\href {\doibase 10.1103/PhysRevX.7.031016} {\bibfield  {journal} {\bibinfo  {journal} {Physical Review X}\ }\textbf {\bibinfo {volume} {7}},\ \bibinfo {pages} {1} (\bibinfo {year} {2017})}\BibitemShut {NoStop}%
\bibitem [{\citenamefont {Zhou}\ and\ \citenamefont {Nahum}(2019)}]{Zhou2019}%
  \BibitemOpen
  \bibfield  {author} {\bibinfo {author} {\bibfnamefont {T.}~\bibnamefont {Zhou}}\ and\ \bibinfo {author} {\bibfnamefont {A.}~\bibnamefont {Nahum}},\ }\bibfield  {title} {{Emergent statistical mechanics of entanglement in random unitary circuits},\ }\href {\doibase 10.1103/PhysRevB.99.174205} {\bibfield  {journal} {\bibinfo  {journal} {Physical Review B}\ }\textbf {\bibinfo {volume} {99}},\ \bibinfo {pages} {174205} (\bibinfo {year} {2019})}\BibitemShut {NoStop}%
\bibitem [{\citenamefont {Fisher}\ \emph {et~al.}(2023)\citenamefont {Fisher}, \citenamefont {Khemani}, \citenamefont {Nahum},\ and\ \citenamefont {Vijay}}]{Fisher2022}%
  \BibitemOpen
  \bibfield  {author} {\bibinfo {author} {\bibfnamefont {M.~P.}\ \bibnamefont {Fisher}}, \bibinfo {author} {\bibfnamefont {V.}~\bibnamefont {Khemani}}, \bibinfo {author} {\bibfnamefont {A.}~\bibnamefont {Nahum}}, \ and\ \bibinfo {author} {\bibfnamefont {S.}~\bibnamefont {Vijay}},\ }\bibfield  {title} {{Random Quantum Circuits},\ }\href {\doibase 10.1146/annurev-conmatphys-031720-030658} {\bibfield  {journal} {\bibinfo  {journal} {Annual Review of Condensed Matter Physics}\ }\textbf {\bibinfo {volume} {14}},\ \bibinfo {pages} {335} (\bibinfo {year} {2023})}\BibitemShut {NoStop}%
\bibitem [{\citenamefont {Lami}\ \emph {et~al.}(2025{\natexlab{a}})\citenamefont {Lami}, \citenamefont {{De Nardis}},\ and\ \citenamefont {Turkeshi}}]{Lami2025a}%
  \BibitemOpen
  \bibfield  {author} {\bibinfo {author} {\bibfnamefont {G.}~\bibnamefont {Lami}}, \bibinfo {author} {\bibfnamefont {J.}~\bibnamefont {{De Nardis}}}, \ and\ \bibinfo {author} {\bibfnamefont {X.}~\bibnamefont {Turkeshi}},\ }\bibfield  {title} {{Anticoncentration and State Design of Random Tensor Networks},\ }\href {\doibase 10.1103/PhysRevLett.134.010401} {\bibfield  {journal} {\bibinfo  {journal} {Physical Review Letters}\ }\textbf {\bibinfo {volume} {134}},\ \bibinfo {pages} {010401} (\bibinfo {year} {2025}{\natexlab{a}})}\BibitemShut {NoStop}%
\bibitem [{\citenamefont {Lami}\ \emph {et~al.}(2025{\natexlab{b}})\citenamefont {Lami}, \citenamefont {{De Luca}}, \citenamefont {Turkeshi},\ and\ \citenamefont {{De Nardis}}}]{Lami2025}%
  \BibitemOpen
  \bibfield  {author} {\bibinfo {author} {\bibfnamefont {G.}~\bibnamefont {Lami}}, \bibinfo {author} {\bibfnamefont {A.}~\bibnamefont {{De Luca}}}, \bibinfo {author} {\bibfnamefont {X.}~\bibnamefont {Turkeshi}}, \ and\ \bibinfo {author} {\bibfnamefont {J.}~\bibnamefont {{De Nardis}}},\ }\bibfield  {title} {{Quantum State Design and Emergent Confinement Mechanism in Measured Tensor Network States},\ }\href {http://arxiv.org/abs/2504.16995} {\bibfield  {journal} {\bibinfo  {journal} {arXiv:2504.16995}\ } (\bibinfo {year} {2025}{\natexlab{b}})}\BibitemShut {NoStop}%
\bibitem [{\citenamefont {Sauliere}\ \emph {et~al.}(2026)\citenamefont {Sauliere}, \citenamefont {Lami}, \citenamefont {Boyer}, \citenamefont {{De Nardis}},\ and\ \citenamefont {{De Luca}}}]{Sauliere2026}%
  \BibitemOpen
  \bibfield  {author} {\bibinfo {author} {\bibfnamefont {A.}~\bibnamefont {Sauliere}}, \bibinfo {author} {\bibfnamefont {G.}~\bibnamefont {Lami}}, \bibinfo {author} {\bibfnamefont {C.}~\bibnamefont {Boyer}}, \bibinfo {author} {\bibfnamefont {J.}~\bibnamefont {{De Nardis}}}, \ and\ \bibinfo {author} {\bibfnamefont {A.}~\bibnamefont {{De Luca}}},\ }\bibfield  {title} {{Universality in the Anticoncentration of Noisy Quantum Circuits at Finite Depths},\ }\href {\doibase 10.1103/xl16-cdy9} {\bibfield  {journal} {\bibinfo  {journal} {PRX Quantum}\ ,\ \bibinfo {pages} {1}} (\bibinfo {year} {2026})}\BibitemShut {NoStop}%
\bibitem [{\citenamefont {Dowling}\ \emph {et~al.}(2026)\citenamefont {Dowling}, \citenamefont {Turkeshi}, \citenamefont {{De Nardis}},\ and\ \citenamefont {Lami}}]{Dowling2026}%
  \BibitemOpen
  \bibfield  {author} {\bibinfo {author} {\bibfnamefont {N.}~\bibnamefont {Dowling}}, \bibinfo {author} {\bibfnamefont {X.}~\bibnamefont {Turkeshi}}, \bibinfo {author} {\bibfnamefont {J.}~\bibnamefont {{De Nardis}}}, \ and\ \bibinfo {author} {\bibfnamefont {G.}~\bibnamefont {Lami}},\ }\bibfield  {title} {{Noise-induced Simulability Transition from Operator Scrambling},\ }\href {http://arxiv.org/abs/2605.18943} {\bibfield  {journal} {\bibinfo  {journal} {arXiv:2605.18943}\ } (\bibinfo {year} {2026})}\BibitemShut {NoStop}%
\bibitem [{\citenamefont {Stokes}\ \emph {et~al.}(2020)\citenamefont {Stokes}, \citenamefont {Izaac}, \citenamefont {Killoran},\ and\ \citenamefont {Carleo}}]{Stokes2020}%
  \BibitemOpen
  \bibfield  {author} {\bibinfo {author} {\bibfnamefont {J.}~\bibnamefont {Stokes}}, \bibinfo {author} {\bibfnamefont {J.}~\bibnamefont {Izaac}}, \bibinfo {author} {\bibfnamefont {N.}~\bibnamefont {Killoran}}, \ and\ \bibinfo {author} {\bibfnamefont {G.}~\bibnamefont {Carleo}},\ }\bibfield  {title} {{Quantum Natural Gradient},\ }\href {\doibase 10.22331/q-2020-05-25-269} {\bibfield  {journal} {\bibinfo  {journal} {Quantum}\ }\textbf {\bibinfo {volume} {4}},\ \bibinfo {pages} {269} (\bibinfo {year} {2020})}\BibitemShut {NoStop}%
\bibitem [{\citenamefont {Aharonov}\ \emph {et~al.}(2009)\citenamefont {Aharonov}, \citenamefont {Gottesman}, \citenamefont {Irani},\ and\ \citenamefont {Kempe}}]{Aharonov2009}%
  \BibitemOpen
  \bibfield  {author} {\bibinfo {author} {\bibfnamefont {D.}~\bibnamefont {Aharonov}}, \bibinfo {author} {\bibfnamefont {D.}~\bibnamefont {Gottesman}}, \bibinfo {author} {\bibfnamefont {S.}~\bibnamefont {Irani}}, \ and\ \bibinfo {author} {\bibfnamefont {J.}~\bibnamefont {Kempe}},\ }\bibfield  {title} {{The power of quantum systems on a line},\ }\href {\doibase 10.1007/s00220-008-0710-3} {\bibfield  {journal} {\bibinfo  {journal} {Communications in Mathematical Physics}\ }\textbf {\bibinfo {volume} {287}},\ \bibinfo {pages} {41} (\bibinfo {year} {2009})}\BibitemShut {NoStop}%
\bibitem [{\citenamefont {Larocca}\ \emph {et~al.}(2023)\citenamefont {Larocca}, \citenamefont {Ju}, \citenamefont {Garc{\'{i}}a-Mart{\'{i}}n}, \citenamefont {Coles},\ and\ \citenamefont {Cerezo}}]{Larocca2023}%
  \BibitemOpen
  \bibfield  {author} {\bibinfo {author} {\bibfnamefont {M.}~\bibnamefont {Larocca}}, \bibinfo {author} {\bibfnamefont {N.}~\bibnamefont {Ju}}, \bibinfo {author} {\bibfnamefont {D.}~\bibnamefont {Garc{\'{i}}a-Mart{\'{i}}n}}, \bibinfo {author} {\bibfnamefont {P.~J.}\ \bibnamefont {Coles}}, \ and\ \bibinfo {author} {\bibfnamefont {M.}~\bibnamefont {Cerezo}},\ }\bibfield  {title} {{Theory of overparametrization in quantum neural networks},\ }\href {\doibase 10.1038/s43588-023-00467-6} {\bibfield  {journal} {\bibinfo  {journal} {Nature Computational Science}\ }\textbf {\bibinfo {volume} {3}},\ \bibinfo {pages} {542} (\bibinfo {year} {2023})}\BibitemShut {NoStop}%
\bibitem [{\citenamefont {Drury}\ and\ \citenamefont {Love}(2008)}]{Drury2008}%
  \BibitemOpen
  \bibfield  {author} {\bibinfo {author} {\bibfnamefont {B.}~\bibnamefont {Drury}}\ and\ \bibinfo {author} {\bibfnamefont {P.}~\bibnamefont {Love}},\ }\bibfield  {title} {{Constructive quantum Shannon decomposition from Cartan involutions},\ }\href {\doibase 10.1088/1751-8113/41/39/395305} {\bibfield  {journal} {\bibinfo  {journal} {Journal of Physics A: Mathematical and Theoretical}\ }\textbf {\bibinfo {volume} {41}},\ \bibinfo {pages} {395305} (\bibinfo {year} {2008})}\BibitemShut {NoStop}%
\bibitem [{\citenamefont {Kutschan}(2018)}]{Kutschan2018}%
  \BibitemOpen
  \bibfield  {author} {\bibinfo {author} {\bibfnamefont {B.}~\bibnamefont {Kutschan}},\ }\bibfield  {title} {{Tangent cones to tensor train varieties},\ }\href {\doibase 10.1016/j.laa.2018.01.012} {\bibfield  {journal} {\bibinfo  {journal} {Linear Algebra and its Applications}\ }\textbf {\bibinfo {volume} {544}},\ \bibinfo {pages} {370} (\bibinfo {year} {2018})}\BibitemShut {NoStop}%
\bibitem [{\citenamefont {Zaletel}\ and\ \citenamefont {Pollmann}(2020)}]{Zaletel2020}%
  \BibitemOpen
  \bibfield  {author} {\bibinfo {author} {\bibfnamefont {M.~P.}\ \bibnamefont {Zaletel}}\ and\ \bibinfo {author} {\bibfnamefont {F.}~\bibnamefont {Pollmann}},\ }\bibfield  {title} {{Isometric Tensor Network States in Two Dimensions},\ }\href {\doibase 10.1103/PhysRevLett.124.037201} {\bibfield  {journal} {\bibinfo  {journal} {Physical Review Letters}\ }\textbf {\bibinfo {volume} {124}},\ \bibinfo {pages} {037201} (\bibinfo {year} {2020})}\BibitemShut {NoStop}%
\bibitem [{\citenamefont {Zhang}\ \emph {et~al.}(2023{\natexlab{b}})\citenamefont {Zhang}, \citenamefont {Allcock}, \citenamefont {Wan}, \citenamefont {Liu}, \citenamefont {Sun}, \citenamefont {Yu}, \citenamefont {Yang}, \citenamefont {Qiu}, \citenamefont {Ye}, \citenamefont {Chen}, \citenamefont {Lee}, \citenamefont {Zheng}, \citenamefont {Jian}, \citenamefont {Yao}, \citenamefont {Hsieh},\ and\ \citenamefont {Zhang}}]{Zhang2022_z}%
  \BibitemOpen
  \bibfield  {author} {\bibinfo {author} {\bibfnamefont {S.-X.}\ \bibnamefont {Zhang}}, \bibinfo {author} {\bibfnamefont {J.}~\bibnamefont {Allcock}}, \bibinfo {author} {\bibfnamefont {Z.-Q.}\ \bibnamefont {Wan}}, \bibinfo {author} {\bibfnamefont {S.}~\bibnamefont {Liu}}, \bibinfo {author} {\bibfnamefont {J.}~\bibnamefont {Sun}}, \bibinfo {author} {\bibfnamefont {H.}~\bibnamefont {Yu}}, \bibinfo {author} {\bibfnamefont {X.-H.}\ \bibnamefont {Yang}}, \bibinfo {author} {\bibfnamefont {J.}~\bibnamefont {Qiu}}, \bibinfo {author} {\bibfnamefont {Z.}~\bibnamefont {Ye}}, \bibinfo {author} {\bibfnamefont {Y.-Q.}\ \bibnamefont {Chen}}, \bibinfo {author} {\bibfnamefont {C.-K.}\ \bibnamefont {Lee}}, \bibinfo {author} {\bibfnamefont {Y.-C.}\ \bibnamefont {Zheng}}, \bibinfo {author} {\bibfnamefont {S.-K.}\ \bibnamefont {Jian}}, \bibinfo {author} {\bibfnamefont {H.}~\bibnamefont {Yao}}, \bibinfo {author} {\bibfnamefont {C.-Y.}\ \bibnamefont {Hsieh}}, \ and\ \bibinfo {author} {\bibfnamefont {S.}~\bibnamefont {Zhang}},\ }\bibfield  {title} {{TensorCircuit: a Quantum Software Framework for the NISQ Era},\ }\href {\doibase 10.22331/q-2023-02-02-912} {\bibfield  {journal} {\bibinfo  {journal} {Quantum}\ }\textbf {\bibinfo {volume} {7}},\ \bibinfo {pages} {912} (\bibinfo {year} {2023}{\natexlab{b}})}\BibitemShut {NoStop}%
\bibitem [{\citenamefont {Zhang}\ \emph {et~al.}(2026)\citenamefont {Zhang}, \citenamefont {Chen}, \citenamefont {Li}, \citenamefont {Sun}, \citenamefont {Ma}, \citenamefont {Zheng}, \citenamefont {Huang}, \citenamefont {Wang}, \citenamefont {Yu}, \citenamefont {Li}, \citenamefont {Huang}, \citenamefont {Li}, \citenamefont {Wan}, \citenamefont {Liu}, \citenamefont {Qiu}, \citenamefont {Miao}, \citenamefont {Song}, \citenamefont {Yan}, \citenamefont {Tsuoka}, \citenamefont {Zhang}, \citenamefont {Wang}, \citenamefont {Fan}, \citenamefont {Hsieh}, \citenamefont {Yao},\ and\ \citenamefont {Xiang}}]{Zhang2026}%
  \BibitemOpen
  \bibfield  {author} {\bibinfo {author} {\bibfnamefont {S.-X.}\ \bibnamefont {Zhang}}, \bibinfo {author} {\bibfnamefont {Y.-Q.}\ \bibnamefont {Chen}}, \bibinfo {author} {\bibfnamefont {W.}~\bibnamefont {Li}}, \bibinfo {author} {\bibfnamefont {J.}~\bibnamefont {Sun}}, \bibinfo {author} {\bibfnamefont {W.-G.}\ \bibnamefont {Ma}}, \bibinfo {author} {\bibfnamefont {P.-L.}\ \bibnamefont {Zheng}}, \bibinfo {author} {\bibfnamefont {Y.-X.}\ \bibnamefont {Huang}}, \bibinfo {author} {\bibfnamefont {Q.-X.}\ \bibnamefont {Wang}}, \bibinfo {author} {\bibfnamefont {H.}~\bibnamefont {Yu}}, \bibinfo {author} {\bibfnamefont {Z.}~\bibnamefont {Li}}, \bibinfo {author} {\bibfnamefont {X.}~\bibnamefont {Huang}}, \bibinfo {author} {\bibfnamefont {Z.-L.}\ \bibnamefont {Li}}, \bibinfo {author} {\bibfnamefont {Z.-Q.}\ \bibnamefont {Wan}}, \bibinfo {author} {\bibfnamefont {S.}~\bibnamefont {Liu}}, \bibinfo {author} {\bibfnamefont {J.}~\bibnamefont {Qiu}}, \bibinfo {author} {\bibfnamefont {J.}~\bibnamefont {Miao}}, \bibinfo {author} {\bibfnamefont {Z.}~\bibnamefont {Song}}, \bibinfo {author} {\bibfnamefont {Y.}~\bibnamefont {Yan}}, \bibinfo {author} {\bibfnamefont {K.}~\bibnamefont {Tsuoka}}, \bibinfo {author} {\bibfnamefont {P.}~\bibnamefont {Zhang}}, \bibinfo {author} {\bibfnamefont {L.}~\bibnamefont {Wang}}, \bibinfo {author} {\bibfnamefont {H.}~\bibnamefont {Fan}}, \bibinfo {author} {\bibfnamefont {C.-Y.}\ \bibnamefont {Hsieh}}, \bibinfo {author} {\bibfnamefont {H.}~\bibnamefont {Yao}}, \ and\ \bibinfo {author} {\bibfnamefont {T.}~\bibnamefont {Xiang}},\ }\bibfield  {title} {{TensorCircuit-NG: A Universal, Composable, and Scalable Platform for Quantum Computing and Quantum Simulation},\ }\href {http://arxiv.org/abs/2602.14167} {\bibfield  {journal} {\bibinfo  {journal} {arXiv:2602.14167}\ } (\bibinfo {year} {2026})}\BibitemShut {NoStop}%
\end{thebibliography}%

\clearpage
\newpage
\widetext

\begin{center}
\textbf{\large Supplemental Material for \\``Absence of poor local minima in matrix product states''}
\end{center}

\addtocontents{toc}{\protect\setcounter{tocdepth}{0}}
{
\tableofcontents
}

\renewcommand{\theproposition}{S\arabic{proposition}}
\setcounter{proposition}{0}
\renewcommand{\thedefinition}{S\arabic{definition}}
\setcounter{definition}{0}

\renewcommand{\thefigure}{S\arabic{figure}}
\setcounter{figure}{0}
\renewcommand{\theequation}{S\arabic{equation}}
\setcounter{equation}{0}
\renewcommand{\thesection}{\Roman{section}}
\setcounter{section}{0}
\renewcommand{\thetable}{S\arabic{table}}
\setcounter{table}{0}
\setcounter{secnumdepth}{4}


\section{Preliminaries and basic setup}

\subsection{Haar integral and Weingarten calculus}

To analytically calculate integrals over unitaries with respect to the Haar measure, we introduce the Weingarten formula~\cite{Collins2006}, which provides an analytical expression for the $t$-degree twirling channel 
\begin{equation}\label{eq:weingarten}
    \mathcal{T}^{(t)}(\cdot) = \int \mathrm{d}\mu(U)~ U^{\otimes t} (\cdot) U^{\dagger\otimes t} = \sum_{\sigma,\tau\in \mathcal{S}_{t}} W^{(t)}_{\sigma,\tau}(d) ~\opr{Tr} \left[S_\tau^\dagger (\cdot)\right] S_\sigma,
\end{equation}
where ``$\cdot$'' is the placeholder for arbitrary linear operators. $\sigma$ and $\tau$ are elements of the $t$-degree symmetric permutation group $\mathcal{S}_t$. The coefficient
\begin{equation}
    W^{(t)}_{\sigma,\tau}(d) = \opr{Wg}(\sigma^{-1}\tau, d),
\end{equation}
is the $t$-degree Weingarten function of the permutation $\sigma^{-1} \tau$ and the $d$-dimensional unitary group $\mathcal{U}(d)$. $S_\sigma=S_\sigma(d)$ represents the permutation operator which permutes the indices of the $t$ replicas according to the permutation $\sigma$, e.g., the identity and the SWAP operator for $t=2$. The validity of Eq.\,\eqref{eq:weingarten} can be partially understood by the fact that the translational invariance of the Haar measure requires the integral result commutes with an arbitrary $t$-fold unitary operator $V^{\otimes t}$, which is naturally satisfied by the permutation operators and their linear combinations.

After reshaping the twirling channel to its Liouville superoperator representation, the Weingarten formula can be expressed as
\begin{equation}\label{eq:weingarten_Liouville}
    \int \mathrm{d}\mu(U)~ U^{\otimes t} \otimes U^{* \otimes t} = \sum_{\sigma,\tau\in \mathcal{S}_{t}} W^{(t)}_{\sigma,\tau} (d) \ketbra{\sigma}{\tau},
\end{equation}
where $\ket{\sigma}=(S_\sigma \otimes I)\ket{I}$ denotes the vector obtained by reshaping the permutation operator $S_\sigma$, where $I$ represents the identity operator of the same dimension as $S_\sigma$ and $\ket{I}=\left(\sum_{i=1}^{d}\ket{i}\otimes \ket{i}\right)^{\otimes t}$ represents an unnormalized Bell state. The correspondence before and after the reshaping procedure is shown in the following tensor diagrams, where we set $t=4$ as an example
\begin{equation}
\begin{mytikz3}
\draw  [fill={rgb, 255:red, 245; green, 240; blue, 235 }  ,fill opacity=1 ][line width=0.5]  (120,70) -- (160,70) -- (160,110) -- (120,110) -- cycle ;
\draw [line width=0.5]    (100,88) -- (120,88) ;
\draw [line width=0.5]    (160,88) -- (180,88) ;
\draw [line width=0.5]    (100,92) -- (120,92) ;
\draw [line width=0.5]    (160,92) -- (180,92) ;
\draw [line width=0.5]    (100,96) -- (120,96) ;
\draw [line width=0.5]    (160,96) -- (180,96) ;
\draw [line width=0.5]    (100,84) -- (120,84) ;
\draw [line width=0.5]    (160,84) -- (180,84) ;

\draw (142,91) node   [align=left] {$ S_{\sigma }$};
\end{mytikz3}
\quad\longrightarrow\quad
\begin{mytikz3}
\draw  [fill={rgb, 255:red, 245; green, 240; blue, 235 }  ,fill opacity=1 ][line width=0.5]  (330.03,50) -- (370.03,50) -- (370.03,90) -- (330.03,90) -- cycle ;
\draw [line width=0.5]    (310,68) -- (330.03,68) ;
\draw [line width=0.5]    (370.03,68) -- (390.03,68) ;
\draw [line width=0.5]    (310.03,128) -- (390.03,128) ;
\draw [line width=0.5]    (390.03,128) .. controls (430.35,128.08) and (430.35,67.68) .. (390.03,68) ;
\draw [line width=0.5]    (310,72) -- (330.03,72) ;
\draw [line width=0.5]    (370.03,72) -- (390.03,72) ;
\draw [line width=0.5]    (310.03,132) -- (390.03,132) ;
\draw [line width=0.5]    (390.03,132) .. controls (430.35,132.08) and (430.35,71.68) .. (390.03,72) ;
\draw [line width=0.5]    (310,64) -- (330.03,64) ;
\draw [line width=0.5]    (370.03,64) -- (390.03,64) ;
\draw [line width=0.5]    (310.03,124) -- (390.03,124) ;
\draw [line width=0.5]    (390.03,124) .. controls (430.35,124.08) and (430.35,63.68) .. (390.03,64) ;
\draw [line width=0.5]    (310,76) -- (330.03,76) ;
\draw [line width=0.5]    (370.03,76) -- (390.03,76) ;
\draw [line width=0.5]    (310.03,136) -- (390.03,136) ;
\draw [line width=0.5]    (390.03,136) .. controls (430.35,136.08) and (430.35,75.68) .. (390.03,76) ;

\draw (352,71) node   [align=left] {$ S_{\sigma }$};
\end{mytikz3}
\quad=\quad
\begin{mytikz3}
\draw [line width=0.5]    (550,68) -- (570,68) ;
\draw [line width=0.5]    (550,128) -- (570,128) ;
\draw [line width=0.5]    (550,72) -- (570,72) ;
\draw [line width=0.5]    (550,132) -- (570,132) ;
\draw  [fill={rgb, 255:red, 245; green, 240; blue, 235 }  ,fill opacity=1 ][line width=0.5]  (610,100) -- (570,140) -- (570,60) -- cycle ;
\draw [line width=0.5]    (550,76) -- (570,76) ;
\draw [line width=0.5]    (550,64) -- (570,64) ;
\draw [line width=0.5]    (550,124) -- (570,124) ;
\draw [line width=0.5]    (550,136) -- (570,136) ;

\draw (586,100) node   [align=left] {$\ket{\sigma }$};
\end{mytikz3}
\quad,
\end{equation}

\begin{equation}
\begin{mytikz3}
\draw  [fill={rgb, 255:red, 245; green, 240; blue, 235 }  ,fill opacity=1 ][line width=0.5]  (120,70) -- (160,70) -- (160,110) -- (120,110) -- cycle ;
\draw [line width=0.5]    (110,88) -- (120,88) ;
\draw [line width=0.5]    (160,88) -- (179.95,88) ;
\draw [line width=0.5]    (110,92) -- (120,92) ;
\draw [line width=0.5]    (160,92) -- (179.95,92) ;
\draw [line width=0.5]    (110,96) -- (120,96) ;
\draw [line width=0.5]    (160,96) -- (179.95,96) ;
\draw [line width=0.5]    (110,84) -- (120,84) ;
\draw [line width=0.5]    (160,84) -- (179.95,84) ;
\draw [line width=0.5]    (110.16,148) .. controls (69.95,148.08) and (69.95,87.68) .. (110.16,88) ;
\draw [line width=0.5]    (110.16,152) .. controls (69.95,152.08) and (69.95,91.68) .. (110.16,92) ;
\draw [line width=0.5]    (110.16,144) .. controls (69.95,144.08) and (69.95,83.68) .. (110.16,84) ;
\draw [line width=0.5]    (110.16,156) .. controls (69.95,156.08) and (69.95,95.68) .. (110.16,96) ;

\draw [line width=0.5]    (110,148) -- (230,148) ;
\draw [line width=0.5]    (110,152) -- (230,152) ;
\draw [line width=0.5]    (110,156) -- (230,156) ;
\draw [line width=0.5]    (110,144) -- (230,144) ;
\draw [line width=0.5]    (230,148) .. controls (270.32,148.08) and (270.32,87.68) .. (230,88) ;
\draw [line width=0.5]    (230,152) .. controls (270.32,152.08) and (270.32,91.68) .. (230,92) ;
\draw [line width=0.5]    (230,144) .. controls (270.32,144.08) and (270.32,83.68) .. (230,84) ;
\draw [line width=0.5]    (230,156) .. controls (270.32,156.08) and (270.32,95.68) .. (230,96) ;
\draw [line width=0.5]    (220.03,88) -- (229.95,88) ;
\draw [line width=0.5]    (220.03,92) -- (229.95,92) ;
\draw [line width=0.5]    (220.03,96) -- (229.95,96) ;
\draw [line width=0.5]    (220.03,84) -- (229.95,84) ;

\draw (140,90) node   [align=left] {$ S_{\tau }^{\dagger }$};
\end{mytikz3}
\quad\longrightarrow\quad
\begin{mytikz3}
\draw  [fill={rgb, 255:red, 245; green, 240; blue, 235 }  ,fill opacity=1 ][line width=0.5]  (340,130) -- (380,130) -- (380,170) -- (340,170) -- cycle ;
\draw [line width=0.5]    (330,148) -- (340,148) ;
\draw [line width=0.5]    (330,88) -- (399.96,88) ;
\draw [line width=0.5]    (330,152) -- (340,152) ;
\draw [line width=0.5]    (330,92) -- (399.96,92) ;
\draw [line width=0.5]    (330,156) -- (340,156) ;
\draw [line width=0.5]    (330,96) -- (399.96,96) ;
\draw [line width=0.5]    (330,144) -- (340,144) ;
\draw [line width=0.5]    (330,84) -- (399.96,84) ;
\draw [line width=0.5]    (330.16,148) .. controls (289.95,148.08) and (289.95,87.68) .. (330.16,88) ;
\draw [line width=0.5]    (330.16,152) .. controls (289.95,152.08) and (289.95,91.68) .. (330.16,92) ;
\draw [line width=0.5]    (330.16,144) .. controls (289.95,144.08) and (289.95,83.68) .. (330.16,84) ;
\draw [line width=0.5]    (330.16,156) .. controls (289.95,156.08) and (289.95,95.68) .. (330.16,96) ;

\draw [line width=0.5]    (380,148) -- (399.96,148) ;
\draw [line width=0.5]    (380,152) -- (399.96,152) ;
\draw [line width=0.5]    (380,156) -- (399.96,156) ;
\draw [line width=0.5]    (380,144) -- (399.96,144) ;

\draw (360,151) node   [align=left] {$ S_{\tau }^{*}$};
\end{mytikz3}
\quad = \quad
\begin{mytikz3}
\draw [line width=0.5]    (560,88) -- (579.87,88) ;
\draw [line width=0.5]    (560,148) -- (579.95,148) ;
\draw [line width=0.5]    (560,92) -- (579.87,92) ;
\draw [line width=0.5]    (560,152) -- (579.95,152) ;
\draw  [fill={rgb, 255:red, 245; green, 240; blue, 235 }  ,fill opacity=1 ][line width=0.5]  (520,120) -- (560,160) -- (560,80) -- cycle ;
\draw [line width=0.5]    (560,96) -- (579.87,96) ;
\draw [line width=0.5]    (560,84) -- (579.87,84) ;
\draw [line width=0.5]    (560,144) -- (579.95,144) ;
\draw [line width=0.5]    (560,156) -- (579.95,156) ;

\draw (545,120) node   [align=left] {$\bra{\tau}$};
\end{mytikz3}
\quad.
\end{equation}
Thus, the Weingarten coefficients on the right-hand side of Eq.\,\eqref{eq:weingarten_Liouville} can be seen as a matrix representation of the Liouville superoperator of the twirling channel under the non-orthonormal basis consisting of the permutation vectors. The tensor network diagram for the Weingarten formula in Eq.\,\eqref{eq:weingarten_Liouville} can be expressed as
\begin{equation}\label{eq:weingarten_tn}
\int \mathrm{d}\mu(U)~
\begin{mytikz3}
\draw  [fill={rgb, 255:red, 245; green, 240; blue, 235 }  ,fill opacity=1 ][line width=0.5]  (120,120) -- (160,120) -- (160,160) -- (120,160) -- cycle ;
\draw [line width=0.5]    (160,140) -- (190,140) ;
\draw [line width=0.5]    (90,140) -- (119.95,140) ;
\draw  [fill={rgb, 255:red, 245; green, 240; blue, 235 }  ,fill opacity=1 ][line width=0.5]  (115,115) -- (155,115) -- (155,155) -- (115,155) -- cycle ;
\draw [line width=0.5]    (154.6,135) -- (184.6,135) ;
\draw [line width=0.5]    (85,135) -- (114.95,135) ;
\draw  [fill={rgb, 255:red, 245; green, 240; blue, 235 }  ,fill opacity=1 ][line width=0.5]  (110,110) -- (150,110) -- (150,150) -- (110,150) -- cycle ;
\draw [line width=0.5]    (150,130) -- (180,130) ;
\draw [line width=0.5]    (80,130) -- (109.95,130) ;
\draw  [fill={rgb, 255:red, 245; green, 240; blue, 235 }  ,fill opacity=1 ][line width=0.5]  (105,105) -- (145,105) -- (145,145) -- (105,145) -- cycle ;
\draw [line width=0.5]    (145,125) -- (175,125) ;
\draw [line width=0.5]    (75,125) -- (104.95,125) ;
\draw  [fill={rgb, 255:red, 245; green, 240; blue, 235 }  ,fill opacity=1 ][line width=0.5]  (120,190) -- (160,190) -- (160,230) -- (120,230) -- cycle ;
\draw [line width=0.5]    (160,210) -- (190,210) ;
\draw [line width=0.5]    (90,210) -- (119.95,210) ;
\draw  [fill={rgb, 255:red, 245; green, 240; blue, 235 }  ,fill opacity=1 ][line width=0.5]  (115,185) -- (155,185) -- (155,225) -- (115,225) -- cycle ;
\draw [line width=0.5]    (154.6,205) -- (184.6,205) ;
\draw [line width=0.5]    (85,205) -- (114.95,205) ;
\draw  [fill={rgb, 255:red, 245; green, 240; blue, 235 }  ,fill opacity=1 ][line width=0.5]  (110,180) -- (150,180) -- (150,220) -- (110,220) -- cycle ;
\draw [line width=0.5]    (150,200) -- (180,200) ;
\draw [line width=0.5]    (80,200) -- (109.95,200) ;
\draw  [fill={rgb, 255:red, 245; green, 240; blue, 235 }  ,fill opacity=1 ][line width=0.5]  (105,175) -- (145,175) -- (145,215) -- (105,215) -- cycle ;
\draw [line width=0.5]    (145,195) -- (175,195) ;
\draw [line width=0.5]    (75,195) -- (104.95,195) ;

\draw (125,125) node   [align=left] {$ U$};
\draw (127,195) node   [align=left] {$ U^{*}$};
\end{mytikz3}
\quad=\quad
\begin{mytikz3}

\draw [line width=0.5]    (315,118) -- (334.89,118) ;
\draw [line width=0.5]    (315,178) -- (334.97,178) ;
\draw [line width=0.5]    (315,122) -- (334.89,122) ;
\draw [line width=0.5]    (315,182) -- (334.97,182) ;
\draw  [fill={rgb, 255:red, 245; green, 240; blue, 235 }  ,fill opacity=1 ][line width=0.5]  (374.77,150) -- (334.77,190) -- (334.77,110) -- cycle ;
\draw [line width=0.5]    (315,126) -- (334.89,126) ;
\draw [line width=0.5]    (315,114) -- (334.89,114) ;
\draw [line width=0.5]    (315,174) -- (334.97,174) ;
\draw [line width=0.5]    (315,186) -- (334.97,186) ;
\draw [line width=0.5]    (515,118) -- (534.87,118) ;
\draw [line width=0.5]    (515,178) -- (534.95,178) ;
\draw [line width=0.5]    (515,122) -- (534.87,122) ;
\draw [line width=0.5]    (515,182) -- (534.95,182) ;
\draw  [fill={rgb, 255:red, 245; green, 240; blue, 235 }  ,fill opacity=1 ][line width=0.5]  (475,150) -- (515,190) -- (515,110) -- cycle ;
\draw [line width=0.5]    (515,126) -- (534.87,126) ;
\draw [line width=0.5]    (515,114) -- (534.87,114) ;
\draw [line width=0.5]    (515,174) -- (534.95,174) ;
\draw [line width=0.5]    (515,186) -- (534.95,186) ;
\draw  [fill={rgb, 255:red, 245; green, 240; blue, 235 }  ,fill opacity=1 ][line width=0.5]  (395,120) -- (455,120) -- (455,180) -- (395,180) -- cycle ;
\draw [line width=0.5]    (374.77,150) -- (395,150) ;
\draw [line width=0.5]    (475,150) -- (455,150) ;

\draw (425,150) node   [align=left] {$W(d)$};
\draw (350,150) node   [align=left] {$S$};
\draw (500,150) node   [align=left] {$S^{\dagger }$};

\end{mytikz3}
\quad,
\end{equation}
where we set $t=4$ again for the convenience of illustration. The superscript ``$(t)$'' of $W^{(t)}(d)$ is omitted in the diagram for simplicity. The tensors on the right-hand side of Eq.\,\eqref{eq:weingarten_tn} are defined by
\begin{equation}\label{eq:wg_S_def}
\begin{mytikz3}

\draw  [fill={rgb, 255:red, 245; green, 240; blue, 235 }  ,fill opacity=1 ][line width=0.5]  (120,90) -- (180,90) -- (180,150) -- (120,150) -- cycle ;
\draw [line width=0.5]    (90,120) -- (120,120) ;
\draw [line width=0.5]    (210,120) -- (180,120) ;

\draw (150.36,120.21) node   [align=left] {$ W(d)$};
\draw (80,120) node   [align=left] {$\sigma $};
\draw (220,120) node   [align=left] {$\tau $};

\end{mytikz3}
~= W^{(t)}_{\sigma,\tau}(d)~,\quad\quad
\begin{mytikz3}
\draw [line width=0.5]    (310,118) -- (329.89,118) ;
\draw [line width=0.5]    (310,178) -- (329.97,178) ;
\draw [line width=0.5]    (310,122) -- (329.89,122) ;
\draw [line width=0.5]    (310,182) -- (329.97,182) ;
\draw  [fill={rgb, 255:red, 245; green, 240; blue, 235 }  ,fill opacity=1 ][line width=0.5]  (369.77,150) -- (329.77,190) -- (329.77,110) -- cycle ;
\draw [line width=0.5]    (310,126) -- (329.89,126) ;
\draw [line width=0.5]    (310,114) -- (329.89,114) ;
\draw [line width=0.5]    (310,174) -- (329.97,174) ;
\draw [line width=0.5]    (310,186) -- (329.97,186) ;
\draw [line width=0.5]    (369.77,150) -- (390,150) ;

\draw (345,150) node   [align=left] {$ S $};
\draw (400,150) node   [align=left] {$ \sigma $};
\end{mytikz3}
~=~
\begin{mytikz3}
\draw [line width=0.5]    (550,68) -- (570,68) ;
\draw [line width=0.5]    (550,128) -- (570,128) ;
\draw [line width=0.5]    (550,72) -- (570,72) ;
\draw [line width=0.5]    (550,132) -- (570,132) ;
\draw  [fill={rgb, 255:red, 245; green, 240; blue, 235 }  ,fill opacity=1 ][line width=0.5]  (610,100) -- (570,140) -- (570,60) -- cycle ;
\draw [line width=0.5]    (550,76) -- (570,76) ;
\draw [line width=0.5]    (550,64) -- (570,64) ;
\draw [line width=0.5]    (550,124) -- (570,124) ;
\draw [line width=0.5]    (550,136) -- (570,136) ;

\draw (586,100) node   [align=left] {$\ket{\sigma }$};
\end{mytikz3}
~,\quad\quad
\begin{mytikz3}
\draw [line width=0.5]    (500,118) -- (519.87,118) ;
\draw [line width=0.5]    (500,178) -- (519.95,178) ;
\draw [line width=0.5]    (500,122) -- (519.87,122) ;
\draw [line width=0.5]    (500,182) -- (519.95,182) ;
\draw  [fill={rgb, 255:red, 245; green, 240; blue, 235 }  ,fill opacity=1 ][line width=0.5]  (460,150) -- (500,190) -- (500,110) -- cycle ;
\draw [line width=0.5]    (500,126) -- (519.87,126) ;
\draw [line width=0.5]    (500,114) -- (519.87,114) ;
\draw [line width=0.5]    (500,174) -- (519.95,174) ;
\draw [line width=0.5]    (500,186) -- (519.95,186) ;
\draw [line width=0.5]    (460,150) -- (440,150) ;

\draw (485,150) node   [align=left] {$ S^{\dagger }$};
\draw (430,150) node   [align=left] {$ \tau $};
\end{mytikz3}
~=~
\begin{mytikz3}
\draw [line width=0.5]    (560,88) -- (579.87,88) ;
\draw [line width=0.5]    (560,148) -- (579.95,148) ;
\draw [line width=0.5]    (560,92) -- (579.87,92) ;
\draw [line width=0.5]    (560,152) -- (579.95,152) ;
\draw  [fill={rgb, 255:red, 245; green, 240; blue, 235 }  ,fill opacity=1 ][line width=0.5]  (520,120) -- (560,160) -- (560,80) -- cycle ;
\draw [line width=0.5]    (560,96) -- (579.87,96) ;
\draw [line width=0.5]    (560,84) -- (579.87,84) ;
\draw [line width=0.5]    (560,144) -- (579.95,144) ;
\draw [line width=0.5]    (560,156) -- (579.95,156) ;

\draw (545,120) node   [align=left] {$\bra{\tau}$};
\end{mytikz3}
\quad.
\end{equation}
Note that the indices $\sigma,\tau$ of the matrix $W^{(t)}(d)$ are elements in the permutation group $\mathcal{S}_t$. The values of the Weingarten matrix can be determined by inserting both sides of Eq.\,\eqref{eq:weingarten_Liouville} to two arbitrary permutation vectors $\bra{\varsigma'}$ and $\ket{\varsigma}$, i.e.,
\begin{equation}\label{eq:taupp_taup}
    \braoprket{\varsigma'}{ \int \mathrm{d}\mu(U)~ U^{\otimes t} \otimes U^{* \otimes t} }{\varsigma} = \braket{\varsigma'}{\varsigma} = \sum_{\sigma,\tau\in \mathcal{S}_{t}} W^{(t)}_{\sigma,\tau}(d) \braket{\varsigma'}{\sigma} \braket{\tau}{\varsigma},
\end{equation}
where we have used the invariance of $\ket{\varsigma}$ under the action of $U^{\otimes t}\otimes U^{*\otimes t}$. The inner product like $\braket{\sigma}{\tau}$ can be seen as the matrix element of the Gram matrix $G^{(t)}(d)$ of the permutation vectors, which is equal to
\begin{equation}\label{eq:gram_def}
    G^{(t)}_{\sigma,\tau}(d) = \braket{\sigma}{\tau} = \opr{Tr}(S_{\sigma}^{-1} S_{\tau}) = d^{\#(\sigma^{-1}\tau)},
\end{equation}
where the notation $\#(\sigma^{-1}\tau)$ counts the number of cyclic permutations (cycle number) in $\sigma^{-1}\tau$, i.e., the number of closed loops in the tensor diagram of $\opr{Tr}(S_\sigma^{-1} S_\tau)=\langle \sigma|\tau\rangle$. Using the Gram matrix $G^{(t)}(d)$, Eq.\,\eqref{eq:taupp_taup} can be rewritten as
\begin{equation}
    G^{(t)}(d) W^{(t)}(d) G^{(t)}(d) = G^{(t)}(d).
\end{equation}
Thus, the Weingarten matrix $W^{(t)}(d)$ is just the Moore–Penrose pseudo-inverse matrix of the Gram matrix $G^{(t)}(d)$. When $G^{(t)}(d)$ is invertible, i.e., $d\geq t$, the Weingarten matrix $W^{(t)}(d)$ is just the normal inverse matrix of $G^{(t)}(d)$
\begin{equation}\label{eq:wg_inverse}
    W^{(t)}(d) G^{(t)}(d) = I.
\end{equation}
The corresponding tensor diagram of Eq.\,\eqref{eq:wg_inverse} is
\begin{equation}
\begin{mytikz3}

\draw  [fill={rgb, 255:red, 245; green, 240; blue, 235 }  ,fill opacity=1 ][line width=0.5]  (140,110) -- (200,110) -- (200,170) -- (140,170) -- cycle ;
\draw [line width=0.5]    (100,140) -- (140,140) ;
\draw [line width=0.5]    (240,140) -- (200,140.03) ;
\draw [line width=0.5]    (303.64,140) -- (338.79,140) ;
\draw  [fill={rgb, 255:red, 245; green, 240; blue, 235 }  ,fill opacity=1 ][line width=0.5]  (271.82,108.18) -- (303.64,140) -- (271.82,171.82) -- (240,140) -- cycle ;

\draw (170,140) node   [align=left] {$W(d)$};
\draw (271.82,140) node   [align=left] {$G(d)$};

\end{mytikz3}
~~ = ~~
\begin{mytikz3}
\draw [line width=0.5]    (400,118) -- (419.87,118) ;
\draw [line width=0.5]    (400,178) -- (419.95,178) ;
\draw [line width=0.5]    (400,122) -- (419.87,122) ;
\draw [line width=0.5]    (400,182) -- (419.95,182) ;
\draw  [fill={rgb, 255:red, 245; green, 240; blue, 235 }  ,fill opacity=1 ][line width=0.5]  (360,150) -- (400,190) -- (400,110) -- cycle ;
\draw [line width=0.5]    (400,126) -- (419.87,126) ;
\draw [line width=0.5]    (400,114) -- (419.87,114) ;
\draw [line width=0.5]    (400,174) -- (419.95,174) ;
\draw [line width=0.5]    (400,186) -- (419.95,186) ;
\draw  [fill={rgb, 255:red, 245; green, 240; blue, 235 }  ,fill opacity=1 ][line width=0.5]  (260,120) -- (320,120) -- (320,180) -- (260,180) -- cycle ;
\draw [line width=0.5]    (220,150) -- (260,150) ;
\draw [line width=0.5]    (360,150) -- (320,150.03) ;
\draw  [fill={rgb, 255:red, 245; green, 240; blue, 235 }  ,fill opacity=1 ][line width=0.5]  (460,150) -- (420,190) -- (420,110) -- cycle ;
\draw [line width=0.5]    (460,150) -- (487.73,150) ;

\draw (290,150) node   [align=left] {$W(d)$};
\draw (385,150) node   [align=left] {$S^{\dagger }$};
\draw (435.23,150) node   [align=left] {$S$};
\end{mytikz3}
~~=~~
\begin{mytikz3}
\draw [line width=0.5]    (400,150) -- (460,150) ;
\end{mytikz3}
\quad.
\end{equation}
Taking the first two degrees as examples, the Weingarten matrices are
\begin{equation}
    t=1:~W^{(1)}(d)=\frac{1}{d},\quad \quad  t=2:~W^{(2)}(d)=\left(\begin{matrix}
        d^2 & d \\
        d & d^2
    \end{matrix}\right)^{-1}=\frac{1}{d^2-1} \left(\begin{matrix}
        1 & -1/d \\
        -1/d & 1
    \end{matrix}\right),
\end{equation}
where we have used the fact $\braket{(1)}{(1)}=d$, $\braket{(1)(2)}{(1)(2)}=\braket{(12)}{(12)}=d^2$, and $\braket{(1)(2)}{(12)}=d$. Here $(1)$, $(2)$, and $(12)$ are canonical cycle notations of permutations. The Weingarten functions of higher degrees can be expressed in an exact closed form as a summation over all partitions $\lambda$ of the integer $t$~\cite{Collins2006}
\begin{equation}
    \opr{Wg}(\sigma, d) = \frac{1}{t!} \sum_{\lambda} \frac{\chi_\lambda(1)\chi_\lambda(\sigma)}{\prod_{(i,j)\in\mathcal{Y}_\lambda}(d-i+j)},
\end{equation}
where the product is taken over all cells $(i,j)$ in the Young diagram $\mathcal{Y}_\lambda$ of shape $\lambda$. Here $1$ in $\chi_\lambda(1)$ denotes the identity in $\mathcal{S}_t$, and $\chi_\lambda(\sigma)$ is the irreducible character of $\mathcal{S}_t$ corresponding to $\lambda$.

\subsection{Properties of the Weingarten function}

The Gram matrix $G^{(t)}(d)$ is a positive semi-definite symmetric matrix of shape $t!\times t!$, so the Weingarten matrix, as the pseudo-inverse of the Gram matrix, is also a positive semi-definite symmetric matrix of shape $t!\times t!$. 

According to the group rearrangement theorem and Eq.\,\eqref{eq:gram_def}, each row (column) of the Gram matrix $G^{(t)}(d)$ is just a rearrangement of the first row (column) of $G^{(t)}(d)$. The same also holds for the Weingarten matrix $W^{(t)}(d)$ since the Weingarten function also only depends on the product of the indices, i.e., $W^{(t)}_{\sigma,\tau}(d)= \opr{Wg}(\sigma^{-1}\tau, d)$~\cite{Collins2006}. Thus, the all-one vector 
\begin{equation}\label{eq:all-one_vector}
    \vec{1}=(1,1,\ldots,1)^T,
\end{equation}
of dimension $t!$ is an eigenvector for both $G^{(t)}(d)$ and $W^{(t)}(d)$ because the sum of the elements in a column of $G^{(t)}(d)$ or $W^{(t)}(d)$ is the same for all columns. The corresponding eigenvalue of $G^{(t)}(d)$ is
\begin{equation}\label{eq:gram_eigenvalue}
    \sum_{\sigma\in\mathcal{S}_t} G^{(t)}_{\sigma,\tau}(d) = \sum_{\sigma\in\mathcal{S}_t} d^{\#(\sigma^{-1}\tau)} = \sum_{\sigma\in\mathcal{S}_t} d^{\# \sigma} = \sum_{c=1}^{t} |s(t,c)| d^{c},
\end{equation}
where $|s(t,c)|$ is the number of permutations in $\mathcal{S}_t$ whose cycle number equals $c$. By definition, $s(t,c)$ is the so-called Stirling number of the first kind, and $|s(t,c)|$ is the unsigned Stirling number of the first kind, which are the coefficients in the expansion of the falling and rising factorials, respectively. Specifically, $|s(t,c)|$ can be obtained by expanding the polynomial
\begin{equation}\label{eq:unsigned_stirling}
    x(x+1)(x+2)\cdots (x+t-1) = \sum_{c=1}^{t} |s(t,c)| x^{c}.
\end{equation}
This equation can be interpreted as follows.
\begin{itemize}
    \item For $t=1$, the monomial $x$ means that there is only one permutation of cycle number $c=1$ in $\mathcal{S}_1$: $\{(1)\}$.
    \item For $t=2$, on the right-hand side, $x^2+x$ means one permutation of $c=2$ and one permutation of $c=1$: $\{(1)(2),(12)\}$; on the left-hand side, the additional factor $(x+1)$ in $x(x+1)$ means that, to add a new element, there are two different choices: make the new element a $1$-cycle by itself [$x$ in $(x+1)$], or insert it into the existing cycle [$1$ in $(x+1)$].
    \item For an arbitrary $t$, the additional factor $(x+t-1)$ means either making the new element a $1$-cycle by itself [$x$ in $(x+t-1)$], or inserting it into existing cycles at $(t-1)$ different positions [$(t-1)$ in $(x+t-1)$].
\end{itemize}
Therefore, it can be seen that the coefficient $|s(t,c)|$ of $x^c$ indeed counts the number of permutations in $\mathcal{S}_t$ consisting of $c$ cyclic permutations. Combining Eqs.\,\eqref{eq:gram_eigenvalue} and \eqref{eq:unsigned_stirling}, the eigenvalue of $G^{(t)}(d)$ corresponding to the all-one vector $\vec{1}$ is
\begin{equation}\label{eq:sum_G}
    \sum_{\sigma\in\mathcal{S}_t} G^{(t)}_{\sigma,\tau}(d) = \sum_{c=1}^{t} |s(t,c)| d^{c} = d(d+1)(d+2)\cdots (d+t-1).
\end{equation}
As $W^{(t)}(d)$ is the pseudo-inverse of $G^{(t)}(d)$, the eigenvalue of $W^{(t)}(d)$ corresponding to $\vec{1}$ is
\begin{equation}\label{eq:sum_Wg}
    \sum_{\sigma\in\mathcal{S}_t} W^{(t)}_{\sigma,\tau}(d) = \frac{1}{d(d+1)(d+2)\cdots (d+t-1)}.
\end{equation}

\subsection{Weingarten formula for random quantum circuits}

For a random unitary gate acting on multiple qudits with multiple indices, the corresponding Weingarten formula can be obtained by combining the indices into a single one and then applying Eq.\,\eqref{eq:weingarten_Liouville}, i.e., performing the same permutation operation for each index across different replicas simultaneously $\ket{\sigma}\rightarrow \ket{\sigma\sigma\cdots\sigma}$. In tensor network notations, this can be explicitly represented by connecting different indices with a COPY tensor, whose tensor element equals $1$ if and only if all indices take the same value under the chosen basis (otherwise equals $0$), i.e.,
\begin{equation}\label{eq:copy_tensor_def}
    \opr{COPY}_{ijk\ldots} = \delta_{ijk\ldots}~,\quad\quad
\begin{mytikz3}
\draw [line width=0.5]    (190,100) -- (150,120) ;
\draw [line width=0.5]    (110,120) -- (150,120) ;
\draw [line width=0.5]    (190,140) -- (150,120) ;
\draw [line width=0.5]    (192.5,127.5) -- (150,120) ;
\draw [line width=0.5]    (192.5,112.5) -- (150,120) ;
\draw  [fill={rgb, 255:red, 142; green, 142; blue, 142 }  ,fill opacity=1 ][line width=0.5]  (146.5,120) .. controls (146.5,118.07) and (148.07,116.5) .. (150,116.5) .. controls (151.93,116.5) and (153.5,118.07) .. (153.5,120) .. controls (153.5,121.93) and (151.93,123.5) .. (150,123.5) .. controls (148.07,123.5) and (146.5,121.93) .. (146.5,120) -- cycle ;

\draw (94.33,113) node [anchor=north west][inner sep=0.75pt]   [align=left] {$ \sigma $};
\draw (193.73,86.73) node [anchor=north west][inner sep=0.75pt]   [align=left] {$ \sigma $};
\draw (197.73,102.73) node [anchor=north west][inner sep=0.75pt]   [align=left] {$ \sigma $};
\draw (197.73,119.73) node [anchor=north west][inner sep=0.75pt]   [align=left] {$ \sigma $};
\draw (192.73,137.73) node [anchor=north west][inner sep=0.75pt]   [align=left] {$ \sigma $};
\end{mytikz3}
\quad=~1~.
\end{equation}
where $\delta_{ijk\ldots}$ is the Kronecker delta function. Taking a two-qudit gate as an example, the corresponding tensor network diagram for the Weingarten formula, like in Eq.\,\eqref{eq:weingarten_tn}, is
\begin{equation}\label{eq:weingarten_tn_2qudit}
\int \mathrm{d}\mu(U)~
\begin{mytikz3}
\draw [line width=0.5]    (95,175) -- (124.95,175) ;
\draw [line width=0.5]    (92.5,172.5) -- (122.45,172.5) ;
\draw [line width=0.5]    (90,170) -- (119.95,170) ;
\draw [line width=0.5]    (87.5,167.5) -- (117.45,167.5) ;
\draw  [fill={rgb, 255:red, 245; green, 240; blue, 235 }  ,fill opacity=1 ][line width=0.5]  (125,115) -- (165,115) -- (165,195) -- (125,195) -- cycle ;
\draw [line width=0.5]    (165,135) -- (195,135) ;
\draw [line width=0.5]    (95,135) -- (124.95,135) ;
\draw  [fill={rgb, 255:red, 245; green, 240; blue, 235 }  ,fill opacity=1 ][line width=0.5]  (122.5,112.5) -- (162.5,112.5) -- (162.5,192.5) -- (122.5,192.5) -- cycle ;
\draw [line width=0.5]    (162.5,132.5) -- (192.5,132.5) ;
\draw [line width=0.5]    (92.5,132.5) -- (122.45,132.5) ;
\draw  [fill={rgb, 255:red, 245; green, 240; blue, 235 }  ,fill opacity=1 ][line width=0.5]  (120,110) -- (160,110) -- (160,190) -- (120,190) -- cycle ;
\draw [line width=0.5]    (160,130) -- (190,130) ;
\draw [line width=0.5]    (90,130) -- (119.95,130) ;
\draw  [fill={rgb, 255:red, 245; green, 240; blue, 235 }  ,fill opacity=1 ][line width=0.5]  (117.5,107.5) -- (157.5,107.5) -- (157.5,187.5) -- (117.5,187.5) -- cycle ;
\draw [line width=0.5]    (157.5,127.5) -- (187.5,127.5) ;
\draw [line width=0.5]    (87.5,127.5) -- (117.45,127.5) ;
\draw [line width=0.5]    (165,175) -- (195,175) ;
\draw [line width=0.5]    (162.5,172.5) -- (192.5,172.5) ;
\draw [line width=0.5]    (160,170) -- (190,170) ;
\draw [line width=0.5]    (157.5,167.5) -- (187.5,167.5) ;
\draw [line width=0.5]    (82.5,162.5) -- (112.45,162.5) ;
\draw [line width=0.5]    (80,160) -- (109.95,160) ;
\draw [line width=0.5]    (77.5,157.5) -- (107.45,157.5) ;
\draw [line width=0.5]    (75,155) -- (104.95,155) ;
\draw  [fill={rgb, 255:red, 245; green, 240; blue, 235 }  ,fill opacity=1 ][line width=0.5]  (112.5,102.5) -- (152.5,102.5) -- (152.5,182.5) -- (112.5,182.5) -- cycle ;
\draw [line width=0.5]    (152.5,122.5) -- (182.5,122.5) ;
\draw [line width=0.5]    (82.5,122.5) -- (112.45,122.5) ;
\draw  [fill={rgb, 255:red, 245; green, 240; blue, 235 }  ,fill opacity=1 ][line width=0.5]  (110,100) -- (150,100) -- (150,180) -- (110,180) -- cycle ;
\draw [line width=0.5]    (150,120) -- (180,120) ;
\draw [line width=0.5]    (80,120) -- (109.95,120) ;
\draw  [fill={rgb, 255:red, 245; green, 240; blue, 235 }  ,fill opacity=1 ][line width=0.5]  (107.5,97.5) -- (147.5,97.5) -- (147.5,177.5) -- (107.5,177.5) -- cycle ;
\draw [line width=0.5]    (147.5,117.5) -- (177.5,117.5) ;
\draw [line width=0.5]    (77.5,117.5) -- (107.45,117.5) ;
\draw  [fill={rgb, 255:red, 245; green, 240; blue, 235 }  ,fill opacity=1 ][line width=0.5]  (105,95) -- (145,95) -- (145,175) -- (105,175) -- cycle ;
\draw [line width=0.5]    (145,115) -- (175,115) ;
\draw [line width=0.5]    (75,115) -- (104.95,115) ;
\draw [line width=0.5]    (152.5,162.5) -- (182.5,162.5) ;
\draw [line width=0.5]    (150,160) -- (180,160) ;
\draw [line width=0.5]    (147.5,157.5) -- (177.5,157.5) ;
\draw [line width=0.5]    (145,155) -- (175,155) ;

\draw (125,135) node   [align=left] {$ U$};
\end{mytikz3}
\quad=\quad
\begin{mytikz3}

\draw  [fill={rgb, 255:red, 245; green, 240; blue, 235 }  ,fill opacity=1 ][line width=0.5]  (340,130) -- (327.5,142.5) -- (327.5,117.5) -- cycle ;
\draw  [fill={rgb, 255:red, 245; green, 240; blue, 235 }  ,fill opacity=1 ][line width=0.5]  (397.5,110) -- (477.5,110) -- (477.5,190) -- (397.5,190) -- cycle ;
\draw [line width=0.5]    (369.77,150) -- (397.5,150) ;
\draw [line width=0.5]    (505,150) -- (477.5,150) ;
\draw [line width=0.5]    (369.77,150) -- (340,130) ;
\draw [line width=0.5]    (369.77,150) -- (340,170) ;
\draw  [fill={rgb, 255:red, 142; green, 142; blue, 142 }  ,fill opacity=1 ][line width=0.5]  (366.27,150) .. controls (366.27,148.07) and (367.84,146.5) .. (369.77,146.5) .. controls (371.7,146.5) and (373.27,148.07) .. (373.27,150) .. controls (373.27,151.93) and (371.7,153.5) .. (369.77,153.5) .. controls (367.84,153.5) and (366.27,151.93) .. (366.27,150) -- cycle ;
\draw [line width=0.5]    (535,170) -- (505,150) ;
\draw [line width=0.5]    (535,130) -- (505,150) ;
\draw  [fill={rgb, 255:red, 142; green, 142; blue, 142 }  ,fill opacity=1 ][line width=0.5]  (501.5,150) .. controls (501.5,148.07) and (503.07,146.5) .. (505,146.5) .. controls (506.93,146.5) and (508.5,148.07) .. (508.5,150) .. controls (508.5,151.93) and (506.93,153.5) .. (505,153.5) .. controls (503.07,153.5) and (501.5,151.93) .. (501.5,150) -- cycle ;
\draw [line width=0.5]    (307.5,120) -- (327.5,120) ;
\draw [line width=0.5]    (307.5,122.5) -- (327.5,122.5) ;
\draw [line width=0.5]    (307.5,125) -- (327.5,125) ;
\draw [line width=0.5]    (307.5,127.5) -- (327.5,127.5) ;
\draw [line width=0.5]    (307.5,132.5) -- (321.48,132.5) -- (327.5,132.5) ;
\draw [line width=0.5]    (307.5,135) -- (327.5,135) ;
\draw [line width=0.5]    (307.5,137.5) -- (327.5,137.5) ;
\draw [line width=0.5]    (307.5,140) -- (327.5,140) ;
\draw  [fill={rgb, 255:red, 245; green, 240; blue, 235 }  ,fill opacity=1 ][line width=0.5]  (340,170) -- (327.5,182.5) -- (327.5,157.5) -- cycle ;
\draw [line width=0.5]    (307.5,160) -- (327.5,160) ;
\draw [line width=0.5]    (307.5,162.5) -- (327.5,162.5) ;
\draw [line width=0.5]    (307.5,165) -- (327.5,165) ;
\draw [line width=0.5]    (307.5,167.5) -- (327.5,167.5) ;
\draw [line width=0.5]    (307.5,172.5) -- (321.48,172.5) -- (327.5,172.5) ;
\draw [line width=0.5]    (307.5,175) -- (327.5,175) ;
\draw [line width=0.5]    (307.5,177.5) -- (327.5,177.5) ;
\draw [line width=0.5]    (307.5,180) -- (327.5,180) ;
\draw  [fill={rgb, 255:red, 245; green, 240; blue, 235 }  ,fill opacity=1 ][line width=0.5]  (535,130) -- (547.5,117.5) -- (547.5,142.5) -- cycle ;
\draw [line width=0.5]    (547.5,120) -- (567.5,120) ;
\draw [line width=0.5]    (547.5,122.5) -- (567.5,122.5) ;
\draw [line width=0.5]    (547.5,125) -- (567.5,125) ;
\draw [line width=0.5]    (547.5,127.5) -- (567.5,127.5) ;
\draw [line width=0.5]    (547.5,132.5) -- (561.48,132.5) -- (567.5,132.5) ;
\draw [line width=0.5]    (547.5,135) -- (567.5,135) ;
\draw [line width=0.5]    (547.5,137.5) -- (567.5,137.5) ;
\draw [line width=0.5]    (547.5,140) -- (567.5,140) ;
\draw [line width=0.5]    (547.5,160) -- (567.5,160) ;
\draw [line width=0.5]    (547.5,162.5) -- (567.5,162.5) ;
\draw [line width=0.5]    (547.5,165) -- (567.5,165) ;
\draw [line width=0.5]    (547.5,167.5) -- (567.5,167.5) ;
\draw [line width=0.5]    (547.5,172.5) -- (561.48,172.5) -- (567.5,172.5) ;
\draw [line width=0.5]    (547.5,175) -- (567.5,175) ;
\draw [line width=0.5]    (547.5,177.5) -- (567.5,177.5) ;
\draw [line width=0.5]    (547.5,180) -- (567.5,180) ;
\draw  [fill={rgb, 255:red, 245; green, 240; blue, 235 }  ,fill opacity=1 ][line width=0.5]  (535,170) -- (547.5,157.5) -- (547.5,182.5) -- cycle ;

\draw (437.5,150) node  [font=\small] [align=left] {$W(d_{1} d_{2})$};

\end{mytikz3}
\quad\quad,
\end{equation}
where we have stacked the replicas of $U^*$ under those of $U$. $d_1$ and $d_2$ denote the local Hilbert space dimensions for the two qudits. The symbols ``$S$'' and ``$S^\dagger$'' on the triangular tensors are omitted for simplicity. We remark that the values of the Weingarten matrix elements depend on the dimension of the integrated unitary, e.g., the Weingarten matrix in Eq.\,\eqref{eq:weingarten_tn_2qudit} is $W^{(t)}(d_1d_2)$, instead of $W^{(t)}(d)$ in Eq.\,\eqref{eq:weingarten_tn}.

Therefore, for a random unitary circuit $\mathbf{U}$ composed of independently Haar-random unitary gates $\{U_k\}_{k=1}^{n}$, i.e.,
\begin{equation}
    \mathbf{U} = \opr{tTr}\left( \bigotimes_{k} U_k \right),
\end{equation}
where $\opr{tTr}$ represents taking the tensor network contraction according to the circuit graph. The corresponding ensemble averages can always be represented as a tensor network in the same shape as the circuit graph but with each unitary gate replaced by a ``Weingarten gate'' $\opr{W}^{(t)}$ as on the right-hand side of Eq.\,\eqref{eq:weingarten_tn_2qudit}, i.e.,
\begin{equation}\label{eq:weingarten_gate}
    \mathbb{E}_{\mathbb{U}} \left[ \mathbf{U}^{\otimes t} \otimes \mathbf{U}^{* \otimes t} \right] = \opr{tTr}\left( \bigotimes_{k} \opr{W}_k^{(t)} \right).
\end{equation}
Here $\mathbb{U}$ denotes the ensemble of random circuits and $\mathbb{E}$ denotes the ensemble average. $\opr{W}_{k}^{(t)}$ is the $t$-degree Weingarten ``gate'' in the same shape as the corresponding unitary gate $U_k$. As an example, for two adjacent independently random unitary gates, the ensemble averages can be represented as two Weingarten gates with the same contraction relation, i.e.,
\begin{equation}
\begin{mytikz3}
\draw [line width=0.5]    (177.5,112.54) -- (177.5,92.54) ;
\draw  [fill={rgb, 255:red, 245; green, 240; blue, 235 }  ,fill opacity=1 ][line width=0.5]  (107.5,92.54) -- (107.5,52.54) -- (187.5,52.54) -- (187.5,92.54) -- cycle ;
\draw [line width=0.5]    (117.5,52.54) -- (117.5,32.54) ;
\draw [line width=0.5]    (117.5,132.54) -- (117.5,112.54) ;
\draw [line width=0.5]    (177.5,52.54) -- (177.5,32.5) ;
\draw [line width=0.5]    (117.5,192.54) -- (117.5,172.54) ;
\draw  [fill={rgb, 255:red, 245; green, 240; blue, 235 }  ,fill opacity=1 ][line width=0.5]  (47.5,172.54) -- (47.5,132.54) -- (127.5,132.54) -- (127.5,172.54) -- cycle ;
\draw [line width=0.5]    (57.5,192.54) -- (57.5,172.54) ;
\draw [line width=0.5]    (57.5,132.54) -- (57.5,112.54) ;
\draw [line width=0.5]    (117.5,112.54) -- (117.5,92.54) ;

\draw [line width=0.5]    (175,115.04) -- (175,95.04) ;
\draw  [fill={rgb, 255:red, 245; green, 240; blue, 235 }  ,fill opacity=1 ][line width=0.5]  (105,95.04) -- (105,55.04) -- (185,55.04) -- (185,95.04) -- cycle ;
\draw [line width=0.5]    (115,55.04) -- (115,35.04) ;
\draw [line width=0.5]    (115,135.04) -- (115,115.04) ;
\draw [line width=0.5]    (175,55.04) -- (175,35) ;
\draw [line width=0.5]    (115,195.04) -- (115,175.04) ;
\draw  [fill={rgb, 255:red, 245; green, 240; blue, 235 }  ,fill opacity=1 ][line width=0.5]  (45,175.04) -- (45,135.04) -- (125,135.04) -- (125,175.04) -- cycle ;
\draw [line width=0.5]    (55,195.04) -- (55,175.04) ;
\draw [line width=0.5]    (55,135.04) -- (55,115.04) ;
\draw [line width=0.5]    (115,115.04) -- (115,95.04) ;

\draw [line width=0.5]    (172.5,117.54) -- (172.5,97.54) ;
\draw  [fill={rgb, 255:red, 245; green, 240; blue, 235 }  ,fill opacity=1 ][line width=0.5]  (102.5,97.54) -- (102.5,57.54) -- (182.5,57.54) -- (182.5,97.54) -- cycle ;
\draw [line width=0.5]    (112.5,57.54) -- (112.5,37.54) ;
\draw [line width=0.5]    (112.5,137.54) -- (112.5,117.54) ;
\draw [line width=0.5]    (172.5,57.54) -- (172.5,37.5) ;
\draw [line width=0.5]    (112.5,197.54) -- (112.5,177.54) ;
\draw  [fill={rgb, 255:red, 245; green, 240; blue, 235 }  ,fill opacity=1 ][line width=0.5]  (42.5,177.54) -- (42.5,137.54) -- (122.5,137.54) -- (122.5,177.54) -- cycle ;
\draw [line width=0.5]    (52.5,197.54) -- (52.5,177.54) ;
\draw [line width=0.5]    (52.5,137.54) -- (52.5,117.54) ;
\draw [line width=0.5]    (112.5,117.54) -- (112.5,97.54) ;

\draw [line width=0.5]    (170,120) -- (170,100) ;
\draw  [fill={rgb, 255:red, 245; green, 240; blue, 235 }  ,fill opacity=1 ][line width=0.5]  (100,100) -- (100,60) -- (180,60) -- (180,100) -- cycle ;
\draw [line width=0.5]    (110,60) -- (110,40) ;
\draw [line width=0.5]    (110,140) -- (110,120) ;
\draw [line width=0.5]    (170,60) -- (170,39.96) ;
\draw [line width=0.5]    (110,200) -- (110,180) ;
\draw  [fill={rgb, 255:red, 245; green, 240; blue, 235 }  ,fill opacity=1 ][line width=0.5]  (40,180) -- (40,140) -- (120,140) -- (120,180) -- cycle ;
\draw [line width=0.5]    (50,200) -- (50,180) ;
\draw [line width=0.5]    (50,140) -- (50,120) ;
\draw [line width=0.5]    (110,120) -- (110,100) ;

\draw (80,160) node   [align=left] {$ U_{1}$};
\draw (140,80) node   [align=left] {$ U_{2}$};
\end{mytikz3}
\quad\quad \xrightarrow{~~\mathbb{E}~~} \quad\quad
\begin{mytikz2}
\draw  [fill={rgb, 255:red, 245; green, 240; blue, 235 }  ,fill opacity=1 ][line width=0.5]  (160,230) -- (180,230) -- (180,250) -- (160,250) -- cycle ;
\draw [line width=0.5]    (170,210) -- (170,230) ;
\draw [line width=0.5]    (170,250) -- (170,270) ;
\draw [line width=0.5]    (170,210) -- (220,190) ;
\draw [line width=0.5]    (170,210) -- (120,190) ;
\draw  [fill={rgb, 255:red, 245; green, 240; blue, 235 }  ,fill opacity=1 ][line width=0.5]  (220,190) -- (210,180) -- (230,180) -- cycle ;
\draw [line width=0.5]    (120,290) -- (170,270) ;
\draw [line width=0.5]    (220,290) -- (170,270) ;
\draw [line width=0.5]    (220,160) -- (220,180) ;
\draw  [fill={rgb, 255:red, 245; green, 240; blue, 235 }  ,fill opacity=1 ][line width=0.5]  (120,190) -- (110,180) -- (130,180) -- cycle ;
\draw [line width=0.5]    (120,160) -- (120,180) ;
\draw  [fill={rgb, 255:red, 245; green, 240; blue, 235 }  ,fill opacity=1 ][line width=0.5]  (120,290) -- (130,300) -- (110,300) -- cycle ;
\draw  [fill={rgb, 255:red, 245; green, 240; blue, 235 }  ,fill opacity=1 ][line width=0.5]  (220,290) -- (230,300) -- (210,300) -- cycle ;
\draw [line width=0.5]    (220,300) -- (220,320) ;
\draw [line width=0.5]    (120,300) -- (120,320) ;
\draw  [fill={rgb, 255:red, 245; green, 240; blue, 235 }  ,fill opacity=1 ][line width=0.5]  (260,70) -- (280,70) -- (280,90) -- (260,90) -- cycle ;
\draw [line width=0.5]    (270,50) -- (270,70) ;
\draw [line width=0.5]    (270,90) -- (270,110) ;
\draw [line width=0.5]    (270,50) -- (305.82,35.64) -- (320,30) ;
\draw [line width=0.5]    (270,50) -- (220,30) ;
\draw  [fill={rgb, 255:red, 245; green, 240; blue, 235 }  ,fill opacity=1 ][line width=0.5]  (320,30) -- (310,20) -- (330,20) -- cycle ;
\draw [line width=0.5]    (220,130) -- (270,110) ;
\draw [line width=0.5]    (320,130) -- (270,110) ;
\draw [line width=0.5]    (320,0) -- (320,20) ;
\draw  [fill={rgb, 255:red, 245; green, 240; blue, 235 }  ,fill opacity=1 ][line width=0.5]  (220,30) -- (210,20) -- (230,20) -- cycle ;
\draw [line width=0.5]    (220,0) -- (220,20) ;
\draw  [fill={rgb, 255:red, 245; green, 240; blue, 235 }  ,fill opacity=1 ][line width=0.5]  (220,130) -- (230,140) -- (210,140) -- cycle ;
\draw  [fill={rgb, 255:red, 245; green, 240; blue, 235 }  ,fill opacity=1 ][line width=0.5]  (320,130) -- (330,140) -- (310,140) -- cycle ;
\draw [line width=0.5]    (320,140) -- (320,160) ;
\draw [line width=0.5]    (220,140) -- (220,160) ;
\draw  [fill={rgb, 255:red, 142; green, 142; blue, 142 }  ,fill opacity=1 ][line width=0.5]  (166.5,210) .. controls (166.5,208.07) and (168.07,206.5) .. (170,206.5) .. controls (171.93,206.5) and (173.5,208.07) .. (173.5,210) .. controls (173.5,211.93) and (171.93,213.5) .. (170,213.5) .. controls (168.07,213.5) and (166.5,211.93) .. (166.5,210) -- cycle ;
\draw  [fill={rgb, 255:red, 142; green, 142; blue, 142 }  ,fill opacity=1 ][line width=0.5]  (166.5,270) .. controls (166.5,268.07) and (168.07,266.5) .. (170,266.5) .. controls (171.93,266.5) and (173.5,268.07) .. (173.5,270) .. controls (173.5,271.93) and (171.93,273.5) .. (170,273.5) .. controls (168.07,273.5) and (166.5,271.93) .. (166.5,270) -- cycle ;
\draw  [fill={rgb, 255:red, 142; green, 142; blue, 142 }  ,fill opacity=1 ][line width=0.5]  (266.5,110) .. controls (266.5,108.07) and (268.07,106.5) .. (270,106.5) .. controls (271.93,106.5) and (273.5,108.07) .. (273.5,110) .. controls (273.5,111.93) and (271.93,113.5) .. (270,113.5) .. controls (268.07,113.5) and (266.5,111.93) .. (266.5,110) -- cycle ;
\draw  [fill={rgb, 255:red, 142; green, 142; blue, 142 }  ,fill opacity=1 ][line width=0.5]  (266.5,50) .. controls (266.5,48.07) and (268.07,46.5) .. (270,46.5) .. controls (271.93,46.5) and (273.5,48.07) .. (273.5,50) .. controls (273.5,51.93) and (271.93,53.5) .. (270,53.5) .. controls (268.07,53.5) and (266.5,51.93) .. (266.5,50) -- cycle ;
\draw  [dash pattern={on 5.63pt off 4.5pt}][line width=0.5]  (190,10) -- (350,10) -- (350,150) -- (190,150) -- cycle ;
\draw  [dash pattern={on 5.63pt off 4.5pt}][line width=0.5]  (90,170) -- (250,170) -- (250,310) -- (90,310) -- cycle ;
\end{mytikz2}
\quad\quad,
\end{equation}
where we take $t=2$ for the convenience of illustration. The replicas and the conjugates are indicated by the rectangles underneath on the left-hand side. The entire diagram is rotated by $90^\circ$ compared to the diagrams above to align with the conventions in the literature of random quantum circuits. The $2t$-fold indices of the permutation vectors on the right-hand side are omitted and replaced by a single index in the diagram without ambiguity. The Weingarten matrix symbol ``$W$'' on the square tensors is also omitted. The dashed boxes indicate the two Weingarten gates $\opr{W}^{(t)}_1$ and $\opr{W}^{(t)}_2$. The Gram matrix of permutation vectors naturally arises in the contraction of two Weingarten gates. 

Therefore, for each random unitary circuit, the resulting tensor network has a special structure: it comprises some COPY tensor nodes connected by some weight matrices (the Gram matrix of the permutation vectors or the Weingarten matrix). This is exactly the form of tensor network representations of the partition functions of classical statistical mechanics models. The classical ``spins'' take values from the permutation group $\mathcal{S}_t$ and the interaction between spins is determined by the Weingarten matrix or its inverse. However, there is a slight difference between such a tensor network and that of a strict classical statistical mechanics model: the elements of the Weingarten matrix, intended to be used as a weight matrix, may be negative, whereas the elements of the weight matrix in the partition function of a classical statistical model must be positive. This issue can be partially fixed by integrating out certain parts of the tensor network.

Taking the 1D brickwork circuit as an example, the integration result can be represented as
\begin{equation}
\begin{mytikz}

\draw [line width=0.5]    (177.5,112.58) -- (177.5,92.58) ;
\draw  [fill={rgb, 255:red, 245; green, 240; blue, 235 }  ,fill opacity=1 ][line width=0.5]  (107.5,92.58) -- (107.5,52.58) -- (187.5,52.58) -- (187.5,92.58) -- cycle ;
\draw [line width=0.5]    (117.5,52.58) -- (117.5,32.58) ;
\draw [line width=0.5]    (117.5,132.58) -- (117.5,112.58) ;
\draw [line width=0.5]    (177.5,52.58) -- (177.5,32.54) ;
\draw [line width=0.5]    (117.5,192.58) -- (117.5,172.58) ;
\draw  [fill={rgb, 255:red, 245; green, 240; blue, 235 }  ,fill opacity=1 ][line width=0.5]  (47.5,172.58) -- (47.5,132.58) -- (127.5,132.58) -- (127.5,172.58) -- cycle ;
\draw [line width=0.5]    (57.5,192.58) -- (57.5,172.58) ;
\draw [line width=0.5]    (57.5,132.58) -- (57.5,112.58) ;
\draw [line width=0.5]    (117.5,112.58) -- (117.5,92.58) ;
\draw [line width=0.5]    (237.5,192.58) -- (237.5,172.58) ;
\draw  [fill={rgb, 255:red, 245; green, 240; blue, 235 }  ,fill opacity=1 ][line width=0.5]  (167.5,172.58) -- (167.5,132.58) -- (247.5,132.58) -- (247.5,172.58) -- cycle ;
\draw [line width=0.5]    (177.5,132.58) -- (177.5,112.58) ;
\draw [line width=0.5]    (237.5,132.58) -- (237.5,112.54) ;
\draw [line width=0.5]    (177.5,192.58) -- (177.5,172.58) ;
\draw [line width=0.5]    (177.5,272.58) -- (177.5,252.58) ;
\draw  [fill={rgb, 255:red, 245; green, 240; blue, 235 }  ,fill opacity=1 ][line width=0.5]  (107.5,252.58) -- (107.5,212.58) -- (187.5,212.58) -- (187.5,252.58) -- cycle ;
\draw [line width=0.5]    (117.5,212.58) -- (117.5,192.58) ;
\draw [line width=0.5]    (177.5,212.58) -- (177.5,192.54) ;
\draw [line width=0.5]    (117.5,272.58) -- (117.5,252.58) ;
\draw [line width=0.5]    (297.5,112.54) -- (297.5,92.54) ;
\draw  [fill={rgb, 255:red, 245; green, 240; blue, 235 }  ,fill opacity=1 ][line width=0.5]  (227.5,92.54) -- (227.5,52.54) -- (307.5,52.54) -- (307.5,92.54) -- cycle ;
\draw [line width=0.5]    (237.5,52.54) -- (237.5,32.54) ;
\draw [line width=0.5]    (297.5,52.54) -- (297.5,32.5) ;
\draw [line width=0.5]    (237.5,112.54) -- (237.5,92.54) ;
\draw [line width=0.5]    (357.5,192.54) -- (357.5,172.54) ;
\draw  [fill={rgb, 255:red, 245; green, 240; blue, 235 }  ,fill opacity=1 ][line width=0.5]  (287.5,172.54) -- (287.5,132.54) -- (367.5,132.54) -- (367.5,172.54) -- cycle ;
\draw [line width=0.5]    (297.5,132.54) -- (297.5,112.54) ;
\draw [line width=0.5]    (357.5,132.54) -- (357.5,112.5) ;
\draw [line width=0.5]    (297.5,192.54) -- (297.5,172.54) ;
\draw [line width=0.5]    (297.5,272.54) -- (297.5,252.54) ;
\draw  [fill={rgb, 255:red, 245; green, 240; blue, 235 }  ,fill opacity=1 ][line width=0.5]  (227.5,252.54) -- (227.5,212.54) -- (307.5,212.54) -- (307.5,252.54) -- cycle ;
\draw [line width=0.5]    (237.5,212.54) -- (237.5,192.54) ;
\draw [line width=0.5]    (297.5,212.54) -- (297.5,192.5) ;
\draw [line width=0.5]    (237.5,272.54) -- (237.5,252.54) ;

\draw [line width=0.5]    (175,115.08) -- (175,95.08) ;
\draw  [fill={rgb, 255:red, 245; green, 240; blue, 235 }  ,fill opacity=1 ][line width=0.5]  (105,95.08) -- (105,55.08) -- (185,55.08) -- (185,95.08) -- cycle ;
\draw [line width=0.5]    (115,55.08) -- (115,35.08) ;
\draw [line width=0.5]    (115,135.08) -- (115,115.08) ;
\draw [line width=0.5]    (175,55.08) -- (175,35.04) ;
\draw [line width=0.5]    (115,195.08) -- (115,175.08) ;
\draw  [fill={rgb, 255:red, 245; green, 240; blue, 235 }  ,fill opacity=1 ][line width=0.5]  (45,175.08) -- (45,135.08) -- (125,135.08) -- (125,175.08) -- cycle ;
\draw [line width=0.5]    (55,195.08) -- (55,175.08) ;
\draw [line width=0.5]    (55,135.08) -- (55,115.08) ;
\draw [line width=0.5]    (115,115.08) -- (115,95.08) ;
\draw [line width=0.5]    (235,195.08) -- (235,175.08) ;
\draw  [fill={rgb, 255:red, 245; green, 240; blue, 235 }  ,fill opacity=1 ][line width=0.5]  (165,175.08) -- (165,135.08) -- (245,135.08) -- (245,175.08) -- cycle ;
\draw [line width=0.5]    (175,135.08) -- (175,115.08) ;
\draw [line width=0.5]    (235,135.08) -- (235,115.04) ;
\draw [line width=0.5]    (175,195.08) -- (175,175.08) ;
\draw [line width=0.5]    (175,275.08) -- (175,255.08) ;
\draw  [fill={rgb, 255:red, 245; green, 240; blue, 235 }  ,fill opacity=1 ][line width=0.5]  (105,255.08) -- (105,215.08) -- (185,215.08) -- (185,255.08) -- cycle ;
\draw [line width=0.5]    (115,215.08) -- (115,195.08) ;
\draw [line width=0.5]    (175,215.08) -- (175,195.04) ;
\draw [line width=0.5]    (115,275.08) -- (115,255.08) ;
\draw [line width=0.5]    (295,115.04) -- (295,95.04) ;
\draw  [fill={rgb, 255:red, 245; green, 240; blue, 235 }  ,fill opacity=1 ][line width=0.5]  (225,95.04) -- (225,55.04) -- (305,55.04) -- (305,95.04) -- cycle ;
\draw [line width=0.5]    (235,55.04) -- (235,35.04) ;
\draw [line width=0.5]    (295,55.04) -- (295,35) ;
\draw [line width=0.5]    (235,115.04) -- (235,95.04) ;
\draw [line width=0.5]    (355,195.04) -- (355,175.04) ;
\draw  [fill={rgb, 255:red, 245; green, 240; blue, 235 }  ,fill opacity=1 ][line width=0.5]  (285,175.04) -- (285,135.04) -- (365,135.04) -- (365,175.04) -- cycle ;
\draw [line width=0.5]    (295,135.04) -- (295,115.04) ;
\draw [line width=0.5]    (355,135.04) -- (355,115) ;
\draw [line width=0.5]    (295,195.04) -- (295,175.04) ;
\draw [line width=0.5]    (295,275.04) -- (295,255.04) ;
\draw  [fill={rgb, 255:red, 245; green, 240; blue, 235 }  ,fill opacity=1 ][line width=0.5]  (225,255.04) -- (225,215.04) -- (305,215.04) -- (305,255.04) -- cycle ;
\draw [line width=0.5]    (235,215.04) -- (235,195.04) ;
\draw [line width=0.5]    (295,215.04) -- (295,195) ;
\draw [line width=0.5]    (235,275.04) -- (235,255.04) ;

\draw [line width=0.5]    (172.5,117.58) -- (172.5,97.58) ;
\draw  [fill={rgb, 255:red, 245; green, 240; blue, 235 }  ,fill opacity=1 ][line width=0.5]  (102.5,97.58) -- (102.5,57.58) -- (182.5,57.58) -- (182.5,97.58) -- cycle ;
\draw [line width=0.5]    (112.5,57.58) -- (112.5,37.58) ;
\draw [line width=0.5]    (112.5,137.58) -- (112.5,117.58) ;
\draw [line width=0.5]    (172.5,57.58) -- (172.5,37.54) ;
\draw [line width=0.5]    (112.5,197.58) -- (112.5,177.58) ;
\draw  [fill={rgb, 255:red, 245; green, 240; blue, 235 }  ,fill opacity=1 ][line width=0.5]  (42.5,177.58) -- (42.5,137.58) -- (122.5,137.58) -- (122.5,177.58) -- cycle ;
\draw [line width=0.5]    (52.5,197.58) -- (52.5,177.58) ;
\draw [line width=0.5]    (52.5,137.58) -- (52.5,117.58) ;
\draw [line width=0.5]    (112.5,117.58) -- (112.5,97.58) ;
\draw [line width=0.5]    (232.5,197.58) -- (232.5,177.58) ;
\draw  [fill={rgb, 255:red, 245; green, 240; blue, 235 }  ,fill opacity=1 ][line width=0.5]  (162.5,177.58) -- (162.5,137.58) -- (242.5,137.58) -- (242.5,177.58) -- cycle ;
\draw [line width=0.5]    (172.5,137.58) -- (172.5,117.58) ;
\draw [line width=0.5]    (232.5,137.58) -- (232.5,117.54) ;
\draw [line width=0.5]    (172.5,197.58) -- (172.5,177.58) ;
\draw [line width=0.5]    (172.5,277.58) -- (172.5,257.58) ;
\draw  [fill={rgb, 255:red, 245; green, 240; blue, 235 }  ,fill opacity=1 ][line width=0.5]  (102.5,257.58) -- (102.5,217.58) -- (182.5,217.58) -- (182.5,257.58) -- cycle ;
\draw [line width=0.5]    (112.5,217.58) -- (112.5,197.58) ;
\draw [line width=0.5]    (172.5,217.58) -- (172.5,197.54) ;
\draw [line width=0.5]    (112.5,277.58) -- (112.5,257.58) ;
\draw [line width=0.5]    (292.5,117.54) -- (292.5,97.54) ;
\draw  [fill={rgb, 255:red, 245; green, 240; blue, 235 }  ,fill opacity=1 ][line width=0.5]  (222.5,97.54) -- (222.5,57.54) -- (302.5,57.54) -- (302.5,97.54) -- cycle ;
\draw [line width=0.5]    (232.5,57.54) -- (232.5,37.54) ;
\draw [line width=0.5]    (292.5,57.54) -- (292.5,37.5) ;
\draw [line width=0.5]    (232.5,117.54) -- (232.5,97.54) ;
\draw [line width=0.5]    (352.5,197.54) -- (352.5,177.54) ;
\draw  [fill={rgb, 255:red, 245; green, 240; blue, 235 }  ,fill opacity=1 ][line width=0.5]  (282.5,177.54) -- (282.5,137.54) -- (362.5,137.54) -- (362.5,177.54) -- cycle ;
\draw [line width=0.5]    (292.5,137.54) -- (292.5,117.54) ;
\draw [line width=0.5]    (352.5,137.54) -- (352.5,117.5) ;
\draw [line width=0.5]    (292.5,197.54) -- (292.5,177.54) ;
\draw [line width=0.5]    (292.5,277.54) -- (292.5,257.54) ;
\draw  [fill={rgb, 255:red, 245; green, 240; blue, 235 }  ,fill opacity=1 ][line width=0.5]  (222.5,257.54) -- (222.5,217.54) -- (302.5,217.54) -- (302.5,257.54) -- cycle ;
\draw [line width=0.5]    (232.5,217.54) -- (232.5,197.54) ;
\draw [line width=0.5]    (292.5,217.54) -- (292.5,197.5) ;
\draw [line width=0.5]    (232.5,277.54) -- (232.5,257.54) ;

\draw [line width=0.5]    (170,120.04) -- (170,100.04) ;
\draw  [fill={rgb, 255:red, 245; green, 240; blue, 235 }  ,fill opacity=1 ][line width=0.5]  (100,100.04) -- (100,60.04) -- (180,60.04) -- (180,100.04) -- cycle ;
\draw [line width=0.5]    (110,60.04) -- (110,40.04) ;
\draw [line width=0.5]    (110,140.04) -- (110,120.04) ;
\draw [line width=0.5]    (170,60.04) -- (170,40) ;
\draw [line width=0.5]    (110,200.04) -- (110,180.04) ;
\draw  [fill={rgb, 255:red, 245; green, 240; blue, 235 }  ,fill opacity=1 ][line width=0.5]  (40,180.04) -- (40,140.04) -- (120,140.04) -- (120,180.04) -- cycle ;
\draw [line width=0.5]    (50,200.04) -- (50,180.04) ;
\draw [line width=0.5]    (50,140.04) -- (50,120.04) ;
\draw [line width=0.5]    (110,120.04) -- (110,100.04) ;
\draw [line width=0.5]    (230,200.04) -- (230,180.04) ;
\draw  [fill={rgb, 255:red, 245; green, 240; blue, 235 }  ,fill opacity=1 ][line width=0.5]  (160,180.04) -- (160,140.04) -- (240,140.04) -- (240,180.04) -- cycle ;
\draw [line width=0.5]    (170,140.04) -- (170,120.04) ;
\draw [line width=0.5]    (230,140.04) -- (230,120) ;
\draw [line width=0.5]    (170,200.04) -- (170,180.04) ;
\draw [line width=0.5]    (170,280.04) -- (170,260.04) ;
\draw  [fill={rgb, 255:red, 245; green, 240; blue, 235 }  ,fill opacity=1 ][line width=0.5]  (100,260.04) -- (100,220.04) -- (180,220.04) -- (180,260.04) -- cycle ;
\draw [line width=0.5]    (110,220.04) -- (110,200.04) ;
\draw [line width=0.5]    (170,220.04) -- (170,200) ;
\draw [line width=0.5]    (110,280.04) -- (110,260.04) ;
\draw [line width=0.5]    (290,120) -- (290,100) ;
\draw  [fill={rgb, 255:red, 245; green, 240; blue, 235 }  ,fill opacity=1 ][line width=0.5]  (220,100) -- (220,60) -- (300,60) -- (300,100) -- cycle ;
\draw [line width=0.5]    (230,60) -- (230,40) ;
\draw [line width=0.5]    (290,60) -- (290,39.96) ;
\draw [line width=0.5]    (230,120) -- (230,100) ;
\draw [line width=0.5]    (350,200) -- (350,180) ;
\draw  [fill={rgb, 255:red, 245; green, 240; blue, 235 }  ,fill opacity=1 ][line width=0.5]  (280,180) -- (280,140) -- (360,140) -- (360,180) -- cycle ;
\draw [line width=0.5]    (290,140) -- (290,120) ;
\draw [line width=0.5]    (350,140) -- (350,119.96) ;
\draw [line width=0.5]    (290,200) -- (290,180) ;
\draw [line width=0.5]    (290,280) -- (290,260) ;
\draw  [fill={rgb, 255:red, 245; green, 240; blue, 235 }  ,fill opacity=1 ][line width=0.5]  (220,260) -- (220,220) -- (300,220) -- (300,260) -- cycle ;
\draw [line width=0.5]    (230,220) -- (230,200) ;
\draw [line width=0.5]    (290,220) -- (290,199.96) ;
\draw [line width=0.5]    (230,280) -- (230,260) ;

\end{mytikz}
\quad\quad \xrightarrow{~~\mathbb{E}~~} \quad\quad
\begin{mytikz2}
\draw [line width=0.5]    (220,140) -- (270,110) ;
\draw  [fill={rgb, 255:red, 245; green, 240; blue, 235 }  ,fill opacity=1 ][line width=0.5]  (260,70) -- (280,70) -- (280,90) -- (260,90) -- cycle ;
\draw [line width=0.5]    (270,50) -- (270,70) ;
\draw [line width=0.5]    (270,90) -- (270,110) ;
\draw [line width=0.5]    (270,50) -- (303.52,28.32) ;
\draw [line width=0.5]    (270,50) -- (237.12,28.72) ;
\draw  [fill={rgb, 255:red, 245; green, 240; blue, 235 }  ,fill opacity=1 ][line width=0.5]  (286.08,40.29) -- (289.35,26.53) -- (299.83,43.56) -- cycle ;
\draw [line width=0.5]    (320,140) -- (270,110) ;
\draw  [fill={rgb, 255:red, 245; green, 240; blue, 235 }  ,fill opacity=1 ][line width=0.5]  (253.23,39.52) -- (239.33,42.12) -- (250.63,25.62) -- cycle ;
\draw  [fill={rgb, 255:red, 245; green, 240; blue, 235 }  ,fill opacity=1 ][line width=0.5]  (254.43,119.04) -- (251.08,132.78) -- (240.69,115.69) -- cycle ;
\draw  [fill={rgb, 255:red, 142; green, 142; blue, 142 }  ,fill opacity=1 ][line width=0.5]  (266.5,110) .. controls (266.5,108.07) and (268.07,106.5) .. (270,106.5) .. controls (271.93,106.5) and (273.5,108.07) .. (273.5,110) .. controls (273.5,111.93) and (271.93,113.5) .. (270,113.5) .. controls (268.07,113.5) and (266.5,111.93) .. (266.5,110) -- cycle ;
\draw  [fill={rgb, 255:red, 142; green, 142; blue, 142 }  ,fill opacity=1 ][line width=0.5]  (266.5,50) .. controls (266.5,48.07) and (268.07,46.5) .. (270,46.5) .. controls (271.93,46.5) and (273.5,48.07) .. (273.5,50) .. controls (273.5,51.93) and (271.93,53.5) .. (270,53.5) .. controls (268.07,53.5) and (266.5,51.93) .. (266.5,50) -- cycle ;
\draw  [fill={rgb, 255:red, 245; green, 240; blue, 235 }  ,fill opacity=1 ][line width=0.5]  (286.82,119.55) -- (300.87,116.07) -- (290.15,132.96) -- cycle ;
\draw [line width=0.5]    (187.28,219.76) -- (220,200) ;
\draw  [fill={rgb, 255:red, 245; green, 240; blue, 235 }  ,fill opacity=1 ][line width=0.5]  (210,160) -- (230,160) -- (230,180) -- (210,180) -- cycle ;
\draw [line width=0.5]    (220,140) -- (220,160) ;
\draw [line width=0.5]    (220,180) -- (220,200) ;
\draw [line width=0.5]    (220,140) -- (187.28,120) ;
\draw  [fill={rgb, 255:red, 245; green, 240; blue, 235 }  ,fill opacity=1 ][line width=0.5]  (236.08,130.29) -- (239.35,116.53) -- (249.83,133.56) -- cycle ;
\draw [line width=0.5]    (270,230) -- (220,200) ;
\draw  [fill={rgb, 255:red, 245; green, 240; blue, 235 }  ,fill opacity=1 ][line width=0.5]  (203.23,129.52) -- (189.33,132.12) -- (200.63,115.62) -- cycle ;
\draw  [fill={rgb, 255:red, 245; green, 240; blue, 235 }  ,fill opacity=1 ][line width=0.5]  (204.43,209.04) -- (201.08,222.78) -- (190.69,205.69) -- cycle ;
\draw  [fill={rgb, 255:red, 142; green, 142; blue, 142 }  ,fill opacity=1 ][line width=0.5]  (216.5,200) .. controls (216.5,198.07) and (218.07,196.5) .. (220,196.5) .. controls (221.93,196.5) and (223.5,198.07) .. (223.5,200) .. controls (223.5,201.93) and (221.93,203.5) .. (220,203.5) .. controls (218.07,203.5) and (216.5,201.93) .. (216.5,200) -- cycle ;
\draw  [fill={rgb, 255:red, 142; green, 142; blue, 142 }  ,fill opacity=1 ][line width=0.5]  (216.5,140) .. controls (216.5,138.07) and (218.07,136.5) .. (220,136.5) .. controls (221.93,136.5) and (223.5,138.07) .. (223.5,140) .. controls (223.5,141.93) and (221.93,143.5) .. (220,143.5) .. controls (218.07,143.5) and (216.5,141.93) .. (216.5,140) -- cycle ;
\draw  [fill={rgb, 255:red, 245; green, 240; blue, 235 }  ,fill opacity=1 ][line width=0.5]  (235.5,209.56) -- (249.55,206.08) -- (238.83,222.97) -- cycle ;
\draw [line width=0.5]    (270,230) -- (320,200) ;
\draw  [fill={rgb, 255:red, 245; green, 240; blue, 235 }  ,fill opacity=1 ][line width=0.5]  (310,160) -- (330,160) -- (330,180) -- (310,180) -- cycle ;
\draw [line width=0.5]    (320,140) -- (320,160) ;
\draw [line width=0.5]    (320,180) -- (320,200) ;
\draw [line width=0.5]    (320,140) -- (370,110) ;
\draw  [fill={rgb, 255:red, 245; green, 240; blue, 235 }  ,fill opacity=1 ][line width=0.5]  (336.08,130.29) -- (339.35,116.53) -- (349.83,133.56) -- cycle ;
\draw [line width=0.5]    (370,230) -- (320,200) ;
\draw  [fill={rgb, 255:red, 245; green, 240; blue, 235 }  ,fill opacity=1 ][line width=0.5]  (305.25,130.74) -- (291.44,133.83) -- (302.16,116.94) -- cycle ;
\draw  [fill={rgb, 255:red, 245; green, 240; blue, 235 }  ,fill opacity=1 ][line width=0.5]  (304.43,209.04) -- (301.08,222.78) -- (290.69,205.69) -- cycle ;
\draw  [fill={rgb, 255:red, 142; green, 142; blue, 142 }  ,fill opacity=1 ][line width=0.5]  (316.5,200) .. controls (316.5,198.07) and (318.07,196.5) .. (320,196.5) .. controls (321.93,196.5) and (323.5,198.07) .. (323.5,200) .. controls (323.5,201.93) and (321.93,203.5) .. (320,203.5) .. controls (318.07,203.5) and (316.5,201.93) .. (316.5,200) -- cycle ;
\draw  [fill={rgb, 255:red, 142; green, 142; blue, 142 }  ,fill opacity=1 ][line width=0.5]  (316.5,140) .. controls (316.5,138.07) and (318.07,136.5) .. (320,136.5) .. controls (321.93,136.5) and (323.5,138.07) .. (323.5,140) .. controls (323.5,141.93) and (321.93,143.5) .. (320,143.5) .. controls (318.07,143.5) and (316.5,141.93) .. (316.5,140) -- cycle ;
\draw  [fill={rgb, 255:red, 245; green, 240; blue, 235 }  ,fill opacity=1 ][line width=0.5]  (336.82,209.75) -- (350.87,206.27) -- (340.15,223.16) -- cycle ;
\draw  [fill={rgb, 255:red, 245; green, 240; blue, 235 }  ,fill opacity=1 ][line width=0.5]  (360,70) -- (380,70) -- (380,90) -- (360,90) -- cycle ;
\draw [line width=0.5]    (370,50) -- (370,70) ;
\draw [line width=0.5]    (370,90) -- (370,110) ;
\draw [line width=0.5]    (370,50) -- (401.52,30.72) ;
\draw [line width=0.5]    (370,50) -- (338.72,29.52) ;
\draw  [fill={rgb, 255:red, 245; green, 240; blue, 235 }  ,fill opacity=1 ][line width=0.5]  (386.08,40.29) -- (389.35,26.53) -- (399.83,43.56) -- cycle ;
\draw [line width=0.5]    (420,140) -- (370,110) ;
\draw  [fill={rgb, 255:red, 245; green, 240; blue, 235 }  ,fill opacity=1 ][line width=0.5]  (353.23,39.52) -- (339.33,42.12) -- (350.63,25.62) -- cycle ;
\draw  [fill={rgb, 255:red, 245; green, 240; blue, 235 }  ,fill opacity=1 ][line width=0.5]  (354.43,119.04) -- (351.08,132.78) -- (340.69,115.69) -- cycle ;
\draw  [fill={rgb, 255:red, 142; green, 142; blue, 142 }  ,fill opacity=1 ][line width=0.5]  (366.5,110) .. controls (366.5,108.07) and (368.07,106.5) .. (370,106.5) .. controls (371.93,106.5) and (373.5,108.07) .. (373.5,110) .. controls (373.5,111.93) and (371.93,113.5) .. (370,113.5) .. controls (368.07,113.5) and (366.5,111.93) .. (366.5,110) -- cycle ;
\draw  [fill={rgb, 255:red, 142; green, 142; blue, 142 }  ,fill opacity=1 ][line width=0.5]  (366.5,50) .. controls (366.5,48.07) and (368.07,46.5) .. (370,46.5) .. controls (371.93,46.5) and (373.5,48.07) .. (373.5,50) .. controls (373.5,51.93) and (371.93,53.5) .. (370,53.5) .. controls (368.07,53.5) and (366.5,51.93) .. (366.5,50) -- cycle ;
\draw [line width=0.5]    (237.84,309.6) -- (270,290) ;
\draw  [fill={rgb, 255:red, 245; green, 240; blue, 235 }  ,fill opacity=1 ][line width=0.5]  (260,250) -- (280,250) -- (280,270) -- (260,270) -- cycle ;
\draw [line width=0.5]    (270,230) -- (270,250) ;
\draw [line width=0.5]    (270,270) -- (270,290) ;
\draw  [fill={rgb, 255:red, 245; green, 240; blue, 235 }  ,fill opacity=1 ][line width=0.5]  (286.08,220.29) -- (289.35,206.53) -- (299.83,223.56) -- cycle ;
\draw [line width=0.5]    (302.64,310) -- (270,290) ;
\draw  [fill={rgb, 255:red, 245; green, 240; blue, 235 }  ,fill opacity=1 ][line width=0.5]  (253.97,220.78) -- (240.17,223.86) -- (250.89,206.98) -- cycle ;
\draw  [fill={rgb, 255:red, 245; green, 240; blue, 235 }  ,fill opacity=1 ][line width=0.5]  (254.43,299.04) -- (251.08,312.78) -- (240.69,295.69) -- cycle ;
\draw  [fill={rgb, 255:red, 142; green, 142; blue, 142 }  ,fill opacity=1 ][line width=0.5]  (266.5,290) .. controls (266.5,288.07) and (268.07,286.5) .. (270,286.5) .. controls (271.93,286.5) and (273.5,288.07) .. (273.5,290) .. controls (273.5,291.93) and (271.93,293.5) .. (270,293.5) .. controls (268.07,293.5) and (266.5,291.93) .. (266.5,290) -- cycle ;
\draw  [fill={rgb, 255:red, 142; green, 142; blue, 142 }  ,fill opacity=1 ][line width=0.5]  (266.5,230) .. controls (266.5,228.07) and (268.07,226.5) .. (270,226.5) .. controls (271.93,226.5) and (273.5,228.07) .. (273.5,230) .. controls (273.5,231.93) and (271.93,233.5) .. (270,233.5) .. controls (268.07,233.5) and (266.5,231.93) .. (266.5,230) -- cycle ;
\draw  [fill={rgb, 255:red, 245; green, 240; blue, 235 }  ,fill opacity=1 ][line width=0.5]  (286.86,299.67) -- (300.97,296.45) -- (289.95,313.14) -- cycle ;
\draw [line width=0.5]    (370,230) -- (420,200) ;
\draw  [fill={rgb, 255:red, 245; green, 240; blue, 235 }  ,fill opacity=1 ][line width=0.5]  (410,160) -- (430,160) -- (430,180) -- (410,180) -- cycle ;
\draw [line width=0.5]    (420,140) -- (420,160) ;
\draw [line width=0.5]    (420,180) -- (420,200) ;
\draw [line width=0.5]    (420,140) -- (453.92,120.24) ;
\draw  [fill={rgb, 255:red, 245; green, 240; blue, 235 }  ,fill opacity=1 ][line width=0.5]  (436.08,130.29) -- (439.35,116.53) -- (449.83,133.56) -- cycle ;
\draw [line width=0.5]    (453.52,220.24) -- (420,200) ;
\draw  [fill={rgb, 255:red, 245; green, 240; blue, 235 }  ,fill opacity=1 ][line width=0.5]  (404.43,209.04) -- (401.08,222.78) -- (390.69,205.69) -- cycle ;
\draw  [fill={rgb, 255:red, 142; green, 142; blue, 142 }  ,fill opacity=1 ][line width=0.5]  (416.5,200) .. controls (416.5,198.07) and (418.07,196.5) .. (420,196.5) .. controls (421.93,196.5) and (423.5,198.07) .. (423.5,200) .. controls (423.5,201.93) and (421.93,203.5) .. (420,203.5) .. controls (418.07,203.5) and (416.5,201.93) .. (416.5,200) -- cycle ;
\draw  [fill={rgb, 255:red, 142; green, 142; blue, 142 }  ,fill opacity=1 ][line width=0.5]  (416.5,140) .. controls (416.5,138.07) and (418.07,136.5) .. (420,136.5) .. controls (421.93,136.5) and (423.5,138.07) .. (423.5,140) .. controls (423.5,141.93) and (421.93,143.5) .. (420,143.5) .. controls (418.07,143.5) and (416.5,141.93) .. (416.5,140) -- cycle ;
\draw  [fill={rgb, 255:red, 245; green, 240; blue, 235 }  ,fill opacity=1 ][line width=0.5]  (436.86,209.67) -- (450.97,206.45) -- (439.95,223.14) -- cycle ;
\draw  [fill={rgb, 255:red, 245; green, 240; blue, 235 }  ,fill opacity=1 ][line width=0.5]  (385.9,119.56) -- (399.95,116.08) -- (389.23,132.97) -- cycle ;
\draw  [fill={rgb, 255:red, 245; green, 240; blue, 235 }  ,fill opacity=1 ][line width=0.5]  (404.37,130.78) -- (390.57,133.86) -- (401.29,116.98) -- cycle ;
\draw [line width=0.5]    (337.44,309.6) -- (370,290) ;
\draw  [fill={rgb, 255:red, 245; green, 240; blue, 235 }  ,fill opacity=1 ][line width=0.5]  (360,250) -- (380,250) -- (380,270) -- (360,270) -- cycle ;
\draw [line width=0.5]    (370,230) -- (370,250) ;
\draw [line width=0.5]    (370,270) -- (370,290) ;
\draw  [fill={rgb, 255:red, 245; green, 240; blue, 235 }  ,fill opacity=1 ][line width=0.5]  (386.08,220.29) -- (389.35,206.53) -- (399.83,223.56) -- cycle ;
\draw [line width=0.5]    (403.44,310.8) -- (370,290) ;
\draw  [fill={rgb, 255:red, 245; green, 240; blue, 235 }  ,fill opacity=1 ][line width=0.5]  (354.43,299.04) -- (351.08,312.78) -- (340.69,295.69) -- cycle ;
\draw  [fill={rgb, 255:red, 142; green, 142; blue, 142 }  ,fill opacity=1 ][line width=0.5]  (366.5,290) .. controls (366.5,288.07) and (368.07,286.5) .. (370,286.5) .. controls (371.93,286.5) and (373.5,288.07) .. (373.5,290) .. controls (373.5,291.93) and (371.93,293.5) .. (370,293.5) .. controls (368.07,293.5) and (366.5,291.93) .. (366.5,290) -- cycle ;
\draw  [fill={rgb, 255:red, 142; green, 142; blue, 142 }  ,fill opacity=1 ][line width=0.5]  (366.5,230) .. controls (366.5,228.07) and (368.07,226.5) .. (370,226.5) .. controls (371.93,226.5) and (373.5,228.07) .. (373.5,230) .. controls (373.5,231.93) and (371.93,233.5) .. (370,233.5) .. controls (368.07,233.5) and (366.5,231.93) .. (366.5,230) -- cycle ;
\draw  [fill={rgb, 255:red, 245; green, 240; blue, 235 }  ,fill opacity=1 ][line width=0.5]  (386.86,299.67) -- (400.97,296.45) -- (389.95,313.14) -- cycle ;
\draw  [fill={rgb, 255:red, 245; green, 240; blue, 235 }  ,fill opacity=1 ][line width=0.5]  (355.34,221.04) -- (341.54,224.12) -- (352.26,207.24) -- cycle ;
\end{mytikz2}
\quad\quad,
\end{equation}
where the nodes of the COPY tensors form a hexagonal lattice. The degrees of freedom on these nodes interact with each other through the Weingarten matrix $W^{(t)}(d^2)$ on the vertical bonds or the Gram matrix $G^{(t)}(d)$ on the non-vertical bonds. Although some of the elements in the Weingarten matrix are negative, if we integrate out the degrees of freedom on one sublattice of the hexagonal lattice
\begin{equation}
\begin{mytikz2}
\draw [line width=0.5]    (320,140) -- (270,110) ;
\draw  [fill={rgb, 255:red, 245; green, 240; blue, 235 }  ,fill opacity=1 ][line width=0.5]  (286.82,119.55) -- (300.87,116.07) -- (290.15,132.96) -- cycle ;
\draw  [fill={rgb, 255:red, 245; green, 240; blue, 235 }  ,fill opacity=1 ][line width=0.5]  (310,160) -- (330,160) -- (330,180) -- (310,180) -- cycle ;
\draw [line width=0.5]    (320,140) -- (320,160) ;
\draw [line width=0.5]    (320,180) -- (320,200) ;
\draw [line width=0.5]    (320,140) -- (370,110) ;
\draw  [fill={rgb, 255:red, 245; green, 240; blue, 235 }  ,fill opacity=1 ][line width=0.5]  (336.08,130.29) -- (339.35,116.53) -- (349.83,133.56) -- cycle ;
\draw  [fill={rgb, 255:red, 245; green, 240; blue, 235 }  ,fill opacity=1 ][line width=0.5]  (305.25,130.74) -- (291.44,133.83) -- (302.16,116.94) -- cycle ;
\draw  [fill={rgb, 255:red, 142; green, 142; blue, 142 }  ,fill opacity=1 ][line width=0.5]  (316.5,140) .. controls (316.5,138.07) and (318.07,136.5) .. (320,136.5) .. controls (321.93,136.5) and (323.5,138.07) .. (323.5,140) .. controls (323.5,141.93) and (321.93,143.5) .. (320,143.5) .. controls (318.07,143.5) and (316.5,141.93) .. (316.5,140) -- cycle ;
\draw  [fill={rgb, 255:red, 245; green, 240; blue, 235 }  ,fill opacity=1 ][line width=0.5]  (354.43,119.04) -- (351.08,132.78) -- (340.69,115.69) -- cycle ;

\draw (260,100) node   [align=left] {$\sigma $};
\draw (390,97) node   [align=left] {$\sigma' $};
\draw (325,215) node   [align=left] {$\varsigma $};
\end{mytikz2}
\quad\quad=\quad\quad
\begin{mytikz2}
\draw [line width=0.5]    (320,140) -- (270,110) ;
\draw [line width=0.5]    (320,140) -- (320,200) ;
\draw [line width=0.5]    (320,140) -- (370,110) ;
\draw  [fill={rgb, 255:red, 245; green, 240; blue, 235 }  ,fill opacity=1 ][line width=0.5]  (319.84,149.96) -- (310.34,134.06) -- (329.66,134.06) -- cycle ;

\draw (260,100) node   [align=left] {$\sigma $};
\draw (390,97) node   [align=left] {$\sigma'$};
\draw (325,215) node   [align=left] {$\varsigma$};
\end{mytikz2}
\quad,
\end{equation}
the tensor elements of the resulting three-body interaction vertex could be non-negative in certain cases. For example, the three-body ``weight'' matrix for the replica number $t=2$ can be calculated as
\begin{equation}
    \sum_{\tau} \braket{\sigma}{\tau} \braket{\sigma'}{\tau} W^{(2)}_{\tau,\varsigma}(d^2) = \left\{ \begin{array}{ll}
        \delta_{\sigma\varsigma} , & ~\sigma=\sigma', \\
        \displaystyle{\frac{d}{d^2+1}}, & ~\sigma\neq \sigma',
    \end{array}
    \right.
\end{equation}
whose elements are indeed non-negative. Although it contains some zero elements that are not strictly positive, they can be seen as the ``weights'' associated with some infinitely high-energy configurations. For higher orders $t\geq 3$, the resulting tensor elements of the interaction vertex may still be negative~\cite{Zhou2019}. After integrating out one sublattice, the remaining degrees of freedom on the other sublattice form a triangular lattice with three-body interactions existing on half of the triangular plaquettes where the integrated nodes were located. The initial state and the observable to be measured become the bottom and top boundary conditions of this tensor network. Especially, for a single-qudit state such as $\ket{0}$, it holds that
\begin{equation}\label{eq:permutation_vector_zero_state}
    \bra{\sigma}\left(\ket{0} \right)^{\otimes t,t} = (\braket{0}{0})^t = 1,
\end{equation}
for any $t$-degree permutation vector $\ket{\sigma}$. Here, the notation $\ket{\cdot}^{\otimes t,t}$ represents $t$ replicas of a state vector and $t$ replicas of its conjugate.



\section{Theorems and proofs}
\subsection{Matrix product states and sequential preparation circuits}
A matrix product state (MPS) refers to a parametrized pure quantum state where the wave function is given by a product of matrices~\cite{Schollwock2011, Orus2019, Cirac2021}, i.e.,
\begin{equation}\label{eq:mps_def}
    \ket{\Psi_\mathrm{MPS}(A)} = \sum_{n_1 n_2\ldots n_N} \opr{Tr}(A_{1,n_1} A_{2,n_2} \ldots A_{N,n_N}) \ket{n_1 n_2 \ldots n_N},
\end{equation}
where the variable $n_i\in\{0,1,\ldots,d-1\}$ labels the basis states of site $i$ with local dimension $d$. $\ket{n_1 n_2 \ldots n_N}=\ket{n_1} \otimes \ket{n_2} \otimes \ldots \otimes \ket{n_N}$ represents the computational basis of the $N$ sites. $A_{1,n_1},A_{2,n_2},\ldots,A_{N,n_N}$ are $N$ parameter matrices corresponding to the basis state $\ket{n_1 n_2 \ldots n_N}$. The trace operation signifies a periodic boundary condition (PBC), which is reduced to an open boundary condition (OBC) if the first and last matrices are reduced to row and column vectors, respectively. In the following, we focus on MPS with OBC for simplicity, which is frequently used in practical computations. If the basis state variable $n_i$ is regarded as the third index, the set of matrices $A_{i, n_i}$ as a whole becomes a $3$-order tensor $A_{i}$ and MPS can be represented as a chain of tensors $A = \{A_{1}, A_{2},\ldots, A_{N}\}$ whose contraction gives the wave function. For example, the tensor network diagram of an $8$-site MPS is as follows
\begin{equation}
\begin{mytikz3}

\draw  [fill={rgb, 255:red, 245; green, 240; blue, 235 }  ,fill opacity=1 ][line width=0.5]  (50,150) -- (80,150) -- (80,180) -- (50,180) -- cycle ;
\draw [line width=0.5]    (120,165) -- (80,165) ;
\draw [line width=0.5]    (64.5,110) -- (64.5,150) ;
\draw  [fill={rgb, 255:red, 245; green, 240; blue, 235 }  ,fill opacity=1 ][line width=0.5]  (120,150) -- (150,150) -- (150,180) -- (120,180) -- cycle ;
\draw [line width=0.5]    (190,165) -- (150,165) ;
\draw [line width=0.5]    (134.5,110) -- (134.5,150) ;
\draw  [fill={rgb, 255:red, 245; green, 240; blue, 235 }  ,fill opacity=1 ][line width=0.5]  (190,150) -- (220,150) -- (220,180) -- (190,180) -- cycle ;
\draw [line width=0.5]    (260,165) -- (220,165) ;
\draw [line width=0.5]    (204.5,110) -- (204.5,150) ;
\draw  [fill={rgb, 255:red, 245; green, 240; blue, 235 }  ,fill opacity=1 ][line width=0.5]  (260,150) -- (290,150) -- (290,180) -- (260,180) -- cycle ;
\draw [line width=0.5]    (330,165) -- (290,165) ;
\draw [line width=0.5]    (274.5,110) -- (274.5,150) ;
\draw  [fill={rgb, 255:red, 245; green, 240; blue, 235 }  ,fill opacity=1 ][line width=0.5]  (330,150) -- (360,150) -- (360,180) -- (330,180) -- cycle ;
\draw [line width=0.5]    (400,165) -- (360,165) ;
\draw [line width=0.5]    (344.5,110) -- (344.5,150) ;
\draw  [fill={rgb, 255:red, 245; green, 240; blue, 235 }  ,fill opacity=1 ][line width=0.5]  (400,150) -- (430,150) -- (430,180) -- (400,180) -- cycle ;
\draw [line width=0.5]    (470,165) -- (430,165) ;
\draw [line width=0.5]    (414.5,110) -- (414.5,150) ;
\draw  [fill={rgb, 255:red, 245; green, 240; blue, 235 }  ,fill opacity=1 ][line width=0.5]  (470,150) -- (500,150) -- (500,180) -- (470,180) -- cycle ;
\draw [line width=0.5]    (540,165) -- (500,165) ;
\draw [line width=0.5]    (484.5,110) -- (484.5,150) ;
\draw  [fill={rgb, 255:red, 245; green, 240; blue, 235 }  ,fill opacity=1 ][line width=0.5]  (540,150) -- (570,150) -- (570,180) -- (540,180) -- cycle ;
\draw [line width=0.5]    (554.5,110) -- (554.5,150) ;

\draw (65,165) node   [align=left] {$A_{1}$};
\draw (135,165) node   [align=left] {$A_{2}$};
\draw (205,165) node   [align=left] {$A_{3}$};
\draw (275,165) node   [align=left] {$A_{4}$};
\draw (345,165) node   [align=left] {$A_{5}$};
\draw (415,165) node   [align=left] {$A_{6}$};
\draw (485,165) node   [align=left] {$A_{7}$};
\draw (555,165) node   [align=left] {$A_{8}$};

\end{mytikz3}\quad.
\end{equation}
The index $n_i$ of the local tensor $A_i$ is called the physical index (bond), and the other two indices are called the virtual indices (bonds). The dimension of the tensor $A_i$ along a virtual bond is often referred to as the bond dimension of the virtual bond. Typically, the bond dimensions of all the virtual bonds in an MPS are set to be uniform, usually denoted as $D$.

MPS has an inherent 1D chain geometry and has achieved great success in solving 1D quantum many-body systems, thanks to the entanglement area law in 1D gapped ground states~\cite{Hastings2007} with at most logarithmic corrections in 1D gapless ground states. For quantum systems in two or higher spatial dimensions, MPS can be naturally generalized to the so-called projected entangled pair states (PEPS) by generalizing the $3$-order local tensors in MPS to higher-order tensors, i.e.,
\begin{equation}\label{eq:peps_def}
   { \ket{\Psi_\mathrm{PEPS}(A)} = \sum_{\{n_i\} } \opr{tTr}(\{ A_{i,n_i} \}) \ket{ \{n_i\} }},
\end{equation}
where $\opr{tTr}$ denotes the contraction over the indices of the tensors according to a given graph $(V, E)$, which encodes the underlying spatial geometry. If the vertex $i\in V$ has $z_i$ edges $e\in E$ on the graph, $A_{i,n_i}$ is a $(z_i+1)$-order parameter tensor with $1$ physical index and $z_i$ virtual indices that are contracted with the corresponding indices of the adjacent tensors in $\opr{tTr}(\cdot)$. The name ``projected entangled pair'' originates from an alternative definition of $\ket{\Psi_\mathrm{PEPS}(A)}$, i.e.,
\begin{equation}\label{eq:peps_def2}
    \ket{\Psi_\mathrm{PEPS}(A)} = \bigotimes_{i\in V} \hat{A}_i \left(\bigotimes_{e\in E} \ket{I_e} \right),
\end{equation}
where $\ket{I_e}=\sum_{m}\ket{m,m}_{e}$ represents an unnormalized maximally entangled Bell state defined on the edge $e$. $\hat{A}_i$ is a linear map from the virtual Hilbert spaces to the physical Hilbert space defined by
\begin{equation}
    \hat{A}_i = \sum_{n_i, m_{i}} A_{i,n_i,m_i} \ketbra{n_i}{m_i},
\end{equation}
where the indices $m_i=(m_{i,1}, m_{i,2}, \ldots, m_{i,z_i})$ label the basis states in the tensor product of the $z_i$ virtual spaces corresponding to site $i$. Eq.\,\eqref{eq:peps_def2} is equivalent to Eq.\,\eqref{eq:peps_def} after the virtual basis states $\{\ket{m_i}\}$ are integrated out. Thus, $\ket{\Psi_\mathrm{PEPS}(A)}$ is indeed a state that is ``projected'' by $\hat{A}_i$ from the entangled pairs $\ket{I_e}$. The tensor network diagram has different meanings under these two definitions: the former represents the contraction relation among concrete high-order arrays; the latter represents the composition relation among abstract linear maps. Importantly, MPS is efficient on classical computers because its chain structure guarantees the non-increasing overhead of contraction along the chain direction. However, in two or higher spatial dimensions, the dense loop structures of PEPS generally lead to exponential complexity in classical simulations.


Different MPS representations can represent the same quantum state due to the internal ``gauge'' freedom of MPS. Namely, inserting an arbitrary invertible matrix $B$ and its inverse $B^{-1}$ between $A_{i}$ and $A_{i+1}$ does not change the wave function since the inserted matrices are canceled out by $BB^{-1}=I$; on the other hand, the local tensors are indeed changed to $A_{i}\rightarrow A_{i} B$ and $A_{i+1}\rightarrow B^{-1} A_{i+1}$. Using the gauge freedom, the local tensors of MPS can be restricted to isometries without loss of generality. Specifically, the local tensor $A_i$ is left (right) isometric if the matrix obtained by combining the physical index with the left (right) virtual index of $A_i$ is an isometry, i.e.,
\begin{equation}\label{eq:tn_isometric}
\begin{mytikz3}

\draw [line width=0.5]    (140,126) -- (125,126) ;
\draw [line width=0.5]    (140,205) -- (125,205) ;
\draw [line width=0.5]    (200,205.18) -- (170,205) ;
\draw [line width=0.5]    (154.5,142.63) -- (154.5,190) ;
\draw  [fill={rgb, 255:red, 245; green, 240; blue, 235 }  ,fill opacity=1 ][line width=0.5]  (130.33,208.67) -- (158.62,180.38) -- (179.83,201.6) -- (151.55,229.88) -- cycle ;
\draw [line width=0.5]    (200,126) -- (170,126) ;
\draw  [fill={rgb, 255:red, 245; green, 240; blue, 235 }  ,fill opacity=1 ][line width=0.5]  (151,101) -- (179.28,129.28) -- (158.07,150.5) -- (129.79,122.21) -- cycle ;
\draw [line width=0.5]    (125,126) .. controls (105.4,126.02) and (95.65,145.43) .. (95.65,165.43) .. controls (95.65,185.43) and (105.65,204.77) .. (125,205) ;

\draw (155,205) node   [align=left] {$ A_{i}$};
\draw (155,125) node   [align=left] {$ A_{i}^{*}$};

\end{mytikz3}
~~=~~
\begin{mytikz3}

\draw [line width=0.5]    (160,126) -- (125,126) ;
\draw [line width=0.5]    (160,205) -- (125,205) ;
\draw [line width=0.5]    (125,126) .. controls (105.4,126.02) and (95.7,145.43) .. (95.7,165.43) .. controls (95.7,185.43) and (105.65,204.77) .. (125,205) ;

\end{mytikz3}
\quad\quad\text{or}\quad\quad
\begin{mytikz3}

\draw [line width=0.5]    (120.05,104.97) -- (90,104.97) ;
\draw [line width=0.5]    (120.05,184.97) -- (90,185) ;
\draw [line width=0.5]    (165,184.97) -- (150.05,184.8) ;
\draw [line width=0.5]    (134.55,122.6) -- (134.55,169.97) ;
\draw  [fill={rgb, 255:red, 245; green, 240; blue, 235 }  ,fill opacity=1 ][line width=0.5]  (131.6,160.35) -- (159.88,188.64) -- (138.67,209.85) -- (110.38,181.57) -- cycle ;
\draw [line width=0.5]    (165,104.97) -- (150.05,104.97) ;
\draw  [fill={rgb, 255:red, 245; green, 240; blue, 235 }  ,fill opacity=1 ][line width=0.5]  (109.84,108.68) -- (138.12,80.4) -- (159.33,101.61) -- (131.05,129.9) -- cycle ;
\draw [line width=0.5]    (165,104.97) .. controls (185.16,104.81) and (195.16,125.21) .. (195.16,145.21)(195,144.97) .. controls (195,164.97) and (185.16,184.61) .. (165,184.97) ;

\draw (135.13,185.1) node   [align=left] {$ A_{i}$};
\draw (134.59,105.15) node   [align=left] {$ A_{i}^{*}$};

\end{mytikz3}
~~=~~
\begin{mytikz3}

\draw [line width=0.5]    (165,184.97) -- (130,184.97) ;
\draw [line width=0.5]    (165,104.97) -- (130,104.97) ;
\draw [line width=0.5]    (165,104.97) .. controls (185.16,104.81) and (195.16,125.21) .. (195.16,145.21)(195,144.97) .. controls (195,164.97) and (185.16,184.61) .. (165,184.97) ;

\end{mytikz3}
\quad,
\end{equation}
respectively, where we use inclined rectangles instead of the squares used above to highlight the left or right isometric property. An MPS is left (right) canonical if all the local tensors in the MPS are left (right) isometric. More generally, an MPS is mixed canonical if $\{A_1,A_2,\ldots,A_{i-1}\}$ are left isometric and $\{A_{i+1}, A_{i+2}, \ldots, A_{N}\}$ are right isometric, with the local tensor $A_{i}$ called the orthogonality center. For example, the following are the tensor network diagrams of an $8$-site MPS in the left canonical form
\begin{equation}
\begin{mytikz3}

\draw [line width=0.5]    (135,60) -- (135,100) ;
\draw [line width=0.5]    (160,100) -- (135,100) ;
\draw  [fill={rgb, 255:red, 245; green, 240; blue, 235 }  ,fill opacity=1 ][line width=0.5]  (117.32,103.54) -- (138.54,82.32) -- (152.68,96.46) -- (131.46,117.68) -- cycle ;
\draw [line width=0.5]    (185,100) -- (160,100) ;
\draw [line width=0.5]    (185,60) -- (185,100) ;
\draw [line width=0.5]    (210,100) -- (185,100) ;
\draw  [fill={rgb, 255:red, 245; green, 240; blue, 235 }  ,fill opacity=1 ][line width=0.5]  (167.32,103.54) -- (188.54,82.32) -- (202.68,96.46) -- (181.46,117.68) -- cycle ;
\draw [line width=0.5]    (235,100) -- (210,100) ;
\draw [line width=0.5]    (235,60) -- (235,100) ;
\draw [line width=0.5]    (260,100) -- (235,100) ;
\draw  [fill={rgb, 255:red, 245; green, 240; blue, 235 }  ,fill opacity=1 ][line width=0.5]  (217.32,103.54) -- (238.54,82.32) -- (252.68,96.46) -- (231.46,117.68) -- cycle ;
\draw [line width=0.5]    (285,100) -- (260,100) ;
\draw [line width=0.5]    (285,60) -- (285,100) ;
\draw [line width=0.5]    (310,100) -- (285,100) ;
\draw  [fill={rgb, 255:red, 245; green, 240; blue, 235 }  ,fill opacity=1 ][line width=0.5]  (267.32,103.54) -- (288.54,82.32) -- (302.68,96.46) -- (281.46,117.68) -- cycle ;
\draw [line width=0.5]    (335,100) -- (310,100) ;
\draw [line width=0.5]    (335,60) -- (335,100) ;
\draw [line width=0.5]    (360,100) -- (335,100) ;
\draw  [fill={rgb, 255:red, 245; green, 240; blue, 235 }  ,fill opacity=1 ][line width=0.5]  (317.32,103.54) -- (338.54,82.32) -- (352.68,96.46) -- (331.46,117.68) -- cycle ;
\draw [line width=0.5]    (385,100) -- (360,100) ;
\draw [line width=0.5]    (385,60) -- (385,100) ;
\draw [line width=0.5]    (410,100) -- (385,100) ;
\draw  [fill={rgb, 255:red, 245; green, 240; blue, 235 }  ,fill opacity=1 ][line width=0.5]  (367.32,103.54) -- (388.54,82.32) -- (402.68,96.46) -- (381.46,117.68) -- cycle ;
\draw [line width=0.5]    (435,100) -- (410,100) ;
\draw [line width=0.5]    (435,60) -- (435,100) ;
\draw [line width=0.5]    (460,100) -- (435,100) ;
\draw  [fill={rgb, 255:red, 245; green, 240; blue, 235 }  ,fill opacity=1 ][line width=0.5]  (417.32,103.54) -- (438.54,82.32) -- (452.68,96.46) -- (431.46,117.68) -- cycle ;
\draw [line width=0.5]    (485,100) -- (460,100) ;
\draw [line width=0.5]    (485,60) -- (485,100) ;
\draw  [fill={rgb, 255:red, 245; green, 240; blue, 235 }  ,fill opacity=1 ][line width=0.5]  (467.32,103.54) -- (488.54,82.32) -- (502.68,96.46) -- (481.46,117.68) -- cycle ;

\end{mytikz3},
\end{equation}
mixed canonical form (centered at $A_{4}$ as an example)
\begin{equation}\label{eq:mixed_canonical_mps}
\begin{mytikz3}

\draw [line width=0.5]    (135,60) -- (135,100) ;
\draw [line width=0.5]    (160,100) -- (135,100) ;
\draw  [fill={rgb, 255:red, 245; green, 240; blue, 235 }  ,fill opacity=1 ][line width=0.5]  (117.32,103.54) -- (138.54,82.32) -- (152.68,96.46) -- (131.46,117.68) -- cycle ;
\draw [line width=0.5]    (185,100) -- (160,100) ;
\draw [line width=0.5]    (185,60) -- (185,100) ;
\draw [line width=0.5]    (210,100) -- (185,100) ;
\draw  [fill={rgb, 255:red, 245; green, 240; blue, 235 }  ,fill opacity=1 ][line width=0.5]  (167.32,103.54) -- (188.54,82.32) -- (202.68,96.46) -- (181.46,117.68) -- cycle ;
\draw [line width=0.5]    (235,100) -- (210,100) ;
\draw [line width=0.5]    (235,60) -- (235,100) ;
\draw [line width=0.5]    (260,100) -- (235,100) ;
\draw  [fill={rgb, 255:red, 245; green, 240; blue, 235 }  ,fill opacity=1 ][line width=0.5]  (217.32,103.54) -- (238.54,82.32) -- (252.68,96.46) -- (231.46,117.68) -- cycle ;
\draw [line width=0.5]    (285,100) -- (260,100) ;
\draw [line width=0.5]    (285,60) -- (285,100) ;
\draw [line width=0.5]    (310,100) -- (285,100) ;
\draw  [fill={rgb, 255:red, 245; green, 240; blue, 235 }  ,fill opacity=1 ][line width=0.5]  (270,85) -- (300,85) -- (299.31,115) -- (269.31,115) -- cycle ;
\draw [line width=0.5]    (335,100) -- (310,100) ;
\draw [line width=0.5]    (335,60) -- (335,100) ;
\draw [line width=0.5]    (360,100) -- (335,100) ;
\draw  [fill={rgb, 255:red, 245; green, 240; blue, 235 }  ,fill opacity=1 ][line width=0.5]  (331.46,82.32) -- (352.68,103.54) -- (338.54,117.68) -- (317.32,96.46) -- cycle ;
\draw [line width=0.5]    (385,100) -- (360,100) ;
\draw [line width=0.5]    (385,60) -- (385,100) ;
\draw [line width=0.5]    (410,100) -- (385,100) ;
\draw  [fill={rgb, 255:red, 245; green, 240; blue, 235 }  ,fill opacity=1 ][line width=0.5]  (381.46,82.32) -- (402.68,103.54) -- (388.54,117.68) -- (367.32,96.46) -- cycle ;
\draw [line width=0.5]    (435,100) -- (410,100) ;
\draw [line width=0.5]    (435,60) -- (435,100) ;
\draw [line width=0.5]    (460,100) -- (435,100) ;
\draw  [fill={rgb, 255:red, 245; green, 240; blue, 235 }  ,fill opacity=1 ][line width=0.5]  (431.46,82.32) -- (452.68,103.54) -- (438.54,117.68) -- (417.32,96.46) -- cycle ;
\draw [line width=0.5]    (485,100) -- (460,100) ;
\draw [line width=0.5]    (485,60) -- (485,100) ;
\draw  [fill={rgb, 255:red, 245; green, 240; blue, 235 }  ,fill opacity=1 ][line width=0.5]  (481.46,82.32) -- (502.68,103.54) -- (488.54,117.68) -- (467.32,96.46) -- cycle ;

\end{mytikz3},
\end{equation}
and right canonical form
\begin{equation}\label{eq:right_canonical_mps}
\begin{mytikz3}

\draw [line width=0.5]    (135,60) -- (135,100) ;
\draw [line width=0.5]    (160,100) -- (135,100) ;
\draw  [fill={rgb, 255:red, 245; green, 240; blue, 235 }  ,fill opacity=1 ][line width=0.5]  (131.46,82.32) -- (152.68,103.54) -- (138.54,117.68) -- (117.32,96.46) -- cycle ;
\draw [line width=0.5]    (185,100) -- (160,100) ;
\draw [line width=0.5]    (185,60) -- (185,100) ;
\draw [line width=0.5]    (210,100) -- (185,100) ;
\draw  [fill={rgb, 255:red, 245; green, 240; blue, 235 }  ,fill opacity=1 ][line width=0.5]  (181.46,82.32) -- (202.68,103.54) -- (188.54,117.68) -- (167.32,96.46) -- cycle ;
\draw [line width=0.5]    (235,100) -- (210,100) ;
\draw [line width=0.5]    (235,60) -- (235,100) ;
\draw [line width=0.5]    (260,100) -- (235,100) ;
\draw  [fill={rgb, 255:red, 245; green, 240; blue, 235 }  ,fill opacity=1 ][line width=0.5]  (231.46,82.32) -- (252.68,103.54) -- (238.54,117.68) -- (217.32,96.46) -- cycle ;
\draw [line width=0.5]    (285,100) -- (260,100) ;
\draw [line width=0.5]    (285,60) -- (285,100) ;
\draw [line width=0.5]    (310,100) -- (285,100) ;
\draw [line width=0.5]    (335,100) -- (310,100) ;
\draw [line width=0.5]    (335,60) -- (335,100) ;
\draw [line width=0.5]    (360,100) -- (335,100) ;
\draw  [fill={rgb, 255:red, 245; green, 240; blue, 235 }  ,fill opacity=1 ][line width=0.5]  (331.46,82.32) -- (352.68,103.54) -- (338.54,117.68) -- (317.32,96.46) -- cycle ;
\draw [line width=0.5]    (385,100) -- (360,100) ;
\draw [line width=0.5]    (385,60) -- (385,100) ;
\draw [line width=0.5]    (410,100) -- (385,100) ;
\draw  [fill={rgb, 255:red, 245; green, 240; blue, 235 }  ,fill opacity=1 ][line width=0.5]  (381.46,82.32) -- (402.68,103.54) -- (388.54,117.68) -- (367.32,96.46) -- cycle ;
\draw [line width=0.5]    (435,100) -- (410,100) ;
\draw [line width=0.5]    (435,60) -- (435,100) ;
\draw [line width=0.5]    (460,100) -- (435,100) ;
\draw  [fill={rgb, 255:red, 245; green, 240; blue, 235 }  ,fill opacity=1 ][line width=0.5]  (431.46,82.32) -- (452.68,103.54) -- (438.54,117.68) -- (417.32,96.46) -- cycle ;
\draw [line width=0.5]    (485,100) -- (460,100) ;
\draw [line width=0.5]    (485,60) -- (485,100) ;
\draw  [fill={rgb, 255:red, 245; green, 240; blue, 235 }  ,fill opacity=1 ][line width=0.5]  (481.46,82.32) -- (502.68,103.54) -- (488.54,117.68) -- (467.32,96.46) -- cycle ;
\draw  [fill={rgb, 255:red, 245; green, 240; blue, 235 }  ,fill opacity=1 ][line width=0.5]  (281.46,82.32) -- (302.68,103.54) -- (288.54,117.68) -- (267.32,96.46) -- cycle ;

\end{mytikz3},
\end{equation}
where the symbols on the rectangles are omitted for simplicity. The canonical form of MPS offers convenience for computing the expectation values of local observables, as extensive isometries can be canceled out for a mixed canonical MPS with the orthogonality center positioned near the local observable. In practice, the canonical form can be obtained by a sequence of singular value decompositions (SVD) or QR decompositions of adjacent local tensors. The orthogonality center can also be moved in this way, i.e., by successively performing the SVD or QR decompositions for adjacent tensors to make any one of them an isometry, e.g.,
\begin{equation}\label{eq:mps_svd_qr}
\begin{mytikz3}
\draw [line width=0.5]    (284.65,140) -- (284.65,100) ;
\draw [line width=0.5]    (335,140) -- (335,100) ;
\draw  [fill={rgb, 255:red, 245; green, 240; blue, 235 }  ,fill opacity=1 ][line width=0.5]  (335,145) -- (342.5,135) -- (327.5,135) -- cycle ;
\draw [line width=0.5]    (285,100) -- (250,100) ;
\draw [line width=0.5]    (285,60) -- (285,100) ;
\draw [line width=0.5]    (310,100) -- (285,100) ;
\draw  [fill={rgb, 255:red, 245; green, 240; blue, 235 }  ,fill opacity=1 ][line width=0.5]  (270,85) -- (300,85) -- (300,115) -- (270,115) -- cycle ;
\draw [line width=0.5]    (335,100) -- (310,100) ;
\draw [line width=0.5]    (335,60) -- (335,100) ;
\draw [line width=0.5]    (370,100) -- (335,100) ;
\draw  [fill={rgb, 255:red, 245; green, 240; blue, 235 }  ,fill opacity=1 ][line width=0.5]  (331.46,82.32) -- (352.68,103.54) -- (338.54,117.68) -- (317.32,96.46) -- cycle ;
\draw  [fill={rgb, 255:red, 245; green, 240; blue, 235 }  ,fill opacity=1 ][line width=0.5]  (284.65,145) -- (292.15,135) -- (277.15,135) -- cycle ;

\draw (286,160) node  [font=\footnotesize] [align=left] {$i$};
\draw (336,160) node  [font=\footnotesize] [align=left] {$i+1$};

\end{mytikz3}\quad\xrightarrow{\text{contraction}}\quad
\begin{mytikz3}
\draw [line width=0.5]    (285,100) -- (250,100) ;
\draw [line width=0.5]    (285,60) -- (285,100) ;
\draw [line width=0.5]    (310,100) -- (285,100) ;
\draw [line width=0.5]    (335,100) -- (310,100) ;
\draw [line width=0.5]    (335,60) -- (335,100) ;
\draw [line width=0.5]    (370,100) -- (335,100) ;
\draw  [fill={rgb, 255:red, 245; green, 240; blue, 235 }  ,fill opacity=1 ][line width=0.5]  (270,85) -- (350,85) -- (350,115) -- (270,115) -- cycle ;
\draw [line width=0.5]  [dash pattern={on 3pt off 2pt}]  (310,65) -- (310,140) ;

\draw (286,160) node  [font=\footnotesize] [align=left] {$i$};
\draw (336,160) node  [font=\footnotesize] [align=left] {$i+1$};

\end{mytikz3}\quad\xrightarrow{\text{SVD/QR}}\quad
\begin{mytikz3}

\draw [line width=0.5]    (285,140) -- (285,100) ;
\draw [line width=0.5]    (335,140) -- (335,100) ;
\draw  [fill={rgb, 255:red, 245; green, 240; blue, 235 }  ,fill opacity=1 ][line width=0.5]  (335,145) -- (342.5,135) -- (327.5,135) -- cycle ;
\draw [line width=0.5]    (285,100) -- (250,100) ;
\draw [line width=0.5]    (285,60) -- (285,100) ;
\draw [line width=0.5]    (310,100) -- (285,100) ;
\draw [line width=0.5]    (335,100) -- (310,100) ;
\draw [line width=0.5]    (335,60) -- (335,100) ;
\draw [line width=0.5]    (370,100) -- (335,100) ;
\draw  [fill={rgb, 255:red, 245; green, 240; blue, 235 }  ,fill opacity=1 ][line width=0.5]  (302.33,96.46) -- (281.12,117.68) -- (266.98,103.54) -- (288.19,82.32) -- cycle ;
\draw  [fill={rgb, 255:red, 245; green, 240; blue, 235 }  ,fill opacity=1 ][line width=0.5]  (284.65,145) -- (292.15,135) -- (277.15,135) -- cycle ;
\draw  [fill={rgb, 255:red, 245; green, 240; blue, 235 }  ,fill opacity=1 ][line width=0.5]  (320,85) -- (350,85) -- (350,115) -- (320,115) -- cycle ;

\draw (286,160) node  [font=\footnotesize] [align=left] {$i$};
\draw (336,160) node  [font=\footnotesize] [align=left] {$i+1$};

\end{mytikz3}\quad,
\end{equation}
where the orthogonality center is moved from site $i$ to site $i+1$ without changing the product of the two adjacent tensors. The dashed line in Eq.\,\eqref{eq:mps_svd_qr} indicates the way of index partition when performing the decomposition.

Using the canonical forms, the preparation quantum circuits that generate the MPS can be easily constructed by completing the isometries $A_i$ to unitaries $U_i$, e.g.,
\begin{equation}
\begin{mytikz3}

\draw [line width=0.5]    (185,100) -- (135,100) ;
\draw [line width=0.5]    (185,50) -- (185,100) ;
\draw [line width=0.5]    (235,100) -- (185,100) ;
\draw  [fill={rgb, 255:red, 245; green, 240; blue, 235 }  ,fill opacity=1 ][line width=0.5]  (181.46,75.25) -- (209.75,103.54) -- (188.54,124.75) -- (160.25,96.46) -- cycle ;

\draw (185,99) node   [align=left] {$ A_{i}$};
\draw (186,154) node  [font=\footnotesize] [align=left] {$\ $};

\end{mytikz3}
~~=~~
\begin{mytikz3}

\draw [line width=0.5]    (185,100) -- (185,140) ;
\draw [line width=0.5]    (185,100) -- (135,100) ;
\draw [line width=0.5]    (185,50) -- (185,100) ;
\draw [line width=0.5]    (235,100) -- (185,100) ;
\draw  [fill={rgb, 255:red, 245; green, 240; blue, 235 }  ,fill opacity=1 ][line width=0.5]  (181.46,75.25) -- (209.75,103.54) -- (188.54,124.75) -- (160.25,96.46) -- cycle ;
\draw  [fill={rgb, 255:red, 245; green, 240; blue, 235 }  ,fill opacity=1 ][line width=0.5]  (185,160.2) -- (200,140.2) -- (170,140.2) -- cycle ;

\draw (185,100) node   [align=left] {$U_{i}$};
\draw (185,148) node  [font=\footnotesize] [align=left] {$0$};

\end{mytikz3}\quad,
\end{equation}
where the triangle with $0$ represents an input state $\ket{0}$. Thus, the circuit diagrams can be obtained by ``$45^{\circ}$-rotating'' the MPS diagrams depicted above. For example, for an $8$-site MPS with uniform bond dimension $D=d^{2}$, the exact preparation circuit starting from an initial product state $\ket{0}^{\otimes N}$ corresponding to the right canonical form in Eq.\,\eqref{eq:right_canonical_mps} is
\begin{equation}\label{eq:right_canonical_circuit_exact}
\begin{mytikz3}

\draw [line width=0.5]    (230,300) -- (230,30) ;
\draw [line width=0.5]    (190,300) -- (190,30) ;
\draw [line width=0.5]    (150,300) -- (150,30) ;
\draw [line width=0.5]    (110,300) -- (110,30) ;
\draw [line width=0.5]    (390,300) -- (390,30) ;
\draw [line width=0.5]    (350,300) -- (350,30) ;
\draw [line width=0.5]    (310,300) -- (310,30) ;
\draw [line width=0.5]    (270,300) -- (270,30) ;
\draw  [fill={rgb, 255:red, 245; green, 240; blue, 235 }  ,fill opacity=1 ][line width=0.5]  (90,260) -- (170,260) -- (170,280) -- (90,280) -- cycle ;
\draw  [fill={rgb, 255:red, 245; green, 240; blue, 235 }  ,fill opacity=1 ][line width=0.5]  (130,230) -- (250,230) -- (250,250) -- (130,250) -- cycle ;
\draw  [fill={rgb, 255:red, 245; green, 240; blue, 235 }  ,fill opacity=1 ][line width=0.5]  (170,200) -- (290,200) -- (290,220) -- (170,220) -- cycle ;
\draw  [fill={rgb, 255:red, 245; green, 240; blue, 235 }  ,fill opacity=1 ][line width=0.5]  (210,170) -- (330,170) -- (330,190) -- (210,190) -- cycle ;
\draw  [fill={rgb, 255:red, 245; green, 240; blue, 235 }  ,fill opacity=1 ][line width=0.5]  (250,140) -- (370,140) -- (370,160) -- (250,160) -- cycle ;
\draw  [fill={rgb, 255:red, 245; green, 240; blue, 235 }  ,fill opacity=1 ][line width=0.5]  (290,110) -- (410,110) -- (410,130) -- (290,130) -- cycle ;
\draw  [fill={rgb, 255:red, 245; green, 240; blue, 235 }  ,fill opacity=1 ][line width=0.5]  (330,80) -- (410,80) -- (410,100) -- (330,100) -- cycle ;
\draw  [fill={rgb, 255:red, 245; green, 240; blue, 235 }  ,fill opacity=1 ][line width=0.5]  (370,50) -- (410,50) -- (410,70) -- (370,70) -- cycle ;

\end{mytikz3}
\quad.
\end{equation}
The gates are arranged sequentially in a staircase pattern of size $\beta=\log_d D + 1$ in the bulk. The deformation near the boundaries occurs because the Schmidt ranks are automatically cut off by the Hilbert space of the smaller part of the bipartition. If the gates are universal on their supports, the circuit in Eq.\,\eqref{eq:right_canonical_circuit_exact} can be simplified to a ``homogeneous'' canonical form
\begin{equation}\label{eq:right_canonical_circuit_regular}
\begin{mytikz3}

\draw [line width=0.5]    (230,300) -- (230,90.02) ;
\draw [line width=0.5]    (190,300) -- (190,90.02) ;
\draw [line width=0.5]    (150,300) -- (150,90.02) ;
\draw [line width=0.5]    (110,300) -- (110,90.02) ;
\draw [line width=0.5]    (390,300) -- (390,90.02) ;
\draw [line width=0.5]    (350,300) -- (350,90.02) ;
\draw [line width=0.5]    (310,300) -- (310,90.02) ;
\draw [line width=0.5]    (270,300) -- (270,90.02) ;
\draw  [fill={rgb, 255:red, 245; green, 240; blue, 235 }  ,fill opacity=1 ][line width=0.5]  (90,260) -- (210,260) -- (210,280) -- (90,280) -- cycle ;
\draw  [fill={rgb, 255:red, 245; green, 240; blue, 235 }  ,fill opacity=1 ][line width=0.5]  (130,230) -- (250,230) -- (250,250) -- (130,250) -- cycle ;
\draw  [fill={rgb, 255:red, 245; green, 240; blue, 235 }  ,fill opacity=1 ][line width=0.5]  (170,200) -- (290,200) -- (290,220) -- (170,220) -- cycle ;
\draw  [fill={rgb, 255:red, 245; green, 240; blue, 235 }  ,fill opacity=1 ][line width=0.5]  (210,170) -- (330,170) -- (330,190) -- (210,190) -- cycle ;
\draw  [fill={rgb, 255:red, 245; green, 240; blue, 235 }  ,fill opacity=1 ][line width=0.5]  (250,140) -- (370,140) -- (370,160) -- (250,160) -- cycle ;
\draw  [fill={rgb, 255:red, 245; green, 240; blue, 235 }  ,fill opacity=1 ][line width=0.5]  (290,110) -- (410,110) -- (410,130) -- (290,130) -- cycle ;

\end{mytikz3}\quad,
\end{equation}
with the set of generated states unchanged. The smaller gates $U_7,U_8$ in Eq.\,\eqref{eq:right_canonical_circuit_exact} at the right edge are ``absorbed'' into the larger gate $U_6$ whose support already covers the right edge. Similarly, the smaller gate $U_1$ at the left edge can be embedded into a larger gate by adding an extra input qudit. Transforming into the same horizontal arrangement as in Eq.\,\eqref{eq:right_canonical_mps}, Eq.\,\eqref{eq:right_canonical_circuit_regular} becomes
\begin{equation}\label{eq:right_canonical_circuit_horizontal}
\begin{mytikz3}

\draw [line width=0.5]    (135,140) -- (135,100) ;
\draw [line width=0.5]    (138,140) -- (138,100) ;
\draw [line width=0.5]    (132,140) -- (132,100) ;
\draw [line width=0.5]    (185,140) -- (185,100) ;
\draw [line width=0.5]    (235,140) -- (235,100) ;
\draw [line width=0.5]    (285,140) -- (285,100) ;
\draw [line width=0.5]    (335,140) -- (335,100) ;
\draw [line width=0.5]    (385,140) -- (385,100) ;
\draw [line width=0.5]    (160,103) -- (135,103) ;
\draw [line width=0.5]    (185,103) -- (160,103) ;
\draw [line width=0.5]    (210,103) -- (185,103) ;
\draw [line width=0.5]    (235,103) -- (210,103) ;
\draw [line width=0.5]    (260,103) -- (235,103) ;
\draw [line width=0.5]    (285,103) -- (260,103) ;
\draw [line width=0.5]    (310,103) -- (285,103) ;
\draw [line width=0.5]    (335,103) -- (310,103) ;
\draw [line width=0.5]    (360,103) -- (335,103) ;
\draw [line width=0.5]    (385,103) -- (360,103) ;
\draw [line width=0.5]    (410,103) -- (385,103) ;
\draw [line width=0.5]    (435,103) -- (410,103) ;
\draw [line width=0.5]    (460,103) -- (435,103) ;
\draw [line width=0.5]    (485,103) -- (460,103) ;
\draw [line width=0.5]    (135,60) -- (135,100) ;
\draw [line width=0.5]    (160,100) -- (135,100) ;
\draw  [fill={rgb, 255:red, 245; green, 240; blue, 235 }  ,fill opacity=1 ][line width=0.5]  (127.93,78.79) -- (156.21,107.07) -- (142.07,121.21) -- (113.79,92.93) -- cycle ;
\draw [line width=0.5]    (185,100) -- (160,100) ;
\draw [line width=0.5]    (185,60) -- (185,100) ;
\draw [line width=0.5]    (210,100) -- (185,100) ;
\draw  [fill={rgb, 255:red, 245; green, 240; blue, 235 }  ,fill opacity=1 ][line width=0.5]  (177.93,78.79) -- (206.21,107.07) -- (192.07,121.21) -- (163.79,92.93) -- cycle ;
\draw [line width=0.5]    (235,100) -- (210,100) ;
\draw [line width=0.5]    (235,60) -- (235,100) ;
\draw [line width=0.5]    (260,100) -- (235,100) ;
\draw  [fill={rgb, 255:red, 245; green, 240; blue, 235 }  ,fill opacity=1 ][line width=0.5]  (227.93,78.79) -- (256.21,107.07) -- (242.07,121.21) -- (213.79,92.93) -- cycle ;
\draw [line width=0.5]    (285,100) -- (260,100) ;
\draw [line width=0.5]    (285,60) -- (285,100) ;
\draw [line width=0.5]    (310,100) -- (285,100) ;
\draw [line width=0.5]    (335,100) -- (310,100) ;
\draw [line width=0.5]    (335,60) -- (335,100) ;
\draw [line width=0.5]    (360,100) -- (335,100) ;
\draw  [fill={rgb, 255:red, 245; green, 240; blue, 235 }  ,fill opacity=1 ][line width=0.5]  (327.93,78.79) -- (356.21,107.07) -- (342.07,121.21) -- (313.79,92.93) -- cycle ;
\draw [line width=0.5]    (385,100) -- (360,100) ;
\draw [line width=0.5]    (385,60) -- (385,100) ;
\draw [line width=0.5]    (410,100) -- (385,100) ;
\draw  [fill={rgb, 255:red, 245; green, 240; blue, 235 }  ,fill opacity=1 ][line width=0.5]  (377.93,78.79) -- (406.21,107.07) -- (392.07,121.21) -- (363.79,92.93) -- cycle ;
\draw [line width=0.5]    (435,100) -- (410,100) ;
\draw [line width=0.5]    (435,60) -- (435,100) ;
\draw [line width=0.5]    (485,60) -- (485,103) ;
\draw  [fill={rgb, 255:red, 245; green, 240; blue, 235 }  ,fill opacity=1 ][line width=0.5]  (277.93,78.79) -- (306.21,107.07) -- (292.07,121.21) -- (263.79,92.93) -- cycle ;

\end{mytikz3}\quad.
\end{equation}
Namely, all the unitary gates are of shape $(D d)\times (D d)$ and act on the input states $\ket{0}$ of dimension $D d$, where $d$ is the local Hilbert space dimension and $D$ is the MPS bond dimension. The number of (non-trivial) gates in this homogeneous right canonical form becomes
\begin{equation}
    N_g=N-(\beta-1),
\end{equation}
which is different from $N_g=N$ in conventional MPS representations. The circuit corresponding to the left canonical form is similar to Eqs.\,\eqref{eq:right_canonical_circuit_exact} and \eqref{eq:right_canonical_circuit_regular} except it is spatially reversed, i.e., the gates are arranged from right to left.

For the mixed canonical form, the preparation circuit is a combination of the cases of the left and right canonical forms, with the tensor at the orthogonality center generated by the first gate of size $(2\beta-1)$. For example, the exact preparation circuit for Eq.\,\eqref{eq:mixed_canonical_mps} with $D=d^2$ is
\begin{equation}\label{eq:mixed_canonical_circuit_exact}
\begin{mytikz3}
\draw [line width=0.5]    (230,210.02) -- (230,30) ;
\draw [line width=0.5]    (190,210.02) -- (190,30) ;
\draw [line width=0.5]    (150,210.02) -- (150,30) ;
\draw [line width=0.5]    (110,210.02) -- (110,30) ;
\draw [line width=0.5]    (390,210.02) -- (390,30) ;
\draw [line width=0.5]    (350,210.02) -- (350,30) ;
\draw [line width=0.5]    (310,210.02) -- (310,30) ;
\draw [line width=0.5]    (270,210.02) -- (270,30) ;
\draw  [fill={rgb, 255:red, 245; green, 240; blue, 235 }  ,fill opacity=1 ][line width=0.5]  (90,140) -- (210,140) -- (210,160) -- (90,160) -- cycle ;
\draw  [fill={rgb, 255:red, 245; green, 240; blue, 235 }  ,fill opacity=1 ][line width=0.5]  (129.84,170) -- (329.84,170) -- (329.84,190) -- (129.84,190) -- cycle ;
\draw  [fill={rgb, 255:red, 245; green, 240; blue, 235 }  ,fill opacity=1 ][line width=0.5]  (250,140) -- (370,140) -- (370,160) -- (250,160) -- cycle ;
\draw  [fill={rgb, 255:red, 245; green, 240; blue, 235 }  ,fill opacity=1 ][line width=0.5]  (290,110) -- (410,110) -- (410,130) -- (290,130) -- cycle ;
\draw  [fill={rgb, 255:red, 245; green, 240; blue, 235 }  ,fill opacity=1 ][line width=0.5]  (330,80) -- (410,80) -- (410,100) -- (330,100) -- cycle ;
\draw  [fill={rgb, 255:red, 245; green, 240; blue, 235 }  ,fill opacity=1 ][line width=0.5]  (370,50) -- (410,50) -- (410,70) -- (370,70) -- cycle ;
\draw  [fill={rgb, 255:red, 245; green, 240; blue, 235 }  ,fill opacity=1 ][line width=0.5]  (90,110) -- (170,110) -- (170,130) -- (90,130) -- cycle ;
\draw  [fill={rgb, 255:red, 245; green, 240; blue, 235 }  ,fill opacity=1 ][line width=0.5]  (90,80) -- (130,80) -- (130,100) -- (90,100) -- cycle ;

\end{mytikz3}\quad,
\end{equation}
which can be simplified to
\begin{equation}\label{eq:mixed_canonical_circuit_regular}
\begin{mytikz3}
\draw [line width=0.5]    (230,210.02) -- (230,89.98) ;
\draw [line width=0.5]    (190,210.02) -- (190,89.98) ;
\draw [line width=0.5]    (150,210.02) -- (150,89.98) ;
\draw [line width=0.5]    (110,210.02) -- (110,89.98) ;
\draw [line width=0.5]    (390,210.02) -- (390,89.98) ;
\draw [line width=0.5]    (350,210.02) -- (350,89.98) ;
\draw [line width=0.5]    (310,210.02) -- (310,89.98) ;
\draw [line width=0.5]    (270,210.02) -- (270,89.98) ;
\draw  [fill={rgb, 255:red, 245; green, 240; blue, 235 }  ,fill opacity=1 ][line width=0.5]  (90,140) -- (210,140) -- (210,160) -- (90,160) -- cycle ;
\draw  [fill={rgb, 255:red, 245; green, 240; blue, 235 }  ,fill opacity=1 ][line width=0.5]  (130,170) -- (330,170) -- (330,190) -- (130,190) -- cycle ;
\draw  [fill={rgb, 255:red, 245; green, 240; blue, 235 }  ,fill opacity=1 ][line width=0.5]  (250,140) -- (370,140) -- (370,160) -- (250,160) -- cycle ;
\draw  [fill={rgb, 255:red, 245; green, 240; blue, 235 }  ,fill opacity=1 ][line width=0.5]  (290,110) -- (410,110) -- (410,130) -- (290,130) -- cycle ;

\end{mytikz3}\quad,
\end{equation}
provided that the gates are universal on their supports. In the horizontal arrangement as in Eq.\,\eqref{eq:mixed_canonical_mps}, Eq.\,\eqref{eq:mixed_canonical_circuit_regular} becomes
\begin{equation}\label{eq:mixed_canonical_circuit_horizontal}
\begin{mytikz3}

\draw [line width=0.5]    (291,140) -- (291,100) ;
\draw [line width=0.5]    (279,140) -- (279,100) ;
\draw [line width=0.5]    (288,140) -- (288,100) ;
\draw [line width=0.5]    (282,140) -- (282,100) ;
\draw [line width=0.5]    (235,140) -- (235,100) ;
\draw [line width=0.5]    (285,140) -- (285,100) ;
\draw [line width=0.5]    (335,140) -- (335,100) ;
\draw [line width=0.5]    (385,140) -- (385,100) ;
\draw [line width=0.5]    (160,103) -- (135,103) ;
\draw [line width=0.5]    (185,103) -- (160,103) ;
\draw [line width=0.5]    (210,103) -- (185,103) ;
\draw [line width=0.5]    (235,103) -- (210,103) ;
\draw [line width=0.5]    (260,103) -- (235,103) ;
\draw [line width=0.5]    (285,103) -- (260,103) ;
\draw [line width=0.5]    (310,103) -- (285,103) ;
\draw [line width=0.5]    (335,103) -- (310,103) ;
\draw [line width=0.5]    (360,103) -- (335,103) ;
\draw [line width=0.5]    (385,103) -- (360,103) ;
\draw [line width=0.5]    (410,103) -- (385,103) ;
\draw [line width=0.5]    (435,103) -- (410,103) ;
\draw [line width=0.5]    (460,103) -- (435,103) ;
\draw [line width=0.5]    (485,103) -- (460,103) ;
\draw [line width=0.5]    (135,60) -- (135,103) ;
\draw [line width=0.5]    (185,60) -- (185,100) ;
\draw [line width=0.5]    (210,100) -- (185,100) ;
\draw [line width=0.5]    (235,100) -- (210,100) ;
\draw [line width=0.5]    (235,60) -- (235,100) ;
\draw [line width=0.5]    (260,100) -- (235,100) ;
\draw  [fill={rgb, 255:red, 245; green, 240; blue, 235 }  ,fill opacity=1 ][line width=0.5]  (213.79,107.07) -- (242.07,78.79) -- (256.21,92.93) -- (227.93,121.21) -- cycle ;
\draw [line width=0.5]    (285,100) -- (260,100) ;
\draw [line width=0.5]    (285,60) -- (285,100) ;
\draw [line width=0.5]    (310,100) -- (285,100) ;
\draw [line width=0.5]    (335,100) -- (310,100) ;
\draw [line width=0.5]    (335,60) -- (335,100) ;
\draw [line width=0.5]    (360,100) -- (335,100) ;
\draw  [fill={rgb, 255:red, 245; green, 240; blue, 235 }  ,fill opacity=1 ][line width=0.5]  (327.93,78.79) -- (356.21,107.07) -- (342.07,121.21) -- (313.79,92.93) -- cycle ;
\draw [line width=0.5]    (385,100) -- (360,100) ;
\draw [line width=0.5]    (385,60) -- (385,100) ;
\draw [line width=0.5]    (410,100) -- (385,100) ;
\draw  [fill={rgb, 255:red, 245; green, 240; blue, 235 }  ,fill opacity=1 ][line width=0.5]  (377.93,78.79) -- (406.21,107.07) -- (392.07,121.21) -- (363.79,92.93) -- cycle ;
\draw [line width=0.5]    (435,100) -- (410,100) ;
\draw [line width=0.5]    (435,60) -- (435,100) ;
\draw [line width=0.5]    (485,60) -- (485,103) ;
\draw  [fill={rgb, 255:red, 245; green, 240; blue, 235 }  ,fill opacity=1 ][line width=0.5]  (270,85) -- (300,85) -- (300,115) -- (270,115) -- cycle ;

\end{mytikz3}\quad.
\end{equation}
Namely, the unitary gate at the orthogonality center is of shape $(D^2 d)\times (D^2 d)$ and acts on the input states $\ket{0}$ of dimension $D^2 d$. The unitary gates elsewhere are of shape $(D d)\times (D d)$ and act on the input states $\ket{0}$ of dimension $d$. The number of gates in this homogeneous mixed canonical form becomes
\begin{equation}
    N_g=N-2(\beta-1).
\end{equation}
Note that in this form, the orthogonality center can also be placed at the left or right edges, but the preparation circuit will be slightly different from the homogeneous left or right canonical forms as in Eqs.\,\eqref{eq:right_canonical_circuit_regular} and \eqref{eq:right_canonical_circuit_horizontal} by the size of the gates at the edges, i.e., $\beta$ for the latter while $(2\beta-1)$ for the former. 


For PEPS in two or higher dimensions, a canonical form no longer exists in general because restricting all the local tensors to isometries will lead to a loss of generality. This restriction gives rise to a subclass of PEPS known as isometric tensor network states~\cite{Zaletel2020, Wei2022}, which remains considerably powerful to represent many complex ground states of physical interest.

\subsection{MPS manifold and MPS variety}\label{sec:mps_manifold}
In this section, we analyze the mapping structure of MPS and introduce the concepts of the MPS manifold~\cite{Haegeman2014} and MPS variety~\cite{Kutschan2018}, which will be used in the proof of the MPS local minimum theorem. We denote the variational set of MPS of size $N$ with open boundary conditions as
\begin{equation}\label{eq:v_mps}
    \mathcal{V}_\mathrm{MPS}(N,\mathrm{D}) = \left\{\ket{\Psi_\mathrm{MPS}(A)}\mid A\in\mathcal{A}_\mathrm{MPS}(N,\mathrm{D}) \right\},
\end{equation}
where $\ket{\Psi_\mathrm{MPS}(A)}$ is an MPS parametrized by $N$ complex-valued local tensors $A=\{A_1,A_2,\ldots,A_N\}$, as defined in Eq.\,\eqref{eq:mps_def}. $\mathrm{D}=\{D_1,D_2,\ldots,D_{N-1}\}$ is the vector of bond dimensions of the $(N-1)$ bonds. The parameter space of the $N$ complex-valued local tensors is denoted as
\begin{equation}
    \mathcal{A}_\mathrm{MPS}(N,\mathrm{D})=\prod_{i=1}^{N} \mathbb{C}^{D_{i-1} d D_{i}},
\end{equation}
subject to the normalization condition, where we set $D_0=D_N=1$ consistent with the open boundary conditions. $d$ is the local Hilbert space dimension. 


The variational set $\mathcal{V}_\mathrm{MPS}$, as a subset of the entire Hilbert space, is not strictly a complex manifold because it contains conical singularities, which correspond to rank-deficient MPS~\cite{Haegeman2014}. Here, the rank of MPS refers to the Schmidt ranks with respect to the $(N-1)$ bonds of the MPS. Full-rank MPS means that the Schmidt rank with respect to each bond is equal to its maximum value, i.e., the bond dimension. For a rank-deficient MPS, the allowed first-order perturbations within $\mathcal{V}_\mathrm{MPS}$ contain any directions such that the Schmidt ranks do not exceed $D$. However, the linear combinations of these directions in principle lead to Schmidt ranks larger than $D$. Therefore, the tangent space at a rank-deficient MPS point is not a well-defined linear space---it forms a non-trivial tangent cone instead, which implies that the rank-deficient MPS points are singular and the variational set $\mathcal{V}_\mathrm{MPS}$ is indeed not a complex manifold as a whole.

In fact, $\mathcal{V}_\mathrm{MPS}$ is an algebraic variety instead, because it can be defined as the set of quantum states whose Schmidt rank with respect to each bond is equal to or smaller than the bond dimension
\begin{equation}
    \mathcal{V}_\mathrm{MPS}(N,\mathrm{D}) = \left\{\ket{\psi}\in\mathcal{H} \mid \mathrm{rank}(\rho_{1:i}) \leq D_i, ~  i=1,2,\ldots,N-1 \right\},
\end{equation}
where $\rho_{1:i}=\mathrm{Tr}_{(i+1):N}\ketbrasame{\psi}$ is the reduced density matrix on the first $i$ sites. $\mathcal{H}=\mathbb{C}^{d^N}$ is the total Hilbert space. Namely, $\mathcal{V}_\mathrm{MPS}$ corresponds to the set of zero points of all $(D_i+1)\times(D_i+1)$ minors of the reduced density matrices from the bipartition at bond $i$. Therefore, we call $\mathcal{V}_\mathrm{MPS}$ an ``MPS variety'' and the full-rank MPS set, denoted as $\mathcal{V}_\mathrm{MPS}^{\mathrm{f}}$, an ``MPS manifold''. It has been rigorously proved that $\mathcal{V}_\mathrm{MPS}^{\mathrm{f}}$ is indeed a complex manifold using the biholomorphic relation with the quotient manifold of the parameter space~\cite{Haegeman2014}. The key step in understanding the geometry of MPS is the ``gauge'' structure of the MPS representation, which is briefly reviewed below.

The map $\Psi_\mathrm{MPS}$ from the local tensors $A$ to the pure quantum state $\ket{\Psi_\mathrm{MPS}(A)}$ is holomorphic because each component in $\ket{\Psi_\mathrm{MPS}(A)}$ is a polynomial of local tensor entries. It is also evident that the map $\Psi_\mathrm{MPS}$ is not injective because inserting an arbitrary invertible matrix $G$ with its inverse $G^{-1}$ between adjacent local tensors $A_{i},A_{i+1}$ does not change the MPS as $GG^{-1}=I$ but does change the local tensors by $A_{i}\rightarrow A_{i}G$ and $A_{i+1}\rightarrow G^{-1}A_{i+1}$. Thus, the redundancy in the MPS representation $A$ is described by the so-called gauge group of MPS
\begin{equation}
    \mathcal{G}_{\text{MPS}} = \prod_{i=1}^{N-1}\mathrm{GL}(D_i;\mathbb{C}),
\end{equation}
where $\mathrm{GL}(D_i;\mathbb{C})$ is the complex general linear group of dimension $D_i$. The action of the group element $\mathrm{G}=\{G_{1},G_{2},\ldots,G_{N-1}\}$ as a local gauge transformation on the MPS representation is
\begin{equation}
    A_{1}\rightarrow A_{1}G_{1},~\ldots,~~A_i\rightarrow G_{i-1}^{-1} A_{i} G_{i},~\ldots,~~A_{N}\rightarrow G_{N-1}^{-1}A_{N}.
\end{equation}
The physical indices are omitted for simplicity. For MPS in canonical forms, the local tensors are isometric, and hence the gauge group of MPS is reduced to
\begin{equation}
    \mathcal{G}_{\text{MPS}} = \prod_{i=1}^{N-1}\mathcal{U}(D_i),
\end{equation}
where $\mathcal{U}(D_i)$ is the unitary group of dimension $D_i$.

To remove the gauge redundancy in the MPS representation, one can construct a ``gauge orbit space'' by taking the quotient manifold of the full-rank MPS representation manifold with respect to the gauge group. Specifically, the full-rank MPS representation manifold is defined by
\begin{equation}
    \mathcal{A}_{\text{MPS}}^{\mathrm{f}} = \{A\in\mathcal{A}_{\text{MPS}}\mid \mathrm{rank}(A_{1:i})=\mathrm{rank}(A_{(i+1):N})=D_i,~\forall i=1,\ldots,N-1\},
\end{equation}
where $A_{i:j}=\prod_{k=i}^j A_k$ represents the matrix obtained by contracting the local tensors from site $i$ to $j$ and merging the physical indices as the row or column index (for $A_{1:i}$ or $A_{(i+1):N}$, respectively) and merging the virtual indices as the other. Note that choosing bond dimensions with $D_i>d_iD_{i-1}$ or $D_{i-1}>d_iD_i$ is useless, because in such cases the local tensor $A_i$ cannot be full-rank under the corresponding matricization, and every MPS in the parameter space is necessarily rank-deficient. Therefore, without loss of generality, we assume $D_i \leq d_iD_{i-1}$ and $D_{i-1} \leq d_iD_i$ for every site $i$.


Given that the action of the gauge group $\mathcal{G}_\mathrm{MPS}$ is a holomorphic, free and proper on $\mathcal{A}_\mathrm{MPS}^\mathrm{f}$~\cite{Haegeman2014}, the orbit space $\mathcal{A}_\mathrm{MPS}^\mathrm{f}/\mathcal{G}_\mathrm{MPS}$ is a complex manifold and the quotient map $\pi:\mathcal{A}_\mathrm{MPS}^\mathrm{f}\rightarrow \mathcal{A}_\mathrm{MPS}^\mathrm{f}/\mathcal{G}_\mathrm{MPS}$ is holomorphic. Then, since the MPS contraction map $\Psi_\mathrm{MPS}$ is holomorphic and invariant under the action of $\mathcal{G}_\mathrm{MPS}$, the induced map $\Phi_\mathrm{MPS}:\mathcal{A}_\mathrm{MPS}^\mathrm{f}/\mathcal{G}_\mathrm{MPS}\rightarrow \mathcal{V}_\mathrm{MPS}^{\mathrm{f}}$ from the composition
\begin{equation}
    \Psi_\mathrm{MPS}=\Phi_\mathrm{MPS}\circ\pi,
\end{equation}
is also holomorphic. Combined with the fact that $\Phi_\mathrm{MPS}$ is injective as $A$ can be uniquely obtained by performing a series of Schmidt decompositions of the coefficient vector of $\ket{\Psi_\mathrm{MPS}(A)}$ up to the gauge transformation, $\Phi_\mathrm{MPS}$ is a biholomorphism, and the full-rank MPS set $\mathcal{V}_\mathrm{MPS}^{\mathrm{f}}$ is a complex manifold that is biholomorphic to the orbit space $\mathcal{A}_\mathrm{MPS}^\mathrm{f}/\mathcal{G}_\mathrm{MPS}$ with dimension
\begin{equation}
    \mathrm{dim}\mathcal{V}_\mathrm{MPS}^{\mathrm{f}} = \mathrm{dim}\mathcal{A}_\mathrm{MPS}^\mathrm{f} - \mathrm{dim}\mathcal{G}_\mathrm{MPS}.
\end{equation}
As a whole, the full-rank MPS representations form a principal fiber bundle with total manifold $\mathcal{A}_\mathrm{MPS}^\mathrm{f}$, base manifold $\mathcal{V}_\mathrm{MPS}^{\mathrm{f}}$, bundle projection $\Psi_\mathrm{MPS}$, and structure group $\mathcal{G}_\mathrm{MPS}$.

For a rank-deficient MPS representation $A$, i.e., a point in $\mathcal{A}_\mathrm{MPS}$ but not in $\mathcal{A}_\mathrm{MPS}^\mathrm{f}$, the action of $\mathcal{G}_\mathrm{MPS}$ is not free, which means that the stabilizer subgroup of $\mathcal{G}_\mathrm{MPS}$ on $A$ is not trivial. For example, if the Schmidt rank between site $i$ and $i+1$ is $r_i<D_i$, the stabilizer subgroup acting on this bond is
\begin{equation}
    \mathrm{Stab}(A_{1:i}, A_{(i+1):N}) = \left\{G_{i}\in\mathrm{GL}(D_i;\mathbb{C})\mid A_{1:i}G_i=A_{1:i},~ G_i^{-1}A_{(i+1):N}=A_{(i+1):N} \right\}.
\end{equation}
In particular, this stabilizer subgroup equals $\mathrm{GL}(D_i-r_i;\mathbb{C})$ if $\mathrm{rank}(A_{1:i})=\mathrm{rank}(A_{(i+1):N})=r_i$. Moreover, for rank-deficient MPS, the induced map $\Phi$ from the orbit space $\mathcal{A}_\mathrm{MPS}/\mathcal{G}_\mathrm{MPS}$ to $\mathcal{V}_\mathrm{MPS}$ is not even injective because there may exist different MPS representations that cannot be connected by any gauge transformations in $\mathcal{G}_\mathrm{MPS}$. Specifically, given the Schmidt rank $r_i<D_i$, the ranks of the left part $A_{1:i}$ and right part $A_{(i+1):N}$ only need to satisfy the inequalities
\begin{equation}
    r_i\leq \mathrm{rank}(A_{1:i})\leq D_i,\quad r_i\leq \mathrm{rank}(A_{(i+1):N})\leq D_i,
\end{equation}
instead of the equalities $\mathrm{rank}(A_{1:i})=\mathrm{rank}(A_{(i+1):N})=D_i$ in the full-rank case. Namely, there is arbitrariness in specifying the ranks of $A_{1:i}$ and $A_{(i+1):N}$, whereas the gauge transformations in $\mathcal{G}_\mathrm{MPS}$ cannot change the ranks because they are all invertible. Therefore, there exist multiple orbits in $(\mathcal{A}_\mathrm{MPS}-\mathcal{A}_\mathrm{MPS}^\mathrm{f})/\mathcal{G}_\mathrm{MPS}$ that are mapped to the same point in $\mathcal{V}_\mathrm{MPS}$ by the map $\Phi$.

In particular, for MPS in canonical forms, the rank arbitrariness in rank-deficient MPS representations is fixed because isometric matrices are always full-rank, and hence it holds either $\mathrm{rank}(A_{1:i})=D_i$ or $\mathrm{rank}(A_{(i+1):N})=D_i$, depending on whether site $i$ is located at the left or right side of the orthogonality center. Given the Schmidt rank $r_i<D_i$, this must lead to $\mathrm{rank}(A_{(i+1):N})=r_i$ or $\mathrm{rank}(A_{1:i})=r_i$, respectively, and hence the ranks of $A_{1:i}$ and $A_{(i+1):N}$ are indeed fixed. Thus, the stabilizer subgroup becomes trivial. Combined with the fact that the gauge group is compact for MPS in canonical forms, the group action becomes free and proper even for rank-deficient MPS representations.

However, the induced map $\Phi$ from $\mathcal{A}_\mathrm{MPS}/\mathcal{G}_\mathrm{MPS}$ to $\mathcal{V}_\mathrm{MPS}$ is still not injective for rank-deficient canonical MPS representations because there is still ambiguity in the column space of $A_{1:i}$ or the row space of $A_{(i+1):N}$. For instance, if $\mathrm{rank}(A_{1:i})=D_i$ and $\mathrm{rank}(A_{(i+1):N})=r_i$, the column space of $A_{1:i}$ can be an arbitrary $D_i$-dimensional linear subspace in $\mathbb{C}^{d^i}$ that contains the $r_i$-dimensional column space of $A_{1:i}A_{(i+1):N}$. Different choices of these $D_i$-dimensional subspaces cannot be connected by the gauge transformation in $\mathcal{G}_\mathrm{MPS}$ as right multiplication does not change the column space of a matrix. Therefore, there exists a continuum of orbits in $(\mathcal{A}_\mathrm{MPS}-\mathcal{A}_\mathrm{MPS}^\mathrm{f})/\mathcal{G}_\mathrm{MPS}$ that are mapped to a single point in $\mathcal{V}_\mathrm{MPS}$ by the map $\Phi$ even for MPS in canonical forms. When performing the Schmidt decompositions of the coefficient vector mentioned above, this ambiguity corresponds to the extra degrees of freedom in the so-called ``thin SVD'' compared to the ``compact SVD''. This obstruction of non-injectivity further showcases that the points of rank-deficient MPS are indeed singularities in the MPS variety $\mathcal{V}_\mathrm{MPS}$, which should be excluded to obtain the complex manifold $\mathcal{V}_\mathrm{MPS}^{\mathrm{f}}$.

We remark that, although the rank-deficient MPS are singular with respect to the geometry of the full-rank MPS, they themselves form lower-dimensional complex manifolds. In other words, the MPS variety can be regarded as a hierarchy of manifolds, where the ones with lower bond dimensions serve as the singular regions (boundaries) of those with higher bond dimensions. Mathematically, this corresponds to the so-called Whitney stratification of $\mathcal{V}_\mathrm{MPS}$, i.e.,
\begin{equation}
    \mathcal{V}_\mathrm{MPS}(N,\mathrm{D}) = \bigcup_{\mathrm{D}'\leq \mathrm{D}} \mathcal{V}_\mathrm{MPS}^{\mathrm{f}}(N,\mathrm{D}'),
\end{equation}
where $\mathrm{D}'=\{D_1', D_2', \ldots, D_{N-1}'\}$ is the vector of bond dimensions that is equal to or element-wise smaller than $\mathrm{D}=\{D_1, D_2, \ldots, D_{N-1}\}$, under the constraints $D_{i} \leq d_{i}D_{i-1}$ and $D_{i-1} \leq d_{i}D_{i}$. $\mathcal{V}_\mathrm{MPS}^{\mathrm{f}}(N,\mathrm{D})$ is the top-dimensional open dense stratum, and the other strata are complex sub-manifolds of lower dimensions.

\subsection{Random MPS ensemble}
We further clarify the concept of ``MPS ensemble''. An ensemble is defined by a set together with a well-defined probability distribution over it. A unitary ensemble usually refers to a unitary group together with the Haar measure over it. Hence, a quantum circuit ensemble can be defined as the combination of the ensembles of the gates comprising the circuit, where the probability distribution is defined by the push-forward measure induced by the circuit contraction map from the product of the independent Haar measures over respective unitary gates. Here, we note that the push-forward measure $\mu'$ with respect to a map $g:\mathcal{X}\rightarrow\mathcal{Y}$ and a given measure $\mu$ on $\mathcal{X}$ is defined by
\begin{equation}
    \mu'(\mathcal{B}) = \mu\left( g^{-1}(\mathcal{B}) \right),
\end{equation}
where $\mathcal{B}$ is a Borel measurable subset of $\mathcal{Y}$. Namely, the volume of a subset in $\mathcal{Y}$ is given by the volume of its preimage in $\mathcal{X}$. 

Similarly, an MPS ensemble $\mathbb{V}_\mathrm{MPS}$ is defined by an MPS variety $\mathcal{V}_\mathrm{MPS}$ together with the probability distribution induced by the random sequential circuit generating the MPS, as illustrated in Eq.\,\eqref{eq:mixed_canonical_circuit_regular}. A canonical MPS representation ensemble $\mathbb{A}_\mathrm{MPS}$ is defined by a canonical MPS representation space $\mathcal{A}_\mathrm{MPS}$ with the probability distribution induced by the unitaries generating the local tensors, as illustrated in Eq.\,\eqref{eq:mixed_canonical_circuit_horizontal}. In particular, a full-rank MPS ensemble $\mathbb{V}_\mathrm{MPS}^\mathrm{f}$ and a full-rank MPS representation ensemble $\mathbb{A}_\mathrm{MPS}^\mathrm{f}$ are defined in the same way as $\mathbb{V}_\mathrm{MPS}$ and $\mathbb{A}_\mathrm{MPS}$ but with the subset of rank-deficient MPS points removed (which has zero measure).

\subsection{Ensemble equivalence: Proof of Theorem~1}
In this section, we provide a detailed proof of Theorem~\textcolor{darkblue1}{1} in the main text. To this end, we first prove the following Lemma~\ref{lemma:Wg_G_commute} for clarity.
\begin{lemma}\label{lemma:Wg_G_commute}
The $t$-degree Weingarten matrix $W^{(t)}(d)$ commutes with the $t$-degree Gram matrix $G^{(t)}(d')$ of permutation vectors for arbitrary dimension parameters $d$ and $d'$.
\end{lemma}
\begin{proof}
If $d=d'$, $W^{(t)}(d)$ is just the pseudo-inverse of $G^{(t)}(d)$, and hence they commute trivially. Thus, in the following, we focus on the cases of $d\neq d'$. We first define the ``Gram function'' of a permutation $\sigma$ as
\begin{equation}
    \opr{G}(\sigma, d) = \opr{Tr}(S_{\sigma}) = d^{\#(\sigma)},
\end{equation}
which can be viewed as the character of the group representation of $\mathcal{S}_t$ on the $t$-fold Hilbert space of dimension $d$: $\sigma\rightarrow S_{\sigma}$. Then, the Gram matrix can be rewritten as
\begin{equation}
    G^{(t)}_{\sigma,\tau}(d) = \opr{G}(\sigma^{-1}\tau, d).
\end{equation}
Using the Weingarten function
\begin{equation}
    W^{(t)}_{\sigma,\tau}(d) = \opr{Wg}(\sigma^{-1}\tau, d),
\end{equation}
the product of the two matrices can be expressed as
\begin{equation}
    \left[ W^{(t)}(d) G^{(t)}(d') \right]_{\sigma,\tau} = \sum_{\varsigma} \opr{Wg}(\sigma^{-1}\varsigma, d) \opr{G}(\varsigma^{-1}\tau, d') = \sum_{\varsigma\varsigma'=\sigma^{-1}\tau} \opr{Wg}(\varsigma, d) \opr{G}(\varsigma', d'),
\end{equation}
where the summation runs over $\varsigma,\varsigma'\in\mathcal{S}_t$ under the constraint $\varsigma\varsigma'=\sigma^{-1}\tau$. This is exactly the convolution of the two functions, i.e.,
\begin{equation}\label{eq:WG=WgG}
    \left[ W^{(t)}(d) G^{(t)}(d') \right]_{\sigma,\tau} = \left[ \opr{Wg}^{(t)}(d) * \opr{G}^{(t)}(d') \right](\sigma^{-1}\tau).
\end{equation}
The symbol ``$*$'' denotes the convolution product of two functions over groups. Similarly, we also have
\begin{equation}\label{eq:GW=GWg}
    \left[ G^{(t)}(d') W^{(t)}(d) \right]_{\sigma,\tau} = \left[ \opr{G}^{(t)}(d') * \opr{Wg}^{(t)}(d) \right](\sigma^{-1}\tau).
\end{equation}
Therefore, the matrix commutation of the Weingarten matrix $W^{(t)}(d)$ and the Gram matrix $G^{(t)}(d')$ is equivalent to the convolution commutation of the Weingarten function $\opr{Wg}^{(t)}(d)$ and the Gram function $\opr{G}^{(t)}(d')$.

Then, we consider two elements in the group algebra $\mathbb{C}[\mathcal{S}_t]$ over the field of complex numbers $\mathbb{C}$ whose coefficients are $\opr{Wg}^{(t)}(d)$ and $\opr{G}^{(t)}(d')$, i.e.,
\begin{equation}
    w = \sum_{\sigma} \opr{Wg}(\sigma, d) \sigma, \quad g = \sum_{\sigma} \opr{G}(\sigma, d') \sigma.
\end{equation}
The product of these two elements is
\begin{equation}\label{eq:wg=WgG}
\begin{aligned}
    wg &= \sum_{\sigma,\sigma'} \opr{Wg}(\sigma, d) \opr{G}(\sigma', d') \sigma\sigma' = \sum_{\tau} \left(\sum_{\sigma\sigma'=\tau} \opr{Wg}(\sigma, d) \opr{G}(\sigma', d')\right) \tau \\
    &= \sum_{\tau} \left[\opr{Wg}^{(t)}(d) * \opr{G}^{(t)}(d')\right](\tau) \,\tau.
\end{aligned}
\end{equation}
Similarly, we have
\begin{equation}\label{eq:gw=GWg}
    gw = \sum_{\tau} \left[\opr{G}^{(t)}(d') * \opr{Wg}^{(t)}(d)\right](\tau) \,\tau. 
\end{equation}
That is to say, the convolution commutation of $\opr{Wg}^{(t)}(d)$ and $\opr{G}^{(t)}(d')$ is equivalent to the commutation of the elements $w$ and $g$ under the group algebra, where the multiplication is defined by extending the group multiplication linearly.

The Gram function $\opr{G}^{(t)}(d')$ is a class function over $\mathcal{S}_t$ because the number of cycles $\#(\sigma)$ is invariant under the conjugation operation, i.e.,
\begin{equation}
    \opr{G}(\sigma, d') = \opr{Tr}(S_\sigma) = \opr{Tr}(S_\tau S_\sigma S_\tau^{-1}) = \opr{G}(\tau \sigma \tau^{-1}, d').
\end{equation}
Thus, the element $g$ commutes with an arbitrary group element $\tau\in\mathcal{S}_t$, i.e.,
\begin{equation}
    g\tau = \sum_{\sigma} \opr{G}(\sigma, d') \sigma \tau = \sum_{\varsigma} \opr{G}(\tau\varsigma\tau^{-1}, d') (\tau\varsigma\tau^{-1}) \tau = \sum_{\varsigma} \opr{G}(\varsigma, d') \tau \varsigma = \tau g.
\end{equation}
Hence, $g$ belongs to the center of $\mathbb{C}[\mathcal{S}_t]$, and commutes with any linear combination of group elements, including $w$, i.e.,
\begin{equation}
    gw = wg.
\end{equation}
Therefore, according to Eqs.\,\eqref{eq:wg=WgG} and \eqref{eq:gw=GWg}, we have
\begin{equation}
    \opr{Wg}^{(t)}(d) * \opr{G}^{(t)}(d') = \opr{G}^{(t)}(d') * \opr{Wg}^{(t)}(d),
\end{equation}
and then according to Eqs.\,\eqref{eq:WG=WgG} and \eqref{eq:GW=GWg}, we have
\begin{equation}
    W^{(t)}(d) G^{(t)}(d') = G^{(t)}(d') W^{(t)}(d).
\end{equation}
This completes the proof.
\end{proof}

With Lemma~\ref{lemma:Wg_G_commute} and the preliminaries on the Weingarten calculus and MPS introduced above, we are ready to prove Theorem~\textcolor{darkblue1}{1} in the main text, which is reiterated below for convenience. We will focus on MPS in the ``homogeneous'' mixed canonical form as illustrated in Eq.\eqref{eq:mixed_canonical_circuit_regular}. We use $\mathbb{V}_\mathrm{MPS}^{[i]}$ to denote the random MPS ensemble with the orthogonality center at site $i$.

\renewcommand{\theproposition}{1}
\begin{theorem}
The random MPS ensembles with different orthogonality centers are identical, i.e.,
\begin{equation}
    \mathbb{V}_\mathrm{MPS}^{[i]} = \mathbb{V}_\mathrm{MPS}^{[j]},
\end{equation}
for any two sites $i$ and $j$.
\end{theorem}
\addtocounter{proposition}{-1}
\renewcommand{\theproposition}{S\arabic{proposition}}

\begin{proof}
Two ensembles are identical if and only if their sample spaces and the corresponding probability distributions are equal to each other. The sample spaces of $\mathbb{V}_\mathrm{MPS}^{[i]}$ and $\mathbb{V}_\mathrm{MPS}^{[j]}$ are the same, i.e., both of them equal the MPS variational set $\mathcal{V}_\mathrm{MPS}(N,\mathrm{D})$ defined in Eq.\,\eqref{eq:v_mps} of size $N$ and uniform bond dimension $D$. Since random MPS can be viewed as generated from a series of random unitaries sampled from \textit{compact} groups, the distributions of $\mathbb{V}_\mathrm{MPS}^{[i]}$ and $\mathbb{V}_\mathrm{MPS}^{[j]}$ are identical if and only if they have equal moments for arbitrarily high orders. Therefore, in the following, our goal is to prove that the moments of $\mathbb{V}_\mathrm{MPS}^{[i]}$ and $\mathbb{V}_\mathrm{MPS}^{[j]}$ are equal for an arbitrary order $t$. 

By definition, the MPS in $\mathbb{V}_\mathrm{MPS}^{[i]}$ are generated from a random sequential circuit acting on an input product state, i.e.,
\begin{equation}
    \mathbb{V}_\mathrm{MPS}^{[i]} = \left\{\ket{\Psi_\mathrm{MPS}} \mid \ket{\Psi_\mathrm{MPS}} = \mathbf{U} \ket{\boldsymbol{0}}, \mathbf{U}\in \mathbb{U}^{[i]} \right\},
\end{equation}
where $\ket{\boldsymbol{0}}=\ket{0}^{\otimes N}$ is the all-zero product state. $\mathbb{U}^{[i]}$ is the random sequential circuit ensemble centered at site $i$ where the independent Haar-random unitary gates are arranged sequentially as in Eqs.\,\eqref{eq:mixed_canonical_circuit_regular} and \eqref{eq:mixed_canonical_circuit_horizontal}. The $t$-degree moment of $\mathbb{V}_\mathrm{MPS}^{[i]}$ is
\begin{equation}
    \mathbb{E}_{\mathbb{V}_\mathrm{MPS}^{[i]}}\left[ \left( \ketbrasame{\Psi_\mathrm{MPS}} \right)^{\otimes t} \right] = \mathbb{E}_{\mathbb{U}^{[i]}} \left[ \mathbf{U}^{\otimes t}  \left(\ketbrasame{\boldsymbol{0}}\right)^{\otimes t} \mathbf{U}^{ \dagger \otimes t} \right],
\end{equation}
which is the expectation of the $t$-fold copies of the density matrix. After vectorization, the $t$-degree moment of $\mathbb{V}_\mathrm{MPS}^{[i]}$ becomes
\begin{equation}
    \mathbb{E}_{\mathbb{U}^{[i]}}\left[ \mathbf{U}^{\otimes t} \otimes \mathbf{U}^{ *\otimes t} \ket{\boldsymbol{0}}^{\otimes 2t} \right].
\end{equation}
The tensor network diagram of the MPS replicas $\mathbf{U}^{\otimes t} \otimes \mathbf{U}^{ *\otimes t} \ket{\boldsymbol{0}}^{\otimes 2t} $ is
\begin{equation}\label{eq:mps_replica}
\begin{mytikz3}

\draw [line width=0.5]    (152.5,132.5) -- (152.5,92.5) ;
\draw  [fill={rgb, 255:red, 245; green, 240; blue, 235 }  ,fill opacity=1 ][line width=0.5]  (152.5,137.5) -- (160,127.5) -- (145,127.5) -- cycle ;
\draw [line width=0.5]    (152.5,92.5) -- (127.5,92.5) ;
\draw [line width=0.5]    (152.5,52.5) -- (152.5,92.5) ;
\draw [line width=0.5]    (177.5,92.5) -- (152.5,92.5) ;
\draw  [fill={rgb, 255:red, 245; green, 240; blue, 235 }  ,fill opacity=1 ][line width=0.5]  (134.82,96.04) -- (156.04,74.82) -- (170.18,88.96) -- (148.96,110.18) -- cycle ;
\draw [line width=0.5]    (127.5,92.5) -- (102.5,92.5) ;
\draw [line width=0.5]    (102.5,52.5) -- (102.5,92.5) ;

\draw [line width=0.5]    (150,135) -- (150,95) ;
\draw  [fill={rgb, 255:red, 245; green, 240; blue, 235 }  ,fill opacity=1 ][line width=0.5]  (150,140) -- (157.5,130) -- (142.5,130) -- cycle ;
\draw [line width=0.5]    (150,95) -- (125,95) ;
\draw [line width=0.5]    (150,55) -- (150,95) ;
\draw [line width=0.5]    (175,95) -- (150,95) ;
\draw  [fill={rgb, 255:red, 245; green, 240; blue, 235 }  ,fill opacity=1 ][line width=0.5]  (132.32,98.54) -- (153.54,77.32) -- (167.68,91.46) -- (146.46,112.68) -- cycle ;
\draw [line width=0.5]    (125,95) -- (100,95) ;
\draw [line width=0.5]    (100,55) -- (100,95) ;

\draw [line width=0.5]    (147.5,137.5) -- (147.5,97.5) ;
\draw  [fill={rgb, 255:red, 245; green, 240; blue, 235 }  ,fill opacity=1 ][line width=0.5]  (147.5,142.5) -- (155,132.5) -- (140,132.5) -- cycle ;
\draw [line width=0.5]    (147.5,97.5) -- (122.5,97.5) ;
\draw [line width=0.5]    (147.5,57.5) -- (147.5,97.5) ;
\draw [line width=0.5]    (172.5,97.5) -- (147.5,97.5) ;
\draw  [fill={rgb, 255:red, 245; green, 240; blue, 235 }  ,fill opacity=1 ][line width=0.5]  (129.82,101.04) -- (151.04,79.82) -- (165.18,93.96) -- (143.96,115.18) -- cycle ;
\draw [line width=0.5]    (122.5,97.5) -- (97.5,97.5) ;
\draw [line width=0.5]    (97.5,57.5) -- (97.5,97.5) ;

\draw [line width=0.5]    (432.5,132.5) -- (432.5,92.5) ;
\draw  [fill={rgb, 255:red, 245; green, 240; blue, 235 }  ,fill opacity=1 ][line width=0.5]  (432.5,137.5) -- (440,127.5) -- (425,127.5) -- cycle ;
\draw [line width=0.5]    (432.5,92.5) -- (407.5,92.5) ;
\draw [line width=0.5]    (432.5,52.5) -- (432.5,92.5) ;
\draw [line width=0.5]    (457.5,92.5) -- (432.5,92.5) ;
\draw  [fill={rgb, 255:red, 245; green, 240; blue, 235 }  ,fill opacity=1 ][line width=0.5]  (428.96,74.82) -- (450.18,96.04) -- (436.04,110.18) -- (414.82,88.96) -- cycle ;
\draw [line width=0.5]    (482.5,92.5) -- (457.5,92.5) ;
\draw [line width=0.5]    (482.5,52.5) -- (482.5,92.5) ;

\draw [line width=0.5]    (430,135) -- (430,95) ;
\draw  [fill={rgb, 255:red, 245; green, 240; blue, 235 }  ,fill opacity=1 ][line width=0.5]  (430,140) -- (437.5,130) -- (422.5,130) -- cycle ;
\draw [line width=0.5]    (430,95) -- (405,95) ;
\draw [line width=0.5]    (430,55) -- (430,95) ;
\draw [line width=0.5]    (455,95) -- (430,95) ;
\draw  [fill={rgb, 255:red, 245; green, 240; blue, 235 }  ,fill opacity=1 ][line width=0.5]  (426.46,77.32) -- (447.68,98.54) -- (433.54,112.68) -- (412.32,91.46) -- cycle ;
\draw [line width=0.5]    (480,95) -- (455,95) ;
\draw [line width=0.5]    (480,55) -- (480,95) ;

\draw [line width=0.5]    (427.5,137.5) -- (427.5,97.5) ;
\draw  [fill={rgb, 255:red, 245; green, 240; blue, 235 }  ,fill opacity=1 ][line width=0.5]  (427.5,142.5) -- (435,132.5) -- (420,132.5) -- cycle ;
\draw [line width=0.5]    (427.5,97.5) -- (402.5,97.5) ;
\draw [line width=0.5]    (427.5,57.5) -- (427.5,97.5) ;
\draw [line width=0.5]    (452.5,97.5) -- (427.5,97.5) ;
\draw  [fill={rgb, 255:red, 245; green, 240; blue, 235 }  ,fill opacity=1 ][line width=0.5]  (423.96,79.82) -- (445.18,101.04) -- (431.04,115.18) -- (409.82,93.96) -- cycle ;
\draw [line width=0.5]    (477.5,97.5) -- (452.5,97.5) ;
\draw [line width=0.5]    (477.5,57.5) -- (477.5,97.5) ;

\draw [line width=0.5]    (292.15,132.5) -- (292.15,92.5) ;
\draw [line width=0.5]    (242.5,132.5) -- (242.5,92.5) ;
\draw  [fill={rgb, 255:red, 245; green, 240; blue, 235 }  ,fill opacity=1 ][line width=0.5]  (242.5,137.5) -- (250,127.5) -- (235,127.5) -- cycle ;
\draw [line width=0.5]    (342.5,132.5) -- (342.5,92.5) ;
\draw  [fill={rgb, 255:red, 245; green, 240; blue, 235 }  ,fill opacity=1 ][line width=0.5]  (342.5,137.5) -- (350,127.5) -- (335,127.5) -- cycle ;
\draw [line width=0.5]    (242.5,92.5) -- (217.5,92.5) ;
\draw [line width=0.5]    (242.5,52.5) -- (242.5,92.5) ;
\draw [line width=0.5]    (267.5,92.5) -- (242.5,92.5) ;
\draw  [fill={rgb, 255:red, 245; green, 240; blue, 235 }  ,fill opacity=1 ][line width=0.5]  (224.82,96.04) -- (246.04,74.82) -- (260.18,88.96) -- (238.96,110.18) -- cycle ;
\draw [line width=0.5]    (292.5,92.5) -- (267.5,92.5) ;
\draw [line width=0.5]    (292.5,52.5) -- (292.5,92.5) ;
\draw [line width=0.5]    (317.5,92.5) -- (292.5,92.5) ;
\draw  [fill={rgb, 255:red, 245; green, 240; blue, 235 }  ,fill opacity=1 ][line width=0.5]  (277.5,77.5) -- (307.5,77.5) -- (306.81,107.5) -- (276.81,107.5) -- cycle ;
\draw [line width=0.5]    (342.5,92.5) -- (317.5,92.5) ;
\draw [line width=0.5]    (342.5,52.5) -- (342.5,92.5) ;
\draw [line width=0.5]    (367.5,92.5) -- (342.5,92.5) ;
\draw  [fill={rgb, 255:red, 245; green, 240; blue, 235 }  ,fill opacity=1 ][line width=0.5]  (338.96,74.82) -- (360.18,96.04) -- (346.04,110.18) -- (324.82,88.96) -- cycle ;
\draw  [fill={rgb, 255:red, 245; green, 240; blue, 235 }  ,fill opacity=1 ][line width=0.5]  (292.15,137.5) -- (299.65,127.5) -- (284.65,127.5) -- cycle ;

\draw [line width=0.5]    (289.65,135) -- (289.65,95) ;
\draw [line width=0.5]    (240,135) -- (240,95) ;
\draw  [fill={rgb, 255:red, 245; green, 240; blue, 235 }  ,fill opacity=1 ][line width=0.5]  (240,140) -- (247.5,130) -- (232.5,130) -- cycle ;
\draw [line width=0.5]    (340,135) -- (340,95) ;
\draw  [fill={rgb, 255:red, 245; green, 240; blue, 235 }  ,fill opacity=1 ][line width=0.5]  (340,140) -- (347.5,130) -- (332.5,130) -- cycle ;
\draw [line width=0.5]    (240,95) -- (215,95) ;
\draw [line width=0.5]    (240,55) -- (240,95) ;
\draw [line width=0.5]    (265,95) -- (240,95) ;
\draw  [fill={rgb, 255:red, 245; green, 240; blue, 235 }  ,fill opacity=1 ][line width=0.5]  (222.32,98.54) -- (243.54,77.32) -- (257.68,91.46) -- (236.46,112.68) -- cycle ;
\draw [line width=0.5]    (290,95) -- (265,95) ;
\draw [line width=0.5]    (290,55) -- (290,95) ;
\draw [line width=0.5]    (315,95) -- (290,95) ;
\draw  [fill={rgb, 255:red, 245; green, 240; blue, 235 }  ,fill opacity=1 ][line width=0.5]  (275,80) -- (305,80) -- (304.31,110) -- (274.31,110) -- cycle ;
\draw [line width=0.5]    (340,95) -- (315,95) ;
\draw [line width=0.5]    (340,55) -- (340,95) ;
\draw [line width=0.5]    (365,95) -- (340,95) ;
\draw  [fill={rgb, 255:red, 245; green, 240; blue, 235 }  ,fill opacity=1 ][line width=0.5]  (336.46,77.32) -- (357.68,98.54) -- (343.54,112.68) -- (322.32,91.46) -- cycle ;
\draw  [fill={rgb, 255:red, 245; green, 240; blue, 235 }  ,fill opacity=1 ][line width=0.5]  (289.65,140) -- (297.15,130) -- (282.15,130) -- cycle ;

\draw [line width=0.5]    (287.15,137.5) -- (287.15,97.5) ;
\draw [line width=0.5]    (237.5,137.5) -- (237.5,97.5) ;
\draw  [fill={rgb, 255:red, 245; green, 240; blue, 235 }  ,fill opacity=1 ][line width=0.5]  (237.5,142.5) -- (245,132.5) -- (230,132.5) -- cycle ;
\draw [line width=0.5]    (337.5,137.5) -- (337.5,97.5) ;
\draw  [fill={rgb, 255:red, 245; green, 240; blue, 235 }  ,fill opacity=1 ][line width=0.5]  (337.5,142.5) -- (345,132.5) -- (330,132.5) -- cycle ;
\draw [line width=0.5]    (237.5,97.5) -- (212.5,97.5) ;
\draw [line width=0.5]    (237.5,57.5) -- (237.5,97.5) ;
\draw [line width=0.5]    (262.5,97.5) -- (237.5,97.5) ;
\draw  [fill={rgb, 255:red, 245; green, 240; blue, 235 }  ,fill opacity=1 ][line width=0.5]  (219.82,101.04) -- (241.04,79.82) -- (255.18,93.96) -- (233.96,115.18) -- cycle ;
\draw [line width=0.5]    (287.5,97.5) -- (262.5,97.5) ;
\draw [line width=0.5]    (287.5,57.5) -- (287.5,97.5) ;
\draw [line width=0.5]    (312.5,97.5) -- (287.5,97.5) ;
\draw  [fill={rgb, 255:red, 245; green, 240; blue, 235 }  ,fill opacity=1 ][line width=0.5]  (272.5,82.5) -- (302.5,82.5) -- (301.81,112.5) -- (271.81,112.5) -- cycle ;
\draw [line width=0.5]    (337.5,97.5) -- (312.5,97.5) ;
\draw [line width=0.5]    (337.5,57.5) -- (337.5,97.5) ;
\draw [line width=0.5]    (362.5,97.5) -- (337.5,97.5) ;
\draw  [fill={rgb, 255:red, 245; green, 240; blue, 235 }  ,fill opacity=1 ][line width=0.5]  (333.96,79.82) -- (355.18,101.04) -- (341.04,115.18) -- (319.82,93.96) -- cycle ;
\draw  [fill={rgb, 255:red, 245; green, 240; blue, 235 }  ,fill opacity=1 ][line width=0.5]  (287.15,142.5) -- (294.65,132.5) -- (279.65,132.5) -- cycle ;

\draw [line width=0.5]    (284.65,140) -- (284.65,100) ;
\draw [line width=0.5]    (235,140) -- (235,100) ;
\draw  [fill={rgb, 255:red, 245; green, 240; blue, 235 }  ,fill opacity=1 ][line width=0.5]  (235,145) -- (242.5,135) -- (227.5,135) -- cycle ;
\draw [line width=0.5]    (335,140) -- (335,100) ;
\draw  [fill={rgb, 255:red, 245; green, 240; blue, 235 }  ,fill opacity=1 ][line width=0.5]  (335,145) -- (342.5,135) -- (327.5,135) -- cycle ;
\draw [line width=0.5]    (235,100) -- (210,100) ;
\draw [line width=0.5]    (235,60) -- (235,100) ;
\draw [line width=0.5]    (260,100) -- (235,100) ;
\draw  [fill={rgb, 255:red, 245; green, 240; blue, 235 }  ,fill opacity=1 ][line width=0.5]  (217.32,103.54) -- (238.54,82.32) -- (252.68,96.46) -- (231.46,117.68) -- cycle ;
\draw [line width=0.5]    (285,100) -- (260,100) ;
\draw [line width=0.5]    (285,60) -- (285,100) ;
\draw [line width=0.5]    (310,100) -- (285,100) ;
\draw  [fill={rgb, 255:red, 245; green, 240; blue, 235 }  ,fill opacity=1 ][line width=0.5]  (270,85) -- (300,85) -- (299.31,115) -- (269.31,115) -- cycle ;
\draw [line width=0.5]    (335,100) -- (310,100) ;
\draw [line width=0.5]    (335,60) -- (335,100) ;
\draw [line width=0.5]    (360,100) -- (335,100) ;
\draw  [fill={rgb, 255:red, 245; green, 240; blue, 235 }  ,fill opacity=1 ][line width=0.5]  (331.46,82.32) -- (352.68,103.54) -- (338.54,117.68) -- (317.32,96.46) -- cycle ;
\draw  [fill={rgb, 255:red, 245; green, 240; blue, 235 }  ,fill opacity=1 ][line width=0.5]  (284.65,145) -- (292.15,135) -- (277.15,135) -- cycle ;

\draw [line width=0.5]    (425,140) -- (425,100) ;
\draw  [fill={rgb, 255:red, 245; green, 240; blue, 235 }  ,fill opacity=1 ][line width=0.5]  (425,145) -- (432.5,135) -- (417.5,135) -- cycle ;
\draw [line width=0.5]    (425,100) -- (400,100) ;
\draw [line width=0.5]    (425,60) -- (425,100) ;
\draw [line width=0.5]    (450,100) -- (425,100) ;
\draw  [fill={rgb, 255:red, 245; green, 240; blue, 235 }  ,fill opacity=1 ][line width=0.5]  (421.46,82.32) -- (442.68,103.54) -- (428.54,117.68) -- (407.32,96.46) -- cycle ;
\draw [line width=0.5]    (475,100) -- (450,100) ;
\draw [line width=0.5]    (475,60) -- (475,100) ;

\draw [line width=0.5]    (145,140) -- (145,100) ;
\draw  [fill={rgb, 255:red, 245; green, 240; blue, 235 }  ,fill opacity=1 ][line width=0.5]  (145,145) -- (152.5,135) -- (137.5,135) -- cycle ;
\draw [line width=0.5]    (145,100) -- (120,100) ;
\draw [line width=0.5]    (145,60) -- (145,100) ;
\draw [line width=0.5]    (170,100) -- (145,100) ;
\draw  [fill={rgb, 255:red, 245; green, 240; blue, 235 }  ,fill opacity=1 ][line width=0.5]  (127.32,103.54) -- (148.54,82.32) -- (162.68,96.46) -- (141.46,117.68) -- cycle ;
\draw [line width=0.5]    (120,100) -- (95,100) ;
\draw [line width=0.5]    (95,60) -- (95,100) ;

\draw (192.5,95) node   [align=left] {$ ...$};
\draw (382.5,95) node   [align=left] {$ ...$};
\draw (286,160) node  [font=\footnotesize] [align=left] {$i$};
\draw (336,160) node  [font=\footnotesize] [align=left] {$i+1$};
\draw (426,160) node  [font=\footnotesize] [align=left] {$N_g$};
\draw (236,160) node  [font=\footnotesize] [align=left] {$i-1$};
\draw (146,160) node  [font=\footnotesize] [align=left] {$1$};

\end{mytikz3}\quad.
\end{equation}
Here we take $t=2$ as an example for illustration, in a manner similar to the diagram in Eq.\,\eqref{eq:weingarten_tn_2qudit}. We use $\{1, \ldots, i-1, i, i+1, \ldots, N_g\}$ to mark the positions of the $N_g$ unitary gates of the MPS. According to the Weingarten formula in Eq.\,\eqref{eq:weingarten_tn_2qudit}, the expectation of the MPS replicas $\mathbb{E}_{\mathbb{U}^{[i]}} [\mathbf{U}^{\otimes t} \otimes \mathbf{U}^{ *\otimes t} \ket{\boldsymbol{0}}^{\otimes 2t}] $ is equal to
\begin{equation}\label{eq:integrated_mps_replica}
\begin{mytikz4}

\draw  [color={rgb, 255:red, 150; green, 150; blue, 150 }  ,draw opacity=1 ][dash pattern={on 3.75pt off 6pt}] (55,85) -- (235,85) -- (235,195) -- (55,195) -- cycle ;
\draw  [color={rgb, 255:red, 150; green, 150; blue, 150 }  ,draw opacity=1 ][dash pattern={on 3.75pt off 6pt}] (295,85) -- (475,85) -- (475,195) -- (295,194) -- cycle ;
\draw  [color={rgb, 255:red, 150; green, 150; blue, 150 }  ,draw opacity=1 ][dash pattern={on 3.75pt off 6pt}] (585,85) -- (765,85) -- (765,195) -- (585,195) -- cycle ;
\draw  [color={rgb, 255:red, 150; green, 150; blue, 150 }  ,draw opacity=1 ][dash pattern={on 3.75pt off 6pt}] (475,85) -- (585,85) -- (585,225) -- (475,225) -- cycle ;
\draw  [color={rgb, 255:red, 150; green, 150; blue, 150 }  ,draw opacity=1 ][dash pattern={on 3.75pt off 6pt}] (825,85) -- (1005,85) -- (1005,195) -- (825,195) -- cycle ;
\draw [line width=0.5]    (530,140) -- (585,140) ;
\draw [line width=0.5]    (475,140) -- (530,140) ;
\draw [line width=0.5]    (585,140) -- (675,140) ;
\draw [line width=0.5]    (745,140) -- (675,140) ;
\draw [line width=0.5]    (745,140) -- (780,140) ;
\draw [line width=0.5]    (710,140.02) -- (710,60) ;
\draw  [fill={rgb, 255:red, 142; green, 142; blue, 142 }  ,fill opacity=1 ][line width=0.5]  (706.5,140.02) .. controls (706.5,138.08) and (708.07,136.52) .. (710,136.52) .. controls (711.93,136.52) and (713.5,138.08) .. (713.5,140.02) .. controls (713.5,141.95) and (711.93,143.52) .. (710,143.52) .. controls (708.07,143.52) and (706.5,141.95) .. (706.5,140.02) -- cycle ;
\draw [line width=0.5]    (640,220.02) -- (640,140) ;
\draw  [fill={rgb, 255:red, 142; green, 142; blue, 142 }  ,fill opacity=1 ][line width=0.5]  (636.5,140) .. controls (636.5,138.07) and (638.07,136.5) .. (640,136.5) .. controls (641.93,136.5) and (643.5,138.07) .. (643.5,140) .. controls (643.5,141.93) and (641.93,143.5) .. (640,143.5) .. controls (638.07,143.5) and (636.5,141.93) .. (636.5,140) -- cycle ;
\draw  [fill={rgb, 255:red, 245; green, 240; blue, 235 }  ,fill opacity=1 ][line width=0.5]  (640,225.02) -- (647.5,215.02) -- (632.5,215.02) -- cycle ;
\draw [line width=0.5]    (385,140) -- (475,140) ;
\draw [line width=0.5]    (385,140) -- (330,140.02) ;
\draw [line width=0.5]    (279.3,140.02) -- (330,140.02) ;
\draw [line width=0.5]    (350,140.02) -- (350,60) ;
\draw  [fill={rgb, 255:red, 142; green, 142; blue, 142 }  ,fill opacity=1 ][line width=0.5]  (346.5,140.02) .. controls (346.5,138.08) and (348.07,136.52) .. (350,136.52) .. controls (351.93,136.52) and (353.5,138.08) .. (353.5,140.02) .. controls (353.5,141.95) and (351.93,143.52) .. (350,143.52) .. controls (348.07,143.52) and (346.5,141.95) .. (346.5,140.02) -- cycle ;
\draw [line width=0.5]    (420,220.02) -- (420,140) ;
\draw  [fill={rgb, 255:red, 142; green, 142; blue, 142 }  ,fill opacity=1 ][line width=0.5]  (416.5,140) .. controls (416.5,138.07) and (418.07,136.5) .. (420,136.5) .. controls (421.93,136.5) and (423.5,138.07) .. (423.5,140) .. controls (423.5,141.93) and (421.93,143.5) .. (420,143.5) .. controls (418.07,143.5) and (416.5,141.93) .. (416.5,140) -- cycle ;
\draw  [fill={rgb, 255:red, 245; green, 240; blue, 235 }  ,fill opacity=1 ][line width=0.5]  (420,225.02) -- (427.5,215.02) -- (412.5,215.02) -- cycle ;
\draw [line width=0.5]    (530,250) -- (530,140) ;
\draw  [fill={rgb, 255:red, 245; green, 240; blue, 235 }  ,fill opacity=1 ][line width=0.5]  (530,255) -- (537.5,245) -- (522.5,245) -- cycle ;
\draw [line width=0.5]    (530,140) -- (530,59.98) ;
\draw  [fill={rgb, 255:red, 142; green, 142; blue, 142 }  ,fill opacity=1 ][line width=0.5]  (526.5,140) .. controls (526.5,138.07) and (528.07,136.5) .. (530,136.5) .. controls (531.93,136.5) and (533.5,138.07) .. (533.5,140) .. controls (533.5,141.93) and (531.93,143.5) .. (530,143.5) .. controls (528.07,143.5) and (526.5,141.93) .. (526.5,140) -- cycle ;
\draw  [fill={rgb, 255:red, 245; green, 240; blue, 235 }  ,fill opacity=1 ][line width=0.5]  (515,160) -- (545,160) -- (545,190) -- (515,190) -- cycle ;
\draw  [fill={rgb, 255:red, 245; green, 240; blue, 235 }  ,fill opacity=1 ][line width=0.5]  (295,155.02) -- (295,125.02) -- (310,140.02) -- cycle ;
\draw  [fill={rgb, 255:red, 245; green, 240; blue, 235 }  ,fill opacity=1 ][line width=0.5]  (475,155) -- (475,125) -- (490,140) -- cycle ;
\draw  [fill={rgb, 255:red, 245; green, 240; blue, 235 }  ,fill opacity=1 ][line width=0.5]  (475,125) -- (475,155) -- (460,140) -- cycle ;
\draw  [fill={rgb, 255:red, 245; green, 240; blue, 235 }  ,fill opacity=1 ][line width=0.5]  (585,155) -- (585,125) -- (600,140) -- cycle ;
\draw  [fill={rgb, 255:red, 245; green, 240; blue, 235 }  ,fill opacity=1 ][line width=0.5]  (585,125) -- (585,155) -- (570,140) -- cycle ;
\draw  [fill={rgb, 255:red, 245; green, 240; blue, 235 }  ,fill opacity=1 ][line width=0.5]  (765,125) -- (765,155) -- (750,140) -- cycle ;
\draw  [fill={rgb, 255:red, 245; green, 240; blue, 235 }  ,fill opacity=1 ][line width=0.5]  (660,125) -- (690,125) -- (690,155) -- (660,155) -- cycle ;
\draw  [fill={rgb, 255:red, 245; green, 240; blue, 235 }  ,fill opacity=1 ][line width=0.5]  (370,125) -- (400,125) -- (400,155) -- (370,155) -- cycle ;
\draw  [fill={rgb, 255:red, 245; green, 240; blue, 235 }  ,fill opacity=1 ][line width=0.5]  (335,85) -- (365,85) -- (350,100) -- cycle ;
\draw  [fill={rgb, 255:red, 245; green, 240; blue, 235 }  ,fill opacity=1 ][line width=0.5]  (695,85.01) -- (725,85.01) -- (710,100.01) -- cycle ;
\draw  [fill={rgb, 255:red, 245; green, 240; blue, 235 }  ,fill opacity=1 ][line width=0.5]  (545,225) -- (515,225) -- (530,210) -- cycle ;
\draw  [fill={rgb, 255:red, 245; green, 240; blue, 235 }  ,fill opacity=1 ][line width=0.5]  (435,195.01) -- (405,195.01) -- (420,180.01) -- cycle ;
\draw  [fill={rgb, 255:red, 245; green, 240; blue, 235 }  ,fill opacity=1 ][line width=0.5]  (515,85) -- (545,85) -- (530,100) -- cycle ;
\draw  [fill={rgb, 255:red, 245; green, 240; blue, 235 }  ,fill opacity=1 ][line width=0.5]  (655,195.01) -- (625,195.01) -- (640,180.01) -- cycle ;
\draw [line width=0.5]    (145,140) -- (250.8,140) ;
\draw [line width=0.5]    (145,140) -- (90,140.02) ;
\draw [line width=0.5]    (20,140.02) -- (90,140.02) ;
\draw [line width=0.5]    (110,140.02) -- (110,60) ;
\draw  [fill={rgb, 255:red, 142; green, 142; blue, 142 }  ,fill opacity=1 ][line width=0.5]  (106.5,140.02) .. controls (106.5,138.08) and (108.07,136.52) .. (110,136.52) .. controls (111.93,136.52) and (113.5,138.08) .. (113.5,140.02) .. controls (113.5,141.95) and (111.93,143.52) .. (110,143.52) .. controls (108.07,143.52) and (106.5,141.95) .. (106.5,140.02) -- cycle ;
\draw [line width=0.5]    (180,220.02) -- (180,140) ;
\draw  [fill={rgb, 255:red, 142; green, 142; blue, 142 }  ,fill opacity=1 ][line width=0.5]  (176.5,140) .. controls (176.5,138.07) and (178.07,136.5) .. (180,136.5) .. controls (181.93,136.5) and (183.5,138.07) .. (183.5,140) .. controls (183.5,141.93) and (181.93,143.5) .. (180,143.5) .. controls (178.07,143.5) and (176.5,141.93) .. (176.5,140) -- cycle ;
\draw  [fill={rgb, 255:red, 245; green, 240; blue, 235 }  ,fill opacity=1 ][line width=0.5]  (180,225.02) -- (187.5,215.02) -- (172.5,215.02) -- cycle ;
\draw  [fill={rgb, 255:red, 245; green, 240; blue, 235 }  ,fill opacity=1 ][line width=0.5]  (55,155.02) -- (55,125.02) -- (70,140.02) -- cycle ;
\draw  [fill={rgb, 255:red, 245; green, 240; blue, 235 }  ,fill opacity=1 ][line width=0.5]  (235,125) -- (235,155) -- (220,140) -- cycle ;
\draw  [fill={rgb, 255:red, 245; green, 240; blue, 235 }  ,fill opacity=1 ][line width=0.5]  (130,125) -- (160,125) -- (160,155) -- (130,155) -- cycle ;
\draw  [fill={rgb, 255:red, 245; green, 240; blue, 235 }  ,fill opacity=1 ][line width=0.5]  (94.99,85.02) -- (124.99,85) -- (110,100.01) -- cycle ;
\draw  [fill={rgb, 255:red, 245; green, 240; blue, 235 }  ,fill opacity=1 ][line width=0.5]  (195,195.01) -- (165,195.01) -- (180,180.01) -- cycle ;
\draw [line width=0.5]    (20,140.02) -- (20,60) ;
\draw [line width=0.5]    (810,140) -- (915,140) ;
\draw [line width=0.5]    (985,140) -- (915,140) ;
\draw [line width=0.5]    (985,140) -- (1040,140) ;
\draw [line width=0.5]    (950,140.02) -- (950,60) ;
\draw  [fill={rgb, 255:red, 142; green, 142; blue, 142 }  ,fill opacity=1 ][line width=0.5]  (946.5,140.02) .. controls (946.5,138.08) and (948.07,136.52) .. (950,136.52) .. controls (951.93,136.52) and (953.5,138.08) .. (953.5,140.02) .. controls (953.5,141.95) and (951.93,143.52) .. (950,143.52) .. controls (948.07,143.52) and (946.5,141.95) .. (946.5,140.02) -- cycle ;
\draw [line width=0.5]    (880,220.02) -- (880,140) ;
\draw  [fill={rgb, 255:red, 142; green, 142; blue, 142 }  ,fill opacity=1 ][line width=0.5]  (876.5,140) .. controls (876.5,138.07) and (878.07,136.5) .. (880,136.5) .. controls (881.93,136.5) and (883.5,138.07) .. (883.5,140) .. controls (883.5,141.93) and (881.93,143.5) .. (880,143.5) .. controls (878.07,143.5) and (876.5,141.93) .. (876.5,140) -- cycle ;
\draw  [fill={rgb, 255:red, 245; green, 240; blue, 235 }  ,fill opacity=1 ][line width=0.5]  (880,225.02) -- (887.5,215.02) -- (872.5,215.02) -- cycle ;
\draw  [fill={rgb, 255:red, 245; green, 240; blue, 235 }  ,fill opacity=1 ][line width=0.5]  (825,155) -- (825,125) -- (840,140) -- cycle ;
\draw  [fill={rgb, 255:red, 245; green, 240; blue, 235 }  ,fill opacity=1 ][line width=0.5]  (1005,125) -- (1005,155) -- (990,140) -- cycle ;
\draw  [fill={rgb, 255:red, 245; green, 240; blue, 235 }  ,fill opacity=1 ][line width=0.5]  (900,125) -- (930,125) -- (930,155) -- (900,155) -- cycle ;
\draw  [fill={rgb, 255:red, 245; green, 240; blue, 235 }  ,fill opacity=1 ][line width=0.5]  (935,85) -- (965,85) -- (950,100) -- cycle ;
\draw  [fill={rgb, 255:red, 245; green, 240; blue, 235 }  ,fill opacity=1 ][line width=0.5]  (895,195.01) -- (865,195.01) -- (880,180.01) -- cycle ;
\draw [line width=0.5]    (1040,140) -- (1040,59.98) ;

\draw (385,113) node  [font=\footnotesize] [align=left] {$ W( Dd)$};
\draw (675,113) node  [font=\footnotesize] [align=left] {$ W( Dd)$};
\draw (583,175) node  [font=\footnotesize] [align=left] {$ W(D^{2} d)$};
\draw (585,113) node  [font=\footnotesize] [align=left] {$ G( D)$};
\draw (475,113) node  [font=\footnotesize] [align=left] {$ G( D)$};
\draw (795,140) node  [font=\normalsize] [align=left] {$...$};
\draw (265,140) node  [font=\normalsize] [align=left] {$...$};
\draw (145,113) node  [font=\footnotesize] [align=left] {$ W( Dd)$};
\draw (915,113) node  [font=\footnotesize] [align=left] {$ W( Dd)$};

\end{mytikz4}~~,
\end{equation}
where the symbols of the tensors are written beside the nodes instead of inside them for convenience. The squares with diagonal lines represent the $t$-degree Gram matrices $G^{(t)}(D)$, and the squares without diagonal lines represent the $t$-degree Weingarten matrices $W^{(t)}(Dd)$. The replica index $t$ is omitted in Eq.\,\eqref{eq:integrated_mps_replica}. The small triangles represent the replicas of the input zero states $\ket{0}$ of corresponding dimensions. The large triangles represent the permutation tensors $S(D)$ or $S(d)$ defined in Eq.\,\eqref{eq:wg_S_def}. The grey dots represent the COPY tensors defined in Eq.\,\eqref{eq:copy_tensor_def}. The dashed boxes indicate the ``Weingarten gates'' defined in Eq.\,\eqref{eq:weingarten_gate}, as the integration results of individual unitary operators $\{U_1, \ldots, U_{i-1}, U_{i}, U_{i+1}, \ldots, U_{N_g}\}$. According to the identities in Eqs.\,\eqref{eq:sum_G}, \eqref{eq:sum_Wg}, and \eqref{eq:permutation_vector_zero_state}, the diagram in Eq.\eqref{eq:integrated_mps_replica} is simplified to
\begin{equation}\label{eq:integrated_mps_replica_simplified}
\begin{mytikz3}

\draw [line width=0.5]    (430,140) -- (470,140) ;
\draw [line width=0.5]    (390,140) -- (430,140) ;
\draw [line width=0.5]    (470,140) -- (520,140) ;
\draw [line width=0.5]    (560,140) -- (520,140) ;
\draw [line width=0.5]    (560,140) -- (600,140) ;
\draw [line width=0.5]    (560,140) -- (560,80) ;
\draw  [fill={rgb, 255:red, 142; green, 142; blue, 142 }  ,fill opacity=1 ][line width=0.5]  (556.5,140) .. controls (556.5,138.07) and (558.07,136.5) .. (560,136.5) .. controls (561.93,136.5) and (563.5,138.07) .. (563.5,140) .. controls (563.5,141.93) and (561.93,143.5) .. (560,143.5) .. controls (558.07,143.5) and (556.5,141.93) .. (556.5,140) -- cycle ;
\draw [line width=0.5]    (340,140) -- (390,140) ;
\draw [line width=0.5]    (340,140) -- (300,140) ;
\draw [line width=0.5]    (260,140) -- (300,140) ;
\draw [line width=0.5]    (300,140) -- (300,80) ;
\draw  [fill={rgb, 255:red, 142; green, 142; blue, 142 }  ,fill opacity=1 ][line width=0.5]  (296.51,140.02) .. controls (296.51,138.08) and (298.07,136.52) .. (300.01,136.52) .. controls (301.94,136.52) and (303.51,138.08) .. (303.51,140.02) .. controls (303.51,141.95) and (301.94,143.52) .. (300.01,143.52) .. controls (298.07,143.52) and (296.51,141.95) .. (296.51,140.02) -- cycle ;
\draw [line width=0.5]    (430,139.79) -- (430,80) ;
\draw  [fill={rgb, 255:red, 142; green, 142; blue, 142 }  ,fill opacity=1 ][line width=0.5]  (426.5,140) .. controls (426.5,138.07) and (428.07,136.5) .. (430,136.5) .. controls (431.93,136.5) and (433.5,138.07) .. (433.5,140) .. controls (433.5,141.93) and (431.93,143.5) .. (430,143.5) .. controls (428.07,143.5) and (426.5,141.93) .. (426.5,140) -- cycle ;
\draw  [fill={rgb, 255:red, 245; green, 240; blue, 235 }  ,fill opacity=1 ][line width=0.5]  (260,155.02) -- (260,125.02) -- (275,140.02) -- cycle ;
\draw  [fill={rgb, 255:red, 245; green, 240; blue, 235 }  ,fill opacity=1 ][line width=0.5]  (390,155) -- (390,125) -- (405,140) -- cycle ;
\draw  [fill={rgb, 255:red, 245; green, 240; blue, 235 }  ,fill opacity=1 ][line width=0.5]  (390,125) -- (390,155) -- (375,140) -- cycle ;
\draw  [fill={rgb, 255:red, 245; green, 240; blue, 235 }  ,fill opacity=1 ][line width=0.5]  (470,155) -- (470,125) -- (485,140) -- cycle ;
\draw  [fill={rgb, 255:red, 245; green, 240; blue, 235 }  ,fill opacity=1 ][line width=0.5]  (470,125) -- (470,155) -- (455,140) -- cycle ;
\draw  [fill={rgb, 255:red, 245; green, 240; blue, 235 }  ,fill opacity=1 ][line width=0.5]  (600,125) -- (600,155) -- (585,140) -- cycle ;
\draw  [fill={rgb, 255:red, 245; green, 240; blue, 235 }  ,fill opacity=1 ][line width=0.5]  (505,125) -- (535,125) -- (535,155) -- (505,155) -- cycle ;
\draw  [fill={rgb, 255:red, 245; green, 240; blue, 235 }  ,fill opacity=1 ][line width=0.5]  (325,125) -- (355,125) -- (355,155) -- (325,155) -- cycle ;
\draw [line width=0.5]    (600,140) -- (620,140) ;
\draw [line width=0.5]    (240,140) -- (260,140) ;
\draw  [fill={rgb, 255:red, 245; green, 240; blue, 235 }  ,fill opacity=1 ][line width=0.5]  (545,100) -- (574.98,100) -- (559.99,114.99) -- cycle ;
\draw [line width=0.5]    (660,140) -- (710,140) ;
\draw [line width=0.5]    (750,140) -- (710,140) ;
\draw [line width=0.5]    (750,140) -- (790,140) ;
\draw [line width=0.5]    (750,140) -- (750,80) ;
\draw  [fill={rgb, 255:red, 142; green, 142; blue, 142 }  ,fill opacity=1 ][line width=0.5]  (746.5,140) .. controls (746.5,138.07) and (748.07,136.5) .. (750,136.5) .. controls (751.93,136.5) and (753.5,138.07) .. (753.5,140) .. controls (753.5,141.93) and (751.93,143.5) .. (750,143.5) .. controls (748.07,143.5) and (746.5,141.93) .. (746.5,140) -- cycle ;
\draw  [fill={rgb, 255:red, 245; green, 240; blue, 235 }  ,fill opacity=1 ][line width=0.5]  (660,155) -- (660,125) -- (675,140) -- cycle ;
\draw  [fill={rgb, 255:red, 245; green, 240; blue, 235 }  ,fill opacity=1 ][line width=0.5]  (790,125) -- (790,155) -- (775,140) -- cycle ;
\draw  [fill={rgb, 255:red, 245; green, 240; blue, 235 }  ,fill opacity=1 ][line width=0.5]  (695,125) -- (725,125) -- (725,155) -- (695,155) -- cycle ;
\draw [line width=0.5]    (790,140) -- (820,140) ;
\draw [line width=0.5]    (640,140) -- (660,140) ;
\draw [line width=0.5]    (150,140) -- (200,140) ;
\draw [line width=0.5]    (150,140) -- (110,140) ;
\draw [line width=0.5]    (70,140) -- (110,140) ;
\draw [line width=0.5]    (110,140) -- (110,80) ;
\draw  [fill={rgb, 255:red, 142; green, 142; blue, 142 }  ,fill opacity=1 ][line width=0.5]  (106.51,140.02) .. controls (106.51,138.08) and (108.07,136.52) .. (110.01,136.52) .. controls (111.94,136.52) and (113.51,138.08) .. (113.51,140.02) .. controls (113.51,141.95) and (111.94,143.52) .. (110.01,143.52) .. controls (108.07,143.52) and (106.51,141.95) .. (106.51,140.02) -- cycle ;
\draw  [fill={rgb, 255:red, 245; green, 240; blue, 235 }  ,fill opacity=1 ][line width=0.5]  (70,155.02) -- (70,125.02) -- (85,140.02) -- cycle ;
\draw  [fill={rgb, 255:red, 245; green, 240; blue, 235 }  ,fill opacity=1 ][line width=0.5]  (200,125) -- (200,155) -- (185,140) -- cycle ;
\draw  [fill={rgb, 255:red, 245; green, 240; blue, 235 }  ,fill opacity=1 ][line width=0.5]  (135,125) -- (165,125) -- (165,155) -- (135,155) -- cycle ;
\draw [line width=0.5]    (40,140) -- (70,140) ;
\draw [line width=0.5]    (200,140) -- (220,140) ;
\draw [line width=0.5]    (40,140) -- (40,80) ;
\draw [line width=0.5]    (820,140) -- (820,80) ;
\draw  [fill={rgb, 255:red, 245; green, 240; blue, 235 }  ,fill opacity=1 ][line width=0.5]  (735,100) -- (764.98,100) -- (749.99,114.99) -- cycle ;
\draw  [fill={rgb, 255:red, 245; green, 240; blue, 235 }  ,fill opacity=1 ][line width=0.5]  (415,100) -- (444.98,100) -- (429.99,114.99) -- cycle ;
\draw  [fill={rgb, 255:red, 245; green, 240; blue, 235 }  ,fill opacity=1 ][line width=0.5]  (285.01,100) -- (314.99,100) -- (300,114.99) -- cycle ;
\draw  [fill={rgb, 255:red, 245; green, 240; blue, 235 }  ,fill opacity=1 ][line width=0.5]  (95,100) -- (124.98,100) -- (109.99,114.99) -- cycle ;

\draw (630,140) node  [font=\normalsize] [align=left] {$...$};
\draw (230,140) node  [font=\normalsize] [align=left] {$...$};
\draw (430,165) node  [font=\footnotesize] [align=left] {$i$};
\draw (560,165) node  [font=\footnotesize] [align=left] {$i+1$};
\draw (300,165) node  [font=\footnotesize] [align=left] {$i-1$};
\draw (110,165) node  [font=\footnotesize] [align=left] {$1$};
\draw (750,165) node  [font=\footnotesize] [align=left] {$N_g$};

\end{mytikz3}\quad,
\end{equation}
times a factor of
\begin{equation}\label{eq:factor_D2d}
    \frac{1}{D^2 d(D^2 d+1)\cdots (D^2d+t-1)} ~~,
\end{equation}
where we have used the identities
\begin{equation}\label{eq:S0=1_C1=C_W1=1}
\begin{mytikz3}

\draw [line width=0.5]    (180,210) -- (180,110) ;
\draw  [fill={rgb, 255:red, 245; green, 240; blue, 235 }  ,fill opacity=1 ][line width=0.5]  (180,222.5) -- (198.84,197.5) -- (161.16,197.5) -- cycle ;
\draw  [fill={rgb, 255:red, 245; green, 240; blue, 235 }  ,fill opacity=1 ][line width=0.5]  (215,165) -- (145,165) -- (180,130) -- cycle ;

\draw (171,142) node [anchor=north west][inner sep=0.75pt]   [align=left] {$ S$};
\draw (173,199) node [anchor=north west][inner sep=0.75pt]   [align=left] {$ 0$};

\end{mytikz3}
~~=~~
\begin{mytikz3}

\draw [line width=0.5]    (180,210) -- (180,140.03) ;
\draw  [fill={rgb, 255:red, 245; green, 240; blue, 235 }  ,fill opacity=1 ][line width=0.5]  (180,224) -- (201.17,196) -- (158.83,196) -- cycle ;
\draw (173,196) node [anchor=north west][inner sep=0.75pt]  [font=\normalsize] [align=left] {$ \vec{1}$};

\end{mytikz3}
~~,~~\quad
\begin{mytikz3}

\draw [line width=0.5]    (220,150) -- (180,150) ;
\draw [line width=0.5]    (180,210) -- (180,150) ;
\draw  [fill={rgb, 255:red, 245; green, 240; blue, 235 }  ,fill opacity=1 ][line width=0.5]  (180,224) -- (201.17,196) -- (158.83,196) -- cycle ;
\draw [line width=0.5]    (180,150) -- (140,150) ;
\draw [line width=0.5]    (180,150) -- (180,110) ;
\draw  [fill={rgb, 255:red, 142; green, 142; blue, 142 }  ,fill opacity=1 ][line width=0.5]  (176.5,150) .. controls (176.5,148.07) and (178.07,146.5) .. (180,146.5) .. controls (181.93,146.5) and (183.5,148.07) .. (183.5,150) .. controls (183.5,151.93) and (181.93,153.5) .. (180,153.5) .. controls (178.07,153.5) and (176.5,151.93) .. (176.5,150) -- cycle ;

\draw (173,196) node [anchor=north west][inner sep=0.75pt]  [font=\normalsize] [align=left] {$ \vec{1}$};

\end{mytikz3}
~~=~~
\begin{mytikz3}

\draw [line width=0.5]    (220,150) -- (180,150) ;
\draw [line width=0.5]    (180,150) -- (140,150) ;
\draw [line width=0.5]    (180,150) -- (180,110) ;
\draw  [fill={rgb, 255:red, 142; green, 142; blue, 142 }  ,fill opacity=1 ][line width=0.5]  (176.5,150) .. controls (176.5,148.07) and (178.07,146.5) .. (180,146.5) .. controls (181.93,146.5) and (183.5,148.07) .. (183.5,150) .. controls (183.5,151.93) and (181.93,153.5) .. (180,153.5) .. controls (178.07,153.5) and (176.5,151.93) .. (176.5,150) -- cycle ;

\end{mytikz3}
~~,~~\quad
\begin{mytikz3}

\draw [line width=0.5]    (175,200) -- (175,125) ;
\draw [line width=0.5]    (175,125) -- (175,60) ;
\draw  [fill={rgb, 255:red, 245; green, 240; blue, 235 }  ,fill opacity=1 ][line width=0.5]  (145,95) -- (205,95) -- (205,155) -- (145,155) -- cycle ;
\draw  [fill={rgb, 255:red, 245; green, 240; blue, 235 }  ,fill opacity=1 ][line width=0.5]  (175,214) -- (196.17,186) -- (153.83,186) -- cycle ;

\draw (175,125) node  [font=\normalsize] [align=left] {$ W(d)$};
\draw (175,197) node  [font=\normalsize] [align=left] {$ \vec{1}$};

\end{mytikz3}
~~=~~\frac{1}{d(d+1)\cdots (d+t-1)}\times
\begin{mytikz3}

\draw [line width=0.5]    (175,200) -- (175,125) ;
\draw  [fill={rgb, 255:red, 245; green, 240; blue, 235 }  ,fill opacity=1 ][line width=0.5]  (175,214) -- (196.17,186) -- (153.83,186) -- cycle ;

\draw (175,197) node  [font=\normalsize] [align=left] {$ \vec{1}$};

\end{mytikz3}~~.
\end{equation}
The replica index $t$ is again omitted for the tensor notations in Eq.\,\eqref{eq:S0=1_C1=C_W1=1}, i.e., $S$ refers to the permutation tensor corresponding to the $t$-degree symmetric group $\mathcal{S}_t$, the symbol $0$ refers to the zero state $\ket{0}^{\otimes 2t}$, $\vec{1}$ refers to the $t!$-dimensional all-one vector defined in Eq.\,\eqref{eq:all-one_vector}.

According to Lemma~\ref{lemma:Wg_G_commute}, we know the Weingarten matrix $W^{(t)}(Dd)$ commutes with the Gram matrix $G^{(t)}(D)$, i.e.,
\begin{equation}
\begin{mytikz3}

\draw [line width=0.5]    (430,140) -- (470,140) ;
\draw [line width=0.5]    (470,140) -- (520,140) ;
\draw [line width=0.5]    (560,140) -- (520,140) ;
\draw  [fill={rgb, 255:red, 245; green, 240; blue, 235 }  ,fill opacity=1 ][line width=0.5]  (470,155) -- (470,125) -- (485,140) -- cycle ;
\draw  [fill={rgb, 255:red, 245; green, 240; blue, 235 }  ,fill opacity=1 ][line width=0.5]  (470,125) -- (470,155) -- (455,140) -- cycle ;
\draw  [fill={rgb, 255:red, 245; green, 240; blue, 235 }  ,fill opacity=1 ][line width=0.5]  (505,125) -- (535,125) -- (535,155) -- (505,155) -- cycle ;

\draw (520,110) node  [font=\footnotesize] [align=left] {$W(Dd)$};
\draw (470,110) node  [font=\footnotesize] [align=left] {$G(D)$};
\draw (470,170) node  [font=\footnotesize] [align=left] {$ $};

\end{mytikz3}
\quad=\quad
\begin{mytikz3}

\draw [line width=0.5]    (430,140) -- (470,140) ;
\draw [line width=0.5]    (470,140) -- (520,140) ;
\draw [line width=0.5]    (560,140) -- (520,140) ;
\draw  [fill={rgb, 255:red, 245; green, 240; blue, 235 }  ,fill opacity=1 ][line width=0.5]  (520,155) -- (520,125) -- (535,140) -- cycle ;
\draw  [fill={rgb, 255:red, 245; green, 240; blue, 235 }  ,fill opacity=1 ][line width=0.5]  (520,125) -- (520,155) -- (505,140) -- cycle ;
\draw  [fill={rgb, 255:red, 245; green, 240; blue, 235 }  ,fill opacity=1 ][line width=0.5]  (455,125) -- (485,125) -- (485,155) -- (455,155) -- cycle ;

\draw (470,110) node  [font=\footnotesize] [align=left] {$W(Dd)$};
\draw (520,110) node  [font=\footnotesize] [align=left] {$G(D)$};
\draw (470,170) node  [font=\footnotesize] [align=left] {$ $};

\end{mytikz3}\quad.
\end{equation}
Thus, the integration result of the MPS replicas in Eq.\,\eqref{eq:integrated_mps_replica_simplified} is equal to
\begin{equation}\label{eq:integrated_mps_replica_exchanged}
\begin{mytikz3}
\draw [line width=0.5]    (430,140) -- (470,140) ;
\draw [line width=0.5]    (390,140) -- (430,140) ;
\draw [line width=0.5]    (470,140) -- (520,140) ;
\draw [line width=0.5]    (560,140) -- (520,140) ;
\draw [line width=0.5]    (560,140) -- (600,140) ;
\draw [line width=0.5]    (560,140) -- (560,80) ;
\draw  [fill={rgb, 255:red, 142; green, 142; blue, 142 }  ,fill opacity=1 ][line width=0.5]  (556.5,140) .. controls (556.5,138.07) and (558.07,136.5) .. (560,136.5) .. controls (561.93,136.5) and (563.5,138.07) .. (563.5,140) .. controls (563.5,141.93) and (561.93,143.5) .. (560,143.5) .. controls (558.07,143.5) and (556.5,141.93) .. (556.5,140) -- cycle ;
\draw [line width=0.5]    (340,140) -- (390,140) ;
\draw [line width=0.5]    (340,140) -- (300,140) ;
\draw [line width=0.5]    (260,140) -- (300,140) ;
\draw [line width=0.5]    (300,140) -- (300,80) ;
\draw  [fill={rgb, 255:red, 142; green, 142; blue, 142 }  ,fill opacity=1 ][line width=0.5]  (296.51,140.02) .. controls (296.51,138.08) and (298.07,136.52) .. (300.01,136.52) .. controls (301.94,136.52) and (303.51,138.08) .. (303.51,140.02) .. controls (303.51,141.95) and (301.94,143.52) .. (300.01,143.52) .. controls (298.07,143.52) and (296.51,141.95) .. (296.51,140.02) -- cycle ;
\draw [line width=0.5]    (430,139.79) -- (430,80) ;
\draw  [fill={rgb, 255:red, 142; green, 142; blue, 142 }  ,fill opacity=1 ][line width=0.5]  (426.5,140) .. controls (426.5,138.07) and (428.07,136.5) .. (430,136.5) .. controls (431.93,136.5) and (433.5,138.07) .. (433.5,140) .. controls (433.5,141.93) and (431.93,143.5) .. (430,143.5) .. controls (428.07,143.5) and (426.5,141.93) .. (426.5,140) -- cycle ;
\draw  [fill={rgb, 255:red, 245; green, 240; blue, 235 }  ,fill opacity=1 ][line width=0.5]  (260,155.02) -- (260,125.02) -- (275,140.02) -- cycle ;
\draw  [fill={rgb, 255:red, 245; green, 240; blue, 235 }  ,fill opacity=1 ][line width=0.5]  (390,155) -- (390,125) -- (405,140) -- cycle ;
\draw  [fill={rgb, 255:red, 245; green, 240; blue, 235 }  ,fill opacity=1 ][line width=0.5]  (390,125) -- (390,155) -- (375,140) -- cycle ;
\draw  [fill={rgb, 255:red, 245; green, 240; blue, 235 }  ,fill opacity=1 ][line width=0.5]  (520,155) -- (520,125) -- (535,140) -- cycle ;
\draw  [fill={rgb, 255:red, 245; green, 240; blue, 235 }  ,fill opacity=1 ][line width=0.5]  (520,125) -- (520,155) -- (505,140) -- cycle ;
\draw  [fill={rgb, 255:red, 245; green, 240; blue, 235 }  ,fill opacity=1 ][line width=0.5]  (600,125) -- (600,155) -- (585,140) -- cycle ;
\draw  [fill={rgb, 255:red, 245; green, 240; blue, 235 }  ,fill opacity=1 ][line width=0.5]  (455,125) -- (485,125) -- (485,155) -- (455,155) -- cycle ;
\draw  [fill={rgb, 255:red, 245; green, 240; blue, 235 }  ,fill opacity=1 ][line width=0.5]  (325,125) -- (355,125) -- (355,155) -- (325,155) -- cycle ;
\draw [line width=0.5]    (600,140) -- (620,140) ;
\draw [line width=0.5]    (240,140) -- (260,140) ;
\draw  [fill={rgb, 255:red, 245; green, 240; blue, 235 }  ,fill opacity=1 ][line width=0.5]  (545,100) -- (574.98,100) -- (559.99,114.99) -- cycle ;
\draw [line width=0.5]    (660,140) -- (710,140) ;
\draw [line width=0.5]    (750,140) -- (710,140) ;
\draw [line width=0.5]    (750,140) -- (790,140) ;
\draw [line width=0.5]    (750,140) -- (750,80) ;
\draw  [fill={rgb, 255:red, 142; green, 142; blue, 142 }  ,fill opacity=1 ][line width=0.5]  (746.5,140) .. controls (746.5,138.07) and (748.07,136.5) .. (750,136.5) .. controls (751.93,136.5) and (753.5,138.07) .. (753.5,140) .. controls (753.5,141.93) and (751.93,143.5) .. (750,143.5) .. controls (748.07,143.5) and (746.5,141.93) .. (746.5,140) -- cycle ;
\draw  [fill={rgb, 255:red, 245; green, 240; blue, 235 }  ,fill opacity=1 ][line width=0.5]  (660,155) -- (660,125) -- (675,140) -- cycle ;
\draw  [fill={rgb, 255:red, 245; green, 240; blue, 235 }  ,fill opacity=1 ][line width=0.5]  (790,125) -- (790,155) -- (775,140) -- cycle ;
\draw  [fill={rgb, 255:red, 245; green, 240; blue, 235 }  ,fill opacity=1 ][line width=0.5]  (695,125) -- (725,125) -- (725,155) -- (695,155) -- cycle ;
\draw [line width=0.5]    (790,140) -- (820,140) ;
\draw [line width=0.5]    (640,140) -- (660,140) ;
\draw [line width=0.5]    (150,140) -- (200,140) ;
\draw [line width=0.5]    (150,140) -- (110,140) ;
\draw [line width=0.5]    (70,140) -- (110,140) ;
\draw [line width=0.5]    (110,140) -- (110,80) ;
\draw  [fill={rgb, 255:red, 142; green, 142; blue, 142 }  ,fill opacity=1 ][line width=0.5]  (106.51,140.02) .. controls (106.51,138.08) and (108.07,136.52) .. (110.01,136.52) .. controls (111.94,136.52) and (113.51,138.08) .. (113.51,140.02) .. controls (113.51,141.95) and (111.94,143.52) .. (110.01,143.52) .. controls (108.07,143.52) and (106.51,141.95) .. (106.51,140.02) -- cycle ;
\draw  [fill={rgb, 255:red, 245; green, 240; blue, 235 }  ,fill opacity=1 ][line width=0.5]  (70,155.02) -- (70,125.02) -- (85,140.02) -- cycle ;
\draw  [fill={rgb, 255:red, 245; green, 240; blue, 235 }  ,fill opacity=1 ][line width=0.5]  (200,125) -- (200,155) -- (185,140) -- cycle ;
\draw  [fill={rgb, 255:red, 245; green, 240; blue, 235 }  ,fill opacity=1 ][line width=0.5]  (135,125) -- (165,125) -- (165,155) -- (135,155) -- cycle ;
\draw [line width=0.5]    (40,140) -- (70,140) ;
\draw [line width=0.5]    (200,140) -- (220,140) ;
\draw [line width=0.5]    (40,140) -- (40,80) ;
\draw [line width=0.5]    (820,140) -- (820,80) ;
\draw  [fill={rgb, 255:red, 245; green, 240; blue, 235 }  ,fill opacity=1 ][line width=0.5]  (735,100) -- (764.98,100) -- (749.99,114.99) -- cycle ;
\draw  [fill={rgb, 255:red, 245; green, 240; blue, 235 }  ,fill opacity=1 ][line width=0.5]  (415,100) -- (444.98,100) -- (429.99,114.99) -- cycle ;
\draw  [fill={rgb, 255:red, 245; green, 240; blue, 235 }  ,fill opacity=1 ][line width=0.5]  (285.01,100) -- (314.99,100) -- (300,114.99) -- cycle ;
\draw  [fill={rgb, 255:red, 245; green, 240; blue, 235 }  ,fill opacity=1 ][line width=0.5]  (95,100) -- (124.98,100) -- (109.99,114.99) -- cycle ;

\draw (630,140) node  [font=\normalsize] [align=left] {$...$};
\draw (230,140) node  [font=\normalsize] [align=left] {$...$};
\draw (430,165) node  [font=\footnotesize] [align=left] {$i$};
\draw (560,165) node  [font=\footnotesize] [align=left] {$i+1$};
\draw (300,165) node  [font=\footnotesize] [align=left] {$i-1$};
\draw (110,165) node  [font=\footnotesize] [align=left] {$1$};
\draw (750,165) node  [font=\footnotesize] [align=left] {$N_g$};

\end{mytikz3}\quad,
\end{equation}
where $G(D)$ and $W(Dd)$ on the bond between site $i$ and $i+1$ are exchanged. On the other hand, Eq.\,\eqref{eq:integrated_mps_replica_exchanged}, times the factor in Eq.\,\eqref{eq:factor_D2d}, is exactly the integration result of the MPS replicas with the orthogonality center at site $i+1$. Since this holds for any moment order $t$, we have
\begin{equation}
    \mathbb{V}_\mathrm{MPS}^{[i]} = \mathbb{V}_\mathrm{MPS}^{[i+1]}.
\end{equation}
Similarly, be successively exchanging the two matrices $G(D)$ and $W(Dd)$ on the bonds between adjacent sites, we have
\begin{equation}
    \mathbb{V}_\mathrm{MPS}^{[i]} = \mathbb{V}_\mathrm{MPS}^{[j]},
\end{equation}
for two arbitrary sites $i$ and $j$. The proof is completed.
\end{proof}

We remark that this theorem holds due to the sequential structure of MPS, i.e., MPS has no loop in the temporal-spatial plane, like in the brickwork circuit. Thus, this theorem also holds for other states with movable orthogonality centers, such as tree tensor network states.






\subsection{Local minimum distribution}
In this section, we provide a detailed proof of Theorem~\textcolor{darkblue1}{2} in the main text. We first clarify the definition of local minimum: for a real-valued function $f:\mathcal{X}\rightarrow\mathbb{R}$ on a topological space $\mathcal{X}$, a point $x\in \mathcal{X}$ is a local minimum of the function $f$ if and only if there is an open neighborhood $\mathcal{N}(x)\subseteq \mathcal{X}$ such that
\begin{equation}
    f(x) \leq f(x'),
\end{equation}
holds for any $x'\in\mathcal{N}(x)$. If $f$ is differentiable up to the second order, this is equivalent to the common local minimum condition: zero gradient and positive semi-definite Hessian matrix
\begin{equation}
    \nabla f(x) = 0,\quad \bm{\mathrm{H}}_f(x)\geq 0.
\end{equation}
Here, the strict positive definiteness is not required in order to include situations where the function $f$ contains redundant parameters implicitly, so that there are directions along which the function is constant and hence the Hessian matrix always has zero eigenvalues. We denote the set of local minima of a function $f$ as $\mathcal{LM}_f\subseteq \mathcal{X}$. The shape of the local minimum set $\mathcal{LM}_f$ is discussed as follows.
\begin{itemize}
    \item In the simplest case, if $f$ has no degenerate local minima (e.g., the Morse functions) so that the Hessian matrix at each local minimum has no zero eigenvalue, then $\mathcal{LM}_f$ is just a discrete set of isolated points.
    \item More generally, if $f$ has some ``symmetries'' so that there are always certain directions along which $y=f(x)$ is constant (e.g., the Morse–Bott functions) and accordingly the Hessian matrix at each local minimum has a fixed number of zero eigenvalues, then $\mathcal{LM}_f$ is a continuum of local minima (like a ``valley'' in the function landscape), which is a sub-manifold of $\mathcal{X}$ with a well-defined dimension $\mathrm{dim}\mathcal{LM}_f$, or, equivalently, a union of several connected sub-manifolds of $\mathcal{X}$ of the same dimension. The non-degenerate case discussed above can be regarded as a special case where $\mathcal{LM}_f$ is zero-dimensional.
    \item In the most general case, $\mathcal{LM}_f$ may be mixed-dimensional, i.e., $\mathcal{LM}_f$ is a union of several sub-manifolds of $\mathcal{X}$ with different dimensions, e.g., some of the local minima are isolated points while the others form continuous lines or areas.
\end{itemize}

Then, we clarify the definition of local minimum distribution. By saying ``local minimum distribution'', we mean the distribution of the function values of local minima. That is to say, if we have a certain measure $\mu_{\mathbb{LM}_f}(\cdot)$ over the local minimum set $\mathcal{LM}_f$ and hence obtain an ensemble $\mathbb{LM}_f$, the distribution of the local minima over the function value $y$ is defined by the following cumulative distribution function
\begin{equation}
    \mathcal{P}_f(y) = \mu_{\mathbb{LM}_f}\left(\{x\in \mathcal{LM}_f \mid f(x)\leq y\}\right).
\end{equation}
Here, we use cumulative distributions rather than density distributions since the function values of local minima form a discrete set. The corresponding density distribution can be obtained formally in Dirac delta functions by taking the derivative of $\mathcal{P}_f(y)$. In other words, $\mathcal{P}_f(y)$ can be seen as the push-forward measure of the measure $\mathrm{vol}(\cdot)$ over $\mathcal{LM}_f$ under the map $f$. Then, the question is how to define a proper measure over the local minimum set $\mathcal{LM}_f$.

If $\mathcal{LM}_f$ is a discrete set of isolated points, a simple choice is to take the counting measure, i.e., counting the number of isolated local minimum points. If $\mathcal{LM}_f$ forms a continuum, a natural generalization is to take certain geometric measures, such as the Hausdorff measure. However, in order to reflect the probability of converging to each local minimum in practice, we will use a more relevant measure by incorporating the optimization dynamics.

Given a deterministic optimization algorithm (optimizer), such as the standard gradient descent algorithm, the convergence from initial parameter points to local minima naturally defines an ``optimization map''
\begin{equation}
    \mathcal{O}_f:\mathcal{X}\rightarrow\mathcal{LM}_f.
\end{equation}
Here, we take the limit of infinitesimally small step size (learning rate) to ensure that the parameter point will converge to a local minimum stably rather than oscillate endlessly without convergence. Taking the standard gradient descent algorithm as an example, the optimization map can be defined by the correspondence between the initial and end point of the negative gradient flow, i.e., the following differential equation
\begin{equation}
    \frac{\mathrm{d}}{\mathrm{d}t} \chi(t;x_0) = - \nabla f(\chi(t;x_0)),\quad \chi(0;x_0)=x_0,
\end{equation}
where $x_0$ is the initial point and 
\begin{equation}\label{eq:x_infty}
    x_\infty=\lim_{t\rightarrow\infty}\chi(t;x_0),
\end{equation}
is the end point. Namely, the long-time limit of the solution $\chi(t=\infty;\cdot)$ exactly defines the optimization map $\mathcal{O}_f$. Note that $\mathcal{X}$ should be closed to ensure that the limit points like $x_\infty$ are still in $\mathcal{X}$. In the language of dynamical systems, each local minimum can be understood as an ``attractor'' and its preimage under the map $\mathcal{O}_f$ is the corresponding basin of attraction. Here, we assume that the end points of the optimization algorithm are local minima for almost every initialization. Saddle points and local maxima are also stationary points, but they are unstable under perturbations and are expected to have zero basin-based measure. Alternatively, one can include all stationary points into a set, replacing the local minimum set defined above, which would not affect the validity of our subsequent conclusion.

Thus, given an ensemble $\mathbb{X}$ associated with the parameter space $\mathcal{X}$ induced from random initialization, the push-forward measure of the random initialization measure $\mu_\mathbb{X}$ under the map $\mathcal{O}_f$
\begin{equation}
    \mu_{\mathbb{LM}_f}(\mathcal{B}) = \mu_\mathbb{X}\left( \mathcal{O}_f^{-1}(\mathcal{B}) \right),
    \label{eq:vol_def}
\end{equation}
naturally gives a probability distribution over the local minimum set $\mathcal{LM}_f$, where $\mathcal{B}\subseteq \mathcal{LM}_f$ is a Borel measurable subset. This optimizer-dependent measure directly reflects the probability of converging to each minimum when the optimization is conducted in practice, and hence, the corresponding local minimum distribution over the function value
\begin{equation}
    \mathcal{P}_f(y) = \mu_{\mathbb{LM}_f}\left(\{x\in \mathcal{LM}_f \mid f(x)\leq y\}\right).
    \label{eq:Pf_def}
\end{equation}
faithfully predicts the performance of optimization. That is to say, $\mathcal{P}_f(y)$ directly describes the distribution of the objective function values in the final converged results of the optimization algorithm. If $\mathcal{P}_f(y)$ concentrates near the global minimum value, there is a high probability that the optimization algorithm can find the optimal solution successfully. Conversely, if $\mathcal{P}_f(y)$ concentrates away from the global minimum value, it is almost impossible for the optimization algorithm to obtain the optimal solution. A schematic of this basin-based measure is depicted in Fig.\,\ref{fig:localmin_distribution_def}.

\begin{figure}
    \centering
    \includegraphics[width=0.45\linewidth]{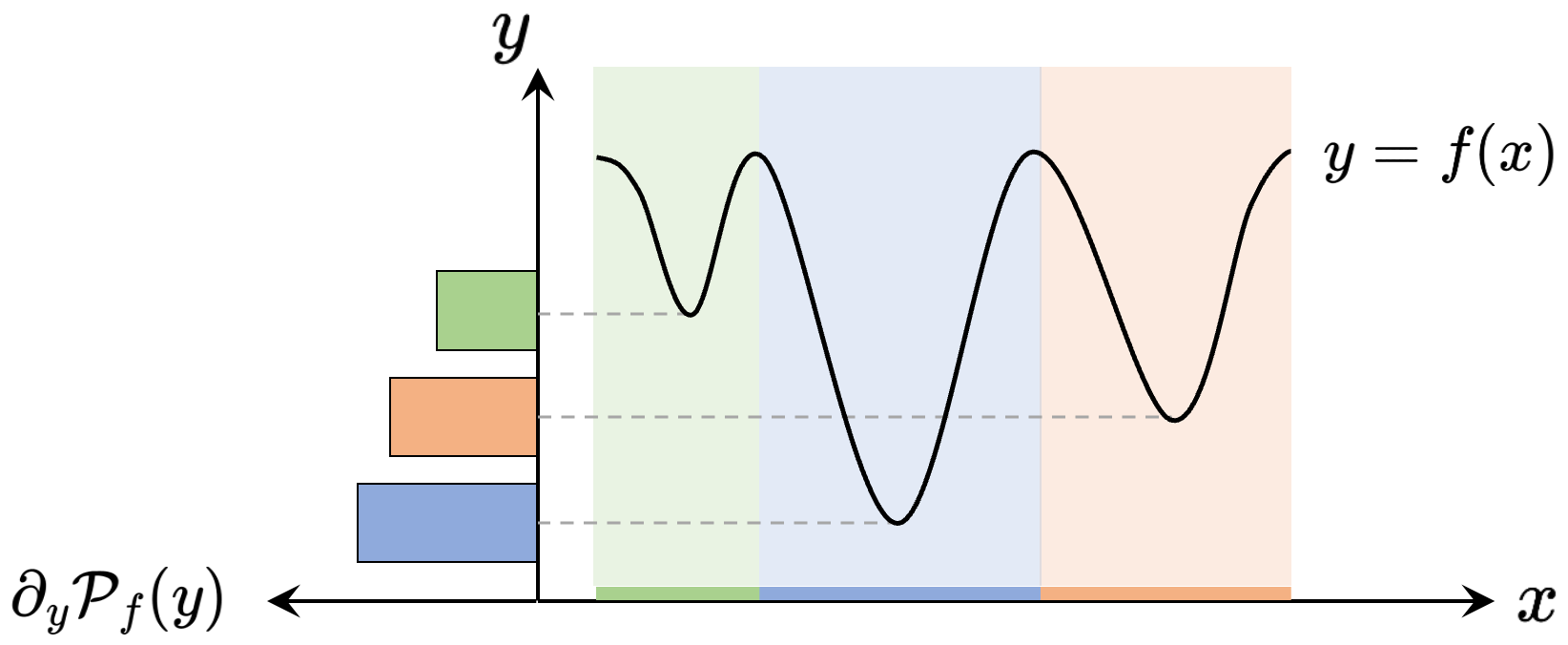}
    \caption{Schematic of the definition of the local minimum distribution. The colored areas represent the respective basins of attraction corresponding to each local minimum. The histogram on the left hand side represent the local minimum distribution, which we formally denote as $\partial_y\mathcal{P}_f(y)$ since $\mathcal{P}_f(y)$ is the cumulative distribution function.}
    \label{fig:localmin_distribution_def}
\end{figure}

\subsection{Invariance of local minimum distribution: Proof of Theorem 2}
With the general definition of local minimum distribution in place, we now focus on its specific application to MPS. The objective function is the energy function
\begin{equation}
    \mathcal{E}:\mathcal{A}_\mathrm{MPS}\rightarrow \mathbb{R},\quad \mathcal{E}(A)=\braoprket{\Psi_\mathrm{MPS}(A)}{H}{\Psi_\mathrm{MPS}(A)},
\end{equation}
where $H$ is the Hamiltonian of the target system. The optimization is conducted within the MPS variety $\mathcal{V}_\mathrm{MPS}(N,\mathrm{D})$. For MPS, there is a natural choice of optimization algorithm---the (imaginary-time) time-dependent variational principle (TDVP) algorithm~\cite{Haegeman2011}, which projects the exact imaginary-time evolution onto the tangent space of MPS. The famous density-matrix renormalization group (DMRG) algorithm can be regarded as a particular limit of the TDVP algorithm~\cite{Haegeman2016}, by taking the time step size to be infinite after performing the Trotter decomposition of the MPS tangent space projector.

Specifically, the evolution flow of the TDVP algorithm in the MPS manifold $\mathcal{V}_\mathrm{MPS}^{\mathrm{f}}$ is defined by the projected imaginary-time Schr\"{o}dinger equation (TDVP equation)
\begin{equation}\label{eq:tdvp_equation}
    \frac{\mathrm{d}}{\mathrm{d}t} \ket{\psi(t)} = -P_\mathcal{T} H \ket{\psi(t)},
\end{equation}
where $P_\mathcal{T}$ is the projector to the tangent space $\mathcal{T}_{\ket{\psi}}\mathcal{V}_\mathrm{MPS}^{\mathrm{f}}$ of the MPS manifold $\mathcal{V}_\mathrm{MPS}^{\mathrm{f}}$ at the base point $\ket{\psi(t)}$. After fixing the gauge properly, the tangent space projector can be represented as
\begin{equation}
    P_\mathcal{T} = \sum_{i=1}^N P_{1:(i-1)} \otimes I_i \otimes P_{(i+1):N} - \sum_{i=1}^{N-1} P_{1:i} \otimes P_{(i+1):N},
\end{equation}
where $I_i$ represents the identity operator at site $i$. $P_{1:i}$ ($P_{(i+1):N}$) is the projector to the subspace spanned by the eigenstates with non-zero eigenvalues of the reduced density matrix of the subsystem on the left (right) of bond $i$ from $\ket{\psi(t)}$. By construction, the resulting evolution will never leave the
MPS variety.

For rank-deficient points in the MPS variety $\mathcal{V}_\mathrm{MPS}$, there is, in general, no well-defined tangent space but rather a non-trivial tangent cone $\mathcal{C}_{\ket{\psi}} \mathcal{V}_\mathrm{MPS}$, since the combination of the allowed rank-increasing directions can lead to Schmidt ranks larger than the bond dimension, and hence leave the tangent cone. Namely, the tangent cone at a rank-deficient point may not be closed with respect to the linear combinations of its elements and hence not a linear space. However, for each trajectory starting from the full-rank stratum and approaching a rank-deficient point at $t_0$, i.e., $\ket{\psi(t)}$ with $t\in [0,t_0)$, the tangent spaces along the trajectory have a well-defined limit, given by the convergence of the tangent space projector
\begin{equation}
    \lim_{t\rightarrow t_0^{-}} P_\mathcal{T}[\psi(t)] = P^*.
\end{equation}
Using this limiting tangent space (the part within the tangent cone), the evolution flow can be extended naturally to rank-deficient points in a trajectory-dependent manner. Mathematically, this corresponds to a differential inclusion
\begin{equation}\label{eq:tdvp_inclusion}
    \frac{\mathrm{d}}{\mathrm{d}t} \ket{\psi(t)} \in -\partial \mathcal{F}(\ket{\psi(t)}),\quad 
\end{equation}
as a generalization of the differential equation defined in smooth manifolds, where $\partial \mathcal{F}$ represents the sub-differential of the energy function 
\begin{equation}
    \mathcal{F}(\ket{\psi})=\braoprket{\psi}{H}{\psi},
\end{equation}
restricted to the MPS variety $\mathcal{V}_\mathrm{MPS}$. For the initial conditions within the full-rank stratum $\mathcal{V}_\mathrm{MPS}^{\mathrm{f}}$, this differential inclusion admits a unique globally defined solution, which selects a consistent limit direction when it passes through (or terminates at) a rank-deficient point. Thus, the TDVP evolution flow exists as a well-defined dynamical system in the whole MPS variety $\mathcal{V}_\mathrm{MPS}$, provided we consider its solutions that start in the dense smooth regions and take limits appropriately. For the initial conditions exactly located at rank-deficient strata, we can define the evolution flow by the TDVP equation defined in the corresponding lower-dimensional manifolds, or simply disregard such cases since they form a set of measure zero. Therefore, the long-time limit of the TDVP evolution flow provides a well-defined optimization map
\begin{equation}
    \mathcal{O}_\mathcal{F}: \mathcal{V}_\mathrm{MPS}^\mathrm{f} \rightarrow \mathcal{LM}_\mathcal{F}^*,
\end{equation}
as described in Eq.\,\eqref{eq:x_infty}. The ``limit local minimum'' set $\mathcal{LM}_\mathcal{F}^*$ is a superset of $\mathcal{LM}_\mathcal{F}$, which additionally contains the ``limit local minima'', i.e., the possible rank-deficient points that are local minima within a certain limit tangent space in the tangent cone. Note that $\mathcal{LM}_\mathcal{F}^*$ can indeed contain rank-deficient points even though the initial points are within the full-rank stratum $\mathcal{V}_\mathrm{MPS}^{\mathrm{f}}$. Importantly, since the TDVP evolution flow is defined intrinsically in $\mathcal{V}_\mathrm{MPS}$, the resulting optimization map $\mathcal{O}_\mathcal{F}$ is independent of the choice of parametrization.

In a more general context, the TDVP algorithm is also known as the quantum natural gradient descent (QNGD) algorithm in variational quantum circuits~\cite{Stokes2020}, Riemannian gradient descent (RGD) in manifold optimization, and the stochastic reconfiguration (SR) in variational Monte Carlo (VMC) algorithms. For a generic parametrized pure quantum state $\ket{\Psi(\bm{\theta})}$ with tunable parameters $\bm{\theta}=\{\theta_\mu\}_{\mu=1}^M$, the evolution flow in the parameter space is given by
\begin{equation}
    \frac{\mathrm{d}}{\mathrm{d}t}\bm{\theta}(t) = - \mathrm{grad}\, \mathcal{E} (\bm{\theta}(t)),
\end{equation}
where $\mathrm{grad}\, \mathcal{E} $ is the representation of the Riemannian gradient of the objective function $\mathcal{F}$ in the parameter space, also known as the quantum natural gradient. It is related to the plain Euclidean gradient $\nabla \mathcal{E}$ by
\begin{equation}\label{eq:Igrad=nabla}
    \mathcal{I}(\bm{\theta}) \mathrm{grad}\, \mathcal{E} (\bm{\theta}) = \nabla \mathcal{E}(\bm{\theta}),
\end{equation}
where $\mathcal{I}(\bm{\theta})$ is the induced metric in the parameter space from the Fubini-Study metric in the Hilbert space by the parametrized quantum state $\ket{\Psi(\bm{\theta})}$, also known as the quantum Fisher information (QFI) matrix. The elements of $\mathcal{I}(\bm{\theta})$ are given by
\begin{equation}
    \mathcal{I}_{\mu\nu}(\bm{\theta}) = \mathrm{Re}\left[ \braket{\partial_\mu\Psi(\bm{\theta})}{\partial_\nu\Psi(\bm{\theta})} - \braket{\partial_\mu\Psi(\bm{\theta})}{\Psi(\bm{\theta})} \braket{\Psi(\bm{\theta})}{\partial_\nu\Psi(\bm{\theta})} \right],
\end{equation}
where $\partial_\mu=\frac{\partial}{\partial \theta_\mu}$ represents taking the partial derivative with respect to $\theta_\mu$. If $\mathcal{I}(\bm{\theta})$ is invertible, given the Euclidean gradient $\nabla\mathcal{E}(\bm{\theta})$, the linear equation in Eq.\,\eqref{eq:Igrad=nabla} has a unique solution
\begin{equation}
    \mathrm{grad}\, \mathcal{E} (\bm{\theta}) = \mathcal{I}^{-1}(\bm{\theta}) \nabla \mathcal{E}(\bm{\theta}).
\end{equation}
If $\mathcal{I}(\bm{\theta})$ is not invertible, such as when there are redundant parameters in $\ket{\Psi(\bm{\theta})}$, it is common to choose a minimum norm solution by taking the Moore–Penrose pseudo-inverse $\mathcal{I}^{+}(\bm{\theta})$ of the metric matrix
\begin{equation}
    \mathrm{grad}\, \mathcal{E} (\bm{\theta}) = \mathcal{I}^{+}(\bm{\theta}) \nabla \mathcal{E}(\bm{\theta}).
\end{equation}
In practice, considering the numerical error, the pseudo-inverse is usually realized approximately by the Tikhonov regularization
\begin{equation}
    \mathcal{I}^{+}(\bm{\theta}) \longrightarrow \mathcal{I}^{\epsilon+}(\bm{\theta}) = \left( \mathcal{I}(\bm{\theta}) + \epsilon I \right)^{-1},
\end{equation}
where $I$ is the identity matrix of the same shape as the metric and $\epsilon$ is a small real number called the Tikhonov regularization parameter. In the limit of infinitesimal regularization parameter, the regularized inverse $\mathcal{I}^{\epsilon+}(\bm{\theta})$ converges to the exact pseudo-inverse $\mathcal{I}^{+}(\bm{\theta})$. Thus, this natural gradient flow exists as a well-defined dynamical system in the whole parameter space, which is mapped to the TDVP evolution flow defined in Eqs.\,\eqref{eq:tdvp_equation} and \eqref{eq:tdvp_inclusion} by the parametrization map $\Psi$ when the variational ansatz is chosen as an MPS. The long-time limit of the natural gradient flow provides a well-defined optimization map
\begin{equation}
    \mathcal{O}_\mathcal{E}: \mathcal{A}_\mathrm{MPS} \rightarrow \mathcal{LM}_\mathcal{E}.
\end{equation}
In the limit of infinitesimal regularization parameter, the natural gradient optimization map $\mathcal{O}_\mathcal{E}$ is equivalent to the TDVP optimization map in the sense of
\begin{equation}\label{eq:PsiOE=OFPsi}
    \Psi \circ \mathcal{O}_\mathcal{E} = \mathcal{O}_\mathcal{F} \circ \Psi,
\end{equation}
for non-singular initial points, because the natural gradient flow is just the representation of the imaginary-time TDVP evolution flow in the parameter space.

Before proving Theorem~\textcolor{darkblue1}{2}, we first prove the following Lemma~\ref{lemma:map_local_min} for convenience.
\begin{lemma}\label{lemma:map_local_min}
Suppose $\mathcal{X}$ and $\mathcal{Y}$ are two topological spaces. $f:\mathcal{Y}\rightarrow\mathbb{R}$ is a real-valued function and $g:\mathcal{X}\rightarrow\mathcal{Y}$ is an open map. If $x\in\mathcal{X}$ is a local minimum of the function $f\circ g$, then $y=g(x)$ is a local minimum of the function $f$.
\end{lemma}
\begin{proof}
As $x\in\mathcal{X}$ is a local minimum of $f\circ g$, there exists an open neighborhood $\mathcal{N}(x)$ such that
\begin{equation}
    f(g(x)) \leq f(g(x')),
\end{equation}
holds for any $x'\in\mathcal{N}(x)$. Since $g$ is an open map, $g(\mathcal{N}(x))$ is an open neighborhood of $y=g(x)$ within $\mathcal{Y}$. Thus, for any $y'\in g(\mathcal{N}(x))$, there exists $x'\in\mathcal{N}(x)$ such that $y'=g(x')$, and hence
\begin{equation}
    f(y')=f(g(x')) \geq f(g(x)) = f(y),
\end{equation}
which means that $y=g(x)$ is indeed a local minimum of the function $f$. The proof is completed.
\end{proof}

According to Lemma~\ref{lemma:map_local_min}, a biholomorphism (which is open in both directions) preserves the local minima of a given real-valued function. More generally, a biholomorphism actually preserves all the critical points: $\nabla f(g(x))=0$ if and only if 
\begin{equation}
    \nabla(f\circ g)(x)= \nabla f(g(x)) \nabla g(x)=0,
\end{equation}
since the Jacobian matrix $\nabla g(x)$ is invertible. The index of a critical point $y=g(x)$ of $f$ (the number of negative eigenvalues of the Hessian matrix) is the same as that of $x$ of $f\circ g$ since
\begin{equation}
   \bm{\mathrm{H}}_{f\circ g}(x) = \left(\nabla g(x)\right)^\dagger \bm{\mathrm{H}}_f(g(x)) \nabla g(x),
\end{equation}
and the congruent transformation preserves the index by Sylvester's law of inertia.

With these clarifications, we are now prepared to prove Theorem~\textcolor{darkblue1}{2} in the main text, which is rephrased below for convenience. We use $\mathcal{P}_\mathcal{E}(E;\mathbb{A}_\mathrm{MPS}^{[i]})$ to denote the cumulative distribution of local minima of the energy function $\mathcal{E}$ at energy level $E$ for the MPS representation ensemble $\mathbb{A}_\mathrm{MPS}^{[i]}$ with orthogonality center at site $i$.

\renewcommand{\theproposition}{2}
\begin{theorem}
The local minimum distributions of the energy landscapes with respect to the MPS representation ensembles with different orthogonality centers are identical, i.e.,
\begin{equation}
    \mathcal{P}_\mathcal{E}(E;\mathbb{A}_\mathrm{MPS}^{[i]}) = \mathcal{P}_\mathcal{E}(E;\mathbb{A}_\mathrm{MPS}^{[j]}),
\end{equation}
for any two sites $i,j$ and any given Hamiltonian.
\end{theorem}
\addtocounter{proposition}{-1}
\renewcommand{\theproposition}{S\arabic{proposition}}

\begin{proof}
We denote the MPS with orthogonality center at site $i$ as $\ket{\Psi_\mathrm{MPS}(A^{[i]})}$ and its parameter space as $\mathcal{A}_\mathrm{MPS}^{[i]}$. $D$ is the bond dimension and $d$ is the local Hilbert space dimension. $N_g$ is the number of local tensors in the MPS. Using the tensor network diagram, $\ket{\Psi_\mathrm{MPS}(A^{[i]})}$ is represented by
\begin{equation}
\begin{mytikz3}
\draw [line width=0.5]    (235,100) -- (210,100) ;
\draw [line width=0.5]    (235,60) -- (235,100) ;
\draw [line width=0.5]    (260,100) -- (235,100) ;
\draw  [fill={rgb, 255:red, 245; green, 240; blue, 235 }  ,fill opacity=1 ][line width=0.5]  (217.32,103.54) -- (238.54,82.32) -- (252.68,96.46) -- (231.46,117.68) -- cycle ;
\draw [line width=0.5]    (285,100) -- (260,100) ;
\draw [line width=0.5]    (285,60) -- (285,100) ;
\draw [line width=0.5]    (310,100) -- (285,100) ;
\draw  [fill={rgb, 255:red, 245; green, 240; blue, 235 }  ,fill opacity=1 ][line width=0.5]  (270,85) -- (300,85) -- (299.31,115) -- (269.31,115) -- cycle ;
\draw [line width=0.5]    (335,100) -- (310,100) ;
\draw [line width=0.5]    (335,60) -- (335,100) ;
\draw [line width=0.5]    (360,100) -- (335,100) ;
\draw  [fill={rgb, 255:red, 245; green, 240; blue, 235 }  ,fill opacity=1 ][line width=0.5]  (331.46,82.32) -- (352.68,103.54) -- (338.54,117.68) -- (317.32,96.46) -- cycle ;
\draw [line width=0.5]    (425,100) -- (400,100) ;
\draw [line width=0.5]    (425,60) -- (425,100) ;
\draw [line width=0.5]    (450,100) -- (425,100) ;
\draw  [fill={rgb, 255:red, 245; green, 240; blue, 235 }  ,fill opacity=1 ][line width=0.5]  (421.46,82.32) -- (442.68,103.54) -- (428.54,117.68) -- (407.32,96.46) -- cycle ;
\draw [line width=0.5]    (475,100) -- (450,100) ;
\draw [line width=0.5]    (475,60) -- (475,100) ;
\draw [line width=0.5]    (145,100) -- (120,100) ;
\draw [line width=0.5]    (145,60) -- (145,100) ;
\draw [line width=0.5]    (170,100) -- (145,100) ;
\draw  [fill={rgb, 255:red, 245; green, 240; blue, 235 }  ,fill opacity=1 ][line width=0.5]  (127.32,103.54) -- (148.54,82.32) -- (162.68,96.46) -- (141.46,117.68) -- cycle ;
\draw [line width=0.5]    (120,100) -- (95,100) ;
\draw [line width=0.5]    (95,60) -- (95,100) ;

\draw (190,99) node   [align=left] {$ ...$};
\draw (380,99) node   [align=left] {$ ...$};
\draw (286,132) node  [font=\footnotesize] [align=left] {$i$};
\draw (336,132) node  [font=\footnotesize] [align=left] {$i+1$};
\draw (426,132) node  [font=\footnotesize] [align=left] {$N_g$};
\draw (236,132) node  [font=\footnotesize] [align=left] {$i-1$};
\draw (146,132) node  [font=\footnotesize] [align=left] {$1$};
\end{mytikz3}\quad.
\end{equation}
By the definition of homogeneous canonical forms discussed above, we have
\begin{equation}
    \mathcal{A}_\mathrm{MPS}^{[i]} = \bigotimes_{k=1}^{i-1}\mathcal{U}_\mathrm{L}(D,d;D) \otimes \mathcal{U}_0(D,d,D) \otimes \bigotimes_{k=i+1}^{N_g} \mathcal{U}_\mathrm{R}(D;d,D),
\end{equation}
where $\mathcal{U}_\mathrm{L}(D,d;D)$ is the space of $3$-order tensors of shape $(D,d,D)$ under the left isometric constraints as in Eq.\,\eqref{eq:tn_isometric}, i.e., for any tensor $A\in \mathcal{U}_\mathrm{L}(D,d;D)$, it satisfies
\begin{equation}
    \sum_{mn} A^*_{mnp} A_{mnq}=\delta_{pq}.
\end{equation}
Similarly, $\mathcal{U}_\mathrm{R}(D,d;D)$ is the space of $3$-order tensors of shape $(D,d,D)$ under the right isometric constraints as in Eq.\,\eqref{eq:tn_isometric}, i.e., for any tensor $A\in \mathcal{U}_\mathrm{R}(D,d;D)$, it satisfies
\begin{equation}
    \sum_{pq} A_{mpq} A^*_{npq} =\delta_{mn}.
\end{equation}
Specifically, $\mathcal{U}_0(D,d,D)$ is the space of $3$-order tensors of shape $(D,d,D)$ under the normalization condition
\begin{equation}
    \sum_{mnp} A^*_{mnp} A_{mnp}=1.
\end{equation}
The dimension of the variational parameter space is
\begin{equation}
    \mathrm{dim}\mathcal{A}_\mathrm{MPS}^{[i]} = (N_g-1)(2dD^2-D^2) + (2dD^2-1) = \left[N_g(2d-1)+1\right]D^2-1.
\end{equation}
As discussed in Sec.\,\ref{sec:mps_manifold}, the codomain of the MPS map $\Psi_\mathrm{MPS}$, as a subset of the total Hilbert space $\mathcal{H}=\mathbb{C}^{d^N}$, is an algebraic variety $\mathcal{V}_\mathrm{MPS}$, whose singularities correspond to rank-deficient MPS. By restricting the domain to the full-rank MPS representations, denoted as $\mathcal{A}_\mathrm{MPS}^{\mathrm{f}[i]}$, the range of the MPS map $\Psi_\mathrm{MPS}$ becomes a complex manifold $\mathcal{V}_\mathrm{MPS}^{\mathrm{f}}$.

The MPS map $\Psi_\mathrm{MPS}$ (i.e., the contraction of local tensors) is not injective because inserting an arbitrary unitary matrix with its inverse between adjacent local tensors does not change the MPS. The redundancy in the MPS representation $A$ is described by the gauge group
\begin{equation}
    \mathcal{G}_{\text{MPS}} = \prod_{i=1}^{N_g-1}\mathcal{U}(D),
\end{equation}
where $\mathcal{U}(D)$ is the unitary group of dimension $D$. The action of the group element $\{G_{1},G_{2},\ldots,G_{N_g-1}\}$ on the MPS representation $A^{[i]}=\{A_1,A_2,\ldots,A_{N_g}\}$ is
\begin{equation}
    A_{1}\rightarrow A_{1}G_{1},~\ldots,~~A_k\rightarrow G_{k-1}^{-1} A_{k} G_{k},~\ldots,~~A_{N_g}\rightarrow G_{N_g-1}^{-1}A_{N_g}.
\end{equation}
The physical indices are omitted for simplicity. By taking the quotient space of the variational parameter space $\mathcal{A}_\mathrm{MPS}^{[i]}$ with respect to the gauge group $\mathcal{G}_\mathrm{MPS}$, we obtain a gauge orbit space $\mathcal{A}_\mathrm{MPS}^{[i]}/\mathcal{G}_\mathrm{MPS}$ with the gauge redundancy in the MPS representation removed.

Given that the action of $\mathcal{G}_\mathrm{MPS}$ on $\mathcal{A}_\mathrm{MPS}^{\mathrm{f}[i]}$ is a holomorphic, free and proper~\cite{Haegeman2014}, the orbit space $\mathcal{A}_\mathrm{MPS}^{\mathrm{f}[i]}/\mathcal{G}_\mathrm{MPS}$ is a complex manifold and the quotient map $\pi^{[i]}:\mathcal{A}_\mathrm{MPS}^{\mathrm{f}[i]}\rightarrow \mathcal{A}_\mathrm{MPS}^{\mathrm{f}[i]}/\mathcal{G}_\mathrm{MPS}$ is holomorphic. Then, since $\Psi_\mathrm{MPS}$ is holomorphic, the induced map $\Phi_\mathrm{MPS}^{[i]}:\mathcal{A}_\mathrm{MPS}^{\mathrm{f}[i]}/\mathcal{G}_\mathrm{MPS}\rightarrow \mathcal{V}_\mathrm{MPS}^{\mathrm{f}}$ from the composition $\Psi_\mathrm{MPS}=\Phi_\mathrm{MPS}^{[i]}\circ\pi^{[i]}$ is also holomorphic. Combined with the fact that $\Phi_\mathrm{MPS}^{[i]}$ is injective as $A^{[i]}$ can be uniquely obtained by performing a series of Schmidt decompositions of the coefficient vector of $\ket{\Psi_\mathrm{MPS}(A^{[i]})}$ up to the gauge transformation, the induced map $\Phi_\mathrm{MPS}^{[i]}$ is a biholomorphism, and the full-rank MPS set $\mathcal{V}_\mathrm{MPS}^{\mathrm{f}}$ is a complex manifold that is \textit{biholomorphic} to the orbit space $\mathcal{A}_\mathrm{MPS}^{\mathrm{f}[i]}/\mathcal{G}_\mathrm{MPS}$ with dimension
\begin{equation}
    \mathrm{dim}\mathcal{V}_\mathrm{MPS}^{\mathrm{f}} = \mathrm{dim}\mathcal{A}_\mathrm{MPS}^{\mathrm{f}[i]} - \mathrm{dim}\mathcal{G}_\mathrm{MPS} = 2\left[N_g(d-1)+1\right]D^2-1.
\end{equation}
Thus, the full-rank MPS representations form a principal fiber bundle with total manifold $\mathcal{A}_\mathrm{MPS}^{\mathrm{f}[i]}$, base manifold $\mathcal{V}_\mathrm{MPS}^{\mathrm{f}}$, bundle projection $\Psi_\mathrm{MPS}$, and structure group $\mathcal{G}_\mathrm{MPS}$.

The energy cost function of MPS representations
\begin{equation}
    \mathcal{E}:\mathcal{A}_\mathrm{MPS}^{[i]} \rightarrow \mathbb{R},\quad \mathcal{E}(A^{[i]}) = \braoprket{ \Psi_\mathrm{MPS}(A^{[i]}) }{H}{ \Psi_\mathrm{MPS}(A^{[i]}) },
\end{equation}
is the composition of the MPS map $\Psi_\mathrm{MPS}$ and the energy function of physical states
\begin{equation}
    \mathcal{F}:\mathcal{V}_\mathrm{MPS} \rightarrow \mathbb{R},\quad \mathcal{F}(\ket{\psi})=\braoprket{ \psi }{H}{ \psi },
\end{equation}
where $\psi$ is a normalized state vector restricted in $\mathcal{V}_\mathrm{MPS}$, i.e., 
\begin{equation}
    \mathcal{E}=\mathcal{F}\circ \Psi_\mathrm{MPS}.
\end{equation}
According to Lemma~\ref{lemma:map_local_min}, because $\Psi_\mathrm{MPS}$ is a submersion in $\mathcal{A}_\mathrm{MPS}^{\mathrm{f}[i]}$ and hence an open map in $\mathcal{A}_\mathrm{MPS}^{\mathrm{f}[i]}$, the image of a local minimum of $\mathcal{E}$ in $\mathcal{A}_\mathrm{MPS}^{\mathrm{f}[i]}$ under $\Psi_\mathrm{MPS}$ is also a local minimum of $\mathcal{F}$ in $\mathcal{V}_\mathrm{MPS}^{\mathrm{f}}$, with the same energy. On the other hand, because the induced map $\Phi_\mathrm{MPS}^{[i]}$ is a biholomorphism between the gauge orbit space $\mathcal{A}_\mathrm{MPS}^{\mathrm{f}[i]}/\mathcal{G}_\mathrm{MPS}$ and the MPS manifold $\mathcal{V}_\mathrm{MPS}^{\mathrm{f}}$, the preimage of a local minimum of $\mathcal{F}$, as a gauge orbit in $\mathcal{A}_\mathrm{MPS}^{\mathrm{f}[i]}$, is a continuum of local minima of $\mathcal{E}$, with the same energy. That is to say, there exists a local minimum correspondence between the energy function of MPS representations and the energy function of MPS in the full-rank stratum, which holds regardless of where the orthogonality center is located (cf. Fig.\,\textcolor{darkblue1}{2} in the main text). For rank-deficient points, this correspondence only holds in a single direction, i.e., from $\mathcal{F}$ to $\mathcal{E}$, because a neighborhood of a point in $\mathcal{A}_\mathrm{MPS}^{[i]}-\mathcal{A}_\mathrm{MPS}^{\mathrm{f}[i]}$ is only mapped to a strict subset of any neighborhood of its image in $\mathcal{V}_\mathrm{MPS}-\mathcal{V}_\mathrm{MPS}^{\mathrm{f}}$ by the map $\Psi_\mathrm{MPS}$. However, as stated below, this does not affect the equivalence of the local minimum distributions.



We denote the local minimum set of the energy function $\mathcal{E}$ on the MPS representation manifold $\mathcal{A}_\mathrm{MPS}^{[i]}$ as $\mathcal{LM}_\mathcal{E}(\mathcal{A}_\mathrm{MPS}^{[i]})$, the local minimum set of the energy function $\mathcal{F}$ on the MPS variety $\mathcal{V}_\mathrm{MPS}$ as $\mathcal{LM}_\mathcal{F}$, and the ``limit local minimum'' set of $\mathcal{F}$ on $\mathcal{V}_\mathrm{MPS}$ as $\mathcal{LM}_\mathcal{F}^*$. According to the discussion above, we have
\begin{equation}
    \mathcal{LM}_\mathcal{F} \subseteq \Psi_\mathrm{MPS} \left( \mathcal{LM}_\mathcal{E}(\mathcal{A}_\mathrm{MPS}^{[i]}) \right) = \mathcal{LM}_\mathcal{F}^*.
\end{equation}
The optimization map of the energy function $\mathcal{E}$ on the MPS representation manifold $\mathcal{A}_\mathrm{MPS}^{[i]}$ from the natural gradient flow is denoted as
\begin{equation}
    \mathcal{O}_\mathcal{E}(\,\cdot\,;\mathcal{A}_\mathrm{MPS}^{[i]}): \mathcal{A}_\mathrm{MPS}^{[i]} \rightarrow \mathcal{LM}_\mathcal{E}(\mathcal{A}_\mathrm{MPS}^{[i]}).
\end{equation}
The local minimum subset of $\mathcal{E}$ in $\mathcal{A}_\mathrm{MPS}^{[i]}$ at energy level $E$ is defined by
\begin{equation}
    \mathcal{B}_E = \left\{ A^{[i]} \in \mathcal{LM}_\mathcal{E}(\mathcal{A}_\mathrm{MPS}^{[i]}) \mid \mathcal{E}(A^{[i]})= E \right\}.
\end{equation}
The basin-based measure of the local minimum subset $\mathcal{B}_E$ of $\mathcal{E}$ can be represented as
\begin{equation}\label{eq:mu1}
    \mu_{\mathbb{LM}_\mathcal{E}}(\mathcal{B}_E;\mathbb{A}_\mathrm{MPS}^{[i]}) = \mu_{\mathbb{A}_\mathrm{MPS}^{[i]}} \left( \mathcal{O}_\mathcal{E}^{-1}(\mathcal{B}_E;\mathcal{A}_\mathrm{MPS}^{[i]}) \right),
\end{equation}
where $\mu_{\mathbb{A}_\mathrm{MPS}^{[i]}}$ is the measure of the MPS representation ensemble $\mathbb{A}_\mathrm{MPS}^{[i]}$. Because the MPS ensemble $\mathbb{V}_\mathrm{MPS}^{[i]}$ is equipped with the push-forward measure from the MPS representation ensemble $\mathbb{A}_\mathrm{MPS}^{[i]}$ under the map $\Psi_\mathrm{MPS}$, the volume of a measurable set in $\mathbb{V}_\mathrm{MPS}^{[i]}$ equals that of its preimage in $\mathbb{A}_\mathrm{MPS}^{[i]}$, and hence we have
\begin{equation}\label{eq:mu2}
    \mu_{\mathbb{A}_\mathrm{MPS}^{[i]}} \left( \mathcal{O}_\mathcal{E}^{-1}(\mathcal{B}_E;\mathcal{A}_\mathrm{MPS}^{[i]}) \right) = \mu_{\mathbb{V}_\mathrm{MPS}^{[i]}} \left( \Psi_\mathrm{MPS} \circ \mathcal{O}_\mathcal{E}^{-1}(\mathcal{B}_E;\mathcal{A}_\mathrm{MPS}^{[i]}) \right),
\end{equation}
According to the relation in Eq.\,\eqref{eq:PsiOE=OFPsi}, we have
\begin{equation}\label{eq:mu3}
    \mu_{\mathbb{V}_\mathrm{MPS}^{[i]}} \left( \Psi_\mathrm{MPS} \circ \mathcal{O}_\mathcal{E}^{-1}(\mathcal{B}_E;\mathcal{A}_\mathrm{MPS}^{[i]}) \right) = \mu_{\mathbb{V}_\mathrm{MPS}^{[i]}} \left( \mathcal{O}_\mathcal{F}^{-1} \circ \Psi_\mathrm{MPS} (\mathcal{B}_E) \right),
\end{equation}
where $\mathcal{O}_\mathcal{F}$ is the optimization map of the energy function $\mathcal{F}$ on the MPS variety $\mathcal{V}_\mathrm{MPS}$ from the TDVP evolution flow. Combining Eqs.\,\eqref{eq:mu1}, \eqref{eq:mu2}, and \eqref{eq:mu3}, the basin-based measure of the local minimum subset $\mathcal{B}_E$ equals
\begin{equation}\label{eq:mu4}
    \mu_{\mathbb{LM}_\mathcal{E}}(\mathcal{B}_E;\mathbb{A}_\mathrm{MPS}^{[i]}) = \mu_{\mathbb{V}_\mathrm{MPS}^{[i]}} \left( \mathcal{O}_\mathcal{F}^{-1} \circ \Psi_\mathrm{MPS} (\mathcal{B}_E) \right),
\end{equation}
where the optimization map $\mathcal{O}_\mathcal{F}$ and the state subset
\begin{equation}
     \Psi_\mathrm{MPS} (\mathcal{B}_E) =  \left\{ \psi \in \mathcal{LM}_\mathcal{F}^* \mid \mathcal{F}(\ket{\psi})= E \right\},
\end{equation}
are intrinsically defined on $\mathcal{V}_\mathrm{MPS}$, independent of MPS parametrization. Importantly, according to Theorem~\textcolor{darkblue1}{1}, the MPS ensembles with different orthogonality centers are identical
\begin{equation}\label{eq:PFi=PFj}
    \mathbb{V}_\mathrm{MPS}^{[i]} = \mathbb{V}_\mathrm{MPS}^{[j]}.
\end{equation}
Thus, the right hand side of Eq.\,\eqref{eq:mu4} as a whole is independent of the position of the orthogonality center, and hence so does the left hand side, i.e., the measure of the local minimum subset of $\mathcal{E}$ in $\mathcal{A}_\mathrm{MPS}^{[i]}$ at energy level $E$. Therefore, the local minimum distributions of $\mathcal{E}$ are identical for MPS representation ensembles with different orthogonality centers, i.e.,
\begin{equation}
    \mathcal{P}_\mathcal{E}(E;\mathbb{A}_\mathrm{MPS}^{[i]}) = \mathcal{P}_\mathcal{E}(E;\mathbb{A}_\mathrm{MPS}^{[j]}).
\end{equation}
The proof is completed.
\end{proof}

For MPS in the form of sequential circuits where the parameters are Pauli angles, i.e., the tensor elements are parametrized by Pauli angles, Theorem~\textcolor{darkblue1}{2} also holds true, with the corresponding proof almost unchanged except that the gauge group is enlarged by the redundancy in the multiplication between the gates and initial product states. This is because the Pauli-angle parametrization is a submersion up to a finite number of pole-like coordinate singularities, so that the local minimum correspondence still in general exists between the parameter spaces with different orthogonality centers. These coordinate singularities can be reasonably omitted when considering the local minimum distribution because they are not intrinsic to the MPS variety like those rank-deficient points. Their positions in the MPS variety depend on the specific parametrization. If a coordinate singularity happens to be a local minimum whose image in the MPS variety is not a local minimum, it can be just removed by a slightly different Pauli-angle parametrization of the corresponding unitary, e.g., multiplying a constant unitary close to the identity.

We remark that the theorem does not imply that the energy landscapes of MPS with different orthogonality centers are the same. Instead, it just implies that the statistical properties of the energy levels of the critical points in the landscapes are the same. We also remark that though the measure we used here requires manifold optimization, our numerical experiments show that our overall conclusion is still correct for various other optimization algorithms, such as the standard gradient descent without considering quantum geometry.

\subsection{Good local minima in MPS energy landscapes}
Based on Theorem~\textcolor{darkblue1}{2}, we can derive some useful conclusions about the trainability of MPS on different Hamiltonians. We denote the Hamiltonian as
\begin{equation}
    H = \sum_j H_j,
\end{equation}
where $H_j$ represents a Hermitian basis operator such as a Pauli string times a real coefficient. In the following, we consider several cases in sequence from easy to difficult.

In the simplest case, suppose that the Hamiltonian $H$ only contains one term $H_1$ which is a single-site operator, such as a Pauli-$Z$ operator. What is the local minimum distribution of the corresponding energy landscape? According to Theorem~\textcolor{darkblue1}{2}, the local minimum distribution is invariant under moves of the orthogonality center, so one can freely move the orthogonality center to investigate the local minimum distribution. A helpful choice is to move the orthogonality center exactly to the support site of the operator, denoted as site $i$, e.g., from the circuit
\begin{equation}
\begin{mytikz4}
\draw [line width=0.5]    (350,220.01) -- (350,30) ;
\draw [line width=0.5]    (310,220) -- (310,30) ;
\draw [line width=0.5]    (150,220.01) -- (150,30) ;
\draw [line width=0.5]    (230,220.03) -- (230,30) ;
\draw [line width=0.5]    (190,220.03) -- (190,30.03) ;
\draw [line width=0.5]    (390,220.03) -- (390,30.03) ;
\draw [line width=0.5]    (270,220.03) -- (270,30) ;
\draw  [fill={rgb, 255:red, 245; green, 240; blue, 235 }  ,fill opacity=1 ][line width=0.5]  (130,160.03) -- (210,160.03) -- (210,180.03) -- (130,180.03) -- cycle ;
\draw  [fill={rgb, 255:red, 255; green, 234; blue, 195 }  ,fill opacity=1 ][line width=0.5]  (170.3,190.03) -- (290.3,190.03) -- (290.3,210.03) -- (170.3,210.03) -- cycle ;
\draw  [fill={rgb, 255:red, 255; green, 234; blue, 195 }  ,fill opacity=1 ][line width=0.5]  (250,160.03) -- (330,160.03) -- (330,180.03) -- (250,180.03) -- cycle ;
\draw  [fill={rgb, 255:red, 255; green, 234; blue, 195 }  ,fill opacity=1 ][line width=0.5]  (330,100.03) -- (410,100.03) -- (410,120.03) -- (330,120.03) -- cycle ;
\draw [line width=0.5]    (110,220) -- (110,30) ;
\draw  [fill={rgb, 255:red, 245; green, 240; blue, 235 }  ,fill opacity=1 ][line width=0.5]  (90,130.03) -- (170,130.03) -- (170,150.03) -- (90,150.03) -- cycle ;
\draw  [fill={rgb, 255:red, 255; green, 234; blue, 195 }  ,fill opacity=1 ][line width=0.5]  (290,130.03) -- (370,130.03) -- (370,150.03) -- (290,150.03) -- cycle ;
\draw [line width=0.5]    (430,220.03) -- (430,30.03) ;
\draw  [fill={rgb, 255:red, 255; green, 234; blue, 195 }  ,fill opacity=1 ][line width=0.5]  (370,70.03) -- (450,70.03) -- (450,90.03) -- (370,90.03) -- cycle ;
\draw [line width=0.5]    (470,220.03) -- (470,30.03) ;
\draw  [fill={rgb, 255:red, 245; green, 240; blue, 235 }  ,fill opacity=1 ][line width=0.5]  (410,40.03) -- (490,40.03) -- (490,60.03) -- (410,60.03) -- cycle ;
\draw  [fill={rgb, 255:red, 245; green, 166; blue, 35 }  ,fill opacity=1 ][line width=0.5]  (382,14) -- (398,14) -- (398,30) -- (382,30) -- cycle ;

\end{mytikz4}\quad,
\end{equation}
to the circuit
\begin{equation}
\begin{mytikz4}
\draw [line width=0.5]    (230.02,249.96) -- (230.02,30) ;
\draw [line width=0.5]    (270.02,249.94) -- (270.02,30) ;
\draw [line width=0.5]    (430.04,249.96) -- (430.04,30) ;
\draw [line width=0.5]    (350.03,249.98) -- (350.03,30) ;
\draw [line width=0.5]    (390.03,249.98) -- (390.03,30.03) ;
\draw [line width=0.5]    (190.01,249.98) -- (190.01,30.03) ;
\draw [line width=0.5]    (310.02,249.98) -- (310.02,30) ;
\draw  [fill={rgb, 255:red, 245; green, 240; blue, 235 }  ,fill opacity=1 ][line width=0.5]  (370.04,190) -- (290.03,190) -- (290.03,210) -- (370.04,210) -- cycle ;
\draw  [fill={rgb, 255:red, 255; green, 234; blue, 195 }  ,fill opacity=1 ][line width=0.5]  (450,220) -- (329.99,220) -- (329.99,240) -- (450,240) -- cycle ;
\draw  [fill={rgb, 255:red, 245; green, 240; blue, 235 }  ,fill opacity=1 ][line width=0.5]  (330.03,160.03) -- (250.02,160.03) -- (250.02,180.03) -- (330.03,180.03) -- cycle ;
\draw  [fill={rgb, 255:red, 245; green, 240; blue, 235 }  ,fill opacity=1 ][line width=0.5]  (250.02,100.03) -- (170.01,100.03) -- (170.01,120.03) -- (250.02,120.03) -- cycle ;
\draw [line width=0.5]    (470.04,249.94) -- (470.04,30) ;
\draw  [fill={rgb, 255:red, 245; green, 240; blue, 235 }  ,fill opacity=1 ][line width=0.5]  (490,190.03) -- (409.99,190.03) -- (409.99,210.03) -- (490,210.03) -- cycle ;
\draw  [fill={rgb, 255:red, 245; green, 240; blue, 235 }  ,fill opacity=1 ][line width=0.5]  (290.02,130.03) -- (210.01,130.03) -- (210.01,150.03) -- (290.02,150.03) -- cycle ;
\draw [line width=0.5]    (150.01,249.99) -- (150.01,30.04) ;
\draw  [fill={rgb, 255:red, 245; green, 240; blue, 235 }  ,fill opacity=1 ][line width=0.5]  (210.01,70.03) -- (130,70.03) -- (130,90.03) -- (210.01,90.03) -- cycle ;
\draw [line width=0.5]    (110,249.99) -- (110,30.04) ;
\draw  [fill={rgb, 255:red, 245; green, 240; blue, 235 }  ,fill opacity=1 ][line width=0.5]  (170.01,40.03) -- (90,40.03) -- (90,60.03) -- (170.01,60.03) -- cycle ;
\draw  [fill={rgb, 255:red, 245; green, 166; blue, 35 }  ,fill opacity=1 ][line width=0.5]  (382,14) -- (398,14) -- (398,30) -- (382,30) -- cycle ;
\end{mytikz4}\quad,
\end{equation}
where the orange square marks the position of the observable $H_1$, and the light yellow rectangles indicates the blocks within the backward causal cone of $H_1$ Thus, all other unitaries outside the backward causal cone of $H_1$ in the sequential circuit are canceled out with their conjugates except for the one at the orthogonality center, and hence the problem is simplified to the optimization of a linear function over the Bloch ball of the qudit at site $i$, with all other parameters being redundant. This is essentially a convex optimization problem where all the local minima are global minima. Therefore, we know that even if we do not explicitly move the orthogonality center, the local minimum distribution keeps the same, i.e., all the local minima concentrate exactly at the ground state energy. Roughly speaking, this can be understood by the following picture: the original optimization problem is approximately simplified to an explicit convex problem by reparametrization (moving the orthogonality center). We point out that from the perspective of Hamiltonian itself, this problem is extremely simple because the single local term $H_1$ is easy to diagonalize; however, from the perspective of circuit optimization, this problem is still non-trivial, because, as a counterexample, brickwork circuits of moderate depth are expected to face challenges in efficiently finding this simple ground state due to the extensive local minima in their landscapes.

More generally, suppose the support of the single term is larger than one site but still local, e.g., $s$ contiguous sites. If we still move the orthogonality center into the support, only the unitaries acting on the support survive, and hence the reduced state on the small support is locally overparametrized in the sense that the number of independent parameters are more than the Hilbert space dimension of the subsystem within the support~\cite{Anschuetz2022}, given the MPS bond dimension is moderately large. For a rough estimation, the number of independent parameters within the causal cone is
\begin{equation}
    (2dD^2-1) + (s-1)(2dD^2-D^2) - (s-1)D^2 - 2D^2 = 2s(d-1)D^2 - 1.
\end{equation}
The Hilbert space dimension of the subsystem is $d^s$, with the number of real degrees of freedom $(2d^s-1)$. Thus, if the bond dimension is larger than the critical value
\begin{equation}\label{eq:D_c}
    D_c = \sqrt{\frac{d^s}{s(d-1)}},
\end{equation}
the reduced state is overparametrized and the local minimum distribution is expected to concentrate near the global minimum value. It is also worth noticing that for a global term with $s\in\mathcal{O}(N)$, according to Eq.\,\eqref{eq:D_c}, the bond dimension should scale exponentially with the system size to realize overparametrization and good trainability, which is consistent with previous results~\cite{Cerezo2021, Larocca2023}.

In the generic case, the Hamiltonian consists of a linear or polynomial number of local terms. The total energy function is a sum of the elementary energy functions of these local terms
\begin{equation}
    \mathcal{E}(A) = \sum_j \mathcal{E}_j(A) = \sum_j \braoprket{\Psi_\mathrm{MPS}(A)}{H_j}{\Psi_\mathrm{MPS}(A)}.
\end{equation}
The question is what the local minimum distribution of the total energy function is, given these elementary energy functions have local minimum distributions that concentrate to their own global minimum values. To inherit the benign landscape properties of the summed functions, we are supposed to impose additional conditions on the relation among the elementary energy functions. A sufficient condition is the ``common local minimum condition'': an arbitrary local minimum $A^*$ of $\mathcal{E}$ is also a local minimum of each $\mathcal{E}_j$. Then, if the local minima of $\mathcal{E}$ are close to its global minimum value within a small error $\epsilon_j$, the energy of each local minimum $A^*$ of $\mathcal{E}$ satisfies
\begin{equation}\label{eq:energy_epsilon}
    \mathcal{E}(A^*) = \sum_j \mathcal{E}_j(A^*) \leq \sum_j \left(\min_{A\in\mathcal{A}_\mathrm{MPS}} \mathcal{E}_j(A) + \epsilon_j\right) \leq \min_{A\in\mathcal{A}_\mathrm{MPS}} \mathcal{E}(A) + \sum_j\epsilon_j.
\end{equation}
That is to say, the local minima of $\mathcal{E}$ also concentrate near its global minimum value up to the additive error $\sum_j\epsilon_j$. Below, we analyze the requisite constraints imposed by the common local minimum condition on the Hamiltonian and ansatz.
\begin{itemize}
    \item Firstly, if $\epsilon_j=0$ for any $j$, the Hamiltonian must be frustration-free, because Eq.\,\eqref{eq:energy_epsilon} directly implies that the ground state of the full Hamiltonian is also the ground states of each individual term. 
    \item Secondly, if $\epsilon_j=0$ for any $j$, the ansatz must be expressive enough to represent the ground state, because Eq.\,\eqref{eq:energy_epsilon} implies the global minimum is actually the true ground state.
    \item Thirdly, if we require an arbitrary critical point of $\mathcal{E}$ is also a critical point of each $\mathcal{E}_j$, a sufficient condition is the following ``compatible gradient condition'': there exist a positive constant $c$ such that
    \begin{equation}\label{eq:compatible_gradient}
        \left\| \mathrm{grad}\,\mathcal{E}(A) \right\|_\mathcal{I}^2 \geq c \sum_j\left\| \mathrm{grad}\,\mathcal{E}_j(A) \right\|_\mathcal{I}^2,
    \end{equation}
    holds for any point in the parameter space, where the norm is taken with respect to the quantum Fisher information metric $\mathcal{I}$. Under the compatible gradient condition, the total gradient is zero only if all the gradients of the subterms are zero, i.e.,
    \begin{equation}
        \mathrm{grad}\,\mathcal{E}(A) = 0 \quad \Rightarrow \quad \mathrm{grad}\,\mathcal{E}_j(A) = 0~~\forall j.
    \end{equation}
    By substituting $\mathcal{E}(A)=\sum_j\mathcal{E}_j(A)$, Eq.\,\eqref{eq:compatible_gradient} can be rewritten as
    \begin{equation}\label{eq:compatible_gradient_cross}
        \sum_{j \neq k}\left\langle \mathrm{grad}\,\mathcal{E}_j(A), \mathrm{grad}\,\mathcal{E}_k(A) \right\rangle_\mathcal{I} \geq (c-1)\sum_j \left\| \mathrm{grad}\,\mathcal{E}_j(A) \right\|_\mathcal{I}^2,
    \end{equation}
    where the inner product is defined by
    \begin{equation}
        \left\langle \mathrm{grad}\,\mathcal{E}_j(A), \mathrm{grad}\,\mathcal{E}_k(A) \right\rangle = \left( \mathrm{grad}\,\mathcal{E}_j(A) \right)^\dagger \mathcal{I}(A) \mathrm{grad}\,\mathcal{E}_k(A).
    \end{equation}
    Eq.\,\eqref{eq:compatible_gradient_cross} means that the gradients of different subterms should not conflict with each other too strongly to be canceled out. Transforming to the Hilbert space, the inner product becomes
    \begin{equation}
        \left\langle \mathrm{grad}\,\mathcal{E}_j(A), \mathrm{grad}\,\mathcal{E}_k(A) \right\rangle_\mathcal{I} = \mathrm{Re}\langle (H_j-\avg{H_j}) P_\mathcal{T} (H_k-\avg{H_k}) \rangle,
    \end{equation}
    where $\langle\cdot\rangle$ represents the expectation with respect to $\ket{\Psi(A)}$ and $P_\mathcal{T}$ is the corresponding tangent space projector. Namely, Eq.\,\eqref{eq:compatible_gradient_cross} can be understood as there being no or only weak ``frustration'' among the imaginary-time evolution directions governed by these subterms within the tangent space. We examine Eq.\,\eqref{eq:compatible_gradient} statistically by numerical simulations on several physical Hamiltonians in Sec.\,\ref{sec:numerical_compatible_gradient}.
\end{itemize}

Therefore, we conclude that, given the individual component energy functions are free from poor local minima, ensuring the same for the total energy function imposes certain requirements on both the Hamiltonian and the ansatz. For instance, the Hamiltonian must exhibit relatively weak frustration, and the ansatz must possess sufficient expressibility to represent the ground states. Furthermore, we note that since finding the ground state of a 1D local Hamiltonian is $\mathsf{QMA}$-complete~\cite{Aharonov2009}, it is natural to expect that MPS are not universally free from the poor local minimum problem across arbitrary Hamiltonians. However, we point out that the trainability of MPS is, to a certain extent, even better than the above theoretical predictions, as demonstrated by our numerical experiments with random Hamiltonians.




\section{Additional numerical results and details}

In this section, we elaborate on the details of the numerical experiments in the main text and provide some additional numerical results to support our conclusions further. Throughout this paper, the numerical experiments are based on the open-source library TensorCircuit-NG~\cite{Zhang2022_z, Zhang2026} for circuit simulation with variational optimization.

\subsection{Training performance of different circuit architectures}

In this section, we introduce the technical details of the numerical experiments on the practical training performance and local minimum distributions of different circuit architectures.

\begin{figure}
    \centering
    \includegraphics[width=0.75\linewidth]{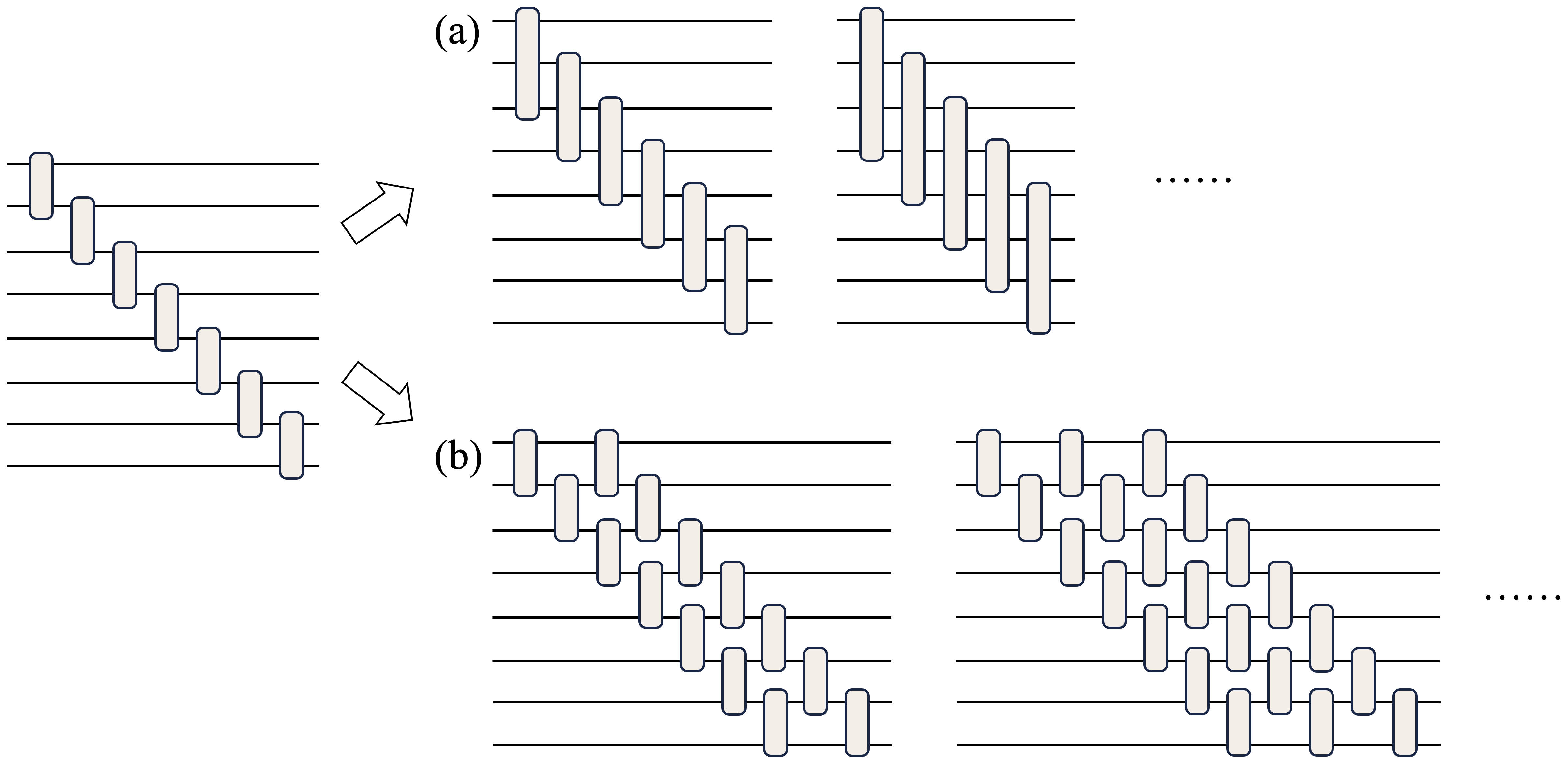}
    \caption{Sequential circuits and sloping brickwork circuits. To increase the expressibility of a minimal 1-dimensional sequential circuit of block size $\beta=2$ and layer depth $L=1$, one can either (a) increase the block size $\beta$ to obtain sequential circuits with larger bond dimension or (b) increase the layer depth $L$ to obtain the so-called sloping brickwork circuits, where the local structure is similar to that of the standard brickwork circuits but the blocks are arranged in a sloping manner globally. Note that each block in the diagram is universal on its support, i.e., it can represent any unitary matrix in the group $\mathcal{SU}(d^\beta)$, where $d$ is the local Hilbert space dimension and $\beta$ is the block size.}
    \label{fig:block_size_vs_depth}
\end{figure}

As mentioned in the main text, we mainly compare three circuit architectures: sequential circuits, brickwork circuits, and sloping brickwork circuits. In particular, the difference between sequential circuits and sloping brickwork circuits is depicted in Fig.\,\ref{fig:block_size_vs_depth}. We refer to a sequential circuit as a circuit where the blocks are arranged in a single-layer staircase pattern, distinguished from the more general class of finite local-depth circuits~\cite{Zhang2024}. Each block in these circuits is a universal unitary on its support. In our experiments, we build a universal unitary by use of the Cartan decomposition and quantum Shannon decomposition~\cite{Drury2008}.

To investigate the typical trainability of the circuits, we design a protocol to generate random Hamiltonians, termed random backward-evolved Hamiltonians, which is defined by
\begin{equation}
    H_{\text{rand}} = \mathbf{V}^\dagger H_{Z} \mathbf{V},
\end{equation}
where $H_{Z}=-\sum_j Z_j$ and $Z_j$ is the Pauli $Z$ operator at site $j$. $\mathbf{V}$ is a random circuit whose architecture is obtained by reversing the ansatz circuit $\mathbf{U}$ in the time direction, i.e., $\mathbf{V}=\mathbf{U}^\dagger$. The parameters in $V$ are independent of those in the ansatz circuit $\mathbf{U}$. This protocol ensures that the ground state of $H_{\text{rand}}$, i.e., $\mathbf{V}^\dagger \ket{\boldsymbol{0}}$, can always be captured by the ansatz. Namely, the expressibility of the ansatz is always sufficient to represent the ground state, and hence the performance (the accuracy to the ground state energy) of optimization is solely determined by the trainability. It is worth mentioning that if we directly use the fidelity between the variational state and a random state in the ansatz space as the cost function, we will encounter the barren plateau phenomenon as the fidelity is a so-called global cost function. The Hamiltonian defined above can be viewed as the local cost function version of the fidelity function~\cite{Cerezo2021}.

In our numerical experiments, the parameters are initialized by uniformly sampling each Pauli angle from $[0, 2\pi)$, and the optimizer is chosen as the Adam optimizer. We point out that although this initialization is not a strict Haar distribution over respective local unitary blocks, it is approximately equal, especially when the block size $\beta$ (the support size of the block) increases. We also point out that although the Adam optimizer used here is not a manifold-based optimizer like the Riemannian gradient descent, the resulting local minimum distribution is not supposed to deviate significantly from that defined in our theorems. This is because, to a large extent, different optimizers mainly affect the convergence speed and have a relatively small influence on the final converged result in the long-time limit, especially considering the fact that we are now focusing on the statistics of the convergence result starting from random initialization, rather than each individual case. We also remark that the $50$ samples at each system size share the same Hamiltonian, i.e., the random unitary $\mathbf{V}$ is sampled at first and then fixed for each system size.

The relative error plotted in Fig.\,\textcolor{darkblue1}{3} of the main text is defined by
\begin{equation}
    \varepsilon=\frac{E-E_0}{E_0},
\end{equation}
where $E$ is the energy of the variational state and $E_0$ is the exact ground state energy. The numerical experiments shown in Fig.\,\textcolor{darkblue1}{3} are conducted with learning rate $\eta=0.01$. Fig.\,\ref{fig:sequential_vs_sloping_lr0.001} shows the results with a smaller learning rate $\eta=0.001$, which is consistent with the data shown in the main text. Fig.\,\ref{fig:sequential_vs_sloping_beta4} further shows the results of sequential circuits with a larger block size $\beta=4$ and sloping brickwork circuits with a larger layer depth $L=16$, where again the optimization of the sequential circuits reaches the global minimum, while that of the sloping brickwork circuits is trapped by poor local minima as the system size increases.

It is worth noting that in the numerical experiments, the performance of sequential circuits is better than we theoretically expected in the sense that the random backward-evolved Hamiltonians may contain non-local terms, though with relatively small weights. This further indicates that even if the Hamiltonian contains certain non-local interactions that decay with distance, the energy landscapes of sequential circuits may remain largely free from poor local minima.




\begin{figure}
    \centering
    \includegraphics[width=0.8\linewidth]{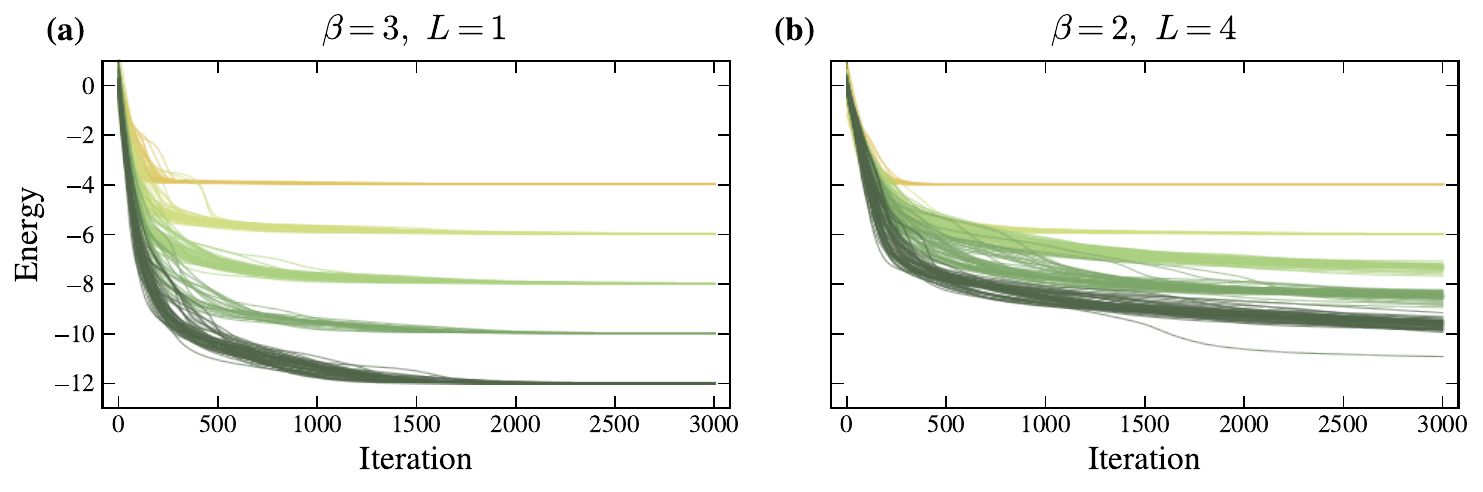}
    \caption{The numerical training curves of (a) sequential circuits of $\beta = 3, L = 1$, (b) sloping brickwork circuits of $\beta = 2, L = 4$. The block size $\beta$ and the number of layers $L$ are selected so that the number of parameters in different circuits are comparable. The loss function is the energy expectation of random backward-evolved Hamiltonians. The darkness of the color marks the system size from $4$ to $12$ qubits. The number of samples is $50$ for each system size. The learning rate is set to $0.001$.}
    \label{fig:sequential_vs_sloping_lr0.001}
\end{figure}

\begin{figure}
    \centering
    \includegraphics[width=0.8\linewidth]{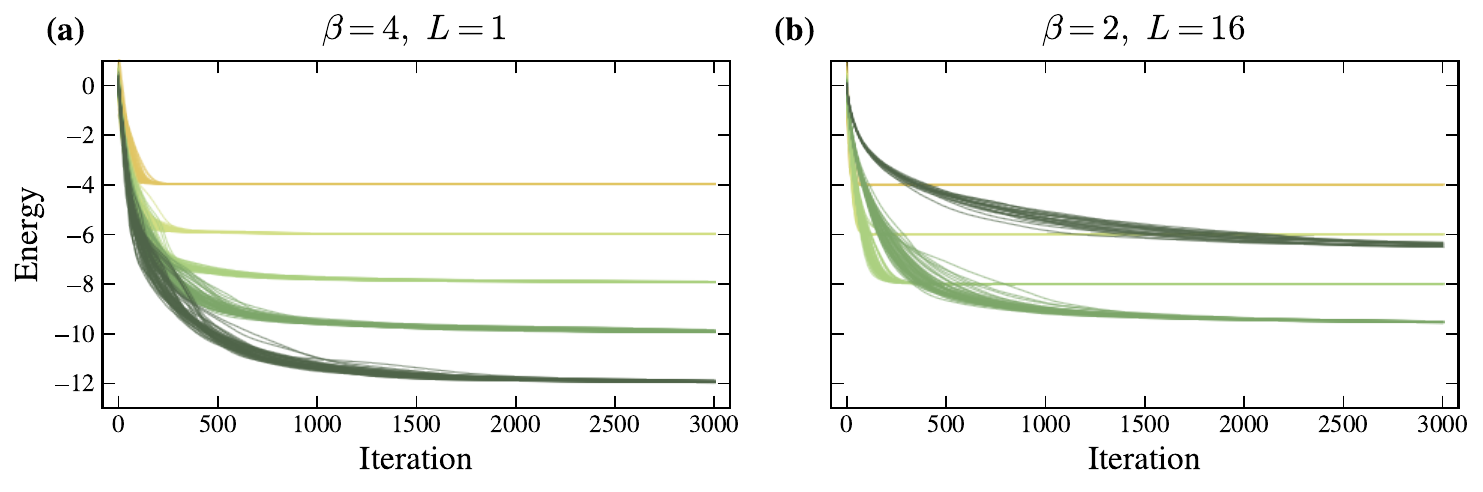}
    \caption{The numerical training curves of (a) sequential circuits of $\beta = 4, L = 1$, (b) sloping brickwork circuits of $\beta = 2, L = 16$. The block size $\beta$ and the number of layers $L$ are selected so that the number of parameters in different circuits are comparable. The loss function is the energy expectation of random backward-evolved Hamiltonians. The darkness of the color marks the system size from $4$ to $12$ qubits. The number of samples is $50$ for each system size. The learning rate is set to $0.001$.}
    \label{fig:sequential_vs_sloping_beta4}
\end{figure}

\subsection{Numerical test of the compatible gradient condition}\label{sec:numerical_compatible_gradient}

In addition to the direct optimization performance, we numerically examine the compatible gradient condition discussed above. For a Hamiltonian decomposition $H=\sum_j H_j$, we denote
\begin{equation}
    \mathcal{E}_j(A)=\braoprket{\Psi_\mathrm{MPS}(A)}{H_j}{\Psi_\mathrm{MPS}(A)},\quad
    g_j(A)=\mathrm{grad}\,\mathcal{E}_j(A),
\end{equation}
where the gradient is the natural gradient with respect to the quantum Fisher information metric $\mathcal{I}$. We quantify the aggregate cancellation among different local gradients by the ratio
\begin{equation}\label{eq:gradient_compatibility_ratio}
    R(A)=
    \frac{\left\|\sum_j g_j(A)\right\|_{\mathcal{I}}^2}
    {\sum_j\left\|g_j(A)\right\|_{\mathcal{I}}^2}.
\end{equation}
More explicitly, for two tangent vectors $u$ and $v$ in the tangent space at parameter point $A$, we use the convention
\begin{equation}
    \left\langle u,v \right\rangle_{\mathcal{I}} = u^\dagger\mathcal{I}(A)v,\quad
    \left\|u\right\|_{\mathcal{I}}^2=\left\langle u,u \right\rangle_{\mathcal{I}},
\end{equation}
and the natural gradient $g_j(A)$ is obtained from $g_j(A)=\mathcal{I}^{+}(A)\nabla\mathcal{E}_j(A)$. Thus, the numerator of Eq.\,\eqref{eq:gradient_compatibility_ratio} is
\begin{equation}
    \left\|\sum_j g_j(A)\right\|_{\mathcal{I}}^2=\left\langle \sum_j g_j(A),\sum_k g_k(A)\right\rangle_{\mathcal{I}}=\sum_j\left\|g_j(A)\right\|_{\mathcal{I}}^2+2\sum_{j<k}\left\langle g_j(A),g_k(A)\right\rangle_{\mathcal{I}}.
\end{equation}
Equivalently,
\begin{equation}
    R(A)=1+
    \frac{2\sum_{j<k}\left\langle g_j(A),g_k(A)\right\rangle_{\mathcal{I}}}
    {\sum_j\left\|g_j(A)\right\|_{\mathcal{I}}^2}.
\end{equation}
Here $R(A)\approx 1$ means that the cross terms among different gradients cancel only weakly in aggregate, while $R(A)\approx 0$ would indicate a nearly complete cancellation of the total gradient. Thus, a pointwise lower bound $R(A)\geq \alpha>0$ is the numerical form of the compatible gradient condition.

We test this ratio for three open-chain spin Hamiltonian settings. The first two belong to the XXZ family
\begin{equation}
    H_\mathrm{XXZ}=\sum_{i=1}^{N-1}\left(X_iX_{i+1}+Y_iY_{i+1}+\Delta Z_iZ_{i+1}\right),
\end{equation}
where each nearest-neighbor bond is treated as one term. We take $\Delta=0.5$ in the first row of Fig.\,\ref{fig:gradient_compatibility_ratio} and $\Delta=1$ in the second row, the latter corresponding to the isotropic Heisenberg Hamiltonian $H_\mathrm{Heis}=H_\mathrm{XXZ}|_{\Delta=1}$. The third setting is the nearest-neighbor ferromagnetic and next-nearest-neighbor antiferromagnetic Heisenberg Hamiltonian
\begin{equation}
    H_{J_1\text{-}J_2}=J_1\sum_{i=1}^{N-1}\left(X_iX_{i+1}+Y_iY_{i+1}+Z_iZ_{i+1}\right)+J_2\sum_{i=1}^{N-2}\left(X_iX_{i+2}+Y_iY_{i+2}+Z_iZ_{i+2}\right),
\end{equation}
with $J_1=-1$ and $J_2=0.5$, where each nearest-neighbor and next-nearest-neighbor bond is treated as one term.

The numerical results are shown in Fig.\,\ref{fig:gradient_compatibility_ratio}. In this test, the $\beta=2$ column uses the full parametrization of two-qubit unitary gate, while the $\beta=3$ and $\beta=4$ columns use Haar-random blocks of size $\beta$, i.e., independent Haar-random $\mathcal{U}(8)$ and $\mathcal{U}(16)$ gates, respectively. For all sampled parameter points, the ratio $R(A)$ is positive and stays clearly away from zero, which means that the sampled sequential circuits satisfy the compatible gradient condition pointwise. Moreover, as the system size increases, the distributions become more concentrated around $R=1$, indicating that aggregate gradient cancellation becomes less significant for larger systems in these examples.

\begin{figure}
    \centering
    \includegraphics[width=0.99\linewidth]{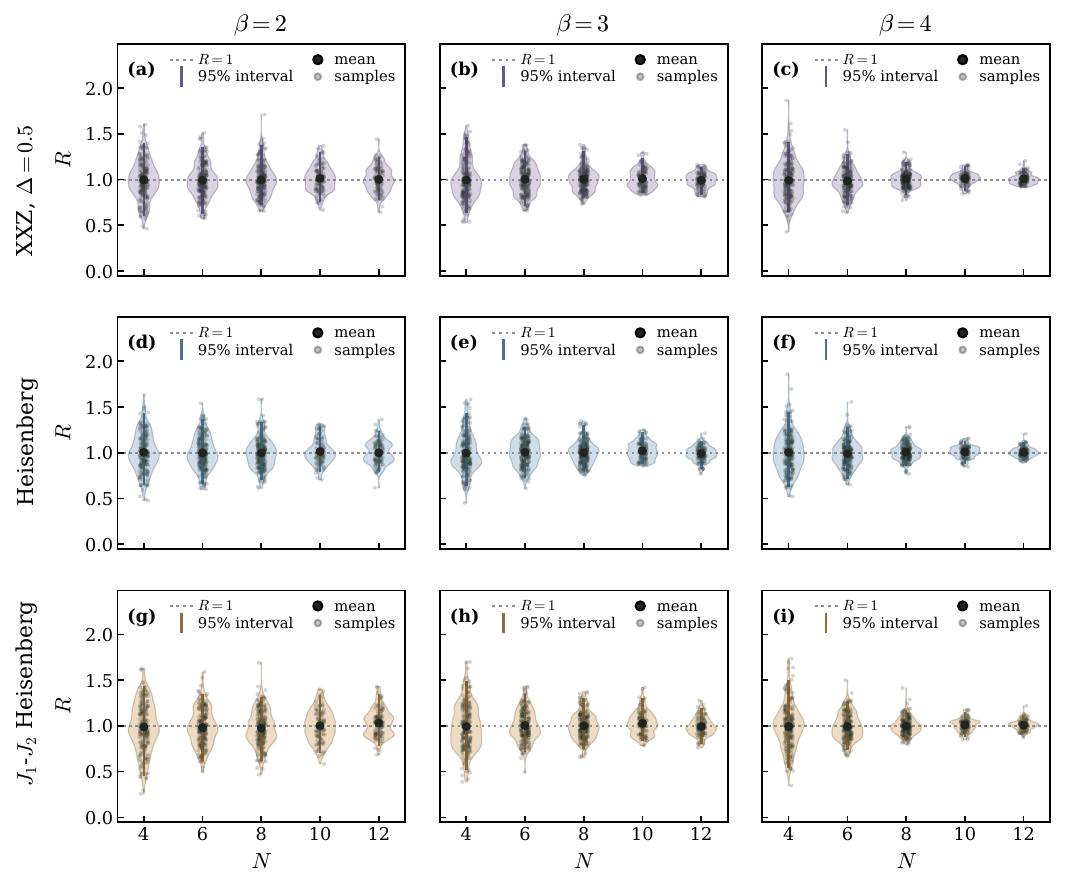}
    \caption{Numerical test of the compatible gradient condition for sequential circuits. The plotted quantity is the ratio $R(A)$ defined in Eq.\,\eqref{eq:gradient_compatibility_ratio}. The three columns correspond to the block sizes $\beta=2,3,4$, respectively; the $\beta=3$ and $\beta=4$ columns use Haar-random $\mathcal{U}(8)$ and $\mathcal{U}(16)$ blocks. The three rows correspond to the XXZ Hamiltonian $H_\mathrm{XXZ}$ with $\Delta=0.5$, the same XXZ Hamiltonian at the Heisenberg point $\Delta=1$, and the $J_1$-$J_2$ Heisenberg Hamiltonian $H_{J_1\text{-}J_2}$, respectively, all with open boundary conditions. The system sizes are $N=4,6,8,10,12$. For each pair of $(\beta,N)$, we use $200$ random samples for $N=4,6,8$ and $100$ random samples for $N=10,12$. The shaded violins show the empirical distribution, gray points are individual samples, black points denote the sample means, and vertical bars indicate the central $95\%$ interval. In all sampled cases, $R(A)$ is bounded away from zero and becomes increasingly concentrated around $R=1$ as the system size grows.}
    \label{fig:gradient_compatibility_ratio}
\end{figure}

Finally, we examine a simple example where the value of $R(A)$ depends sensitively on how the Hamiltonian is decomposed into subterms. Consider
\begin{equation}
    H_{XZ}=\sum_{i=1}^{N}(X_i+Z_i).
\end{equation}
There are two natural decompositions: one may regard $X_i$ and $Z_i$ as separate subterms, or regard $X_i+Z_i$ on the same site as a single subterm. For the split-term decomposition, the exact product ground state of $H_{XZ}$ provides a special cancellation point: the total gradient vanishes, while the individual gradients from $X_i$ and $Z_i$ need not vanish separately and can cancel each other within the tangent space. This gives a decomposition-dependent point with $R(A)=0$. However, as shown in Fig.\,\ref{fig:xz_gradient_compatibility_ratio}, this special point is not observed in typical random samples; for generic sampled parameters, $R(A)$ remains clearly positive and is concentrated near $R=1$, similarly to the Hamiltonians studied above. Moreover, if $X_i+Z_i$ is treated as a single site term, this artificial cancellation disappears: each subterm gradient vanishes separately at the exact ground state.

\begin{figure}
    \centering
    \includegraphics[width=0.99\linewidth]{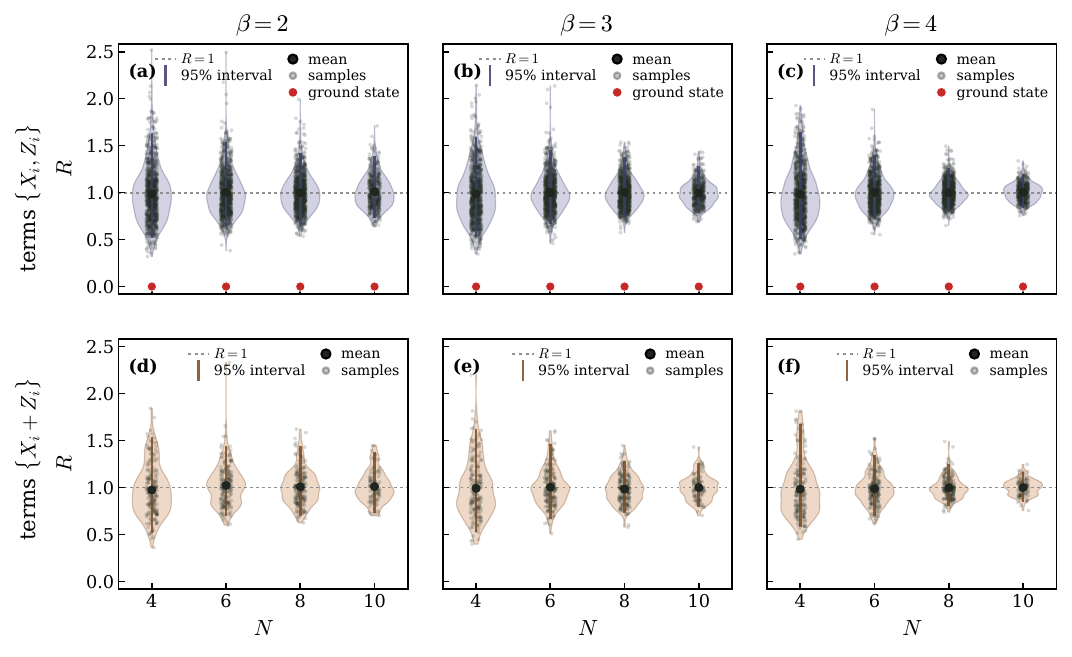}
    \caption{Dependence of the ratio $R(A)$ on the term decomposition for the Hamiltonian $H_{XZ}=\sum_i(X_i+Z_i)$ using sequential circuits. The first row uses the split-term decomposition $\{X_i,Z_i\}$, while the second row uses the site-term decomposition $\{X_i+Z_i\}$. The three columns correspond to $\beta=2,3,4$, respectively; the $\beta=3$ and $\beta=4$ columns use Haar-random $\mathcal{U}(8)$ and $\mathcal{U}(16)$ blocks. For the split-term decomposition, we use $1000$ random samples for $N=4,6,8$ and $500$ random samples for $N=10$; for the site-term decomposition, we use $200$ random samples for $N=4,6,8$ and $100$ random samples for $N=10$. The shaded violins show empirical distributions, gray points are individual samples, black points are sample means, and vertical bars indicate central $95\%$ intervals. The red markers in the first row are not sampled points; they mark the exact ground-state configurations where the split subterm gradients cancel and give $R=0$. Away from these special points, the sampled values remain bounded away from zero and concentrate around $R=1$.}
    \label{fig:xz_gradient_compatibility_ratio}
\end{figure}

\end{document}